\documentclass[11pt,a4paper,twoside,openright, titlepage]{report}
\usepackage[hidelinks]{hyperref}    
\usepackage[margin=30mm]{geometry}  

\usepackage{enumitem}
\usepackage{color}

\usepackage{mathtools}
\usepackage{bbm}	
\usepackage{stmaryrd}	
\usepackage{amsthm}		
\usepackage{amssymb}		
\usepackage{mathpartir}		
\usepackage{listings}		
\usepackage{tikz-cd}		
\usepackage{float}			
\usepackage{caption}		
\usepackage[ddmmyyyy]{datetime}		

\usepackage{appendix}	

\usepackage{proglang}	

\usetikzlibrary{positioning}	

\floatstyle{boxed} 	
\restylefloat{figure}
\newtheorem{theorem}{Theorem}[section]
\newreptheorem{theorem}{Theorem}

\newtheorem{lemma}[theorem]{Lemma}
\newreptheorem{lemma}{Lemma}

\newtheorem{proposition}[theorem]{Proposition}
\newreptheorem{proposition}{Proposition}

\newtheorem{definition}[theorem]{Definition}
\newreptheorem{definition}{Definition}

\newtheorem{corollary}[theorem]{Corollary}
\newreptheorem{corollary}{Corollary}

\theoremstyle{definition}
\newtheorem{example}[theorem]{Example}

\numberwithin{equation}{section}

\begin{document}

\bibliographystyle{alpha}


\pagestyle{empty}

%
%
%
%

\null\vfill
  \begin{center}
    { \Huge {\bfseries {Program Equivalence for \\[1ex] Algebraic Effects via Modalities}} \par}
{\large \vspace*{40mm} 
	{\font\beltcrestfont=oxbeltcrest {\beltcrestfont \char0} \par}
	\vspace*{25mm}}
    {{\Large Cristina Matache} \par}
{\large \vspace*{1ex}
    {{St Cross College} \par}
\vspace*{1ex}
    {University of Oxford \par}
\vspace*{25mm}
    {{A thesis submitted for the degree of} \par}
\vspace*{1ex}
    {\it {MSc in Computer Science} \par}
\vspace*{2ex}
    {Trinity 2018}}
  \end{center}
  \null\vfill

%


\newpage
\thispagestyle{empty}
\mbox{}

\pagestyle{plain}

\newpage
\vspace*{3cm}
{\centering
\section*{Abstract}
}

This dissertation is concerned with the study of program equivalence and algebraic effects as they arise in the theory of programming languages. Algebraic effects represent impure behaviour in a functional programming language, such as input and output, exceptions, nondeterminism etc.~all treated in a generic way. Program equivalence aims to identify which programs can be considered equal in some sense. This question has been studied for a long time but has only recently been extended to languages with algebraic effects, which are a newer development. Much work remains to be done in order to understand program equivalence in the presence of algebraic effects. In particular, there is no characterisation of contextual equivalence using a logic.

We define a logic whose formulas express properties of higher-order programs with algebraic effects. We then investigate three notions of program equivalence for algebraic effects: \emph{logical equivalence} induced by the aforementioned logic, \emph{applicative bisimilarity} and \emph{contextual equivalence}. For the programming language used in this dissertation, we prove that they all coincide.

Therefore, the main novel contribution of the dissertation is defining the first logic for algebraic effects whose induced program equivalence coincides with contextual equivalence.
\vspace*{\fill}

\newpage
\vspace*{3cm}
{\centering
\section*{Acknowledgements}
}

I would like to thank my supervisor, Sam Staton, for his guidance and patience without which this project would not have been possible. I am also grateful to Alex Simpson and Niels Voorneveld for insightful discussions about their work. Finally, I want to thank my parents for their support and for making this year in Oxford possible.
\vspace*{\fill}

\tableofcontents
\listoffigures


\pagestyle{headings}

\chapter{Introduction}

This dissertation is a theoretical study of program equivalence for a higher-order language with algebraic effects. In particular, we are interested in finding a logic of program properties that characterises contextual equivalence. This chapter reviews the history of program equivalence and explains why it is a hard but nevertheless interesting problem. At the same time, we outline the main research questions that led to the present work. Finally, we summarise the contributions of our work and the structure of the dissertation.

\section{Motivation}

Although undecidable in general, program equivalence is a fundamental problem both in the theory of programming languages and in formal verification. Programming language researchers are concerned with  theoretical ways of reasoning about program equivalence based on formal semantics, like denotational or operational semantics. We adopt this point of view in the dissertation. Verification combines theory with building automated tools for checking program equivalence whenever possible.

Program equivalence seeks to establish when two programs are interchangeable, or when they behave the same, for some definition of behaviour. For the semantics of programming languages, this is useful because one can regard the meaning of a program as the equivalence class of programs that it belongs to. From a more practical perspective, checking program equivalence can determine whether a program implements a specification, which is itself given as a more familiar program. For example, establishing whether certain compiler optimisations are safe can be done in this way.

The first definition of program equivalence that springs to mind is that two programs are equivalent when they return the same result. 
In the case of a Turing-complete language, one would also like to take into account the possibility of divergence. Therefore, naively two programs are equivalent if they either return the same result or they both diverge.

This definition has two problems. Firstly, establishing divergence of programs is undecidable due to the halting problem, so program equivalence is in general hard to calculate. Working around this issue is mainly the focus of the formal verification community. Secondly, the above naive definition of program equivalence becomes unclear in the presence of higher-order functions.
A higher-order function can receive as arguments and return as results other functions, rather than just ground data, such as natural numbers or booleans. As an example, consider the following programs:
\begin{align*}
one&=\lambda f.\lambda y.(f\ y)	\\
two&=\lambda f.\lambda y.f\ (f\ y)	\\
two'&=\lambda f.\lambda y.f\ (one\ f\ y)
\end{align*}
where $one$ and $two$ represent the first two Church numerals. Morally, the functions $two$ and $two'$ should be equivalent in the $\lambda$-calculus, although they are not syntactically equal. Therefore the problem is: what does it mean for two higher-order functions to be the same?

Over the years, the programming languages community has proposed many answers to this question. At the beginning, they focussed  on pure higher-order languages, with no side effects, such as the $\lambda$-calculus and PCF \cite{DBLP:journals/tcs/Plotkin77}. PCF is a simply-typed extension of the $\lambda$-calculus with general recursion and a datatype of natural numbers and is therefore Turing-complete. These languages were chosen for their simplicity and  because they  have had well-understood formal semantics for a long time. 

In this context, program equivalence based on denotational semantics was proposed: if two programs have the same denotation, they are equivalent. An alternative notion based on operational semantics is Morris-style contextual equivalence \cite{Mor69}: two programs are equivalent if they have the same \emph{observable} behaviour in all program contexts.
In this dissertation, we are not concerned with denotational equality; contextual equivalence will be discussed more in Section \ref{Sec_prog_equiv_intro}. Notice that it improves on the naive definition of program equivalence because, to test the equality of functions, a context can provide them with arguments and observe their reduction behaviour. 

To address some of the shortcomings of denotational equality and contextual equivalence, other notions of program equivalence that use operational semantics have been proposed. Examples include logical equivalence, that is, equivalence induced by satisfaction of formulas in a logic of program properties (e.g.~\cite{DBLP:journals/jacm/HennessyM85}), applicative bisimilarity \cite{Abr90} and logical relations \cite{DBLP:journals/jsyml/Tait67}. Logical equivalence is the main focus of this dissertation. Notably, it has strong connections with formal verification, where modal logics such as temporal logics \cite{Pnu77} and the modal $\mu$-calculus \cite{DBLP:journals/tcs/Kozen83} are used to specify and verify properties of programs. Applicative bisimilarity will be discussed in more detail in Section~\ref{Sec_prog_equiv_intro}, but we will not be concerned with logical relations for program equivalence.

Given the wide variety of program equivalences mentioned so far, a natural question is comparing them: do they coincide, is one included in the other or are they altogether different? This question has been studied as new definitions arose. For example, Plotkin showed that for PCF denotational equality implies contextual equivalence.

Gradually, the questions identified so far have been posed for more complex languages than the $\lambda$-calculus or PCF. An interesting addition are impure operations such as input and output, exceptions, nondeterminism, state or continuations, whose semantics has been studied more recently. They are known as \emph{computational effects}. In his influential work, Eugenio Moggi \cite{DBLP:journals/iandc/Moggi91} gave a general denotational semantics for computational effects using monads. 

Subsequently, Plotkin and Power \cite{DBLP:conf/fossacs/PlotkinP01, DBLP:conf/fossacs/PlotkinP02, DBLP:journals/acs/PlotkinP03} initiated a program of research concerned with \emph{algebraic effects}, a subset of computational effects whose behaviour can be axiomatised by a set of equations. All the example of effects mentioned before are algebraic with the exception of continuations and exception handling.
The purpose of algebraic effects is to give a unified, uniform treatment of the semantics of effects and of their combinations \cite{DBLP:journals/tcs/HylandPP06}. The latter cannot easily be achieved using Moggi's semantics.  

In this context, a new question becomes apparent: when are two higher-order programs \emph{exhibiting algebraic effects} equivalent? The natural starting point is to extend the existing notions of program equivalence to languages with algebraic effects, and once again compare the resulting relations. A lot of work has been done in this direction for specific effects, such as nondeterminism (e.g.\ \cite{LasPhd}) and probabilistic choice (e.g.\ \cite{DBLP:conf/esop/CrubilleL14}). Ideally however, a notion of program equivalence should be applicable to \emph{any} algebraic effect.

Several such notions of program equivalence for generic algebraic effects have been developed. Johann, Simpson and Voigtl{\"a}nder \cite{DBLP:conf/lics/JohannSV10} study contextual equivalence and a corresponding logical relation. Dal Lago, Gavazzo and Levy \cite{DBLP:conf/lics/LagoGL17} are concerned with applicative bisimilarity. Plotkin and Pretnar \cite{DBLP:conf/lics/PlotkinP08} propose a logic for algebraic effects that is sound with respect to other notions of program equivalence, but not complete in general. Simpson and Voorneveld \cite{SimV18} propose a modal logic whose induced program equivalence coincides with applicative bisimilarity, but not with contextual equivalence.

Thus, the main question of the dissertation arises: can we find a logic that \emph{characterises} contextual equivalence for a higher-order language with generic algebraic effects?

\section{Contributions}

The programming language we consider in this dissertation is named ECPS. It is a call-by-value continuation-passing variant of PCF with generic algebraic effects. It is a higher-order Turing-complete language. 

Being a continuation-passing language means that functions receive an additional argument which specifies how the computation should proceed once the function terminates. This argument is called a \emph{continuation}. Compared to the direct style of programming, continuation-passing style makes control flow and the order of evaluation explicit. Thanks to these features, continuation-passing languages are often used as intermediate languages inside compilers. So ECPS could be seen as a simple variant of an intermediate language.

In the dissertation, we are concerned with three notions of program equivalence for ECPS: logical equivalence, contextual equivalence and applicative bisimilarity, and with the relationship between them. Most importantly, we wish to define a logic whose induced program equivalence coincides with contextual equivalence.

Any notion of program equivalence needs to satisfy two key properties: being an  equivalence relation and being compatible. Compatibility means that equivalent programs can be substituted for a variable in a program equation, thus allowing compositional reasoning about program equivalence. These requirements are both explained in more detail in Section \ref{Sec_prog_equiv_intro}.

The novel contributions of the dissertation are the following:
\begin{enumerate}
\item We study the relationship between ECPS and the language used in the work of Simpson and Voorneveld \cite{SimV18}. We provide a correct embedding of their language into ECPS (Theorem \ref{Thm_CPS_correct}).

\item We develop applicative bisimilarity for ECPS and prove it is a compatible equivalence relation (Lemma \ref{Lemm_(bi)sim_preord} and Theorem \ref{Thm_sim_compat}).

\item We define a logic whose formulas express properties of ECPS programs. We prove that  program equivalence induced by the logic coincides with applicative bisimilarity (Theorem \ref{Cor_bisim_is_log_equiv}). Therefore, logical equivalence is compatible.

\item We present two equivalent definitions of contextual equivalence for ECPS, which are both equivalence relations and compatible. We prove that contextual equivalence coincides with applicative bisimilarity (Theorem \ref{Thm_bisim_is_ctx_equiv}). This leads to the main result of the dissertation: logical equivalence coincides with contextual equivalence (Corollary \ref{Cor_logeq_is_ctxeq}).
\end{enumerate}

\section{Structure of the Dissertation}

Chapter~\ref{Chap_expl_prob} starts with an informal introduction to program equivalence and algebraic effects. It then reviews the work of Simpson and Voorneveld \cite{SimV18} since it is the most closely related to this dissertation.  
Finally, it exemplifies the distinction between contextual equivalence and logical equivalence in the context of their development.

Chapter~\ref{Chap_descr_method} introduces the language ECPS and its operational semantics. It then reviews two proof techniques, namely coinduction and logical relations.

Chapters~\ref{Chap_rel-epcf-ecps} to~\ref{Chap_ctx_equiv} contain the novel technical content. Chapter~\ref{Chap_rel-epcf-ecps} is concerned with justifying the use of the ECPS language. It gives a translation from the programming language used by Simpson and Voorneveld into ECPS and proves the translation correct.

The following three chapters study program equivalence. Chapter~\ref{Chap_bisim} defines applicative bisimilarity and proves its main properties. Chapter~\ref{Chap_mod_logic} introduces a logic for ECPS and proves that logical equivalence coincides with applicative bisimilarity. Finally, Chapter~\ref{Chap_ctx_equiv} develops contextual equivalence and proves it coincides with logical equivalence.

The last chapter reviews the material in the dissertation and surveys previous work. It then sketches several directions for future work and closes with some personal remarks.

Chapters~\ref{Chap_expl_prob} through \ref{Chap_ctx_equiv} all end with an accessible summary of their most important points. For an overview of the dissertation, one can consult the ``Chapter Summaries''.
Almost all the mathematical proofs completed as part of the project appear in the dissertation. To facilitate reading, routine or overlong proofs appear in the appendices rather than in the main body of the text. Therefore, Chapters~\ref{Chap_rel-epcf-ecps} to~\ref{Chap_ctx_equiv} all have a corresponding appendix.

\chapter{Background and Explanation of Problem}\label{Chap_expl_prob}

This chapter starts with an informal discussion of program equivalence and algebraic effects. It then reviews in some detail the work of Simpson and Voorneveld \cite{SimV18} on program equivalence for algebraic effects. The syntax and operational semantics of a programming language with algebraic effects named EPCF is introduced. Two forms of program equivalence are discussed: logical equivalence and applicative bisimilarity, and they are compared to contextual equivalence.

\section{Program Equivalence} \label{Sec_prog_equiv_intro}

As Pitts observes \cite{Pit11}, for a notion of program equivalence, or equality, to be useful it should be a congruence, that is, satisfy two key properties: being an \emph{equivalence relation} and being \emph{compatible}. The former allows reasoning through a chain of equations in order to establish that two programs $P_1$ and $P_n$ are equal:
\begin{equation*}
P_1 \simeq P_2 \simeq \ldots \simeq P_n.
\end{equation*}
Assuming that programs can take parameters, $P(x)$, compatibility means that we can substitute equivalent programs for a parameter in an equation:
\begin{equation*}
P_1(x) \simeq P_2(x) \text{ and } Q_1 \simeq Q_2 \implies P_1(Q_1) \simeq P_2(Q_2).
\end{equation*}
This property is important because it allows us to reason compositionally about programs. In order to decide whether two programs are equivalent, it suffices to investigate whether its subphrases are equivalent. When the parameter is allowed to be a function or a process, rather than just ground data, like booleans or integers, compatibility becomes even more important but also harder to establish.

One of the most intuitive notions of program equivalence is Morris-style \emph{contextual equivalence} \cite{Mor69}. Two programs are contextually equivalent if and only if they have the same observable behaviour in all program contexts:
\begin{equation*}
P_1 \simeq P_2 \quad \Longleftrightarrow \quad \forall C.\ \mathfrak{O}(C[P_1])=\mathfrak{O}(C[P_2]).
\end{equation*}
In the case of the untyped $\lambda$-calculus, a possible definition for the ``observable behaviour'' of program $P$, $\mathfrak{O}(P)$, is whether or not $P$ terminates, written $P{\Downarrow}$. Therefore $\mathfrak{O}(C[P_1])=\mathfrak{O}(C[P_2])$ becomes:
\begin{equation*}
C[P_1]{\Downarrow}\ \Longleftrightarrow\ C[P_2]{\Downarrow}.
\end{equation*}

However, contextual equivalence is difficult to establish for particular programs because of the quantification over all contexts. Therefore, Abramsky \cite{Abr90} proposed another notion of equivalence for the untyped $\lambda$-calculus named \emph{applicative bisimilarity}.

Bisimilarity was first defined by Milner \cite{Mil80} for the process calculus CCS which models concurrency. The main idea is that two processes are bisimilar if whenever one of them can advance by one step, the other can perform a matching step, and the two resulting processes are again bisimilar. The circularity of this definition of bisimilarity suggests that it can be defined coinductively as the greatest relation with a certain property.

In the case of the $\lambda$-calculus, the steps that programs can take are $\beta$-reduction steps. Therefore, applicative similarity, the one-sided version of bisimilarity, is defined as:
\begin{flushleft}
The greatest relation $\precsim$, such that $P_1\precsim P_2$ implies
\begin{multline*}
P_1 \longrightarrow^* \lambda x.P'_1 \implies \\
\exists P'_2 \text{ such that } P_2\longrightarrow^* \lambda x.P'_2 \text{ and for any value } v \text{, } P'_1[v/x]\precsim P'_2[v/x].
\end{multline*}
\end{flushleft}
Applicative bisimilarity is defined analogously.

For the $\lambda$-calculus, applicative bisimilarity coincides with contextual equivalence e.g. \cite{Abr90, Pit11}. Therefore, checking whether two $\lambda$-terms are bisimilar is a sound and complete proof technique for establishing contextual equivalence, easier to use in practice.

Another approach to program equivalence is to define a logic $\mathcal{L}$ whose formulas $\phi$ represent program properties. In this setting, two programs are equivalent if and only if they satisfy the same formulas in this logic:
\begin{equation*}
P_1 \simeq P_2 \quad \Longleftrightarrow \quad (\forall \phi \text{ in }\mathcal{L}.\ P_1\models \phi \Longleftrightarrow P_2\models \phi).
\end{equation*}
An example of a logic that describes program properties is Hennessy-Milner logic \cite{DBLP:journals/jacm/HennessyM85}, which concerns CCS processes. In fact, bisimilarity for CCS coincides with the equivalence induced by Hennessy-Milner logic.

In this dissertation, we are most interested in logical equivalence, so we will consider some informal examples of formulas.

\begin{example}
Consider a call-by-value simply-typed $\lambda$-calculus with natural numbers. A logical formula could be:
\begin{equation*}
\phi = \{3\} \mapsto \{2\}.
\end{equation*}
A function $f$ satisfies $\phi$ if, given argument $3$, the expression $(f\ x)$ reduces to $2$. Consider for example the following function:
\begin{equation*}
f = \lambda n.\ \mathbf{if}\ n=3\ \mathbf{then}\ \mathbf{pred}\ n\ \mathbf{else}\ \mathbf{succ}\ n.
\end{equation*} 
We can see that $f$ indeed satisfies $\phi$, but $f$ does not satisfy $\{4\}\mapsto\{3\}$.
\end{example}

\section{Algebraic Effects and Logical Properties}\label{Sec_alg_effects}

In general, programming languages that are purely functional do not provide ``impure'' operations such as input and output, nondeterministic or probabilistic choice, global state etc. These features are known as algebraic effects. To include them in a programming language, it suffices to add relevant operations to the language. 

As an example consider nondeterministic choice. This can be implemented by adding an operation $or(-,-)$, where, in the term $or(t_1,t_2)$, an external agent chooses nondeterministically whether to execute term $t_1$ or $t_2$. 

\begin{example}\label{Eg_boxdi}
Consider the simply-typed $\lambda$-calculus from the previous example extended with $or$. Now there is more than one value that a term may reduce to. For example:
\begin{equation*}
(g\ 3) \quad \text{where\ } g= \lambda n.\ or(\mathbf{pred}\ n,\ \mathbf{succ}\ n)
\end{equation*}
may reduce to either $2$ or $4$. The reduction behaviour of $(g\ 3)$ could be represented as a tree:
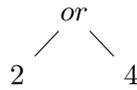
\begin{center}
\begin{tikzpicture}[level distance=0.8cm]
\node {$or$} 
	child { node {$2$} }
	child { node {$4$} };
\end{tikzpicture}
\end{center}

Recall the formula $\{3\}\mapsto \{2\}$ from the previous section. We can interpret it either as: ``the function \emph{always} returns $2$'' or ``\emph{may} return $2$''. As a result we have two new formulas: $\phi_1=\{3\}\mapsto \Box\{2\}$ and $\phi_2=\{3\}\mapsto \Diamond\{2\}$. We can see that $g$ satisfies the latter but not the former.
\end{example}

\begin{example}\label{Eg_dihigher}
As another example, consider a higher-order function:
\begin{equation*}
h = \lambda f.\ \lambda n.\ or(f\ (\mathbf{pred}\ n),\ f\ (\mathbf{succ}\ n)) : (\mathbbm{N}\rightarrow\mathbbm{N})\rightarrow (\mathbbm{N}\rightarrow\mathbbm{N}).
\end{equation*}
If we apply it to arguments $g$ and $2$, and then $g$ and $4$, the trees of $(h\ g\ 2)$ and $(h\ g\ 4)$ respectively are:
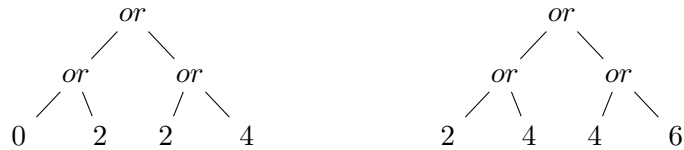
\begin{center}
\begin{tikzpicture}[level distance=0.8cm]
\node (n) {$or$}
	child { node {$or$} 
		child { node {$0$} }
		child { node[left=0.2cm] {$2$} }}
	child { node {$or$} 
		child { node[right=0.2cm] {$2$} }
		child { node {$4$} }};
\node[right=5cm of n] {$or$}
	child { node {$or$} 
		child { node {$2$} }
		child { node[left=0.2cm] {$4$} }}
	child { node {$or$} 
		child { node[right=0.2cm] {$4$} }
		child { node {$6$} }};			
\end{tikzpicture}
\end{center}
Now consider the formula:
\begin{equation*}
\psi = (\{3\}\mapsto\Diamond\{2\})\mapsto (\{2,4\}\mapsto\Diamond\{2\}).
\end{equation*}
It says that, given a function $f$ that satisfies $\{3\}\mapsto\Diamond\{2\}$, $h$ returns another function which when given either $2$ or $4$ as argument \emph{may} return $2$. Function $h$ satisfies $\psi$ because in either case it may call $(f\ 3)$, which we know may return $2$.
\end{example}

\section{Program Equivalence for Algebraic Effects}

The notions of program equivalence presented in  Section~\ref{Sec_prog_equiv_intro} have been extended recently to programming languages with algebraic effects. Johann, Simpson and Voigtl{\"a}nder \cite{DBLP:conf/lics/JohannSV10} study contextual equivalence for a polymorphic language with recursion and generic effects. They characterise contextual equivalence using a logical relation and thus prove some of its fundamental properties.

Dal Lago, Gavazzo and Levy \cite{DBLP:conf/lics/LagoGL17} give an abstract account of applicative bisimilarity for an untyped $\lambda$-calculus with generic algebraic effects. They show that applicative bisimilarity is included in contextual equivalence, but they note that this inclusion is strict in general.

Simpson and Voorneveld \cite{SimV18} consider a simply-typed programming language with recursion and generic algebraic effects. They propose a modal logic in which formulas expressing program behaviour are very similar in spirit to the example formulas we have seen so far. They also define applicative bisimilarity following \cite{DBLP:conf/lics/LagoGL17}. The logical equivalence induced by the modal logic is then proved to coincide with applicative bisimilarity.

In the conclusion of their paper, Simpson and Voorneveld observe that logical equivalence is included in contextual equivalence. However, contextual equivalence equates more programs than logical equivalence does. Therefore, an open research direction is finding a logic that characterises contextual equivalence. This is the main goal of this dissertation.

\section{PCF with Effects -- EPCF}\label{Sec_epcf_def}

The programming language used in the work of Simpson and Voorneveld \cite{SimV18} is a call-by-value, simply-typed $\lambda$-calculus with recursion, a datatype of natural numbers and algebraic effects. Therefore, the language is a variant of Plotkin's PCF \cite{DBLP:journals/tcs/Plotkin77} extended with algebraic effects; in this work we will refer to it as EPCF.

In order to simplify their proofs, Simpson and Voorneveld formulate EPCF as fine-grained call-by-value \cite{LPT03}. This means that there is a distinction between terms that are values and terms that are computations; they form separate syntactic categories. For example, $\lbd{x}{\mathbbm{N}}{S(x)}$ and $3$ are values because they cannot reduce, while $(\lbd{x}{\mathbbm{N}}{S(x)})\ 3$ is a computation. Here $S(x)$ represents the successor of $x$. The fine-grained call-by-value formulation is equivalent to the usual call-by-value formulation.

\begin{definition}[EPCF]\label{Def_epcf}
Types and environments:
\begin{align*}
\tau,\rho &\coloneqq \mathbbm{1} \mid \mathbbm{N} \mid \rho\rightarrow\tau	\\
\Gamma &\coloneqq \emptyset \mid \Gamma,x:\tau.
\end{align*}
Values and computations are defined by the grammar:
\begin{align*}
V,W &\coloneqq \star \mid Z \mid S(V) \mid \lbd{x}{\tau}{M} \mid x	\\
M,N &\coloneqq V\ W \mid \return{V} \mid \letin{M}{x}{N} \mid \fix{V} \mid \casess{V}{M}{x}{N}.
\end{align*}
\end{definition}

The ground types are unit $\mathbbm{1}$ and natural numbers $\mathbbm{N}$. There is a countably infinite set of variables ranged over by $x$. The environment $\Gamma, x:\tau$ assumes that $x$ does not appear in $\Gamma$.

Terms $V,W$ represent values, and $M,N$ represent computations, that is, terms which can be evaluated. The intuitive semantics of computations is the following: $\return{V}$ immediately returns the value $V$. The construct $\letin{M}{x}{N}$ is a sequencing operation: first it evaluates $M$, if this returns a value $V$, $V$ is substituted for $x$ in $N$, then $N[V/x]$ is evaluated. The computation $\fix{V}$ calculates the fixed point of the function $V$. The $\textbf{case}\ V$ construct branches according to whether the natural number $V$ is zero or a successor.

The language EPCF incorporates effects in a general way. Instead of specifying all the effect operations in the language, the definition of EPCF is parametrised by a set of effect operations $\Sigma$. The set $\Sigma$ can be instantiated in turn for nondeterminism, probabilistic choice, global store, input and output etc. This is done in a series of examples at the end of the section.

Each operation $\sigma\in\Sigma$ has an arity which specifies what arguments the operation takes and what type the resulting computation has. The possible arities are:
\begin{equation*}
\alpha^n\rightarrow\alpha \qquad \mathbbm{N}\times \alpha^n \rightarrow \alpha \qquad \alpha^{\mathbbm{N}}\rightarrow\alpha \qquad \mathbbm{N}\times\alpha^{\mathbbm{N}}\rightarrow\alpha
\end{equation*}
where $\alpha$ can be regarded as a type variable. They should be interpreted as follows: $\sigma:\mathbbm{N}\times\alpha^n\rightarrow\alpha$ is an operation that takes a natural number $V$ and $n$ computations of type $\alpha$, $M_0, M_1,\ldots,M_{n-1}$. The resulting computation $\sigma(V;M_0,M_1,\ldots,M_{n-1})$ has type $\alpha$. The expression $\alpha^\mathbbm{N}$ represents a function from natural numbers to the type $\alpha$.

\begin{repdefinition}{Def_epcf}[EPCF -- continued]
Fix a set of effect operations $\Sigma$, with associated arities. The grammar of computations is extended as follows:
\begin{equation*}
M,N \coloneqq \ldots \mid \sigma(M_0,M_1,\ldots,M_{n-1}) \mid \sigma(V;M_0,M_1,\ldots,M_{n-1}) \mid \sigma(V) \mid \sigma(V;W).
\end{equation*}
The typing relations $\Gamma\vdash V:\tau$ and $\Gamma\vdash M:\tau$ are the least relations closed under the rules in Figure \ref{Fig_PFC_type}. The judgement $\Gamma\vdash V:\tau$ should be read as $V$ has type $\tau$ in environment $\Gamma$.
\end{repdefinition}

\begin{figure}
\begin{gather*}
\inferrule
	{ }
	{\Gamma,x:\tau \vdash x:\tau}
(\textsc{var}) \quad
\inferrule
	{ }
	{\Gamma \vdash \star : \mathbbm{1}}
(\textsc{unit}) \quad
\inferrule
	{ }
	{\Gamma \vdash Z : \mathbbm{N}}
(\textsc{zero}) \quad
\inferrule
	{\Gamma \vdash V : \mathbbm{N}}
	{\Gamma \vdash S(V) : \mathbbm{N}}
(\textsc{succ}) \\
\inferrule
	{\Gamma \vdash V : \tau}
	{\Gamma \vdash \return{V} : \tau}
(\textsc{ret}) \quad
\inferrule
	{\Gamma,x:\tau \vdash M : \rho} 
	{\Gamma \vdash \lbd{x}{\tau}{M} : \tau\rightarrow\rho}
(\textsc{lbd})	\quad
\inferrule
	{\Gamma \vdash V : \tau\rightarrow\rho \\ \Gamma \vdash W:\tau}
	{\Gamma \vdash V\ W : \rho}
(\textsc{app})	\\
\inferrule
	{\Gamma \vdash V : (\tau\rightarrow\rho)\rightarrow(\tau\rightarrow\rho)}
	{\Gamma \vdash \fix{V} : \tau\rightarrow\rho}
(\textsc{fix}) \quad
\inferrule
	{\Gamma \vdash M : \tau \\ \Gamma ,x:\tau \vdash N : \rho}
	{\Gamma \vdash \letin{M}{x}{N} : \rho}
(\textsc{let}) \\
\inferrule
	{\Gamma \vdash V : \mathbbm{N} \\ \Gamma \vdash M : \tau \\ \Gamma,x:\mathbbm{N} \vdash N : \tau}
	{\casess{V}{M}{x}{N} : \tau}
(\textsc{case}) \\
\inferrule
	{\sigma: \alpha^n\rightarrow\alpha \\ \Gamma \vdash M_i : \tau}
	{\Gamma \vdash \sigma(M_0,M_1,\ldots,M_{n-1}):\tau}
(\textsc{op1}) \quad
\inferrule
	{\sigma: \alpha^{\mathbbm{N}}\rightarrow\alpha \\ \Gamma\vdash V : \mathbbm{N}\rightarrow\tau}
	{\Gamma \vdash \sigma(V):\tau}
(\textsc{op2}) \\
\inferrule
	{\sigma: \mathbbm{N}\times\alpha^n\rightarrow\alpha \\ \Gamma \vdash V : \mathbbm{N} \\ \Gamma \vdash M_i : \tau}
	{\Gamma \vdash \sigma(V;M_0,M_1,\ldots,M_{n-1}) : \tau}
(\textsc{op3})	\\											
\inferrule
	{\sigma: \mathbbm{N}\times\alpha^{\mathbbm{N}}\rightarrow\alpha \\ \Gamma \vdash V : \mathbbm{N} \\ \Gamma \vdash W : \mathbbm{N}\rightarrow\tau}
	{\Gamma \vdash \sigma(V;W):\tau}
(\textsc{op4})	
\end{gather*}
\caption{Typing judgements for EPCF \cite{SimV18}.} \label{Fig_PFC_type}
\end{figure}

Substitution of values for free variables inside values, $W[V/x]$, and inside computations, $M[V/x]$, is defined by recursion on the structure of $W$ and $M$ in a standard way. An example for effect operations is:
\begin{equation*}
\sigma(W;M_0,M_1,\ldots,M_{n-1})[V/x] = \sigma(W[V/x];M_0[V/x],M_1[V/x],\ldots,M_{n-1}[V/x]).
\end{equation*} 

We use the notation $\textit{Val}(\tau)$ for the set of closed values of type $\tau$ and $\textit{Comp}(\tau)$ for the set of closed computations of type $\tau$. We denote natural numbers $S^n(Z)$ by $\overline{n}$. Everywhere in the dissertation we consider terms up to $\alpha$-conversion.

Below are examples of effect operations from \cite{SimV18} which will be used throughout the dissertation:

\begin{example}[Pure functional computation]\label{Eg_pure}
In this case, the language has no effects so the set $\Sigma$ is empty.
\end{example}

\begin{example}[Nondeterminism]\label{Eg_nondeterminism}
There is one effect operation representing binary choice $or:\alpha^2\rightarrow\alpha$ so $\Sigma=\{or\}$. It takes as arguments two computations of type $\alpha$ and chooses to run one of them. The choice is determined by an external agent.
\end{example}

\begin{example}[Probabilistic choice]\label{Eg_probabilistic}
Define $\Sigma=\{p\text{-}or\}$ where $p\text{-}or:\alpha^2\rightarrow\alpha$ is a binary choice operator. With probability $0.5$ it executes the first computation, otherwise the second computation.
\end{example}

\begin{example}[Global store]\label{Eg_globalstore}
Fix a finite set of locations $\mathbbm{L}$ for storing natural numbers. For each $l\in\mathbbm{L}$, $\Sigma$ contains the following operations:
\begin{gather*}
lookup_l : \alpha^\mathbbm{N}\rightarrow\alpha	\\ 
update_l : \mathbbm{N}\times\alpha\rightarrow\alpha.
\end{gather*}
The intuition is the following: the computation $lookup_l(V)$ looks up the number at location $l$ then passes it to the function $V$; $update_l(\overline{n},M)$ writes $n$ to location $l$ then runs the computation $M$.
\end{example}

\begin{example}[Input/output]\label{Eg_io}
In this case, $\Sigma=\{read,write\}$ where 
\begin{gather*}
read:\alpha^\mathbbm{N}\rightarrow\alpha	\\
write:\mathbbm{N}\times\alpha\rightarrow\alpha.
\end{gather*}
The computation $read(V)$ reads a natural number from the input channel and passes it to the function $V$, which then executes. The computation $write(\overline{n}, M)$ outputs the number $n$ then continues as $M$.

These operations seem very similar to those for global store. One difference is that updating location $l$ with value $n$ then immediately looking up the value will always yield $\overline{n}$. This is not the case for I/O operations $write$ and $read$. There are no guarantees about the values on the output and input channels.
\end{example}

\section{Operational Semantics of EPCF}

Simpson and Voorneveld define an operational semantics for EPCF where computations evaluate to trees, following Plotkin and Power \cite{DBLP:conf/fossacs/PlotkinP01}.

\begin{definition}\label{Def_ecps_oper_sem}
The operational semantics uses evaluation stacks $S$ to implement sequencing. They are defined as:
\begin{equation*}
S \coloneqq id\ |\ S \circ (\letin{(-)}{x}{M}).
\end{equation*}
Define the operation of `filling in the hole' of a stack with a closed computation as:
\begin{gather*}
id\{N\} = N	\\
(S \circ (\letin{(-)}{x}{M}))\{N\} = S\{\letin{N}{x}{M}\}.
\end{gather*}
Write $\textit{Stack}(\tau,\rho)$ for the set of stacks $S$ which when given a computation $N\in\textit{Comp}(\tau)$, return a computation $S\{N\}\in\textit{Comp}(\rho)$.

The operational semantics consists of two relations, one between closed computations, and one between configurations $(S,M)$, where $M\in\textit{Comp}(\tau)$ and $S\in\textit{Stack}(\tau,\rho)$.
\begin{gather*}
(\lbd{x}{\tau}{M})\ V \rightsquigarrow M[V/x]	\\
\fix{F} \rightsquigarrow \return{\lbd{x}{\tau}{\letin{F\ (\lbd{y}{\tau}{\letin{\fix{F}}{z}{z\ y}})}{w}{w\ x}}}	\\
\casess{Z}{M}{x}{N} \rightsquigarrow M	\\
\casess{S(V)}{M}{x}{N} \rightsquigarrow N[V/x]
\end{gather*}
\begin{gather*}
(S,\ \letin{N}{x}{M}) \rightarrowtail (S \circ \letin{(-)}{x}{M},\ N)	\\
(S \circ \letin{(-)}{x}{M},\ \return{V}) \rightarrowtail (S,\ M[V/x])	\\
(S,M) \rightarrowtail (S,M') \text{ if } M \rightsquigarrow M'
\end{gather*}
Denote by $\rightarrowtail^*$ the reflexive-transitive closure of $\rightarrowtail$.
\end{definition}

The reduction rule for the fixed point $\fix{F}$ is somewhat complicated by the syntactic restrictions imposed by fine-grained call-by-value. Intuitively, we can think of $\fix{F}$ as reducing to the thunk of $F(\fix{F})$.

By inspecting the reduction relation $\rightarrowtail$ we can see that it is deterministic. There are two ways $\rightarrowtail$ can get stuck. If $(S,M)\rightarrowtail^* (id,\return{V})$; in this case there is nothing left to do so the computation should terminate. Or if $(S,M)\rightarrowtail^* (S',\sigma(\ldots))$. In this case, an effect operation should take place and the execution should continue from $S'$ with the computation chosen by the effect operation. It is also possible that $\rightarrowtail$ never terminates due to the presence of recursion.

This suggests that a computation of type $\tau$ should evaluate to an \emph{effect tree} with leaves values of type $\tau$. Denote the set of all such trees by $\textit{Trees}(\tau)$. A possibly infinite tree in $\textit{Trees}(\tau)$ can have:
\begin{itemize}
\item a leaf labelled by $\bot$, which signifies nontermination of $\rightarrowtail$;
\item a leaf labelled by a value $V\in\textit{Val}(\tau)$;
\item a node labelled $\sigma$ where $(\sigma:\alpha^n\rightarrow\alpha)\in\Sigma$; this has children $t_0,t_1,\ldots t_{n-1}$;
\item a node labelled $\sigma$ where $(\sigma:\alpha^\mathbbm{N}\rightarrow\alpha)\in\Sigma$; this has infinitely many children $t_0,t_1,\ldots$, one for each natural number;
\item a node labelled $\sigma_m$ where $(\sigma:\mathbbm{N}\times\alpha^n\rightarrow\alpha)\in\Sigma$; the label $m$ is the natural number that $\sigma$ takes as an argument; the node has children $t_0,t_1,\ldots t_{n-1}$;
\item a node labelled $\sigma_m$ where $(\sigma:\mathbbm{N}\times\alpha^\mathbbm{N}\rightarrow\alpha)\in\Sigma$; this has children $t_0,t_1,\ldots$.
\end{itemize}

Before defining the tree associated to each computation, we need to define a partial order on $\textit{Trees}(\tau)$ as follows:
\begin{center}
$tr_1 \leq tr_2 \quad\Longleftrightarrow\quad$ $tr_1$ can be obtained from $tr_2$ by replacing some of its subtrees by $\bot$.
\end{center}
This ordering endows $\textit{Trees}(\tau)$ with an $\omega$-CPO structure \cite{MPF16} with least element $\bot$. This means that every increasing chain $tr_1\leq tr_2\leq\ldots$ has a least upper bound $\bigsqcup_{n\in\mathbb{N}}tr_n$. 

\begin{repdefinition}{Def_ecps_oper_sem}[Continued]\label{Def_epcf_tree_chain}
Define a family of functions 
\begin{equation*}
\trees{-,-}_{(-)}:\textit{Stack}(\tau,\rho)\times\textit{Comp}(\tau)\times\mathbb{N}\longrightarrow\textit{Trees}(\rho).
\end{equation*}
The tree $\trees{S,M}_n$ represents the unfolding of computation $M$ for $n$ steps starting in stack $S$. The formal definition is:
\begin{align*}
\trees{S,M}_0 &= \bot	\\
\trees{S,M}_{n+1} &= \begin{cases}
					V	&\text{if } S=id \text{ and } M=\return{V}	\\
					\trees{S',M'}_n	&\text{if } (S,M)\rightarrowtail(S',M')	\\
					\sigma(\trees{S,M_0}_n,\ldots,\trees{S,M_{k-1}}_n) &\text{if } \sigma:\alpha^k \rightarrow \alpha  \\
						&\text{and } M=\sigma(M_0,M_1,\ldots,M_{k-1})	\\
					\sigma(\trees{S,V\ \overline{0}}_n, \trees{S,V\ \overline{1}}_n,\ldots)	&\text{if } \sigma:\alpha^\mathbbm{N} \rightarrow \alpha	\text{ and } M=\sigma(V) \\
					\sigma_m(\trees{S,M_0}_n,\ldots,\trees{S,M_{k-1}}_n) &\text{if } \sigma:\mathbbm{N}\times\alpha^k \rightarrow \alpha \\
						&\text{and } M=\sigma(\overline{m};M_0,M_1,\ldots,M_{k-1})	\\
					\sigma_m(\trees{S,V\ \overline{0}}_n, \trees{S,V\ \overline{1}}_n,\ldots)	&\text{if } \sigma:\mathbbm{N}\times\alpha^\mathbbm{N} \rightarrow \alpha \text{ and } M=\sigma(\overline{m};V) 	\\
					\bot	&\text{otherwise}.	
					\end{cases}
\end{align*}
From this definition we can see that $\trees{S,M}_n\leq \trees{S,M}_{n+1}$. Therefore, the effect tree associated with a computation is defined as:
\begin{gather*}
\trees{-}:\textit{Comp}(\tau)\longrightarrow \textit{Trees}(\tau)	\\
\trees{M} = \bigsqcup_{n\in\mathbb{N}}\trees{id,M}_n.
\end{gather*}
We call $\trees{M}$ a computation tree.
\end{repdefinition}

The intuitive interpretation of computation trees is that a path through the tree represents a potential execution path of the program. However, the operational semantics is not aware of this interpretation. In this sense, the effect operations are purely formal as in \cite{DBLP:conf/fossacs/PlotkinP01}. This will be illustrated in Example \ref{Eg_globalstore_tree} below.

\begin{example}[Pure functional computation]
In this case, there are no effect operations so all computation trees are leaves.
\end{example}

\begin{example}[Nondeterminism]\label{Eg_nondet_nat}
Consider a computation that generates a natural number nondeterministically (from \cite{SimV18}):
\begin{equation*}
?nat = \letin{\fix{(\lbd{f}{\mathbbm{1}{\rightarrow}{\mathbbm{N}}}{or(\lbd{y}{\mathbbm{1}}{Z},\ \lbd{y}{\mathbbm{1}}{\letin{f\ y}{u}{S(u)}})})}}{w}{w\ \star}.
\end{equation*}
Its computation tree is:
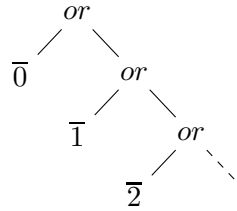
\begin{center}
\begin{tikzpicture}[level distance=0.8cm]
\node {$or$}
	child { node {$\overline{0}$} }
	child { node {$or$}
		child { node {$\overline{1}$} }
		child { node {$or$}
			child { node {$\overline{2}$} }
			child { node {} edge from parent[dashed] } } };
\end{tikzpicture}
\end{center}
Thus, the notion of a tree whose nodes are nondeterministic choice points, discussed informally in Section \ref{Sec_alg_effects}, is formalised by a computation tree.
\end{example}

\begin{example}[Probabilistic choice]\label{Eg_prob_epcf_tree}
Define the following computation whose execution never terminates:
\begin{equation*}
loop = \letin{(\fix{\lbd{f}{\mathbbm{N}{\rightarrow}\mathbbm{1}}{\return{f}}})}{g}{g\ Z}.
\end{equation*}
The computation tree of $p\text{-}or(loop,\return{\star}):\mathbbm{1}$ is therefore:
\begin{center}
\begin{tikzpicture}[level distance=0.8cm]
\node {$p\text{-}or$}
	child { node {$\bot$} }
	child { node {$\star$} };
\end{tikzpicture}
\end{center}
\end{example}

\begin{example}[Global store]\label{Eg_globalstore_tree} 
Consider a location $l_0\in\mathbb{L}$. The following computation writes $\overline{0}$ to location $l_0$ then immediately reads this value:
\begin{equation*}
\trees{update_{l_0}(\overline{0};\ lookup_{l_0}(\lbd{x}{\mathbbm{N}}{\casess{x}{\return{\star}}{y}{loop}}))}=
\end{equation*}
\begin{center}
\begin{tikzpicture}[level distance=1.1cm]
\node {$update_{l_0,0}$}
	child { node {$lookup_{l_0}$} 
		child { node {$\star$} }
		child { node {$\bot$} }
		child { node {$\bot$} }
		child { node[left=0.5cm] {$\ldots$}  edge from parent[draw=none] } 
		};
\end{tikzpicture}
\end{center}
According to the intuitive interpretation that we give to the $lookup$ and $update$ operations only the path that returns a $\star$ can occur. However, the tree contains a path for each natural number that could be in $l_0$ because the $\trees{-}$ function treats operation symbols as syntax without any interpretation. If we replaced $lookup$ and $update$ with the $read$ and $write$ operations from I/O then all paths would be relevant.
\end{example}

\begin{example}[Input/output]
The following computation reads a natural number from the input channel, then returns a function whose behaviour depends on this number:
\begin{equation*}
read(\lbd{x}{\mathbbm{N}}{\return{\lbd{f}{\mathbbm{N}{\rightarrow}\mathbbm{N}}{write(x;\ (f\ x))}}}) \quad :(\mathbbm{N}\rightarrow\mathbbm{N})\rightarrow\mathbbm{N}
\end{equation*}
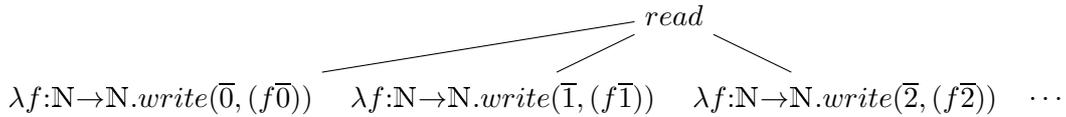
\begin{center}
\begin{tikzpicture}[level distance=1.1cm, sibling distance=4.5cm]
\node {$read$}
		child { node {$\lbd{f}{\mathbbm{N}{\rightarrow}\mathbbm{N}}{write(\overline{0}, (f \overline{0}))}$} }
		child { node {$\lbd{f}{\mathbbm{N}{\rightarrow}\mathbbm{N}}{write(\overline{1}, (f \overline{1}))}$} }
		child { node {$\lbd{f}{\mathbbm{N}{\rightarrow}\mathbbm{N}}{write(\overline{2}, (f \overline{2}))}$} }
		child { node[left=1.5cm] {$\ldots$}  edge from parent[draw=none] }; 
\end{tikzpicture}
\end{center}
A path through the tree is not only a possible execution path of the computation, it also corresponds to an I/O trace. 
\end{example}

\section{Summary of Results about EPCF}\label{Sec_epcf_results}

The main contribution of Simpson's and Voorneveld's work \cite{SimV18} is that they define a modal logic whose formulas represent properties of EPCF programs. We will refer to it as EPCF logic. Moreover, they give a general analysis of the modalities involved.

In EPCF logic, there are two kinds of formulas all of which are attached an EPCF type. A formula $\phi:\tau$ describes a value of type $\tau$, while $\Phi:\tau$ describes a computation of type $\tau$. The definition of the logic starts from a set of basic formulas at each type which is then closed under negation, arbitrary conjunctions and disjunctions.

The basic formulas for values of type $\mathbbm{N}$ are:
\begin{equation*}
\{n\} \text{ where } n\in\mathbb{N}.
\end{equation*}
A closed value $W:\mathbbm{N}$ satisfies $\{n\}$, written $W\models\{n\}$, if and only if $W=\overline{n}$.

For values of function types $\tau\rightarrow\rho$, the basic formulas are:
\begin{equation*}
\phi\mapsto \Phi
\end{equation*}
where $\phi$ is a value formula of type $\tau$ and $\Phi$ is a computation formula of type $\rho$. The satisfaction $W\models \phi\mapsto \Phi$ holds if and only if:
\begin{equation*}
\forall V\models\phi.\ (W\ V) \models\Phi.
\end{equation*}

A value formula $\phi\mapsto\Phi$ tests the behaviour of a function when it is being applied. This behaviour is the fundamental property of a function, hence the choice of value formula is reasonable.

Finally, computation formulas make use of a set of \emph{modalities} $\mathcal{O}$. The set $\mathcal{O}$ contains sets of effect trees of type $\mathbbm{1}$, that is, $\mathcal{O}\subseteq\mathcal{P}(\textit{Trees}(\mathbbm{1}))$. A basic computation formula of type $\tau$ is:
\begin{equation*}
o\phi
\end{equation*}
where $o\in\mathcal{O}$ and $\phi$ is a value formula of type $\tau$. We can see that $o$ lifts a formula for values, $\phi$, to a formula for computations. This is why $o$ is named a \emph{modality}.

For a tree $tr\in\textit{Trees}(\tau)$ and a value formula $\phi:\tau$ denote by:
\begin{equation*}
tr[\models \phi]
\end{equation*}
the tree in $\textit{Trees}(\mathbbm{1})$ obtained by replacing the leaves $V$ of $tr$ by $\star$ if $V\models\phi$ and by $\bot$ otherwise. We can now define satisfaction of computation formulas as:
\begin{equation*}
M \models o\phi \quad \Longleftrightarrow \quad \trees{M}[\models\phi]\in o.
\end{equation*}
A computation formula $o\phi$ tests whether the possible return values of a computation satisfy $\phi$ and also tests the shape of the effect tree of the computation. These two pieces of information form the \emph{observable behaviour} of a computation, in the sense of Section \ref{Sec_prog_equiv_intro}.

The definition of $\mathcal{O}$ depends on the effects present in the language so we will look at the nondeterminism example. The other effects are treated in more detail in Chapter~\ref{Chap_bisim}, in the context of ECPS.

\begin{example}[Nondeterminism]\label{Eg_nondet_epcf_logic}
Define $\mathcal{O}=\{\Diamond,\Box\}$ where: 
\begin{align*}
\Diamond &= \{tr\in\textit{Trees}(\mathbbm{1}) \mid tr \text{ has some } \star \text{ leaf } \}	\\
\Box &= \{tr\in\textit{Trees}(\mathbbm{1}) \mid tr \text{ has finite height and every leaf is a }\star \}.
\end{align*}
The formula $\Diamond\phi$ says that a computation may return a value satisfying $\phi$, whereas $\Box\phi$ asserts that it must return such a value. We can see that $\Box$ and $\Diamond$ which we discussed informally in Examples~\ref{Eg_boxdi} and~\ref{Eg_dihigher} are now defined as modalities.

In fact, the formulas used in those examples are valid formulas in EPCF logic:
\begin{align*}
\phi_1&=\{3\}\mapsto\Box\{2\}	\\
\phi_2&=\{3\}\mapsto\Diamond\{2\}	\\
\psi &= (\{3\}\mapsto\Diamond\{2\})\mapsto ((\{2\}\lor\{4\})\mapsto\Diamond\{2\})
\end{align*}
and the example functions:
\begin{align*}
g &=\lbd{n}{\mathbbm{N}}{or(\mathbf{pred}\ n,\ \mathbf{succ}\ n)}	\\
h &= \lbd{f}{\mathbbm{N{\rightarrow}\mathbbm{N}}}{\ \lbd{n}{\mathbbm{N}}{\ or(f\ (\mathbf{pred}\ n),\ f\ (\mathbf{succ}\ n))}}
\end{align*}
are valid EPCF terms for suitable encodings of $\mathbf{succ}$ and $\mathbf{pred}$.
We can see that $(id,(g\ \overline{3}))$ may return either $\overline{2}$ or $\overline{4}$ so using the definition of logical satisfaction we indeed obtain:
\begin{equation*}
g\not\models\phi_1 \quad \text{and} \quad g\models\phi_2.	
\end{equation*}
Using the same definition we can also see that $h\models\psi$.
\end{example}

Simpson and Voorneveld \cite{SimV18} are concerned with two notions of program equivalence: applicative bisimilarity and logical equivalence induced by satisfaction in EPCF logic. Applicative bisimilarity is a family of relations defined between well-typed EPCF terms. The definition is technically involved so we omit it. However, some key features of applicative bisimilarity are:
\begin{itemize}
\item Two natural numbers are bisimilar if and only if they are equal.

\item Two function values are bisimilar if and only if for all arguments they yield bisimilar computations. This condition is the same as in the informal explanation of similarity from Section \ref{Sec_prog_equiv_intro}.

\item To specify when two computations are bisimilar, their effect trees are inspected. Roughly speaking, $M$ and $N$ are bisimilar if: the possible return values of $N$ approximate those of $M$ well enough to preserve the properties of the effect tree of $M$, and vice-versa.
The set of modalities $\mathcal{O}$ is used express properties of computation trees and to lift bisimilarity of values to a relation between trees. 
\end{itemize}
Because the definition of bisimilarity depends on the set of modalities $\mathcal{O}$, it is named applicative $\mathcal{O}$-bisimilarity.

\begin{definition}
Logical equivalence between well-typed EPCF terms is defined as:
\begin{align*}
V \equiv_{\textit{EPCF}} W \quad &\Longleftrightarrow \quad \forall\phi:\tau.\ V\models\phi \Longleftrightarrow W\models\phi	\\
M \equiv_{\textit{EPCF}} N \quad &\Longleftrightarrow \quad \forall\Phi:\tau.\ M\models\Phi \Longleftrightarrow N\models\Phi.
\end{align*}
\end{definition}

It is easy to prove that logical equivalence and applicative bisimilarity for EPCF are equivalence relations. In order to prove they are \emph{compatible}, Simpson and Voorneveld identify two sufficient conditions that the set of modalities $\mathcal{O}$ should satisfy in general:
\begin{itemize}
\item Each modality $o\in\mathcal{O}$ needs to be Scott-open. This condition will be defined later (Definition~\ref{Def_scott_open}). It refers to the Scott-topology on $\textit{Trees}(\mathbbm{1})$.
\item The set $\mathcal{O}$ needs to be decomposable. This roughly means that, for a valid computation tree, its subtrees can also be described using the modalities in $\mathcal{O}$. Decomposability is a notion introduced by Simpson and Voorneveld whose definition is very technical so we omit it.  
\end{itemize}
They show that for all the example effects in this dissertation $\mathcal{O}$ is decomposable and contains only Scott-open modalities.
Assuming that $\mathcal{O}$ is a decomposable set of Scott-open modalities, Simpson and Voorneveld prove the following theorems:

\begin{theorem}
Applicative $\mathcal{O}$-bisimilarity is compatible.
\end{theorem}

\begin{theorem}\label{Thm_sv_log_is_bisim}
Logical equivalence induced by EPCF logic coincides with applicative $\mathcal{O}$-bisimilarity. Hence, logical equivalence is compatible.
\end{theorem}

Thus, they obtain a logical characterisation of applicative bisimilarity for EPCF.

\section{Problem: Contextual vs.~Logical Equivalence}\label{Sec_epcf_ctx_not_log}

Simpson and Voorneveld \cite{SimV18} briefly note that, in EPCF, contextual equivalence equates more programs than logical equivalance and bisimilarity do. 
In this section we look at example programs, originally due to Lassen \cite{LasPhd} \footnote{I am grateful to Niels Voorneveld for pointing out this example.}, that are contextually equivalent but not logically equivalent.

Consider a computation $min(v,w)$ which returns the minimum of the natural numbers $v$ and $w$. This can be encoded in EPCF as:
\begin{center}
\begin{tabular}{c}
\begin{lstlisting}[mathescape=true]
$min(v,w)=$
  $\mathbf{let}\ (\mathbf{fix}\ \lambda f:\mathbbm{N}{\rightarrow}\mathbbm{N}{\rightarrow}\mathbbm{N}.$
          $\mathbf{return}\ \lambda x:\mathbbm{N}.$
            $\mathbf{return}\ \lambda y:\mathbbm{N}.$
              $\mathbf{case}\ x\ \mathbf{in}\ \{Z\Rightarrow \mathbf{return}\ v,$
                      $S(z_0)\Rightarrow\mathbf{case}\ y\ \mathbf{in}\ \{Z\Rightarrow \mathbf{return}\ w,$
                                     $S(z_1)\Rightarrow (\mathbf{let}\ f\ z_0\Rightarrow h\ \mathbf{in}\ h\ z_1)\}\}$
      $)\Rightarrow g\ \mathbf{in\ let}\ g\ v\Rightarrow g'\ \mathbf{in}\ g'\ w$.
\end{lstlisting}
\end{tabular}
\end{center} 

\begin{example}[Nondeterminism] Recall the computation $?nat$, from Example~\ref{Eg_nondet_nat}, which produces a natural number nondeterministically. Define:
\begin{align*}
&M = \return{\lbd{x}{\mathbbm{1}}{?nat}} \quad: \mathbbm{1}\rightarrow\mathbbm{N}	\\
&P = \letin{?nat}{x}{(\return{\lbd{z}{\mathbbm{1}}{\letin{?nat}{y}{min(x,y)}}})}\quad: \mathbbm{1}\rightarrow\mathbbm{N}.
\end{align*}
Consider the following computation formula in EPCF logic of type $\mathbbm{1}\rightarrow\mathbbm{N}$:
\begin{equation*}
\Phi = \Diamond(true\mapsto\land_{n\in\mathbb{N}}\Diamond\{n\}).
\end{equation*}
A computation that satisfies this formula is one such that: at least one of its possible return values is a function which, when given argument $\star$, \emph{may} return \emph{any} natural number. Note that for type $\mathbbm{1}$, the only formulas are $true$ and $false$, represented as the empty conjunction and disjunction respectively, and the only value is $\star$.

We can show that $M\models\Phi$. We need to check that $\trees{M}[\models true\mapsto\land_{n\in\mathbb{N}}\Diamond\{n\}]\in\Diamond$. Computation $M$ returns immediately so $\trees{M}$ is just a leaf labelled $\lbd{x}{\mathbbm{1}}{?nat}$. The set $\Diamond$ contains the trees with at least one $\star$ leaf. So it is sufficient to check:
\begin{equation*}
\lbd{x}{\mathbbm{1}}{?nat} \models true\mapsto\land_{n\in\mathbb{N}}\Diamond\{n\}
\end{equation*}
that is,
\begin{equation*}
?nat \models \land_{n\in\mathbb{N}}\Diamond\{n\}.
\end{equation*}
Recall that the computation tree of $?nat$ is:
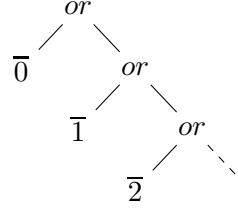
\begin{center}
\begin{tikzpicture}[level distance=0.8cm]
\node {$or$}
	child { node {$\overline{0}$} }
	child { node {$or$}
		child { node {$\overline{1}$} }
		child { node {$or$}
			child { node {$\overline{2}$} }
			child { node {} edge from parent[dashed] } } };
\end{tikzpicture}
\end{center}
It has a leaf labelled with each natural number, therefore $?nat \models \land_{n\in\mathbb{N}}\Diamond\{n\}$ is true.

On the other hand, $P\not\models\Phi$. To see this note that $\trees{P}$ is:
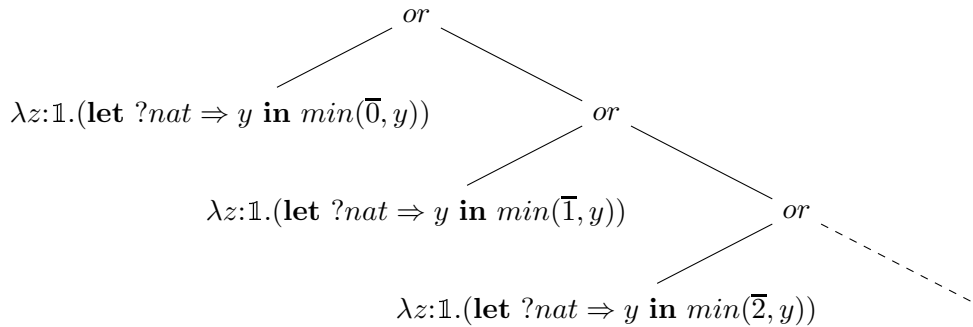
\begin{center}
\begin{tikzpicture}[level distance=1.3cm, sibling distance=5cm]
\node {$or$}
	child { node [text width=5.7cm] {$\lbd{z}{\mathbbm{1}}{(\letin{?nat}{y}{min(\overline{0},y)})}$} }
	child { node {$or$}
		child { node {$\lbd{z}{\mathbbm{1}}{(\letin{?nat}{y}{min(\overline{1},y)})}$} }
		child { node {$or$}
			child { node {$\lbd{z}{\mathbbm{1}}{(\letin{?nat}{y}{min(\overline{2},y)})}$} }
			child { node {} edge from parent[dashed] } } };
\end{tikzpicture}
\end{center}
In order for $\trees{P}[\models true\mapsto\land_{n\in\mathbb{N}}\Diamond\{n\}]\in\Diamond$ to be true we need to find some $m\in\mathbbm{N}$ such that:
\begin{equation*}
\lbd{z}{\mathbbm{1}}{(\letin{?nat}{y}{min(\overline{m},y)})} \models true\mapsto\land_{n\in\mathbb{N}}\Diamond\{n\}
\end{equation*}
that is,
\begin{equation*}
\letin{?nat}{y}{min(\overline{m},y)} \models \land_{n\in\mathbb{N}}\Diamond\{n\}.
\end{equation*}
By contradiction, assume that such an $m$ exists. The tree of $\letin{?nat}{y}{min(\overline{m},y)}$ has as leaves all the numbers from $0$ to $m$, but none greater than $m$:
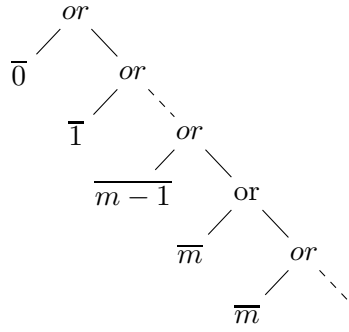
\begin{center}
\begin{tikzpicture}[level distance=0.8cm,
	norm/.style={edge from parent/.style={solid,draw}},
	emph/.style={edge from parent/.style={dashed,draw}} ]
\node {$or$}
	child { node {$\overline{0}$} }
	child { node {$or$}
		child { node {$\overline{1}$} }
		child[emph] { node {$or$}
			child[norm] { node {$\overline{m-1}$} }
			child[norm] { node {or} 
				child { node {$\overline{m}$} } 
				child { node {$or$} 
					child { node {$\overline{m}$} } 
					child[emph] { node {}  } } } } };
\end{tikzpicture}
\end{center}
So it is false that $\letin{?nat}{y}{min(\overline{m},y)} \models \Diamond\{m+1\}$. Therefore $P\not\models\Phi$ so $M$ and $P$ are not logically equivalent.
\end{example}

However, $M$ and $P$ are contextually equivalent. Since we have not defined contextual equivalence for EPCF rigorously, we only argue informally.


A context $C$ can compare computations $M$ and $P$ by supplying their return values, which are functions, with arguments and observing the possible results. The context is a syntactic device, hence it is finite. So it can only perform this test a finite number of times.

Computation $M$ returns a function that can generate any natural number, whereas $P$ returns a bounded number generator but \emph{the bound is arbitrarily large}. Because the context can only check for a finite set of natural numbers, say all smaller than $k$, $P$ and $M$ appear to be equivalent. It is always possible that $P$ returns a bounded number generator with bound larger than $k$.

In contrast, logical equivalence can test for an infinite number of outcomes at once. This is achieved in the formula $\Phi$ by using the infinite conjunction $\land_{n\in\mathbb{N}}\Diamond\{n\}$.

\section{Chapter Summary}

Section~\ref{Sec_prog_equiv_intro} discussed the notions of equivalence relation and compatibility in connection with program equivalence. These properties are useful when establishing equivalence of two particular programs because they allow chaining equations and substituting equals for equals inside an equation. Three kinds of program equivalence were then outlined: contextual equivalence, applicative bisimilarity and equivalence induced by a logic of program properties. These are studied in the rest of the dissertation.

In Section~\ref{Sec_alg_effects}, algebraic effects were briefly discussed. They specify behaviour such as nondeterminism, probabilistic choice, global state and I/O in a generic way. These behaviours are all triggered by a set of operations, for example, a  binary choice operator $or$ for nondeterminism. Because an effectful program has multiple possible execution paths, its execution can be pictured as a tree where the nodes are effect operations and the leaves are return values.

Next, a programming language named EPCF was introduced. It is a call-by-value extension of the simply-typed $\lambda$-calculus with recursion, natural numbers and algebraic effects. EPCF makes a syntactic distinction between values, terms which cannot reduce, and computations, which can reduce. Its operational semantics maps a computation to a tree. For example the computation $M$ where
\begin{align*}
F&= \lbd{f}{\mathbbm{N}{\rightarrow}\mathbbm{N}}{or(f\ \overline{1},\ or(f\ \overline{2},\ f\ \overline{3}))}	\\
M&= F\ (\lbd{n}{\mathbbm{N}}{\return{S(n)}})
\end{align*}
has tree:
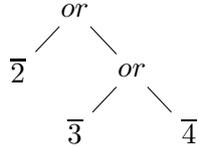
\begin{center}
\begin{tikzpicture}[level distance=0.8cm]
\node {$or$}
	child { node {$\overline{2}$} }
	child { node {$or$} 
		child { node {$\overline{3}$} }
		child { node {$\overline{4}$} }};
\end{tikzpicture}
\end{center}

Simpson and Voorneveld \cite{SimV18} proposed a modal logic that expresses properties of EPCF programs, named EPCF logic (Section~\ref{Sec_epcf_results}). For each effect, there is a set of \emph{modalities} that express properties of computations which exhibit that effect. For nondeterminism these are $\Box$ and $\Diamond$. For example, $\Phi_1=\Box(\{2\}\lor\{3\}\lor\{4\})$ says that computation $M$ \emph{always} returns a result from the set $\{2,3,4\}$, whereas $\Phi_2=\Diamond\{2\}$ says that $M$ \emph{may} return $2$.

For functions, logical formulas have the form $\phi\mapsto\Phi$. This says that, if the argument of the function satisfies $\phi$, then the resulting application satisfies $\Phi$. For example, $F$ satisfies the following property: $F\models (\{1\}\mapsto\Box\{2\})\mapsto\Diamond\{2\}$. 

Simpson and Voorneveld defined applicative bisimilarity for EPCF using modalities and proved it compatible. They showed that program equivalence induced by EPCF logic coincides with applicative bisimilarity but not with contextual equivalence. Section~\ref{Sec_epcf_ctx_not_log} described two EPCF programs exhibiting nondeterminism that are contextually equivalent but not logically equivalent. The problem is that a context can only test a function on a finite number of arguments, whereas a logical formula can test it on infinitely many arguments, using infinitary connectives. Thus, we have identified the main problem of the dissertation: finding a logic that characterises contextual equivalence for a higher-order language with algebraic effects.

\chapter{Introducing the ECPS Language}\label{Chap_descr_method}

This chapter introduces the programming language ECPS and its operational semantics. ECPS will be used in the rest of the dissertation to study program equivalence. Moreover, a few results that will be used later  are outlined: a coinduction proof principle and general intuitions about logical relations.

\section{A New Language -- ECPS}\label{Sec_ecps_def}

To make it easier to formulate a logic that characterises contextual equivalence, we introduce a new programming language ECPS. It is a variant of EPCF in which programs are written in continuation-passing style (CPS)  \cite{DBLP:journals/lisp/Reynolds93}. This means that functions carry an additional argument, namely the \emph{continuation} to which they pass their result. The intuition is that a continuation specifies how the execution should proceed once a function has finished. 

Given a fixed return type $R$, a continuation has type $\alpha\rightarrow R$. It is waiting for an argument of type $\alpha$ to produce a return value of type $R$.
Consider for example a function that adds two natural numbers. Usually, it has type $\mathtt{nat}\rightarrow \mathtt{nat}\rightarrow\mathtt{nat}$. In continuation-passing style, this function would look like:
\begin{align*}
&\mathtt{addc} : \mathtt{nat}\rightarrow\mathtt{nat}\rightarrow (\mathtt{nat}\rightarrow R)\rightarrow R	\\
&\mathtt{addc}\ n\ m\ cont = cont\ (n+m).
\end{align*}
Instead of directly returning the result $n+m$, the function $\mathtt{addc}$ passes it to the continuation $cont$.

The key property of ECPS that allows the formulation of the new logic, which will be introduced in Section \ref{Sec_ecps_modal_log}, is that functions do not have a return type. Once a function has been applied, the resulting computation is expected to run forever. In other words, the return type $R$ of continuations is chosen to be $\bot$.  Thus, programs no longer return values that we can observe. We might however observe termination, which is now treated as an effect, or other side effects such as output values.

\begin{definition}[ECPS]
The types are defined by the following grammar:
\begin{equation*}
A, B \coloneqq \neg(A_1,\ldots,A_n)\ |\ \mathtt{nat}\ |\ \mathtt{unit}.
\end{equation*}
Fix a set $\Sigma$ of effect operations $\sigma$, each with arity $\mathtt{nat}\times\alpha^\mathtt{nat}\rightarrow\alpha$, where $\alpha$ stands for a computation. Values and computations are defined respectively as:
\begin{align*}
v, w &\coloneqq \mathtt{zero}\ |\ \mathtt{succ}(v)\ |\ \star\ |\ \lbd{(x_1,\ldots,x_n)}{(A_1,\ldots,A_n)}{t}\ |\ x	\\
t, u &\coloneqq v(w_1,\ldots, w_n)\ |\ (\mufix{x}{v})(\overrightarrow{w})\ |\ \sigma(v,x.t)\ |\ \downarrow\ |\ \caset{v}{t}{x}{u}.
\end{align*}
There are two typing relations, one for values $\Gamma\vdash v :\tau$, and one for computations which do not have a type, $\Gamma\vdash t$. These are the least relations closed under the rules in Figure~\ref{Fig_ECPS_type}.
\end{definition}

\begin{figure}
\begin{gather*}
\inferrule
	{ }
	{\Gamma ,x:A \vdash x:A}
(\textsc{var}) \quad
\inferrule
	{\Gamma, x_1:A_1, \ldots , x_n:A_n \vdash t}
	{\Gamma \vdash \lbd{(x_1,\ldots,x_n)}{(A_1,\ldots,A_n)}{t} : \neg(A_1, \ldots, A_n)}
(\textsc{lbd})	\\
\inferrule
	{\Gamma \vdash v : \neg (A_1,\ldots,A_n) \\ \Gamma \vdash w_1 : A_1 \\ \ldots \\ \Gamma \vdash w_n : A_n}
	{\Gamma \vdash v\ (w_1,\ldots,w_n)}
(\textsc{app})	\\
\inferrule
	{ }
	{\Gamma \vdash \mathtt{zero} : \mathtt{nat}}
(\textsc{zero}) \quad
\inferrule
	{\Gamma \vdash v : \mathtt{nat}}
	{\Gamma \vdash \texttt{succ}(v) : \mathtt{nat}}	
(\textsc{succ})	\quad
\inferrule
	{ }
	{\Gamma \vdash \star : \mathtt{unit}}
(\textsc{unit})	\\
\inferrule
	{\Gamma \vdash v : \mathtt{nat} \\ \Gamma \vdash t \\ \Gamma,x:\mathtt{nat} \vdash s}
	{\Gamma \vdash \caset{v}{t}{x}{s}}
(\textsc{case})	\\
\inferrule
	{\Gamma,x:\neg(\overrightarrow{A}) \vdash v : \neg(\overrightarrow{A}) \\ \Gamma \vdash \overrightarrow{w} : \overrightarrow{A}}
	{\Gamma \vdash (\mufix{x}{v})(\overrightarrow{w})}
(\textsc{mu}) \\
\inferrule
	{\Gamma,x:\mathtt{nat} \vdash t \\ \Gamma \vdash v : \mathtt{nat}}
	{\Gamma \vdash \sigma(v, x.t)}
\sigma\in\Sigma\ (\textsc{op})	\quad
\inferrule
	{ }
	{\Gamma \vdash \downarrow}
(\textsc{stop})			
\end{gather*}
\caption{Typing judgements for ECPS.}\label{Fig_ECPS_type}
\end{figure}

ECPS is a fine-grained call-by-value language; it makes a distinction between values and computations. Ignoring effects, ECPS is in fact a fragment of Levy's Jump-With-Argument programming language \cite{DBLP:conf/csl/LassenL07}.

The type $\neg(A_1,\ldots,A_n)$ is the type of a function which takes $n$ arguments of types $A_1,\ldots,A_n$ respectively and returns a computation. This computation is not expected to terminate so we can think of $\neg(A_1,\ldots,A_n)$ as $(A_1,\ldots,A_n)\rightarrow\bot$. We can also think of $k:\neg A$ as a continuation of type $A\rightarrow R$, where the return type of all continuations $R$ has been set to $\bot$. The base types $\mathtt{nat}$ and $\mathtt{unit}$ are the same as in EPCF.

In $\mufix{x}{v}$, the expression $v$ is constrained by the typing rules to be a function. Thus, intuitively $\mufix{x}{v}$ is a recursive definition of the function $v$, where $x$ represents $v$ and can appear free inside $v$. Computation $(\mufix{x}{v})(\overrightarrow{w})$ arises by applying this function to $\overrightarrow{w}$.

\begin{example}
We can now implement the function $\mathtt{addc}$ in ECPS using continuations and recursion. The type of this function is:
\begin{equation*}
\mathtt{addc} : \neg(\mathtt{nat},\ \mathtt{nat},\ \neg\mathtt{nat})
\end{equation*}
\begin{lstlisting}[mathescape=true]
$\mathtt{addc}= \lambda(x,y,k):(\mathtt{nat},\ \mathtt{nat},\ \neg\mathtt{nat}).$
   $(\mu f.\ \lambda(x',y',k'):(\mathtt{nat},\ \mathtt{nat},\ \neg\mathtt{nat}).$
      $\mathtt{case}\ x'\ \mathtt{in}\ \{\mathtt{zero}\Rightarrow k'\ y',$
               $\mathtt{succ}(x'')\Rightarrow\mathtt{case}\ y'\ \mathtt{in}$
                         $\{\mathtt{zero}\Rightarrow k'\ x',$
               	          $\mathtt{succ}(y'')\Rightarrow f\ (x'',y'',\ \lambda z:\mathtt{nat}.\ k'\ \mathtt{succ}(\mathtt{succ}(z)))\}\}$
   $)\ (x,y,k)$.
\end{lstlisting}
The behaviour of this function can be explained intuitively as follows: if $x$ is zero then the result of the addition is $y$, so $y$ is passed to the current continuation $k$. Similarly if $y$ is zero. If $x$ and $y$ are both greater than zero, add $x-1$ and $y-1$ and pass the result to continuation $\lambda z:\mathtt{nat}.\ k'\ \mathtt{succ}(\mathtt{succ}(z))$. This continuation adds two to $(x-1)+(y-1)$, then passes the result to the current continuation $k$. This explanation might become clearer if read in conjunction with the operational semantics from the next section.
\end{example}

ECPS does not have a $\mathbf{let}$ constructor for sequencing. Sequencing can be achieved instead by manipulating the continuation passed to a program.

As discussed above, termination in ECPS is an effect. Its associated effect operation is $\downarrow$, which does not take any arguments. 

The different arities for effect operations from EPCF are conflated into the most general one: $\mathtt{nat}\times\alpha^\mathtt{nat}\rightarrow\alpha$. Therefore, operation $\sigma$ takes as arguments a natural number $v$ and a function from a natural number $x$ to a computation $t$, and returns a computation $\sigma(v,x.t)$.

Substitution of values for free variables, $v[w/x]$ and $t[w/x]$, is defined in a standard way by recursion on the structure of $v$ and $t$. We will use $\overline{n}$ to denote the natural number $\mathtt{succ}^n(\mathtt{zero})$. Let $(\vdash)$ be the set of closed computations and $(\vdash A)$ the set of closed values of type $A$.

\section{Operational Semantics of ECPS}\label{Sec_ecps_oper_sem}

Next, we present the operational semantics of ECPS. It does not need to use stacks because any control flow is explicitly encoded inside a computation using continuations.

\begin{definition} 
The operational semantics is given by two families of relations on closed computation terms
\begin{align*}
(\longrightarrow) &\subseteq (\vdash) \times (\vdash)	\\
(\xrightarrow{\sigma(v)}) &\subseteq (\vdash) \times (\vdash)^{\mathbb{N}} \quad \text{for any }\sigma\in\Sigma \text{ and } \vdash v:\mathtt{nat}
\end{align*}
defined as:
\begin{equation*}
\sigma(v,x.t) \xrightarrow{\sigma(v)} (t[\overline{n}/x])_{n\in\mathbb{N}}
\end{equation*}
\begin{gather*}
(\lbd{(\overrightarrow{x})}{(\overrightarrow{A})}{t})\ (\overrightarrow{w}) \longrightarrow t[\overrightarrow{w}/\overrightarrow{x}]	\\
(\mufix{x}{v})\ (\overrightarrow{w}) \longrightarrow (v[(\lbd{(\overrightarrow{y})}{(\overrightarrow{A})}{(\mufix{x}{v})(\overrightarrow{y})})/x])\ (\overrightarrow{w})	\\
\caset{\mathtt{zero}}{s}{x}{t} \longrightarrow s	\\
\caset{\mathtt{succ}(v)}{s}{x}{t} \longrightarrow t[v/x].
\end{gather*}
Denote by $\longrightarrow^*$ the reflexive-transitive closure of $\longrightarrow$.
\end{definition}

We can see that the $\longrightarrow$ reduction relation can only get stuck when encountering an effect operation $\sigma(\ldots)$ or $\downarrow$. It might also be the case that $\longrightarrow$ never terminates.

There are no reduction rules of any kind for $\downarrow$ since it signifies termination. 
For $\sigma\in\Sigma$, the $\xrightarrow{\sigma(v)}$ reduction rule has on its right-hand-side a set of computations, one for each natural number. Therefore, repeated applications of this rule lead to the construction of an infinitely branching tree.

Given this observation we can define effect trees for ECPS. Denote the set of all effect trees by $\textit{Trees}_\Sigma$. A tree in this set can have:
\begin{itemize}
\item a leaf labelled $\bot$, which signifies nontermination of $\longrightarrow$;
\item a leaf labelled $\downarrow$, which signifies termination;
\item nodes labelled $\sigma_n$, where $\sigma\in\Sigma$ and $n\in\mathbb{N}$; such a node has an infinite number of children $t_0,t_1,\ldots$.
\end{itemize}

We can define a partial order on $\textit{Trees}_\Sigma$ which makes it an $\omega$-CPO. This is similar to the order on EPCF trees:
\begin{equation*}
tr_1 \leq tr_2 \quad \Longleftrightarrow \quad tr_1 \text{ can be obtained by replacing subtrees of } tr_2 \text{ by } \bot.
\end{equation*}
Using this order, we can give a domain theoretic definition of the tree associated to a closed computation. The tree might have infinite depth and width.

\begin{definition}[Computation trees for ECPS]\label{Def_ecps_tree_chain}
Define a family of maps 
\begin{equation*}
\treet{-}_{(-)}:(\vdash)\times\mathbb{N}\longrightarrow\textit{Trees}_\Sigma
\end{equation*}
\begin{align*}
&\treet{t}_0 = \bot	\\
&\treet{t}_{n+1} = \begin{cases}
					\treet{s}_n	&\text{if }t\longrightarrow s	\\
					\sigma_m(\treet{s[\overline{0}/x]}_n,\ldots,\treet{s[\overline{k}/x]}_n,\ldots)	&\text{if } t\xrightarrow{\sigma(\overline{m})}(s[\overline{k}/x])_{k\in\mathbb{N}}	\\
					\downarrow	&\text{if } t=\downarrow	\\
					\bot	&\text{otherwise}.
				   \end{cases}
\end{align*}
We can see that $\treet{t}_n \leq \treet{t}_{n+1}$ so we can define $\treet{-}:(\vdash)\longrightarrow\textit{Trees}_\Sigma$ as the least upper bound of the chain $\{\treet{t_n}\}_{n\in\mathbb{N}}$:  
\begin{equation*}
\treet{t} = \bigsqcup_{n\in\mathbb{N}}\treet{t}_n.
\end{equation*}
\end{definition}

\begin{example}[Pure functional computation]
In this case, the signature $\Sigma$ is empty and the only effect operation is $\downarrow$. Therefore, computation trees can only be leaves: $\downarrow$ for a computation that terminates and $\bot$ for one that does not. For example:
\begin{equation*}
\treet{loop} = \treet{(\mufix{f}{\lbd{x}{\mathtt{nat}}{(f\ x)}})\ \mathtt{zero}}=\bot.
\end{equation*} 
\end{example}

\begin{example}[Nondeterminism]
Define $\Sigma=\{or\}$. All effects have arity $\mathbbm{\mathtt{nat}}\times\alpha^\mathtt{nat}\rightarrow\alpha$. The intuitive interpretation of $or(v,x.t)$ is that it ignores $v$ and performs a nondeterministic choice between $t[\overline{0}/x]$ and $t[\overline{1}/x]$.

For example, consider the computation tree of:
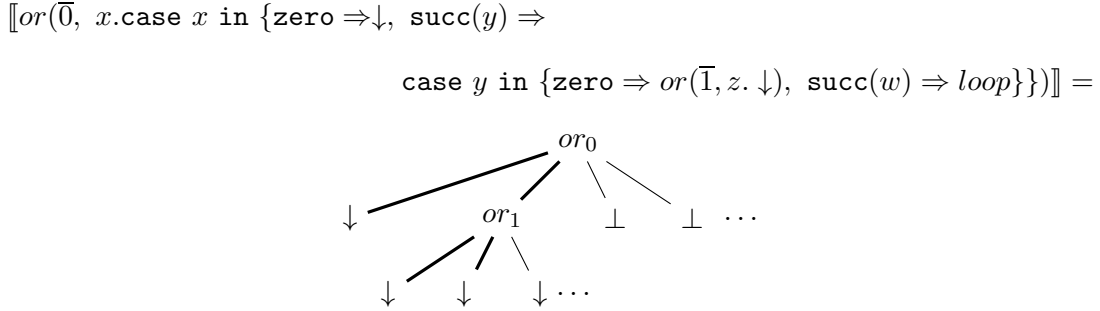
\begin{multline*}
\treet{or(\overline{0},\ x.\caset{x}{\downarrow}{y}{\\ \caset{y}{or(\overline{1},z.\downarrow)}{w}{loop}})}=
\end{multline*}
\begin{center}
\begin{tikzpicture}[level distance=1cm]
\node {$or_0$}
	child { node {$\downarrow$} edge from parent[very thick]}
	child[sibling distance=1cm] { node {$or_1$} edge from parent[very thick] 
		child { node {$\downarrow$} } 
		child { node {$\downarrow$} }
		child[thin] { node {$\downarrow$} }		 
		child[thin] { node[left=0.2cm] {$\ldots$} edge from parent[draw=none] } }
	child { node[right=0.2cm] {$\bot$} }
	child { node {$\bot$} }
	child { node[left=0.5cm] {$\ldots$} edge from parent[draw=none] };	
\end{tikzpicture}
\end{center}
As far as the intuitive interpretation of this computation is concerned, the indices $0$ and $1$ are irrelevant, and only the paths highlighted in bold can occur. However, $\treet{-}$ treats $or$ as uninterpreted syntax, hence the paths that can never occur are still present in the effect tree.
\end{example}

\begin{example}[Probabilistic choice]
Define $\Sigma=\{p\text{-}or\}$. Intuitively, the operation $p\text{-}or(v,x.t)$ chooses between $t[\overline{0}/x]$ and $t[\overline{1}/x]$ with probability $0.5$. The rest of the branches can never occur. The following computation:
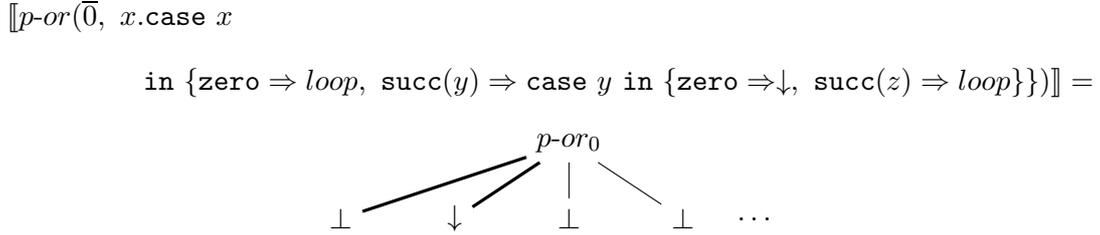
\begin{multline*}
\treet{p\text{-}or(\overline{0},\ x.\caset{x \\}{loop}{y}{\caset{y}{\downarrow}{z}{loop}})}=
\end{multline*}
\begin{center}
\begin{tikzpicture}[level distance=1cm]
\node {$p\text{-}or_0$}
	child { node {$\bot$} edge from parent[very thick]}
	child { node {$\downarrow$} edge from parent[very thick]}
	child { node {$\bot$} }
	child { node {$\bot$} }
	child { node[left=0.2cm] {$\ldots$} edge from parent[draw=none] };
\end{tikzpicture}
\end{center}
is analogous to the EPCF computation $p\text{-}or(loop,\return{\star})$ from Example \ref{Eg_prob_epcf_tree}, irrespective of the index of $p\text{-}or$, here $0$.
\end{example}

\begin{example}[Global store]
There is a finite set of locations $\mathbb{L}$ that can store natural numbers and $\Sigma=\{lookup_l,\ update_l \mid l\in\mathbb{L}\}$. The intuitive interpretation of $lookup_l(v,x.t)$ is that it ignores $v$, it looks up the value at location $l$, if this is $\overline{n}$ it continues with $t[\overline{n}/x]$.
For $update_l(v,x.t)$ the intuition is: write the number $v$ in location $l$ then continue with the computation $t[\overline{0}/x]$.

The EPCF computation tree from  Example \ref{Eg_globalstore_tree} can be adapted here:
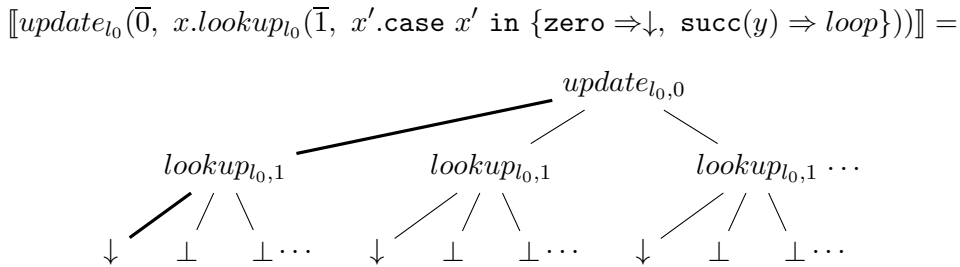
\begin{equation*}
\treet{update_{l_0}(\overline{0},\ x.lookup_{l_0}(\overline{1},\ x'.\caset{x'}{\downarrow}{y}{loop}))}=
\end{equation*}
\begin{center}
\begin{tikzpicture}[level distance=1.1cm, sibling distance=3.5cm]
\node {$update_{l_0,0}$}
	child { node {$lookup_{l_0,1}$} edge from parent[very thick]
		child[sibling distance=1cm] { node {$\downarrow$} }
		child[sibling distance=1cm] { node {$\bot$} edge from parent[thin]}
		child[sibling distance=1cm] { node {$\bot$} edge from parent[thin]}
		child[sibling distance=1cm] { node[left=0.2cm] {$\ldots$}  edge from parent[draw=none] } }
	child { node {$lookup_{l_0,1}$} 
		child[sibling distance=1cm] { node {$\downarrow$} }
		child[sibling distance=1cm] { node {$\bot$} }
		child[sibling distance=1cm] { node {$\bot$} }
		child[sibling distance=1cm] { node[left=0.2cm] {$\ldots$}  edge from parent[draw=none] } }
	child { node {$lookup_{l_0,1}$} 
		child[sibling distance=1cm] { node {$\downarrow$} }
		child[sibling distance=1cm] { node {$\bot$} }
		child[sibling distance=1cm] { node {$\bot$} }
		child[sibling distance=1cm] { node[left=0.2cm] {$\ldots$}  edge from parent[draw=none] } }
	child { node[left=2cm] {$\ldots$}  edge from parent[draw=none] };
\end{tikzpicture}
\end{center}
Only the path in bold can occur in the computation above.
\end{example}

\begin{example}[Input/output]
Define $\Sigma=\{read, write\}$. Intuitively, the computation $read(v,x.t)$ ignores $v$, accepts as input a number $\overline{n}$ and continues with $t[\overline{n}/x]$. The computation $write(v,x.t)$ writes $v$ to the output channel then continues with computation $t[\overline{0}/x]$.

Below is a computation that inputs a number then outputs it immediately. Only the paths in bold can occur:
\begin{equation*}
\treet{read(\overline{0},\ x.write(x,\ x'.\caset{x}{\downarrow}{y}{loop}))}=
\end{equation*}
\begin{center}
\begin{tikzpicture}[level distance=1.1cm, sibling distance=3.5cm]
\node {$read_{0}$}
	child { node {$write_{0}$} edge from parent[very thick]
		child[sibling distance=1cm] { node {$\downarrow$} }
		child[sibling distance=1cm] { node {$\downarrow$} edge from parent[thin]}
		child[sibling distance=1cm] { node {$\downarrow$} edge from parent[thin]}
		child[sibling distance=1cm] { node[left=0.2cm, black] {$\ldots$}  edge from parent[draw=none] } }
	child { node {$write_{1}$} edge from parent[very thick] 
		child[sibling distance=1cm] { node {$\bot$} }
		child[sibling distance=1cm] { node {$\bot$} edge from parent[thin]}
		child[sibling distance=1cm] { node {$\bot$} edge from parent[thin]}
		child[sibling distance=1cm] { node[left=0.2cm] {$\ldots$}  edge from parent[draw=none] } }
	child { node {$write_{2}$} edge from parent[very thick]
		child[sibling distance=1cm] { node {$\bot$} }
		child[sibling distance=1cm] { node {$\bot$} edge from parent[thin]}
		child[sibling distance=1cm] { node {$\bot$} edge from parent[thin]}
		child[sibling distance=1cm] { node[left=0.2cm] {$\ldots$}  edge from parent[draw=none] } }
	child { node[left=2cm] {$\ldots$}  edge from parent[draw=none] };
\end{tikzpicture}
\end{center}
\end{example}

\section{A Coinduction Proof Principle}\label{Sec_coind_princ}

Coinduction is a proof techniques that will be used in the following chapter. We give an abstract overview of a coinduction proof principle using the basic notions of a category, functor, terminal object and coalgebra, all of which can be found in an introduction to category theory such as \cite{AbrT17}. The following definition appears in \cite{JacR11}:

\begin{definition}\label{Def_bisim_abstract}
Let $T:\mathbf{Set}\Longrightarrow\mathbf{Set}$ be a functor. Take two $T$-coalgebras $(X,\ a_X:X\longrightarrow T(X))$ and $(Y,\ a_Y:Y\longrightarrow T(Y))$. A $T$-\emph{bisimulation} between $(X,\ a_X)$ and $(Y,\ a_Y)$ is a relation $\mathcal{R}\subseteq X\times Y$ for which there exists a $T$-coalgebra structure $g:\mathcal{R}\longrightarrow T(\mathcal{R})$ such that the two projection functions $\pi_1:\mathcal{R}\longrightarrow X$ and $\pi_2:\mathcal{R}\longrightarrow Y$ are $T$-coalgebra morphisms:

\begin{equation*}
\begin{tikzcd}
X  \arrow[d, "a_X"] \arrow[r, leftarrow, "\pi_1"] & \mathcal{R} \arrow[r, "\pi_2"]  \arrow[d, "g"] & Y \arrow[d, "a_Y"]	\\
T(X) \arrow[r, leftarrow, "T(\pi_1)"] & T(\mathcal{R})  \arrow[r, "T(\pi_2)"] & T(Y)
\end{tikzcd}
\end{equation*}
Use the following notation:
\begin{equation*}
x\cong y \quad \Longleftrightarrow \quad \text{there exists a } T\text{-bisimulation } \mathcal{R}\subseteq X\times Y \text{ with } (x,y)\in\mathcal{R}.
\end{equation*}
\end{definition}

We can now formulate the following coinduction proof principle:

\begin{proposition}[From \cite{JacR11}]\label{Prop_coind_princip}
Consider the final $T$-coalgebra $(Z,\ c:Z\longrightarrow T(Z))$, if it exists. Let $\alpha_X : X\longrightarrow Z$ and $\alpha_Y : Y\longrightarrow Z$ be coalgebra morphisms. For all $x\in X$ and $y\in Y$:
\begin{center}
if $x\cong y$ then $\alpha_X(x)=\alpha_Y(y)$.
\end{center}
\end{proposition}
\begin{proof}
It suffices to show that the following diagram commutes, where $\mathcal{R}\subseteq X\times Y$ is a bisimulation and $(x,y)\in\mathcal{R}$:
\begin{equation*}
\begin{tikzcd}
Z \arrow[d, "c"] \arrow[r, leftarrow, "\alpha_X"] & X  \arrow[d, "a_X"] \arrow[r, leftarrow, "\pi_1"] & \mathcal{R} \arrow[r, "\pi_2"]  \arrow[d, "g"] & Y \arrow[d, "a_Y"]	\arrow[r, "\alpha_Y"] & Z \arrow[d, "c"] \\
T(Z) \arrow[r, leftarrow, "T(\alpha_X)"] & T(X) \arrow[r, leftarrow, "T(\pi_1)"] & T(\mathcal{R})  \arrow[r, "T(\pi_2)"] & T(Y) \arrow[r, "T(\alpha_Y)"] & T(Z)
\end{tikzcd}
\end{equation*}
This is true because all the small squares commute. Then both $\alpha_X\circ\pi_1$ and $\alpha_Y\circ\pi_2$ are coalgebra morphisms into $Z$. By finality of $Z$, they must be equal.
\end{proof}

\section{Logical Relations}\label{Sec_log_rel_intro}

``Logical relations'' is a proof technique that involves defining a family of relations by induction on the types of a programming language. A relation for type $\tau$ contains only pairs of terms of type $\tau$. The relation is \emph{logical} if, given related functions $f_1$ and $f_2$ of type $\rho\rightarrow\tau$ and related arguments $x_1:\rho$ and $x_2:\rho$, the terms $f_1\ x_1$ and $f_2\ x_2$ are related.

Logical predicates, the unary version of logical relations, have been used to prove strong normalisation of the simply-typed $\lambda$-calculus \cite{DBLP:journals/jsyml/Tait67}, \cite[Chapter 6]{Gir89}. Other versions of logical relations have been used, for example, to characterise program equivalence  \cite{DBLP:conf/esop/Ahmed06} and to prove compiler correctness \cite{DBLP:conf/icfp/BentonH09}. These are examples of syntactic logical relations, based on the operational semantics of a language. These are the kinds of relations used in the next chapter to prove that a continuation-passing translation of EPCF into ECPS is correct. 

There are also logical relations based on denotational models (e.g.~\cite{DBLP:journals/iandc/Pitts96, MPF16}). One of their applications is proving that a denotational model of a programming language in computationally adequate. Adequacy, means that denotational equality is a sound technique for establishing contextual equivalence of programs.

The particular flavour of logical relations we will need is \emph{step-indexed biorthogonal} logical relations. As Jaber and Tabareau explain \cite{JabT11}, biorthogonality allows us to define which terms should be related by specifying their interaction with program contexts. 

A biorthogonal logical relation contains  a collection of relations on values, $\mathcal{V}_\tau$, for each type $\tau$. There is a collection of relations on program contexts defined using the relations on values:
\begin{equation*}
\mathcal{C}_\tau = \{(C_1,C_2) \mid \forall(v_1,v_2)\in\mathcal{V}_\tau.\ \mathfrak{O}(C_1[v_1],C_2[v_2])\}.
\end{equation*}
And finally, a collection of relations on terms defined using the relations on contexts:
\begin{equation*}
\mathcal{T}_\tau = \{(t_1,t_2) \mid \forall(C_1,C_2)\in\mathcal{C}_\tau.\ \mathfrak{O}(C_1[t_1],C_2[t_2])\}.
\end{equation*}
The notation $\mathfrak{O}(C_1[t_1],C_2[t_2])$ stands for an observation about programs $C_1[t_1]$ and $C_2[t_2]$. The notion of observation is chosen on a case-by-case basis, but it usually involves the reduction behaviour of the terms under consideration.

Step-indexing was introduced to deal with programming languages with recursion. This approach, instead of defining a single relation for type $\tau$, defines a family of relations for type $\tau$ indexed by natural numbers. Intuitively, terms in $\mathcal{T}_\tau^n$ are allowed to reduce at most $n$ steps. The natural number indices help to break the vicious circle introduced by recursion in proofs about the logical relation.

Step-indexing and biorthogonality can be combined in a straightforward way as explained for example by Pitts \cite{Pit10}. Because EPCF contains recursion, and because computations can only be evaluated in a stack, it is useful to use both step-indexing and biorthogonality when proving correctness of the translation from EPCF to ECPS. This will be explained in the next chapter.

\section{Chapter Summary}

This chapter introduced a new language ECPS in Sections~\ref{Sec_ecps_def} and \ref{Sec_ecps_oper_sem}. ECPS is a continuation-passing variant of EPCF which will be used in the rest of the dissertation to study program equivalence. 

Being a continuation-passing language means that functions receive an additional argument, which specifies how the computation should proceed once the function has terminated. For example, the successor function in ECPS is:
\begin{equation}
f = \lbd{(n,k)}{(\mathtt{nat},\neg\mathtt{nat})}{(k\ \mathtt{succ}(n))} : \neg(\mathtt{nat},\neg\mathtt{nat}). \label{Eq_ecps_succ_func}
\end{equation}
Here $k:\neg\mathtt{nat}$ is a \emph{continuation}, a function that is waiting for a result of type $\mathtt{nat}$, but is not expected to return.

Because everything is written in continuation-passing style, ECPS computations do not usually return. Therefore, the operational semantics maps an ECPS computation to a tree whose nodes are effect operations and leaves are either $\downarrow$, which signifies the termination effect, or $\bot$ for nontermination. For example, consider the tree of:
\begin{multline*}
m = (\lbd{(n,k)}{(\mathtt{nat},\neg\mathtt{nat})}{\ or(\overline{5},\ y.\caset{y}{(k\ \mathtt{succ}(n))}{y'}{\\ \caset{y'}{(k\ \overline{0})}{y''}{loop}})})\ (\overline{3},\ \lbd{x}{\mathtt{nat}}{\downarrow})
\end{multline*}
\begin{center}
\begin{tikzpicture}[level distance = 1cm, sibling distance=0.8cm]
\node {$or_5$}
	child { node {$\downarrow$} edge from parent[very thick]}
	child { node {$\downarrow$} edge from parent[very thick]}
	child { node {$\bot$} }
	child { node {$\bot$} }
	child { node[left=0.1cm] {$\ldots$} edge from parent[draw=none]};	
\end{tikzpicture}
\end{center}
In ECPS all effect operations have arity $\mathtt{nat}\times\alpha^\mathtt{nat}\rightarrow\alpha$, where $\alpha$ stands for a computation. Therefore, each node in a tree has a child for each natural number. In the case of $or$, the tree carries redundant information: only the paths in bold can occur during computation, and the index $5$ is irrelevant.

Finally, two proof methods were reviewed: coinduction (Section~\ref{Sec_coind_princ}) and logical relations (Section~\ref{Sec_log_rel_intro}), both of which will be used in the next chapter. 

\chapter{Comparison: ECPS vs.~EPCF}\label{Chap_rel-epcf-ecps}

This chapter presents a continuation-passing translation from the language PCF with effects (EPCF), introduced in Chapter \ref{Chap_expl_prob}, to its continuation-passing variant (ECPS), defined in Chapter \ref{Chap_descr_method}. This translation is proved correct using logical relations and coinduction. The last section is an informal argument for why translating ECPS into EPCF is not possible. The contents of this chapter justify the choice of the ECPS language, so the following chapters are only concerned with ECPS.

\section{CPS Translation}\label{Sec_cps_trans}

Preliminary ideas for the translation from EPCF to ECPS appear in a technical report by Lafont, Reus and Streicher \cite{LRS93}, but for simpler languages. In order to simplify the translation, we replace all the EPCF effect operations by operations with arity $\mathbbm{N}\times\alpha^\mathbbm{N} \rightarrow \alpha$, as follows: 
\begin{align*}
&\sigma(M_0,\ldots,M_{n-1})\ \text{becomes}\ \sigma(Z;W)	\\
&\sigma(V;M_0,\ldots,M_{n-1})\ \text{becomes}\ \sigma(V;W)	\\
&\sigma(V)\ \text{becomes}\ \sigma(Z;V)
\end{align*}
where
\begin{align*}
W=\lbd{x}{\mathbbm{N}}{\casess{x}{M_0}{x_1&}
								{\\ \casess{x_1}{M_1}{x_2&}{\ldots \\ \casess{x_{n-1}&}{M_{n-1}}{x_n}{loop} \ldots}}}	\\
loop = \letin{(\fix{\lbd{f}{\mathbbm{N}{\rightarrow}\tau}{\return{f}}})}{x&}{x\ Z}	\tag{where $\tau$ is the type of $M_i$}.	
\end{align*}
The function $W$ chooses between the arguments $M_0,\ldots,M_{n-1}$, just as $\sigma$ used to do. This simplification means that EPCF computation trees carry more redundant information. For example:

\begin{example}[Nondeterminism]
The computation tree of $?nat$

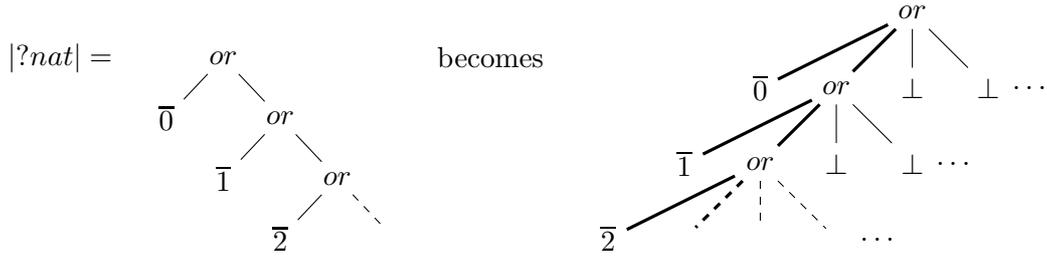
\begin{minipage}[t]{0.49\textwidth}
\begin{center}
\begin{tikzpicture}[level distance=0.8cm]
\node (tr) {$or$}
	child { node {$\overline{0}$} }
	child { node {$or$}
		child { node {$\overline{1}$} }
		child { node {$or$}
			child { node {$\overline{2}$} }
			child { node {} edge from parent[dashed] } } };
\node (comp) [left= 1cm of tr]  {$\trees{?nat}=$};
\node [right=4cm of comp] {becomes};			
\end{tikzpicture}
\end{center}
\end{minipage}\hfill
\begin{minipage}[t]{0.47\textwidth}
\begin{center}
\begin{tikzpicture}[level distance=1cm, sibling distance=1cm]
\node {$or$}
	child { node {$\overline{0}$} edge from parent[very thick] }
	child { node {$or$} edge from parent[very thick]
		child { node {$\overline{1}$} }
		child { node {$or$}
			child { node {$\overline{2}$} }
			child { node {} edge from parent[dashed] } 
			child[thin] { node {} edge from parent[dashed] }
			child[thin] { node {} edge from parent[dashed] } 
			child[thin] { node[left=0.1cm] {$\ldots$} edge from parent[draw=none] } } 
		child[thin] { node {$\bot$} }
		child[thin] { node {$\bot$} }
		child[thin] { node[left=0.1cm] {$\ldots$} edge from parent[draw=none] }	}
	child { node {$\bot$} }
	child { node {$\bot$} }
	child { node[left=0.1cm] {$\ldots$} edge from parent[draw=none] }	;	
\end{tikzpicture}
\end{center}
\end{minipage}
where only the paths in bold can occur. The definitions of modalities in $\mathcal{O}$ need to be adjusted to take this redundancy into account. For the effects considered in the dissertation, Scott-openness, decomposability and program equivalence are preserved.
\end{example}

The CPS translation for types is defined as follows:
\begin{align*}
\mathbbm{N}^* &= \mathtt{nat}	\\
\mathbbm{1}^* &= \mathtt{unit}	\\
(\rho\rightarrow\tau)^* &= \neg(\rho^*,\neg\tau^*). 
\end{align*}

To translate contexts, assume that all variable names from EPCF appear in ECPS:
\begin{align*}
\emptyset^* &=\emptyset	\\
(\Gamma,x:\tau)^* &= \Gamma^*,x:\tau^*.
\end{align*}

The CPS translation for values and computations is given in Figures \ref{Fig_CPS_trans} and \ref{Fig_CPS_trans_continued}. It relates values in context $\Gamma \vdash V:\tau$ to ECPS values $\Gamma^* \vdash V^*:\tau^*$. A function $F=\lbd{x}{\rho}{M}$ of type $\rho\rightarrow\tau$ is translated to $F^*=\lbd{(x,k)}{(\rho^*,\neg\tau^*)}{M^*\ k}$ of type $\neg(\rho^*,\neg\tau^*)$. The function $F^*$ is waiting for an argument of type $\rho^*$ and a continuation of type $\neg\tau^*$. The continuation is itself waiting for the result of computation $M$.

\begin{figure}
\begin{align*}
(\Gamma \vdash Z : \mathbbm{N})^* &= \Gamma^* \vdash \mathtt{zero} : \mathtt{nat} 	\\
(\Gamma \vdash S(V) : \mathbbm{N})^* &= \Gamma^* \vdash \mathtt{succ}(V^*) : \mathtt{nat}	\\
(\Gamma \vdash x : \tau)^* &= \Gamma^* \vdash x : \tau^*	\\
(\Gamma \vdash \star : \mathbbm{1})^* &= \Gamma^* \vdash \star : \mathtt{unit}	\\
(\Gamma \vdash \lbd{x}{\rho}{M} : \rho\rightarrow\tau)^* &= \Gamma^* \vdash \lbd{(x,k)}{(\rho^*,\neg\tau^*)}{(M^*\ k)} : \neg(\rho^*,\neg\tau^*)	\\
(\Gamma \vdash V\ W : \tau)^* &= \Gamma^* \vdash \lbd{k}{\neg\tau^*}{V^*\ (W^*,k)} : \neg\neg\tau^*	\\
(\Gamma \vdash \return{V} : \tau)^* &= \Gamma^* \vdash \lbd{k}{\neg\tau^*}{(k\ V^*)} : \neg\neg\tau^*	\\
(\Gamma \vdash \letin{M}{x}{N} : \tau)^* &= \Gamma^* \vdash \lbd{k}{\neg\tau^*}{M^*\ \lbd{x}{\rho^*}{(N^*\ k)}} : \neg\neg\tau^*	\\
(\Gamma \vdash \casess{V}{M}{x}{N} : \tau)^* &= 	\\
\Gamma^* \vdash\lbd{k}{\neg\tau^*}{\caset{V^*}{&M^*\ k}{x}{N^*\ k} : \neg\neg\tau^*}	\\
(\Gamma \vdash \sigma(V;W) : \tau)^* &= \Gamma^* \vdash \lbd{k}{\neg\tau^*}{\sigma(V^*,x.W^*(x,k))} : \neg\neg\tau^*
\end{align*}
\caption{CPS translation of EPCF into ECPS -- first part.} \label{Fig_CPS_trans}
\end{figure}

\begin{figure}
\begin{lstlisting}[mathescape=true]
$(\Gamma \vdash \mathbf{fix}\ F : \tau\rightarrow\rho)^* =$ 
$\Gamma^* \vdash \lambda k{:}\neg(\neg(\tau^*,\neg\rho^*)).$
   $(\mu x. \lambda c{:}\neg(\neg(\tau^*,\neg\rho^*)).$
       $c\ \lambda(x',k'){:}(\tau^*,\neg\rho^*).$
          $\lambda k''{:}\neg\rho^*.\lambda l{:}\neg(\neg(\tau^*,\neg\rho^*)).(F^*\ (\lambda(y,l'){:}(\tau^*,\neg\rho^*).$
                                  $\lambda p{:}\neg\rho^*.x\ (\lambda z{:}\neg(\tau^*,\neg\rho^*).(\lbd{p'}{\neg\rho^*}{z\ (y,p')})$
                                                     $p)$
                                   $l',$
                                 $l)$
                             $)$
                $\lambda w{:}\neg(\tau^*,\neg\rho^*).((\lambda l''{:}\neg\rho^*.w\ (x',l''))\ k'')$
          $k'$
   $)\ k$ 
   $: \neg\neg(\neg(\tau^*,\neg\rho^*))$
\end{lstlisting}
\caption{CPS translation of EPCF into ECPS -- continued.} \label{Fig_CPS_trans_continued}
\end{figure}

A computation $\Gamma\vdash M:\tau$ is translated to a function $\Gamma^*\vdash M^*:\neg\neg\tau^*$. This is because $M^*$ is waiting for a continuation of type $\neg\tau^*$ to which it can pass its result.

Computation $\return{V}:\tau$ becomes $\lbd{k}{\neg\tau^*}{k\ V^*}$. So value $V^*$ is passed to the current continuation. Application $V\ W:\tau$ becomes $\lbd{k}{\neg\tau^*}{V^*\ (W^*,k)}$, which means that $V^*$ is given arguments $W^*$ and a continuation that is waiting for the result of the application. Sequencing $\letin{M}{x}{N}:\tau$ is implemented as $\lbd{k}{\neg\tau^*}{M^*\ \lbd{x}{\rho^*}{(N^*\ k)}}$. First, $M^*$ is evaluated and its result is passed to the continuation $\lbd{x}{\rho^*}{(N^*\ k)}$ which executes $N^*$ and passes the result to $k$.

To translate effects, assume that the set of effect operations $\Sigma$ is the same in EPCF and ECPS. The translation of $\sigma(V;W):\tau$ is as expected, the function $W^*$ is passed a continuation in addition to argument $x$.

The translation of $(\fix{F}):\tau\rightarrow\rho$ is complicated by the reduction behaviour of $\fix{F}$. Recall that:
\begin{equation*}
\fix{F} \rightsquigarrow \return{\lbd{x}{\tau}{\letin{F\ (\lbd{y}{\tau}{\letin{\fix{F}}{z}{z\ y}})}{w}{w\ x}}}.
\end{equation*}
As a sanity check for the translation, it can be shown that
\begin{equation*}
(\fix{F})^*\ k\longrightarrow^* k\ (\lbd{x}{\tau}{\letin{F\ (\lbd{y}{\tau}{\letin{\fix{F}}{z}{z\ y}})}{w}{w\ x}})^*
\end{equation*}
which matches:
\begin{multline*}
(\return{\lbd{x}{\tau}{\letin{F\ (\lbd{y}{\tau}{\letin{\fix{F}}{z}{z\ y}})}{w}{w\ x}}})^*\ k \longrightarrow 	\\
k\ (\lbd{x}{\tau}{\letin{F\ (\lbd{y}{\tau}{\letin{\fix{F}}{z}{z\ y}})}{w}{w\ x}})^*.
\end{multline*}

\section{Typing and CPS Translation for Stacks}

In the CPS translation $(\return{V})^*=\lbd{k}{\neg\tau^*}{k\ V^*}$, the continuation $k$ plays a similar role to the stack $S\circ\letin{(-)}{x}{M},$ in the reduction:
\begin{equation*}
(S\circ\letin{(-)}{x}{M},\ \return{V}) \rightarrowtail (S,\ M[V/x]).
\end{equation*}
This suggests that EPCF stacks can be translated to ECPS continuations. Since the CPS translation is typed, we allow stacks to have free variables and introduce typing judgements for stacks. The typing judgement $\Gamma \vdash S : \tau\Rightarrow\rho$  says that by substituting a \emph{closed} computation $M$ of type $\tau$ for the hole in $S$ we obtain a computation $\Gamma\vdash S\{M\}:\rho$.
\begin{equation*}
\inferrule
	{ }
	{\Gamma \vdash id : \rho\Rightarrow\rho}
(\textsc{sid}) \qquad
\inferrule
	{\Gamma,x:\tau \vdash M :\tau' \\ \Gamma \vdash S : \tau'\Rightarrow\rho}
	{\Gamma \vdash S \circ (\letin{(-)}{x}{M}) : \tau\Rightarrow\rho}
(\textsc{slet})		
\end{equation*}

The substitution of closed values for free variables can be extended to stacks:
\begin{align*}
id[V/y] &= id	\\
(S \circ \letin{(-)}{x}{M}) &= S[V/y]\circ(\letin{(-)}{x}{M[V/y]}).
\end{align*}

Now EPCF stacks $\Gamma \vdash S : \tau\Rightarrow\rho$ can be translated to ECPS function values $\Gamma^* \vdash S^* : \neg\tau^*$, which are in fact continuations:
\begin{align*}
(\Gamma \vdash id :\rho\Rightarrow\rho)^* &= \Gamma^* \vdash \lbd{x}{\rho^*}{\downarrow}\ : \neg\rho^*	\\
(\Gamma \vdash S \circ (\letin{(-)}{x}{M}) : \tau\Rightarrow\rho)^* &= \Gamma^* \vdash (\lbd{x}{\tau^*}{M^*\ S^*}) : \neg\tau^*.
\end{align*}

The empty stack $id$ is equivalent to a continuation which given any value terminates. Here, we can see an important difference between EPCF and ECPS: the stack $id$ returns the value given to it, while $id^*$ discards the value so we cannot observe it.
The translation of $S \circ (\letin{(-)}{x}{M})$ is a continuation which, when given value $x$, executes $M^*$ and passes the result to $S^*$.

\section{Correctness of the CPS Translation}\label{Sec_cps_correct}

The goal of this section is to prove that EPCF computation trees  preserve their shape when translated into ECPS. This is stated as:

\begin{theorem}[The CPS translation is correct]\label{Thm_CPS_correct}
For any closed computation $M:\tau$ in EPCF and any stack $S:\tau\Rightarrow\rho$ the following holds:
\begin{equation*}
\treet{M^*\ S^*} = \trees{S,M}[\downarrow/l_1,\downarrow/l_2,\ldots].
\end{equation*}
That is, all the value leaves of $\trees{S,M}$ are replaced by $\downarrow$, but the nodes and the $\bot$-leaves stay the same.
\end{theorem}

\begin{example}[Probabilistic choice]
As explained in Section~\ref{Sec_cps_trans}, EPCF operation $p\text{-}or$ is replaced by its version with arity $\mathbbm{N}\times\alpha^\mathbbm{N}\rightarrow\alpha$. So computation $p\text{-}or(loop,\return{\overline{3}})$ becomes:
\begin{equation*}
p\text{-}or(Z;\ \lbd{x}{\mathbbm{N}}{\casess{x}{loop}{x'}{\casess{x'}{\return{\overline{3}}}{x''}{loop}}}).
\end{equation*}
However, we will still use notation $p\text{-}or(loop,\return{\overline{3}})$ for convenience, even when the version with arity $\mathbbm{N}\times\alpha^\mathbbm{N}\rightarrow\alpha$ is meant. Its computation tree is:
\begin{center}
\begin{tikzpicture}[level distance=1cm, sibling distance=0.8cm]
\node {$p\text{-}or_0$}
	child { node {$\bot$} edge from parent[very thick]}
	child { node {$\overline{3}$} edge from parent[very thick]}
	child { node {$\bot$} }
	child { node {$\bot$} }
	child { node[left=0.05cm] {$\ldots$} edge from parent[draw=none] };
\node (tr) [left=1cm] {$\trees{id,\ p\text{-}or(loop,\return{\overline{3}})}=$}	;
\end{tikzpicture}
\end{center}
According to Theorem \ref{Thm_CPS_correct}, when translated to ECPS, this computation tree becomes:
\begin{lstlisting}[mathescape=true, escapechar = !]
$\treet{(p\text{-}or(loop,\return{\overline{3}}))^*\ id^*}=$
  $\llbracket(\lambda k{:}\mathtt{nat}.\ p\text{-}or(\overline{0},\ y.(\lambda(x,k'){:}(\mathtt{nat},\neg\mathtt{nat}).$
                  $(\lambda k''{:}\mathtt{nat}.\ \mathtt{case}\ x\ \mathtt{in}\ \{\mathtt{zero}\Rightarrow loop^*\ k'',$
                     $\mathtt{succ}(x')\Rightarrow(\lambda k'''{:}\mathtt{nat}.\ \mathtt{case}\ x'\ \mathtt{in}\ \{\mathtt{zero}\Rightarrow(\lambda l{:}\mathtt{nat}.l\ \overline{3})\ k''', $
                                 $\mathtt{succ}(x'')\Rightarrow loop^*\ k'''\}$
                              $)\ k''\}$
                  $)\ k'$
                $)\ (y,k)$
            $)$
  $)\ (\lbd{x}{\mathtt{nat}}{\downarrow})\rrbracket=$
\end{lstlisting}
\begin{center}
\begin{tikzpicture}[level distance=1cm, sibling distance=0.8cm, auto]
\begin{scope}[transform canvas={xshift = -3.5cm, yshift=1.1cm}]
\node {$p\text{-}or_0$}
	child { node {$\bot$} edge from parent[very thick]}
	child { node {$\downarrow$} edge from parent[very thick]}
	child { node {$\bot$} }
	child { node {$\bot$} }
	child { node[left=0.05cm] {$\ldots$} edge from parent[draw=none] };
\end{scope}	
\end{tikzpicture}
\end{center}
\end{example}

In order to prove Theorem \ref{Thm_CPS_correct}, we introduce a step-indexed biorthognal logical relation, as discussed in Section \ref{Sec_log_rel_intro}. The purpose of the logical relation is to relate EPCF and ECPS computations that have similar reduction behaviour. Because EPCF computations reduce in a stack, it is useful to use biorthogonality, and have a relation between stacks and continuations as well. Step-indexing is needed to deal with the $\mathbf{fix}$ constructor.

The notion of observation $\mathfrak{O}((S,M),\ t)$ that the logical relation uses is a \emph{step-indexed similarity} relation $\mathcal{S}$: 
\begin{equation*}
\mathcal{S}^n_{\tau,\rho} \subseteq (Stack(\tau,\rho)\times Comp(\tau)) \times (\vdash)
\end{equation*}
defined for each pair of EPCF types $\tau,\rho$ and each $n\in\mathbb{N}$. This is not related to the applicative similarity discussed before because it contains pairs of programs from two different languages.  However, we will show that $\mathcal{S}$ is related to the abstract notion of bisimulation from Section \ref{Sec_coind_princ}. First, we recall the definition of a step-indexed relation from \cite{Pit10} and then define step-indexed similarity:

\begin{definition}\label{Def_step_indexed}
A step-indexed relation on the set $X$ is an $\mathbb{N}$-indexed family of sets:
\begin{equation*}
R = (R_n \mid n\in\mathbb{N}) \text{ satisfying } X \supseteq R_0 \supseteq R_1 \supseteq R_2 \supseteq\ldots.
\end{equation*}
\end{definition}

\begin{definition}\label{Def_step_indexed_sim}
The step-indexed similarity relation $\mathcal{S}$ is the family of greatest relations such that $(S,M)\mathrel{\mathcal{S}^n_{\tau,\rho}} t$ implies:
\begin{enumerate}
\item \label{Def_step_indexed_sim_1} $\forall p<n.\ ((S,M)\rightarrowtail^p(S',\sigma(V;W)) \implies t\longrightarrow^*\sigma(v,x.t')$ where $V=\overline{m}$ and $v=\overline{m}$ and $\forall l\in\mathbb{N}.\ (S',W\ \overline{l})\mathrel{\mathcal{S}^{n-p}_{\tau',\rho}} t'[\overline{l}/x])$.
\item \label{Def_step_indexed_sim_2} $\forall p\leq n.\ ((S,M)\rightarrowtail^p(id,\return{V})$ for some $V:\rho \implies \tau\longrightarrow^*\downarrow)$.
\end{enumerate} 
\end{definition}


The statement $(S,M)\mathrel{\mathcal{S}^n_{\tau,\rho}} t$ means that the ECPS computation $t$ simulates the reduction of the EPCF configuration $(S,M)$ for $n$ steps. The novel feature of this notion of observation is that it tracks effect operations, not only the termination behaviour of related programs, as compared to previous work on logical relations \cite{DBLP:conf/esop/Ahmed06, DBLP:conf/icfp/BentonH09, Pit10, JabT11}.

\begin{example}[Probabilistic choice]
For example, the ECPS computation:
\begin{multline*}
t=(\lbd{k}{\neg\mathtt{nat}}{\ p\text{-}or(\overline{0},\ x.\caset{x}{k\ \overline{2} \\}{x'}{\caset{x'}{k\ \overline{3}}{x''}{loop}})})\ (\lbd{x}{\mathtt{nat}}{\downarrow})
\end{multline*}
with computation tree
\begin{center}
\begin{tikzpicture}[level distance=1cm, sibling distance=1cm]
\node {$p\text{-}or_0$}
	child { node {$\downarrow$} edge from parent[very thick]}
	child { node {$\downarrow$} edge from parent[very thick]}
	child { node {$\bot$} }
	child { node {$\bot$} }
	child { node {$\ldots$} edge from parent[draw=none]};
\end{tikzpicture}
\end{center}
and EPCF computation $M = p\text{-}or(loop,\return{\overline{3}})$ satisfy $(id, M)\mathrel{\mathcal{S}^n_{\mathbbm{N},\mathbbm{N}}} t$ for any $n>0$.
\end{example}

Intuitively, as the index $n$ increases, step-indexed similarity checks for more and more matching reduction steps. Therefore, $\mathcal{S}$ respects the definition of a step-indexed relation. Below is a useful fact about step-indexed similarity which will be used later. The proof appears in Appendix \ref{App_cps}.

\begin{lemma}\label{Lem_sim_eval}
If $(S,M)\rightarrowtail^k(S',M')$ and $t\longrightarrow^*t'$ where $k<n$ then:
\begin{equation*}
(S,M)\mathrel{\mathcal{S}^n_{\tau,\rho}} t \iff (S',M')\mathrel{\mathcal{S}^{n-k}_{\tau,\rho}} t'.
\end{equation*}
\end{lemma}

We can define a well-typed relation $\mathcal{S}$ which encodes the fact that an ECPS computation $t$ simulates configuration $(S,M)$ for \emph{any} number of reduction steps. Step-indexed similarity implies this new notion of similarity.

\begin{definition}\label{Def_sim}
Define similarity as the family $\mathcal{S}$ of greatest relations indexed by EPCF types $\tau$ and $\rho$ such that $(S,M)\mathrel{\mathcal{S}_{\tau,\rho}} t$ implies:
\begin{enumerate}
\item $(S,M)\rightarrowtail^*(S',\sigma(V;W)) \implies t\longrightarrow^*\sigma(v,x.t')$ where $V=\overline{m}$ and $v=\overline{m}$ and $\forall l\in\mathbb{N}.\ (S',W\ \overline{l})\mathrel{\mathcal{S}_{\tau',\rho}} t'[\overline{l}/x]$.
\item $(S,M)\rightarrowtail^*(id,\return{V})$ for some $V:\rho \implies \tau\longrightarrow^*\downarrow$.
\end{enumerate}
\end{definition}

\begin{lemma}\label{Lem_step_simp_impl_sim}
For any configuration $(S,M)$ and closed computation $t$:
\begin{equation*}
(\forall n\in\mathbb{N}.\ (S,M)\mathrel{\mathcal{S}^n_{\tau,\rho}} t) \implies (S,M)\mathrel{ \mathcal{S}_{\tau,\rho}} t.
\end{equation*}
\end{lemma}
\begin{proof}
The strategy is to show $\bigcap_{n\in\mathbb{N}}\mathcal{S}^n$ is a simulation relating $(S,M)$ and $t$, and thus it is included in $\mathcal{S}$. The rest of the proof appears in Appendix \ref{App_cps}.
\end{proof}

Now we are ready to define the logical relation. It is in fact a collection of relations $(\mathcal{R}^{\mathfrak{v},n}_{\tau},\mathcal{R}^{\mathfrak{c},n}_\tau,\mathcal{R}^{\mathfrak{s},n}_{\tau,\rho})$ for all EPCF types $\tau,\rho$ and $n\in\mathbb{N}$.
\begin{align*}
\mathcal{R}^{\mathfrak{v},n}_{\tau} &\subseteq Val(\tau)\times(\vdash\tau^*)	\\
\mathcal{R}^{\mathfrak{c},n}_\tau &\subseteq Comp(\tau)\times(\vdash\neg\neg\tau^*)	\\
\mathcal{R}^{\mathfrak{s},n}_{\tau,\rho} &\subseteq Stack(\tau,\rho)\times(\vdash\neg\tau^*).
\end{align*}
It is defined by well-founded induction on the natural numbers and on EPCF types, following a tutorial paper by Pitts \cite{Pit10} which deals with the untyped $\lambda$-calculus. 

\begin{definition}[Logical Relation]\label{Def_log_rel}
\begin{gather*}
\mathcal{R}^{\mathfrak{v},n}_{\mathbbm{N}} = \{ (V,v) \mid V=\overline{m} \text{ and } v=\overline{m} \text{ for some } m\in\mathbb{N} \}	\\
\mathcal{R}^{\mathfrak{v},n}_{\mathbbm{1}} = \{ (V,v) \mid V=\star \text{ and } v=\star \}	\\
\mathcal{R}^{\mathfrak{v},n}_{\tau{\rightarrow}\rho} = \{ (V,v) \mid \forall p<n.\ \forall(V_1,v_1)\in\mathcal{R}^{\mathfrak{v},p}_{\tau}.\ (V\ V_1,\ \lbd{k}{\neg\rho^*}{v\ (v_1,k)})\in\mathcal{R}^{\mathfrak{c},p}_{\rho} \}	\\
\mathcal{R}^{\mathfrak{s},n}_{\tau,\rho} = \{ (S,k) \mid p\leq n.\ \forall(V,v)\in\mathcal{R}^{\mathfrak{v},p}_{\tau}.\ (S,\return{V})\mathrel{\mathcal{S}^p_{\tau,\rho}} (k\ v) \}	\\
\mathcal{R}^{\mathfrak{c},n}_{\tau} = \{ (M,v) \mid \forall p\leq n.\ \forall(S,k)\in\mathcal{R}^{\mathfrak{s},p}_{\tau,\rho}.\ (S,M)\mathrel{\mathcal{S}^p_{\tau,\rho}} (v\ k) \}.
\end{gather*}
These relations satisfy the definition of a step-indexed relation at each type, $\mathcal{R}^{n+1}_\tau\subseteq\mathcal{R}^n_\tau$, because, as the index of $X\mathrel{\mathcal{R}^n_\tau} x$ increases, the pair $(X,x)$ needs to satisfy more conditions.
\end{definition}

The logical relation is defined for closed values, stacks and computations and makes use of the step-indexed similarity relation $\mathcal{S}^n$. As expected, related functions map related values to related computations. A stack and a continuation are related if, given any related values, they have a similar reduction behaviour. Finally, computations are related if they have similar reduction behaviour in all related stack-continuation pairs.

We can extend the logical relation to open terms and stacks by substituting in related values. The definition below generalises over all indices:
\begin{multline*}
\overrightarrow{x_i:\tau_i} \vdash V\mathrel{\mathcal{R}^{\mathfrak{v}}_{\rho}} v \iff \overrightarrow{x_i:\tau_i} \vdash V :\rho \text{ and } (\overrightarrow{x_i:\tau_i})^* \vdash v : \rho^*	\text{ and }	\\
(\forall n\in\mathbb{N}.\ \forall\overrightarrow{(V_i,v_i):\tau_i}.\ (\overrightarrow{(V_i,v_i)\in\mathcal{R}^{\mathfrak{v},n}_{\tau_i}} \implies (V[\overrightarrow{V_i/x_i}],v[\overrightarrow{v_i/x_i}])\in\mathcal{R}^{\mathfrak{v},n}_{\rho}))
\end{multline*}
\begin{multline*}
\overrightarrow{x_i:\tau_i} \vdash M\mathrel{\mathcal{R}^{\mathfrak{c}}_{\rho}} v \iff \overrightarrow{x_i:\tau_i} \vdash M :\rho \text{ and } (\overrightarrow{x_i:\tau_i})^* \vdash v : \neg\neg\rho^*	\text{ and }	\\
(\forall n\in\mathbb{N}.\ \forall\overrightarrow{(V_i,v_i):\tau_i}.\ (\overrightarrow{(V_i,v_i)\in\mathcal{R}^{\mathfrak{v},n}_{\tau_i}}\implies (M[\overrightarrow{V_i/x_i}],v[\overrightarrow{v_i/x_i}])\in\mathcal{R}^{\mathfrak{c},n}_{\rho}))
\end{multline*}
\begin{multline*}
\overrightarrow{x_i:\tau_i} \vdash S\mathrel{\mathcal{R}^{\mathfrak{s}}_{\rho,\rho'}} k \iff \overrightarrow{x_i:\tau_i} \vdash S :\rho\Rightarrow\rho' \text{ and } (\overrightarrow{x_i:\tau_i})^* \vdash k : \neg\rho^*	\text{ and }	\\
(\forall n\in\mathbb{N}.\ \forall\overrightarrow{(V_i,v_i):\tau_i}.\ (\overrightarrow{(V_i,v_i)\in\mathcal{R}^{\mathfrak{v},n}_{\tau_i}} \implies (S[\overrightarrow{V_i/x_i}],k[\overrightarrow{v_i/x_i}])\in\mathcal{R}^{\mathfrak{s},n}_{\rho,\rho'})).
\end{multline*}

An important property of the logical relation is reflexivity. Since $\mathcal{R}$ relates EPCF terms to ECPS terms the property is formulated as below:

\begin{lemma}[Fundamental property of the logical relation]\label{Lem_fund_prop}
For any value $\overrightarrow{x_i:\tau_i} \vdash V :\rho$, any computation $\overrightarrow{x_i:\tau_i} \vdash M :\rho$ and any stack $\overrightarrow{x_i:\tau_i} \vdash S :\rho\Rightarrow\rho'$ in EPCF:
\begin{enumerate}
\item $\overrightarrow{x_i:\tau_i} \vdash V\mathrel{\mathcal{R}^\mathfrak{v}_\rho} V^*$.
\item $\overrightarrow{x_i:\tau_i} \vdash M\mathrel{\mathcal{R}^\mathfrak{c}_\rho} M^*$.
\item $\overrightarrow{x_i:\tau_i} \vdash S\mathrel{\mathcal{R}^\mathfrak{s}_{\rho,\rho'}} S^*$.
\end{enumerate}
\end{lemma}
\begin{proof}
The proof is by induction on the typing derivations of $V$, $M$ and $S$. The most interesting case is $M=\fix{W}$. It uses the indices of the logical relation to reason about the number of times the fixed point of $W$ is unfolded. All the cases appear in Appendix~\ref{App_cps}.
\end{proof}

The similarity relation $\mathcal{S}^n$ expresses the fact that computation $t$ simulates configuration $(S,M)$ for $n$ steps. Analogously, we can define a notion of similarity $\mathcal{S'}^n$ which says that $(S,M)$ simulates $t$ for $n$ steps. This will be a family of relations parametrised by $n\in\mathbb{N}$, and EPCF types $\tau$ and $\rho$: 
\begin{equation*}
\mathcal{S'}^n_{\tau,\rho} \subseteq (Stack(\tau,\rho)\times Comp(\tau)) \times (\vdash).
\end{equation*} 

\begin{definition}
The step-indexed similarity relation $\mathcal{S'}$ is the  family of greatest relations such that $(S,M)\mathrel{\mathcal{S'}^n_{\tau,\rho}} t$ implies:
\begin{enumerate}
\item $\forall p<n.\ (t\longrightarrow^p\sigma(v,x.t')  \implies (S,M)\rightarrowtail^*(S',\sigma(V;W))$ where $V=\overline{m}$ and $v=\overline{m}$ and $\forall l\in\mathbb{N}.\ (S',W\ \overline{l})\mathrel{\mathcal{S'}^{n-p}_{\tau',\rho}} t'[\overline{l}/x])$.
\item $\forall p\leq n.\ (t\longrightarrow^p\downarrow\implies (S,M)\rightarrowtail^*(id,\return{V})$ for some $V:\rho)$.
\end{enumerate} 
\end{definition}

Using $\mathcal{S'}^n$, we can define a family of relations $\mathcal{S'}$ indexed by types $\tau$ and $\rho$, such that $(S,M)\mathrel{\mathcal{S'}}t$ says that $(S,M)$ simulates $t$ for any number of steps. The definition is similar to that of $\mathcal{S}$ (Definition~\ref{Def_sim}). The analogous of Lemma \ref{Lem_step_simp_impl_sim} can be proved for $\mathcal{S'}^n$ and $\mathcal{S'}$.

Finally, we can define a logical relation $\mathcal{R'}$ exactly as $\mathcal{R}$ was defined but using $\mathcal{S'}$ instead of $\mathcal{S}$. The logical relation $\mathcal{R'}$ also has the fundamental property.

The next step is proving that the trees of a configuration and a computation which simulate each other are closely related. The proof is done by coinduction. We first discuss some properties of the sets of closed ECPS computations and EPCF stack-computation pairs.

Consider the functor:
\begin{equation*}
T : \mathbf{Set}\Longrightarrow\mathbf{Set} \text{ where } T(X)=\{\downarrow,\bot\}+(\Sigma\times\mathbb{N}\times X^\mathbb{N}).
\end{equation*}
Consider the following coalgebra for this functor: $(\textit{Trees}_\Sigma,c:\textit{Trees}_\Sigma\longrightarrow T(\textit{Trees}_\Sigma))$ where
\begin{equation*}
c(tr) = \begin{cases}
		\downarrow	&	\text{if } tr=\downarrow	\\
		\bot		&	\text{if } tr=\bot	\\
		(\sigma,n,\overrightarrow{tr'})	&	\text{if } tr=\sigma_n(\overrightarrow{tr'}).
	   \end{cases}
\end{equation*}
It is a standard result that the category of coalgebras and coalgebra morphisms for the functor $T$ has a terminal object. This is shown for example by Jacobs \cite[Theorem~2.3.9]{Jac16}. Moreover, the terminal object is $(\textit{Trees}_\Sigma,c)$. The proof of this is very similar to the proof of Proposition 2.3.5 from \cite{Jac16}.

\begin{lemma}\label{Lem_ecps_tree_coalgmor}
The function $\treet{-}:(\vdash)\longrightarrow \textit{Trees}_\Sigma$  is a coalgebra morphism in the category of coalgebras for the functor $T$.
\end{lemma}
\begin{proof}
First, note that $(\vdash)$ is indeed a coalgebra by considering the following function $a:(\vdash)\longrightarrow T(\vdash)$ on closed ECPS computations:
\begin{equation*}
a(t) = \begin{cases}
		\downarrow	&	\text{if } t\longrightarrow^*\downarrow	\\
		\bot	&	\text{if } t\longrightarrow^\infty	\\		
		(\sigma,k,\overrightarrow{t'[\overline{n}/x]})	&	\text{if } t\longrightarrow^*\sigma(\overline{k},x.t').
	   \end{cases}
\end{equation*}
By definition of $\longrightarrow$ exactly one of the cases above will occur, so $a$ is a well-defined function.

To prove that $\treet{-}$ is a coalgebra morphism, we use standard domain theoretic techniques. All the necessary results can be found for example in \cite{MPF16}. The proof involves making the order $\leq$ on $\textit{Trees}_\Sigma$ more precise, and defining an order $\leq_T$ on $T(\textit{Trees}_\Sigma)$ which makes it an $\omega$-CPO. The full development can be found in Appendix \ref{App_cps}.
\end{proof}

Define the following function indexed by EPCF types $\tau$ and $\rho$:
\begin{equation*}
\trees{-,-}_{(-)}^*:Stack(\tau,\rho)\times Comp(\tau)\times \mathbb{N} \longrightarrow \textit{Trees}_\Sigma.
\end{equation*}
similarly to how $\trees{-,-}_{(-)}$ was defined. Assuming all EPCF effect operations have been replaced with operations of arity $\mathbbm{N}\times\alpha^\mathbbm{N}\rightarrow\alpha$, the definition is:
\begin{align*}
\trees{S,M}_0^* &= \bot	\\
\trees{S,M}^*_{n+1} &= \begin{cases}
					\downarrow	&\text{if } S=id \text{ and } M=\return{V}	\\
					\trees{S',M'}_n^*	&\text{if } (S,M)\rightarrowtail(S',M')	\\
					\sigma_k(\trees{S,V\ \overline{0}}_n^*, \trees{S,V\ \overline{1}}_n^*,\ldots)	&\text{if } \sigma:\mathbbm{N}\times\alpha^\mathbbm{N} \rightarrow \alpha \text{ and } M=\sigma(\overline{k};V) 	\\
					\bot	&\text{otherwise}.	
					\end{cases}
\end{align*}
The tree $\trees{S,M}_{n}^*$ is different from $\trees{S,M}_{n}$ because all value leaves are replaced by $\downarrow$. We can see that $\trees{S,M}_n^*\leq \trees{S,M}_{n+1}^*$ in $\textit{Trees}_\Sigma$ so we can define:
\begin{gather*}
\trees{-,-}^*:Stack(\tau,\rho)\times\textit{Comp}(\tau)\longrightarrow \textit{Trees}_\Sigma	\\
\trees{S,M}^* = \bigsqcup_{n\in\mathbb{N}}\trees{S,M}_n^*.
\end{gather*}

Define the set of all well-formed stack-computation pairs as:
\begin{equation*}
Stack\times Comp = \{(S,M) \mid S\in Stack(\tau,\rho),\ M\in Comp(\tau) \text{ for some EPCF types }\tau,\rho\}
\end{equation*}
and extend the function $\trees{-,-}^*$ to this set:
\begin{gather*}
\beta: Stack\times Comp \longrightarrow \textit{Trees}_\Sigma	\\
\beta (S,M)= \trees{S,M}^*.
\end{gather*}

\begin{lemma}\label{Lem_epcf_tree_colagmor}
The function $\beta:Stack\times Comp\longrightarrow \textit{Trees}_\Sigma$  is a coalgebra morphism in the category of coalgebras for the functor $T$.
\end{lemma}
\begin{proof}
To see that $Stack\times Comp$ is indeed a coalgebra, consider the following function $(Stack\times Comp,\ b:Stack\times Comp\longrightarrow T(Stack\times Comp))$:
\begin{equation*}
b(S,M) = \begin{cases}
			\downarrow	&	\text{if } (S,M)\rightarrowtail^*(id,\return{V})	\\
			\bot	&	\text{if } (S,M)\rightarrowtail^\infty	\\			
			(\sigma,k,((S',V\overline{0}),(S',V\overline{1}),\ldots))	&	\text{if } (S,M)\rightarrowtail^*(S',\sigma(\overline{k};V)).
		 \end{cases}
\end{equation*}
By definition of $\rightarrowtail$ the cases above are exhaustive, and by determinacy only one of them can occur, so $b$ is a well-defined function.

The proof that $\beta$ is a coalgebra morphism is very similar to the proof for $\treet{-}$ so we omit it. It uses the $\omega$-CPO structure of $\textit{Trees}_\Sigma$ and $T(\textit{Trees}_\Sigma)$ to show that the coalgebra morphism diagram commutes.
\end{proof}

\begin{proposition}\label{Prop_bisim_impl_trees}
For any well-typed EPCF configuration $(S,M)$, where $S:\tau\Rightarrow\rho$, and any ECPS computation $t$:
\begin{equation*}
((S,M),t) \in \mathcal{S}_{\tau,\rho} \cap \mathcal{S'}_{\tau,\rho} \implies \treet{t} = \trees{S,M}[\downarrow/l_1,\downarrow/l_2,\ldots].
\end{equation*}
\end{proposition}
\begin{proof}
Apply the coinduction proof principle from Proposition~\ref{Prop_coind_princip}, use Lemmas~\ref{Lem_ecps_tree_coalgmor} and \ref{Lem_epcf_tree_colagmor} and the fact that $(\textit{Trees$_\Sigma$}, c:\textit{Trees$_\Sigma$}\longrightarrow T(\textit{Trees$_\Sigma$}))$ is a final coalgebra. The full proof appears in Appendix~\ref{App_cps}.
\end{proof}

Finally, we can prove the correctness of the CPS translation:

\begin{reptheorem}{Thm_CPS_correct}
For any closed computation $M:\tau$ in EPCF and any stack $S:\tau\Rightarrow\rho$ the following holds:
\begin{equation*}
\treet{M^*\ S^*} = \trees{S,M}[\downarrow/l_1,\downarrow/l_2,\ldots].
\end{equation*}
That is, all the value leaves of $\trees{S,M}$ are replaced by $\downarrow$, but the nodes and the $\bot$-leaves stay the same.
\end{reptheorem}
\begin{proof}
From the fundamental property of the logical relation (Lemma~\ref{Lem_fund_prop} and its analogue for $\mathcal{R'}$) we know that for any closed configuration $(S,M)$, where $S:\tau\Rightarrow\rho$:
\begin{equation*}
\forall n\in\mathbb{N}.\ ((S,M),(M^*\ S^*)) \in \mathcal{S}_{\tau,\rho}^n \cap \mathcal{S'}_{\tau,\rho}^n.
\end{equation*}

Then using Lemma \ref{Lem_step_simp_impl_sim} and its analogue for $\mathcal{S'}$ we can deduce that:
\begin{equation*}
((S,M),(M^*\ S^*)) \in \mathcal{S}_{\tau,\rho} \cap \mathcal{S'}_{\tau,\rho}.
\end{equation*}

Applying Proposition \ref{Prop_bisim_impl_trees} we know that:
\begin{equation*}
\treet{M^*\ S^*} = \lvert S,M \rvert[\downarrow/l_1,\downarrow/l_2,\ldots].
\end{equation*}
\end{proof}

\section{ECPS Is More Expressive than EPCF}\label{Sec_callcc}

From the point of view of the propositions-as-types correspondence, or the Curry-Howard isomorphism, the $\lambda$-calculus corresponds to intuitionistic logic \cite[Chapter 6]{SorU06}. Griffin showed that a language with \emph{control operators} can extend this correspondence to classical logic \cite{DBLP:conf/popl/Griffin90}.

One such control operator is \texttt{call-cc} (call-with-current-continuation) from the programming language Scheme. Griffin showed that \texttt{call-cc} can be assigned the type of Peirce's law, a proposition which is only provable classically:
\begin{equation*}
((A\rightarrow B) \rightarrow A)\rightarrow A.
\end{equation*}

It is known that PCF and EPCF do not contain any control operators but we will show that ECPS does. In ECPS we can write a term which has the type of Peirce's law and the reduction behaviour of \texttt{call-cc}. In this sense, ECPS is more expressive than EPCF. Below is an informal explanation of \texttt{call-cc} and how it arises in ECPS.

One way of adding \texttt{call-cc} to PCF \cite{DBLP:conf/icalp/RieckeT99} is to add a new type $Cont(A)$ which stands for a continuation waiting for a term of type $A$. Such a continuation can be written as $\gamma x{:}A.E[x]$ where $E$ is a PCF evaluation context and $E[x]$ is a PCF term. Then add the following term constructors to PCF:
\begin{equation*}
\inferrule{ }{\Gamma \vdash \mathtt{callcc}:(Cont(A)\rightarrow A)\rightarrow A} \qquad
\inferrule{ }{\Gamma \vdash \mathtt{throw}:Cont(A)\rightarrow A \rightarrow B}
\end{equation*}

The reduction rules for these new constructors are:
\begin{align}
&E[\mathtt{callcc}\ (\lbd{x}{Cont(A)}{M})] \longrightarrow E[M[x\coloneqq (\gamma y{:}A.E[y])]]	\label{eq_callcc_1}\\
&E'[\mathtt{throw}\ (\gamma y{:}A.E[y])\ V] \longrightarrow E[V]. \label{eq_calcc_2}
\end{align}
The first rule says that when $\mathtt{callcc}$ is encountered the current evaluation context $E$ is bound to the continuation $x=\gamma y{:}A.E[y]$. Then evaluation of $M$ proceeds normally. If $M$ never throws, the control flow is not changed, the return value of $\mathtt{callcc}$ was $M$. If at some point $M$ invokes $\mathtt{throw}\ (\gamma y{:}A.E[y])\ V$, the current evaluation context $E'$ is abandoned. Instead, the context $E$ is restored with value $V$. This looks as if the return value of $\mathtt{callcc}$ was $V$.

The type $Cont(A)$ can be interpreted as $A\rightarrow B$ for some type $B$. Thus $\mathtt{callcc}$ has type $((A\rightarrow B)\rightarrow A)\rightarrow A$. And $\mathtt{throw}$ has type $(A\rightarrow B)\rightarrow A \rightarrow B$. When the continuation $(\gamma y{:}A.E[y]):Cont(A)$ is thrown, the current context $E'$ that is waiting for type $B$ is discarded. Therefore, $B$ can be anything.

In ECPS, the PCF type $A\rightarrow B$ is interpreted as $\neg(A,\neg B)$. In other words, the implication $A\rightarrow B$ is expressed as $\neg(A\land\neg B)$. Peirce's law then becomes:
\begin{equation*}
\neg(\ \mathbf{\neg(\neg(A\land\neg B)\land\neg A)}\ \land\ \neg A).
\end{equation*}
In ECPS, a term of this type is:
\begin{equation*}
\mathtt{callcc}^* = \lbd{(f,k)}{(\mathbf{\neg(\neg(A, \neg B),\ \neg A)},\ \neg A)}{\ f\ (\lbd{(y,k')}{(A, \neg B)}{k\ y},\ k)}.
\end{equation*}
A term with the type of $\mathtt{throw}$:
\begin{equation*}
(A\rightarrow B)\rightarrow A \rightarrow B = \neg(\neg(A\land\neg B)\land\ \neg\neg(A\land\neg B))
\end{equation*}
is
\begin{equation*}
\mathtt{throw}^*= \lbd{(k,l)}{(\neg(A,\neg B),\ \neg\neg(A,\neg B))}{(l\ k)}.
\end{equation*}

To illustrate the reduction behaviour of $\mathtt{callcc}^*$ and $\mathtt{throw}^*$ consider the following examples where $A=\mathbbm{N}$:
\begin{equation*}
f_1 = \lbd{(x,k)}{(\neg(\mathbbm{N},\neg B),\ \neg \mathbbm{N})}{\ \mathtt{throw}^*\ (x,\ \lbd{g}{\neg(\mathbbm{N},\neg B)}{g\ (\overline{3},\ \lbd{w}{B}{loop})})}.
\end{equation*}
Here $x$ is the analogous of $x$ from $\lbd{x}{Cont(A)}{M}$, equation~\ref{eq_callcc_1}. The function $f_1$ throws continuation $x$ with value $V=\overline{3}$, in the context $\lbd{w}{B}{loop}$ analogous to $E'$, from equation~\ref{eq_calcc_2}. 

Now consider the following computation, where $\lbd{z}{\mathbbm{N}}{\downarrow}$ stands for the context $E$ from equation \ref{eq_callcc_1}:
\begin{align*}
\mathtt{callcc}^*\ &(f_1,\ \lbd{z}{\mathbbm{N}}{\downarrow})	\\
&\longrightarrow f_1\ (\lbd{(y,k')}{(\mathbbm{N}, \neg B)}{(\lbd{z}{\mathbbm{N}}{\downarrow})\ y},\ \lbd{z}{\mathbbm{N}}{\downarrow})\\
&\longrightarrow \mathtt{throw}^*\ (\lbd{(y,k')}{(\mathbbm{N}, \neg B)}{(\lbd{z}{\mathbbm{N}}{\downarrow})\ y},\ \lbd{g}{\neg(\mathbbm{N},\neg B)}{g\ (\overline{3},\ \lbd{w}{B}{loop})})	\\
&\longrightarrow^2 (\lbd{(y,k')}{(\mathbbm{N}, \neg B)}{(\lbd{z}{\mathbbm{N}}{\downarrow})\ y})\ (\overline{3},\ \lbd{w}{B}{loop})	\\
&\longrightarrow (\lbd{z}{\mathbbm{N}}{\downarrow})\ \overline{3}. 
\end{align*}
When $\mathtt{callcc}^*$ is called, the current continuation $\lbd{z}{\mathbbm{N}}{\downarrow}$ is saved inside 
\begin{equation*}
\lbd{(y,k')}{(\mathbbm{N}, \neg B)}{(\lbd{z}{\mathbbm{N}}{\downarrow})\ y}.
\end{equation*}
Then when $\mathtt{throw}^*$ occurs, the now current continuation $\lbd{w}{B}{loop}$ is abandoned and the continuation $\lbd{z}{\mathbbm{N}}{\downarrow}$ is restored with value $\overline{3}$.

As another example, consider a function $f_2$ which does not throw. It just invokes the continuation $k$ that was passed to it:
\begin{equation*}
f_2 = \lbd{(x,k)}{(\neg(\mathbbm{N},\neg B),\ \neg \mathbbm{N})}{\ k\ \overline{3}}
\end{equation*}
\begin{align*}
\mathtt{callcc}^*\ &(f_2,\ \lbd{z}{\mathbbm{N}}{\downarrow})	\\
&\longrightarrow f_2\ (\lbd{(y,k')}{(\mathbbm{N}, \neg B)}{(\lbd{z}{\mathbbm{N}}{\downarrow})\ y},\ \lbd{z}{\mathbbm{N}}{\downarrow})\\
&\longrightarrow (\lbd{z}{\mathbbm{N}}{\downarrow})\ \overline{3}.
\end{align*}
Here, the control flow is not changed by the use of $\mathtt{callcc}^*$.

This section is not a full proof that ECPS is more expressive than EPCF. In particualar, we have only shown examples that $\mathtt{callcc}^*$ has the desired behaviour. However, this is strong evidence to suggest that an embedding of ECPS in EPCF is not possible. This is why we only studied a translation of EPCF into ECPS in this chapter.

\section{Chapter Summary}

This chapter started the exposition of the novel technical content of the dissertation. In Section~\ref{Sec_cps_trans}, we presented a continuation-passing translation from EPCF to ECPS. For example, an EPCF function $V$ of type $\tau\rightarrow\mathbbm{N}$ is translated to an ECPS function $V^*:\neg(\tau^*,\neg\mathtt{nat})$, where $\neg\mathtt{nat}$ is the type of a continuation waiting for the result of $V$. An EPCF computation $M:\tau$ is translated to a function $M^*=(\lbd{k}{\neg\tau^*}{\ldots}):\neg\neg\tau^*$, where $k$ is the continuation waiting for the result of $M$.

The correctness theorem of this translation (Theorem~\ref{Thm_CPS_correct}) implies that: the tree of $M^*\ (\lbd{x}{\neg\tau^*}{\downarrow})$ is obtained from the tree of EPCF computation $M$ by replacing all value leaves with $\downarrow$. For example:
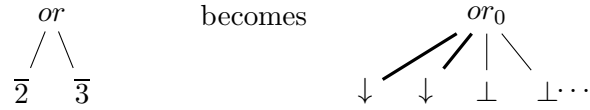
\begin{center}
\begin{tikzpicture}[level distance = 1cm, sibling distance=0.8cm]
\node (n) {$or$}
	child { node {$\overline{2}$} }
	child { node {$\overline{3}$} };
\node[right=1.5cm of n] {becomes};	
\node[right=5cm of n] {$or_0$}
	child { node {$\downarrow$} edge from parent[very thick]}
	child { node {$\downarrow$} edge from parent[very thick]}
	child { node {$\bot$} }
	child { node {$\bot$} }
	child { node[left=0.1cm] {$\ldots$} edge from parent[draw=none]};	
\end{tikzpicture}
\end{center}
This means that the reduction behaviour of terms is preserved by the translation.

We could not prove Theorem~\ref{Thm_CPS_correct} directly by induction on EPCF terms. To obtain a stronger induction hypothesis, we defined a \emph{logical relation} (Definition~\ref{Def_log_rel}) between EPCF terms and ECPS terms. Most importantly, two functions are related if and only if they send related arguments to related computations. Two computations are related
when they simulate each other's behaviour, including effect operations. This is a custom notion of similarity introduced in Definition~\ref{Def_step_indexed_sim}.

We then proved the fundamental property of the logical relation (Lemma~\ref{Lem_fund_prop}): for any EPCF term $T$, the pair $(T,T^*)$ is in the relation. Using this, and the coinduction proof principle from the previous chapter (Proposition~\ref{Prop_coind_princip}), we proved the correctness of the translation.

Finally, we argued informally that ECPS is strictly more expressive than EPCF because it contains the control operator \texttt{call-cc}. Overall, this chapter showed that ECPS is a reasonable choice of language for studying program equivalence. Moreover, because ECPS contains more program contexts than EPCF, it becomes more likely that contextual equivalence equals applicative bisimilarity, which we will prove in Chapter~\ref{Chap_ctx_equiv}, even though this is false for EPCF.

\chapter{Applicative Bisimilarity for ECPS}\label{Chap_bisim}

This chapter starts by defining observations for ECPS for all the running examples of effects. Using them, applicative $\mathfrak{P}$-bisimilarity is defined. Two sufficient conditions for bisimilarity to be compatible are identified: Scott-openness and a novel notion of decomposability. All the example observations are proved decomposable. The final section uses Howe's method to prove bisimilarity is indeed compatible. In the following chapters, applicative $\mathfrak{P}$-bisimilarity is compared with other notions of program equivalence.

\section{Observations for ECPS}\label{Sec_ecps_observations}

To define applicative bisimulation for ECPS we first fix a set of observations $\mathfrak{P}$, which contains subsets of $\textit{Trees}_\Sigma$. The set $\mathfrak{P}$ depends on the effects that are present in the language. It can be used to define various forms of program equivalence which check whether computation trees are in $P\in\mathfrak{P}$.

Observations $P\in\mathfrak{P}$ play a similar role to modalities from $\mathcal{O}$ and to the observations defined by Johann, Simpson and Voigtl{\"a}nder \cite{DBLP:conf/lics/JohannSV10}.  For the example effects considered so far, observations are defined as follows:

\begin{example}[Pure functional computation]
Define $\mathfrak{P}=\{\Downarrow\}$ where $\Downarrow=\{\downarrow\}$. There are no effect operations so the $\Downarrow$ observation only checks for termination.
\end{example}

\begin{example}[Nondeterminism]
Define $\mathfrak{P}=\{\textit{Trees}_\Sigma,\Diamond,\Box\}$ where:
\begin{align*}
\Diamond &= \{tr\in\textit{Trees}_\Sigma \mid \text{at least one of the paths in } tr \text{ that can occur has a } \downarrow \text{ leaf}\}	\\
\Box &=\{tr\in\textit{Trees}_\Sigma \mid \text{the paths in } tr \text{ that can occur are all finite and finish with a } \downarrow\}.
\end{align*}

The intuition is that, if $\treet{t}\in\Diamond$, then computation $t$ \emph{may} terminate. Whereas, if $\treet{t}\in\Box$, $t$ \emph{must} terminate. Notice that there is no condition on the value with which $t$ terminates because ECPS computations do not have a return value.

As discussed in Section \ref{Sec_ecps_oper_sem}, every node in a computation tree has infinitely many children so some paths in the tree can never be executed. For example, the $or$ operation always chooses between its first two children. The definitions of $\Diamond$ and $\Box$ take this into account.

The set of all trees $\textit{Trees}_\Sigma$ is chosen to be an observation for technical reasons that will become clear in the next section. However, the fact that $\textit{Trees}_\Sigma$ is an observation will not affect any notion of program equivalence because for all computations $t$, $\treet{t}\in\textit{Trees}_\Sigma$.
\end{example}

\begin{example}[Probabilistic choice]\label{Eg_prob_ecps}
Define the set of observations as:
\begin{equation*}
\mathfrak{P}=\{\textbf{P}_{>q} \mid q\in\mathbb{Q},\ 0\leq q< 1\}\cup\{\textit{Trees}_\Sigma\}.
\end{equation*}
Define $\mathbb{P}:\textit{Trees}_\Sigma\longrightarrow [0,1]$ to be the least function, by the pointwise order, such that:
\begin{equation*}
\mathbb{P}(tr)=\begin{cases}
				1 	& \text{if } tr= {\downarrow}	\\
				\frac{1}{2}\mathbb{P}(tr_0)+ \frac{1}{2}\mathbb{P}(tr_1)	& \text{if } tr=p\text{-}or(tr_0, tr_1).
                              \end{cases}
\end{equation*}                              
Given functions $f_1,f_2:\textit{Trees}_\Sigma\longrightarrow[0,1]$ the pointwise order is defined as:
\begin{equation*}
f_1\leq f_2 \quad \Longleftrightarrow \quad \forall tr\in\textit{Trees}_\Sigma .\ f_1(tr)\leq f_2(tr).
\end{equation*}

Observations are defined as below:
\begin{equation*}
\textbf{P}_{>q} = \{tr\in\textit{Trees}_\Sigma \mid  \mathbb{P}(tr)>q \}.
\end{equation*}
This means that $tr\in\textbf{P}_{>q}$ if the probability that $tr$ terminates is greater than $q$.
A $p\text{-}or$ node chooses between its first two children with probability $0.5$, so the probability that tree $tr$ terminates is calculated over these choices. Notice that $\mathbb{P}(\bot)=0$.
\end{example}

\begin{example}[Global store]
Define the set of states as the set of functions from storage locations to natural numbers: $State=\mathbb{L}\longrightarrow\mathbb{N}$. The set $\mathfrak{P}$ is defined as:
\begin{equation*}
\mathfrak{P}=\{ (s\rightarrowtail r) \mid s,r\in State\}\cup\{\textit{Trees}_\Sigma\}.
\end{equation*}
Define the execution of a tree starting in a state as the \emph{least} partial function:
\begin{equation*}
exec: \textit{Trees}_\Sigma \times State \longrightarrow \{\downarrow\}\times State
\end{equation*}
which satisfies
\begin{equation*}
exec(tr,s)= \begin{cases}
			(\downarrow,s) & \text{if } tr=\downarrow	\\
			exec(tr_{s(l)},s) & \text{if } tr=lookup_{l,n}(tr_0,tr_1,\ldots) 	\\
			&\text{and } exec(tr_{s(l)},s) \text{ is defined}	\\
			exec(tr_0,s[l\coloneqq n]) & \text{if } tr=update_{l,n}(tr_0, tr_1,\ldots)	\\
			& \text{and } exec(tr_0,s[l\coloneqq n]) \text{ is defined}.
		   \end{cases}
\end{equation*}

Now define observations as:
\begin{equation*}
(s\rightarrowtail r) = \{tr \in\textit{Trees}_\Sigma \mid \\
 exec(tr,s) \text{ is defined and } exec(tr,s)=(\downarrow,r) \}.
\end{equation*}
Notice that $exec(tr,s)$ is defined only when the execution of $tr$ terminates. So $tr\in(s\rightarrowtail r)$ only if the execution of $tr$ started in state $s$ terminates in state $r$.
\end{example}

\begin{example}[Input/output]
An I/O-trace is a finite word $w$ over the alphabet
\begin{equation*}
\{?n \mid n\in\mathbb{N}\} \cup \{!n \mid n\in\mathbb{N} \}.
\end{equation*}
Thus, a trace is a sequence of input and output operations, where $?n$ means that the number $n$ was given as input to a $read$ operation, and $!n$ means that the number $n$ was output by a $write$ operation. We can use them to define the set of observation:
\begin{equation*}
\mathfrak{P} = \{\langle w\rangle_{\ldots} \mid w \text{ an I/O-trace}\}
\end{equation*}
where
\begin{equation*}
\langle w\rangle_{\ldots} = \{tr \in\textit{Trees}_\Sigma \mid \text{the execution of } tr \text{ produces I/O-trace }w \}.
\end{equation*}
To specify rigorously when ``the execution of a tree produces an I/O trace'', we can define a relation between trees and I/O traces, by induction on traces. Denote this relation by $\models$.
\begin{align*}
tr \models \langle\epsilon\rangle_{\ldots} \quad &\Longleftrightarrow \quad \text{true}	\\
tr \models \langle (?n)w \rangle_{\ldots} \quad &\Longleftrightarrow\quad tr=read_k(tr_0,tr_1,\ldots) \text{ and } tr_n\models\langle w\rangle_{\ldots}	\\
tr \models \langle (!n)w\rangle_{\ldots} \quad &\Longleftrightarrow\quad tr=write_n(tr_0,tr_1,\ldots) \text{ and } tr_0\models\langle w\rangle_{\ldots}.
\end{align*}

\end{example}

\section{Applicative $\mathfrak{P}$-Bisimilarity}\label{Sec_ecps_bisim} 


\begin{definition}[Applicative $\mathfrak{P}$-simulation]\label{Def_ecps_simulation}
A collection of relations $\mathcal{R}_A^\mathfrak{v}\subseteq(\vdash A)\times{(\vdash A)}$ for each type $A$ and $\mathcal{R}^\mathfrak{c}\subseteq(\vdash)\times(\vdash)$ is an applicative $\mathfrak{P}$-\emph{simulation} if:
\begin{enumerate}
\item \label{ecps_sim1} $v\ \mathcal{R}^\mathfrak{v}_{\mathtt{unit}}\ w \implies v=w=\star$.
\item \label{ecps_sim2} $v\ \mathcal{R}^\mathfrak{v}_{\mathtt{nat}}\ w \implies v=w$.
\item \label{ecps_sim3} $s\ \mathcal{R}^\mathfrak{c}\ t \implies \forall P\in\mathfrak{P}.\ (\treet{s} \in P \implies \treet{t} \in P)$.
\item \label{ecps_sim4} $v\ \mathcal{R}^\mathfrak{v}_{\neg(A_1,\ldots,A_n)}\ u \implies \forall\vdash w_1:A_1,\ldots,\vdash w_n:A_n.\ v(w_1,\ldots,w_n)\ \mathcal{R}^\mathfrak{c}\ u(w_1,\ldots,w_n)$.
\end{enumerate}
Applicative $\mathfrak{P}$-\emph{similarity} $\precsim$ is the union of all applicative $\mathfrak{P}$-simulations. Therefore, it is the greatest applicative $\mathfrak{P}$-simulation. 
\end{definition}

According to the definition above, unit values and natural number values are similar if and only if they are equal. The third clause says that $t$ simulates $s$ if the computation tree of $t$ has all the properties of the computation tree of $s$. The properties are specified using the set of observations $\mathfrak{P}$. Notice that simulation for computations is not defined using simulations for values, since computations do not have a return value. 
The last clause compares the behaviour of functions for all possible arguments.

\begin{definition}[Applicative $\mathfrak{P}$-bisimulation]\label{Def_ecps_bisim}
An applicative $\mathfrak{P}$-\emph{bisimulation} is a symmetric $\mathfrak{P}$-simulation. Applicative $\mathfrak{P}$-\emph{bisimilarity} $\sim$ is the union of all applicative $\mathfrak{P}$-bisimulations. Therefore it is the greatest applicative $\mathfrak{P}$-bisimulation.
\end{definition}

Below are two properties of bisimilarity and similarity which will be used later. They are followed by two examples of how bisimilarity can be established.

\begin{proposition}\label{Prop_bisim_is_sim_and_simop}
Applicative $\mathfrak{P}$-bisimilarity coincides with the intersection between applicative $\mathfrak{P}$-similarity and its converse:
\begin{equation*}
(\sim) = (\precsim)\cap(\precsim)^{\textit{op}}.
\end{equation*}
\end{proposition}
\begin{proof}
The proof is done using the definitions of similarity and bisimilarity. The $\supseteq$ inclusion is shown by proving $(\precsim)\cap(\precsim)^{\textit{op}}$ is a symmetric simulation. The full proof appears in Appendix \ref{App_bisim}.
\end{proof}

\begin{lemma}\label{Lem_red_pres_sim}
Similarity is preserved by the reduction relation, that is:
\begin{equation*}
\forall s,t.\ \vdash s\precsim^\mathfrak{c} t \text{ and } s\longrightarrow^*s' \text{ and } t\longrightarrow^* t' \implies \vdash s'\precsim^\mathfrak{c} t'.
\end{equation*}
\end{lemma}
\begin{proof}
By the definition of $\treet{-}$ we know that $\treet{s}=\treet{s'}$ and $\treet{t}=\treet{t'}$. So because $\vdash s\precsim^\mathfrak{c} t$ we know that:
\begin{equation*}
\forall P\in\mathfrak{P}.\ (\treet{s}=\treet{s'} \in P \implies \treet{t}=\treet{t'} \in P).
\end{equation*}
This is enough to establish that $\vdash s'\precsim^\mathfrak{c} t'$.
\end{proof}

\begin{example}[Probabilistic choice]\label{Eg_prob_ecps_bisim}
Consider the following computations:
\begin{multline*}
m_{\overline{1},\overline{2}}=p\text{-}or(\overline{0},\ y.\caset{y}{(\lbd{x}{\mathtt{nat}}{\downarrow})\ \overline{1}}{y'}{\\ \caset{y'}{(\lbd{x}{\mathtt{nat}}{\downarrow})\ \overline{2}}{y''}{loop}}).
\end{multline*}
Computation $m_{\overline{1},\overline{3}}$ is defined analogously to $m_{\overline{1},\overline{2}}$ where the value $\overline{2}$ inside the computation is replaced by the value $\overline{3}$.
\begin{multline*}
n_{\overline{1},\overline{2},\overline{1},\overline{3}} = p\text{-}or(\overline{5},\ y.\caset{y}{m_{\overline{1},\overline{2}}}{y'}{\\ \caset{y'}{m_{\overline{1},\overline{3}}}{y''}{loop}}).
\end{multline*}
Their computation trees are:
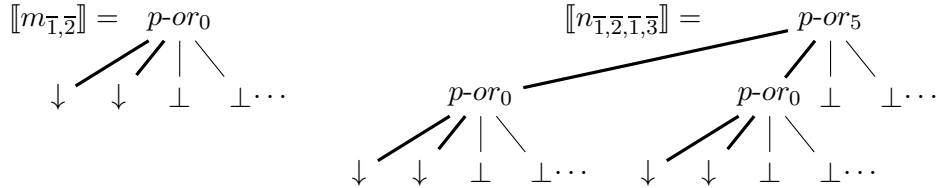
\begin{center}
\begin{tikzpicture}[level distance=1cm, sibling distance=0.8cm]
\node (m) {$p\text{-}or_0$}
	child {node {$\downarrow$} edge from parent[very thick]}
	child {node {$\downarrow$} edge from parent[very thick]}
	child {node {$\bot$} }
	child {node {$\bot$} }
	child {node[left=0.05cm] {$\ldots$} edge from parent[draw=none] };
\node [left=0.1cm of m] {$\treet{m_{\overline{1},\overline{2}}}=$};	
\node (n) [right = 8cm] {$p\text{-}or_5$}
	child[sibling distance=2.3cm] {node {$p\text{-}or_0$} edge from parent[very thick]
		child[sibling distance=0.8cm] {node {$\downarrow$} }
		child[sibling distance=0.8cm] {node {$\downarrow$} } 
		child[sibling distance=0.8cm, thin] {node {$\bot$} }
		child[sibling distance=0.8cm, thin] {node {$\bot$} }
		child[sibling distance=0.8cm, thin] {node[left=0.05cm] {$\ldots$} edge from parent[draw=none] }}
	child {node (m) {$p\text{-}or_0$} edge from parent[very thick]
		child {node {$\downarrow$} }
		child {node {$\downarrow$} } 
		child[thin] {node {$\bot$} }
		child[thin] {node {$\bot$} }
		child[thin] {node[left=0.05cm] {$\ldots$} edge from parent[draw=none] }}
	child {node {$\bot$} }	
	child {node {$\bot$} }
	child {node[left=0.05cm] {$\ldots$} edge from parent[draw=none] };	
\node [left=1cm of n] {$\treet{n_{\overline{1},\overline{2},\overline{1},\overline{3}}}=$};	
\end{tikzpicture}
\end{center}
Both computations terminate with a $\downarrow$ with probability $1$ so
\begin{equation*}
\forall P\in\mathfrak{P}.\ \treet{m_{\overline{1},\overline{2}}} \in P \Longleftrightarrow \treet{n_{\overline{1},\overline{2},\overline{1},\overline{3}}} \in P
\end{equation*}
is true. Therefore $m_{\overline{1},\overline{2}}$ and $n_{\overline{1},\overline{2},\overline{1},\overline{3}}$ are bisimilar.  Notice that the subscripts $0$ and $5$ on the $p\text{-}or$ nodes do not play any role in establishing bisimilarity.

However, if we consider computation $n'_{\overline{1},\overline{2},\overline{1},\overline{3}}$ with tree:
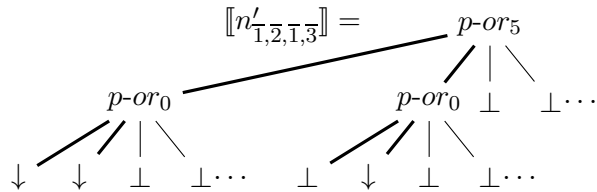
\begin{center}
\begin{tikzpicture}[level distance=1cm, sibling distance=0.8cm]
\node (n) {$p\text{-}or_5$}
	child[sibling distance=2.3cm] {node {$p\text{-}or_0$} edge from parent[very thick]
		child[sibling distance=0.8cm] {node {$\downarrow$} }
		child[sibling distance=0.8cm] {node {$\downarrow$} } 
		child[sibling distance=0.8cm, thin] {node {$\bot$} }
		child[sibling distance=0.8cm, thin] {node {$\bot$} }
		child[sibling distance=0.8cm, thin] {node[left=0.05cm] {$\ldots$} edge from parent[draw=none] }}
	child {node (m) {$p\text{-}or_0$} edge from parent[very thick]
		child {node {$\bot$} }
		child {node {$\downarrow$} } 
		child[thin] {node {$\bot$} }
		child[thin] {node {$\bot$} }
		child[thin] {node[left=0.05cm] {$\ldots$} edge from parent[draw=none] }}
	child {node {$\bot$} }	
	child {node {$\bot$} }
	child {node[left=0.05cm] {$\ldots$} edge from parent[draw=none] };	
\node [left=1cm of n] {$\treet{n'_{\overline{1},\overline{2},\overline{1},\overline{3}}}=$};	
\end{tikzpicture}
\end{center}
it has a probability of $3/4$ of terminating. Therefore, $\treet{m_{\overline{1},\overline{2}}}\in\mathbf{P}_{>0.9}$ but $\treet{n'_{\overline{1},\overline{2},\overline{1},\overline{3}}}\not\in\mathbf{P}_{>0.9}$ so the two computations are not bisimilar.
\end{example}

\begin{example}[Global store]
Plotkin and Power \cite{DBLP:conf/fossacs/PlotkinP02} axiomatise the behaviour of the global store operations $lookup$ and $update$ using a set of program equations. We can show that these equations are in fact induced by applicative $\mathfrak{P}$-bisimilarity. For example, consider the following equation:
\begin{multline*}
\forall loc,loc'\in\mathbb{L} \text{ where } loc\not=loc'.\ \forall n,n'\in\mathbb{N}.\ \forall (y:\mathtt{nat},y':\mathtt{nat}\vdash t).	\\
l=update_{loc}(\overline{n},\ y.update_{loc'}(\overline{n'},\ y'.t)) \mathrel{\sim^\mathfrak{c}} update_{loc'}(\overline{n'},\ y'.update_{loc}(\overline{n},\ y.t))=r.
\end{multline*}
It says that writes to two different locations can be interchanged. Because we are dealing with computations, to prove bisimilarity it suffices to show:
\begin{equation*}
\forall s_1,s_2\in State.\ \treet{l}\in(s_1\rightarrowtail s_2) \Longleftrightarrow \treet{r}\in(s_1\rightarrowtail s_2).
\end{equation*}
The observation $(s_1\rightarrowtail s_2)$ was defined using the partial function $exec$ which formalises the execution of a computation tree. So in fact, we need to show that:
\begin{multline}
\forall s_1,s_2\in State.\ exec(\treet{l},s_1) \text{ is defined and } exec(\treet{l},s_1)=(\downarrow,s_2) 	\\
\Longleftrightarrow exec(\treet{r},s_1) \text{ is defined and } exec(\treet{r},s_1)=(\downarrow,s_2). \label{Eg_globalstore_eq}
\end{multline}
The computation trees of $l$ and $r$ respectively are:
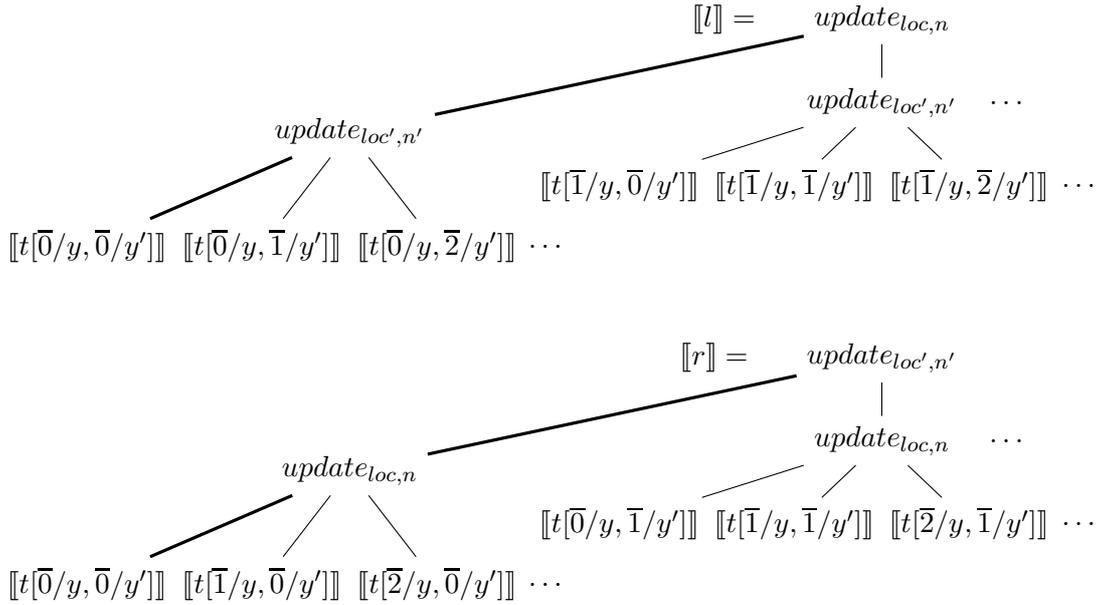
\begin{center}
\begin{tikzpicture}[level distance=1.1cm, sibling distance=7cm]
\node (n) {$update_{loc,n}$}
	child[level distance=1.5cm] { node {$update_{loc',n'}$} edge from parent[very thick]
		child[sibling distance=2.3cm] { node {$\treet{t[\overline{0}/y,\overline{0}/y']}$} }
		child[sibling distance=2.3cm] { node {$\treet{t[\overline{0}/y,\overline{1}/y']}$} edge from parent[thin]}
		child[sibling distance=2.3cm] { node {$\treet{t[\overline{0}/y,\overline{2}/y']}$} edge from parent[thin]}
		child[sibling distance=2.3cm] { node[left=0.5cm] {$\ldots$}  edge from parent[draw=none] } }
	child { node {$update_{loc',n'}$} 
		child[sibling distance=2.3cm] { node {$\treet{t[\overline{1}/y,\overline{0}/y']}$} }
		child[sibling distance=2.3cm] { node {$\treet{t[\overline{1}/y,\overline{1}/y']}$} }
		child[sibling distance=2.3cm] { node {$\treet{t[\overline{1}/y,\overline{2}/y']}$} }
		child[sibling distance=2.3cm] { node[left=0.5cm] {$\ldots$}  edge from parent[draw=none] } }
	child { node[left=5cm] {$\ldots$}  edge from parent[draw=none] };
\node[left=0.5cm of n] {$\treet{l}=$};	
\end{tikzpicture}
\end{center}
\begin{center}
\begin{tikzpicture}[level distance=1.1cm, sibling distance=7cm]
\node (n) {$update_{loc',n'}$}
	child[level distance=1.5cm] { node {$update_{loc,n}$} edge from parent[very thick]
		child[sibling distance=2.3cm] { node {$\treet{t[\overline{0}/y,\overline{0}/y']}$} }
		child[sibling distance=2.3cm] { node {$\treet{t[\overline{1}/y,\overline{0}/y']}$} edge from parent[thin]}
		child[sibling distance=2.3cm] { node {$\treet{t[\overline{2}/y,\overline{0}/y']}$} edge from parent[thin]}
		child[sibling distance=2.3cm] { node[left=0.5cm] {$\ldots$}  edge from parent[draw=none] } }
	child { node {$update_{loc,n}$} 
		child[sibling distance=2.3cm] { node {$\treet{t[\overline{0}/y,\overline{1}/y']}$} }
		child[sibling distance=2.3cm] { node {$\treet{t[\overline{1}/y,\overline{1}/y']}$} }
		child[sibling distance=2.3cm] { node {$\treet{t[\overline{2}/y,\overline{1}/y']}$} }
		child[sibling distance=2.3cm] { node[left=0.5cm] {$\ldots$}  edge from parent[draw=none] } }
	child { node[left=5cm] {$\ldots$}  edge from parent[draw=none] };
\node[left=0.5cm of n] {$\treet{r}=$};	
\end{tikzpicture}
\end{center}
Using the trees and the definition of $exec$, we can deduce the following chain of equations, which is enough to prove equation~\ref{Eg_globalstore_eq}:
\begin{align*}
exec(\treet{l},\ s_1)&=exec(update_{loc',n'}(\overrightarrow{\treet{t[\overline{0}/y,\overline{i}/y]}}),\ s_1[loc:=n])	\\
&=exec(\treet{t[\overline{0}/y,\overline{0}/y']},\ (s_1[loc=n])[loc'=n'])	\\
&=exec(\treet{t[\overline{0}/y,\overline{0}/y']},\ s_1[loc=n,loc'=n'])	\tag{because $loc\not=loc'$}	\\
&=exec(\treet{t[\overline{0}/y,\overline{0}/y']},\ (s_1[loc'=n'])[loc=n])	\\
&=exec(update_{loc,n}(\overrightarrow{\treet{t[\overline{i}/y,\overline{0}/y]}}),\ s_1[loc'=n'])=exec(\treet{r},\ s_1) = (\downarrow,s_2).
\end{align*}
The other six equations for global store from \cite{DBLP:conf/fossacs/PlotkinP02} can be proved analogously. Moreover, Plotkin and Power observe that adding more equations to this set of seven leads to inconsistency. Therefore, these are all the equations between computations for the global store effect.
\end{example}

\section{Applicative $\mathfrak{P}$-Bisimilarity is a Congruence}\label{Sec_ecps_bisim_congr}

This section discusses the two main properties required for $\mathfrak{P}$-bisimilarity to be a well-behaved program equivalence: being an equivalence relation and compatibility. The proof of the next lemma appears in Appendix~\ref{App_bisim}. 

\begin{lemma} \label{Lemm_(bi)sim_preord}
Applicative $\mathfrak{P}$-similarity is a preorder. Applicative $\mathfrak{P}$-bisimilarity is an equivalence relation.
\end{lemma}

Applicative similarity is a \emph{well-typed relation} on \emph{closed} ECPS terms. Compatibility says we can substitute related programs for a variable inside related programs. Therefore, we need to talk about bisimilarity of programs with free variables.
Bisimilarity can be extended to open terms in a standard way \cite{DBLP:conf/lics/LagoGL17}.

\begin{definition}[Open extension]
Given a well-typed relation on closed terms, $\mathcal{R}=(\mathcal{R}^\mathfrak{v}_A,\mathcal{R}^\mathfrak{c})$, the open extension of $\mathcal{R}$ is $\mathcal{R}^\circ=(\mathcal{R}^{\circ,\mathfrak{v}}_A,\mathcal{R}^{\circ,\mathfrak{c}})$ where:
\begin{gather*}
\overrightarrow{x_i:A_i} \vdash v\ \mathcal{R}^{\circ,\mathfrak{v}}_B\ w \quad\Longleftrightarrow\quad \forall \overrightarrow{u_i:A_i}.\ v[\overrightarrow{u_i/x_i}] \mathrel{\mathcal{R}^{\mathfrak{v}}_B} w[\overrightarrow{u_i/x_i}]	\\
\overrightarrow{x_i:A_i} \vdash s\ \mathcal{R}^{\circ,\mathfrak{c}}\ t \quad\Longleftrightarrow\quad \forall \overrightarrow{u_i:A_i}.\ s[\overrightarrow{u_i/x_i}] \mathrel{\mathcal{R}^{\mathfrak{c}}} t[\overrightarrow{u_i/x_i}].
\end{gather*}
\end{definition}

\begin{definition}[Compatibility \cite{DBLP:conf/lics/LagoGL17}] \label{Def_compat}
A well-typed open relation $\mathcal{R}=(\mathcal{R}^\mathfrak{v}_A, \mathcal{R}^\mathfrak{c})$ is compatible if it is closed under the rules in Figure \ref{Fig_compat}. We define $\mathcal{R}$ to be a precongruence if it is a compatible preorder, and a congruence if it is a compatible equivalence relation.
\end{definition}

\begin{figure}
\begin{gather*}
\inferrule{ }{\Gamma\vdash x\ \mathcal{R}^\mathfrak{v}_{A}\ x} \textsc{(Comp1)} \quad
\inferrule{ }{\Gamma\vdash \star\ \mathcal{R}^\mathfrak{v}_{\mathtt{unit}}\ \star}\textsc{(Comp2)}	\\
\inferrule{\Gamma,x_1:A_1,\ldots,x_n:A_n \vdash s\ \mathcal{R}^\mathfrak{c}\ t}{\Gamma \vdash \lbd{\overrightarrow{x}}{\overrightarrow{A}}{s}\ \mathcal{R}^\mathfrak{v}_{\neg(A_1,\ldots,A_n)}\ \lbd{\overrightarrow{x}}{\overrightarrow{A}}{t}} \textsc{(Comp3)}	\\
\inferrule{ }{\Gamma\vdash \mathtt{zero}\ \mathcal{R}^\mathfrak{v}_{\mathtt{nat}}\ \mathtt{zero}} \textsc{(Comp4)} \quad
\inferrule{\Gamma\vdash v\ \mathcal{R}^\mathfrak{v}_{\mathtt{nat}}\ v'}{\Gamma\vdash \mathtt{succ}(v)\ \mathcal{R}^\mathfrak{v}_{\mathtt{nat}}\ \mathtt{succ}(v')} \textsc{(Comp5)}	\\
\inferrule{\Gamma\vdash v\ \mathcal{R}^\mathfrak{v}_{\neg(A_1,\ldots,A_n)}\ v' \\ \Gamma\vdash w_1\ \mathcal{R}^\mathfrak{v}_{A_1}\ w'_1,\ldots,\Gamma\vdash w_n\ \mathcal{R}^\mathfrak{v}_{A_n}\ w'_n}{\Gamma \vdash v(w_1,\ldots,w_n)\ \mathcal{R}^\mathfrak{c}\ v'(w'_1,\ldots,w'_n)} \textsc{(Comp6)} \\
\inferrule{\Gamma,x:\neg(\overrightarrow{A}) \vdash v\ \mathcal{R}^\mathfrak{v}_{\neg(\overrightarrow{A})}\ v' \\ \Gamma\vdash w_i\ \mathcal{R}^\mathfrak{v}_{A_i}\ w'_i \text{ for each i}}{\Gamma\vdash (\mufix{x}{v})(\overrightarrow{w})\ \mathcal{R}^\mathfrak{c}\ (\mufix{x}{v'})(\overrightarrow{w'})} \textsc{(Comp7)} \\
\inferrule{\Gamma\vdash v\ \mathcal{R}^\mathfrak{v}_{\mathtt{nat}}\ v' \\ \Gamma,x:\mathtt{nat}\vdash t\ \mathcal{R}^\mathfrak{c}\ t'}{\Gamma\vdash\sigma(v,x.t)\ \mathcal{R}^\mathfrak{c}\ \sigma(v',x.t')}\sigma\in\Sigma\ \textsc{(Comp8)} \quad
\inferrule{ }{\Gamma\vdash \downarrow\ \mathcal{R}^\mathfrak{c}\ \downarrow}\textsc{(Comp9)}	\\
\inferrule{\Gamma\vdash v\ \mathcal{R}^\mathfrak{v}_{\mathtt{nat}}\ v' \\ \Gamma\vdash s\ \mathcal{R}^\mathfrak{c}\ s' \\ \Gamma,x:\mathtt{nat}\vdash t\ \mathcal{R}^\mathfrak{c}\ t'}{\Gamma \vdash \caset{v}{s}{x}{t}\ \mathcal{R}^\mathfrak{c}\ \caset{v'}{s'}{x}{t'}}\\ \textsc{(Comp10)}
\end{gather*}
\caption{Compatibility rules.}\label{Fig_compat}
\end{figure}

The following lemma identifies some alternative compatibility rules which will be used in later proofs. Its proof can be found in Appendix \ref{App_bisim}.

\begin{lemma} \label{Lem_comp_1premise_rules}
Consider a well-typed relation $\mathcal{R}$ that is a preorder. The compatibility rules \textsc{(comp6)}, \textsc{(comp7)}, \textsc{(comp8)} and \textsc{(comp10)} from Figure~\ref{Fig_compat} are equivalent to the conjunction of their single-premise versions. More explicitly:
\begin{itemize}
\item Rule \textsc{(comp6)} is equivalent to the conjunction of the rules:
\begin{gather*}
\inferrule{\Gamma\vdash v\ \mathcal{R}^\mathfrak{v}_{\neg(A_1,\ldots,A_n)}\ v'	\\ \Gamma\vdash \overrightarrow{w_j:A_j}}{\Gamma \vdash v(w_1,\ldots,w_n)\ \mathcal{R}^\mathfrak{c}\ v'(w_1,\ldots,w_n)} \textsc{(Comp6L)} \\
\inferrule{\Gamma \vdash v:\neg(\overrightarrow{A_j}) \\ \Gamma \vdash (\overrightarrow{w_{1,i-1}:A_{1,i-1}}) \\ \Gamma\vdash w_i\ \mathcal{R}^\mathfrak{v}_{A_i}\ w'_i	\\	\Gamma \vdash (\overrightarrow{w_{i+1,n}:A_{i+1,n}})}{\Gamma \vdash v(w_1,\ldots,w_i,\ldots,w_n)\ \mathcal{R}^\mathfrak{c}\ v(w_1,\ldots,w'_i,\ldots,w_n)} \\ \textsc{(Comp6R$_i$)} \text{ for each }i=\overline{1,n}
\end{gather*}

\item Rule \textsc{(comp7)} is equivalent to the conjunction of the rules:
\begin{gather*}
\inferrule{\Gamma,x:\neg(\overrightarrow{A_j}) \vdash v\ \mathcal{R}^\mathfrak{v}_{\neg(\overrightarrow{A_j})}\ v'	\\	\Gamma \vdash \overrightarrow{w_j:A_j}}{\Gamma\vdash (\mufix{x}{v})(\overrightarrow{w_j})\ \mathcal{R}^\mathfrak{c}\ (\mufix{x}{v'})(\overrightarrow{w_j})} \textsc{(Comp7L)}	\\
\inferrule{\Gamma,x:\neg(\overrightarrow{A_j}) \vdash v:\neg(\overrightarrow{A_j})	\\ \Gamma \vdash (\overrightarrow{w_{1,i-1}:A_{1,i-1}})	\\ \Gamma\vdash w_i\ \mathcal{R}^\mathfrak{v}_{A_i}\ w'_i	\\ \Gamma \vdash (\overrightarrow{w_{i+1,n}:A_{i+1,n}})}{\Gamma\vdash (\mufix{x}{v})(w_1,\ldots,w_i,\ldots,w_n)\ \mathcal{R}^\mathfrak{c}\ (\mufix{x}{v})(w_1,\ldots,w'_i,\ldots,w_n)} \\ \textsc{(Comp7R$_i$)} \text{ for each }i=\overline{1,n}
\end{gather*}

\item Rule \textsc{(comp8)} is equivalent to the conjunction of the rules:
\begin{gather*}
\inferrule{\Gamma\vdash v\ \mathcal{R}^\mathfrak{v}_{\mathtt{nat}}\ v'	\\	\Gamma,x:\mathtt{nat}\vdash t}{\Gamma\vdash\sigma(v,x.t)\ \mathcal{R}^\mathfrak{c}\ \sigma(v',x.t)}\sigma\in\Sigma\ \textsc{(Comp8L)}	\\
\inferrule{\Gamma\vdash v:\mathtt{nat}	\\	\Gamma,x:\mathtt{nat}\vdash t\ \mathcal{R}^\mathfrak{c}\ t'}{\Gamma\vdash\sigma(v,x.t)\ \mathcal{R}^\mathfrak{c}\ \sigma(v,x.t')}\sigma\in\Sigma\ \textsc{(Comp8R)}
\end{gather*}

\item Rule \textsc{(comp10)} is equivalent to the conjunction of the rules:
\begin{gather*}
\inferrule{\Gamma\vdash v\ \mathcal{R}^\mathfrak{v}_{\mathtt{nat}}\ v'	\\ \Gamma\vdash s	\\	\Gamma,x:\mathtt{nat}\vdash t}{\Gamma \vdash \caset{v}{s}{x}{t}\ \mathcal{R}^\mathfrak{c}\ \caset{v'}{s}{x}{t}}\\ \textsc{(Comp10V)}	\\
\inferrule{\Gamma\vdash v:\mathtt{nat}	\\	\Gamma\vdash s\ \mathcal{R}^\mathfrak{c}\ s'	\\	\Gamma,x:\mathtt{nat}\vdash t}{\Gamma \vdash \caset{v}{s}{x}{t}\ \mathcal{R}^\mathfrak{c}\ \caset{v}{s'}{x}{t}}\\ \textsc{(Comp10L)}	\\
\inferrule{\Gamma\vdash v:\mathtt{nat}	\\	\Gamma\vdash s	\\	\Gamma,x:\mathtt{nat}\vdash t\ \mathcal{R}^\mathfrak{c}\ t'}{\Gamma \vdash \caset{v}{s}{x}{t}\ \mathcal{R}^\mathfrak{c}\ \caset{v}{s}{x}{t'}}\\ \textsc{(Comp10R)}
\end{gather*}
\end{itemize}
\end{lemma}

To prove $\mathfrak{P}$-bisimilarity is a congruence, we have identified two sufficient conditions that the set of observations $\mathfrak{P}$ should satisfy. One of them is that every observation needs to be Scott-open, as in the work of Simpson and Voorneveld \cite{SimV18}. The second condition is that $\mathfrak{P}$ needs to be decomposable for a novel definition of decomposability.

\begin{definition}[Scott-openness]\label{Def_scott_open}
A set of trees $X$ is Scott-open if:
\begin{enumerate}
\item It is upwards closed, that is: $tr\in X$ and $tr\leq tr'$ imply $tr'\in X$.
\item Whenever $tr_1\leq tr_2\leq\ldots$ is an ascending chain with least upper bound $\bigsqcup tr_i \in X$, then $tr_j\in X$ for some $j$.
\end{enumerate}
\end{definition}

\begin{definition}[Decomposability]\label{Def_decomposability}
The set of observations $\mathfrak{P}$ is decomposable if for any  $P\in\mathfrak{P}$ and for any $tr\in P$:
\begin{equation*}
tr=\sigma_n(\overrightarrow{tr'}) \implies \exists\overrightarrow{P'}\in\mathfrak{P}.\ \overrightarrow{tr'}\in\overrightarrow{P'}  \text{ and } \forall \overrightarrow{p'}\in\overrightarrow{P'}.\ \sigma_n(\overrightarrow{p'}) \in P.
\end{equation*}
\end{definition}

Decomposability says that, whenever a tree $tr$ is part of an observation $P$, the children of $tr$'s root should themselves be part of some observations which fully capture the restrictions that $P$ places on them. This is true for all examples of effects considered so far. Using the last two definitions we can state the main theorem of this chapter:

\begin{theorem}\label{Thm_sim_compat}
Given a decomposable set of Scott-open observations $\mathfrak{P}$:
\begin{enumerate}
\item The open extension of applicative $\mathfrak{P}$-similarity, $\precsim^\circ$, is compatible, and hence it is a precongruence.
\item The open extension of applicative $\mathfrak{P}$-bisimilarity, $\sim^\circ$, is compatible, and hence it is a congruence.
\end{enumerate}
\end{theorem}

It is easy to check that all the running examples of observations are upwards closed. The proof that they satisfy the second condition in the definition of Scott-openness is the same as in \cite{SimV18}. It remains to check that $\mathfrak{P}$ is decomposable:

\begin{example}[Pure functional computation]
The only observation is $\Downarrow=\{\downarrow\}$. There are no trees in $\Downarrow$ whose root has children, so decomposability is satisfied.
\end{example}

\begin{example}[Nondeterminism]
Recall that $\mathfrak{P}=\{\textit{Trees}_\Sigma,\Diamond,\Box\}$. For any $tr\in\textit{Trees}_\Sigma$ we can choose each $P'_n=\textit{Trees}_\Sigma$ to fulfil decomposability.

Consider $tr\in\Diamond$. Either $tr=\downarrow$, in which case we are done, or $tr=or_n(tr'_0,tr'_1,tr'_2\ldots)$. It must be the case that either $tr'_0$ or $tr'_1$ have a reachable $\downarrow$-leaf. Without loss of generality, assume $tr'_0$ has a reachable $\downarrow$-leaf. Then we know $tr'_0\in\Diamond$ so we can choose $P'_0=\Diamond,P'_1=\textit{Trees}_\Sigma,P'_2=\textit{Trees}_\Sigma,\ldots$. For any $\overrightarrow{p'}\in\overrightarrow{P'}$ we know $or_n(\overrightarrow{p'})\in\Diamond$ because $p'_0$ has a reachable $\downarrow$-leaf.

The argument for $tr\in \Box$ is analogous, if we choose $P'_0=\Box,P'_1=\Box,P'_2=\textit{Trees}_\Sigma,P'_3=\textit{Trees}_\Sigma,\ldots$. This proof relies on the fact that $\textit{Trees}_\Sigma$ is an observation, which is way we decided to include it in $\mathfrak{P}$.
\end{example}

\begin{example}[Probabilistic choice]
Consider $tr=p\text{-}or_n(tr'_0,tr'_1,tr'_2,\ldots) \in\mathbf{P}_{> q}$ for some $q\in\mathbb{Q}$, $0\leq q<1$. Recall the definition of the partial function $\mathbb{P}$ from Example \ref{Eg_prob_ecps}. We know that:
\begin{equation}
\mathbb{P}(tr)=\frac{1}{2}\mathbb{P}(tr'_0) + \frac{1}{2}\mathbb{P}(tr'_1) > q. \label{eq_prob_decomp_proof1}
\end{equation}
Define:
\begin{align*}
q_0 &= \frac{\mathbb{P}(tr'_0)}{\mathbb{P}(tr'_0)+\mathbb{P}(tr'_1)}\cdot 2q	\\
q_1 &= \frac{\mathbb{P}(tr'_1)}{\mathbb{P}(tr'_0)+\mathbb{P}(tr'_1)}\cdot 2q.
\end{align*}
The two probabilities $\mathbb{P}(tr'_0)$ and $\mathbb{P}(tr'_1)$ are rational numbers because they are defined as a sum of rational numbers. The threshold $q$ is rational by assumption, so $q_1$ and $q_2$ are rational.

From equation \ref{eq_prob_decomp_proof1}  we can deduce that:
\begin{align*}
1&\geq \mathbb{P}(tr'_0) >q_0\\
1&\geq \mathbb{P}(tr'_1) >q_1.
\end{align*}
So we can choose $P'_0=\mathbf{P}_{>q_0},P'_1=\mathbf{P}_{>q_1},P'_2=\textit{Trees}_\Sigma, P'_3=\textit{Trees}_\Sigma,\ldots$.
Therefore, $\overrightarrow{tr'}\in\overrightarrow{P'}$ as required.

Consider some other subtrees $\overrightarrow{p'}\in\overrightarrow{P'}$. By the way we defined $\overrightarrow{P'}$ it follows that:
\begin{equation*}
\frac{1}{2}\mathbb{P}(p'_0) + \frac{1}{2}\mathbb{P}(p'_1) > \frac{1}{2}(q_0+q_1) = q
\end{equation*}
so $p\text{-}or(\overrightarrow{p'})\in\mathbf{P}_{>q}$ as required.
\end{example}

\begin{example}[Global store] 
Consider a tree $tr=\sigma_n(tr'_0,tr'_1,tr'_2,\ldots)\in (s\rightarrowtail r)$. It must be the case that $exec(tr,s)=(\downarrow,r)$. 

If $\sigma_n=lookup_{l,n}$: because $exec(tr,s)$ is defined it must be the case that $exec(tr'_{s(l)},s)$ is also defined and  $exec(tr'_{s(l)},s)=exec(tr,s)=(\downarrow,r)$ so we know that $tr'_{s(l)}\in(s\rightarrowtail r)$. In the definition of decomposability, choose $P'_{s(l)}=(s\rightarrowtail r)$ and $P'_{k\not=s(l)}=\textit{Trees}_\Sigma$ and we are done.

If $\sigma_n=update_{l,n}$: because $exec(tr,s)$ is defined it must be the case that $exec(tr'_0,s[l\coloneqq n])$ is also defined and $exec(tr'_0,s[l\coloneqq n])=exec(tr,s)=(\downarrow,r)$. Therefore $tr'_0\in(s[l\coloneqq n]\rightarrowtail r)$. We can choose $P'_0=(s[l\coloneqq n]\rightarrowtail r)$ and $P'_{k\not=0}=\textit{Trees}_\Sigma$ and we are done.
\end{example}

\begin{example}[Input/output]
Consider a tree $tr=\sigma_n(tr'_0,tr'_1,tr'_2,\ldots)\in \langle w\rangle_{\ldots}$. If $w=\epsilon$, then decomposability is immediately satisfied by choosing $P'_{k}=\langle\epsilon\rangle_{\ldots}$. Assume $w\not=\epsilon$.

If $\sigma_n=read_n$, it must be the case that $w=(?k)w'$ and $tr'_k\models\langle w'\rangle_{\ldots}$. We can choose $P'_k=\langle w'\rangle_{\ldots}$ and $P'_{m\not=k}=\langle\epsilon\rangle_{\ldots}$ and we are done.

If $\sigma_n=write_n$, then $w=(!n)w'$ and $tr'_0\models\langle w'\rangle_{\ldots}$. Choose $P'_0=\langle w'\rangle_{\ldots}$ and $P'_{k\not=0}=\langle\epsilon\rangle_{\ldots}$ and we are done.
\end{example}

\section{Howe's Method}\label{Sec_howe}

To prove Theorem \ref{Thm_sim_compat} we will use a method originally due to Howe \cite{DBLP:journals/iandc/Howe96}. The strategy is to define a relation $\precsim^\mathcal{H}$ named the Howe extension of $\precsim$, prove that it is compatible, and then prove that it coincides with $\precsim^\circ$. 

The justification that Pitts \cite{Pit11} gives for the use of Howe's method in the case of the untyped $\lambda$-calculus applies to ECPS as well. A direct proof that $\precsim^\circ$ is compatible is problematic because it requires proving a substitutivity property of $\precsim^\circ$ which is very close to compatibility. Howe's method avoids this problem because $\precsim^\mathcal{H}$ is compatible by construction, and can be proved substitutive.

The proof of compatibility of applicative $\mathfrak{P}$-bisimilarity follows a similar structure to that for applicative bisimilarity for EPCF, found in \cite[Appendix]{SimV17} and \cite[Section 6]{SimV18}. We present all the proofs in detail, filling in gaps in Simpson's and Voorneveld's presentation, and adapting to the setting of ECPS.

\begin{definition}[Compatibe refinement]\label{Def_compat_refinement}
Given a  well-typed open relation $\mathcal{R}$ its compatible refinement $\widehat{\mathcal{R}}$ is inductively defined by the rules in Figure \ref{Fig_compat_refinement}.
\end{definition}

\begin{figure}
\begin{gather*}
\inferrule{ }{\Gamma\vdash x\ \widehat{\mathcal{R}}^\mathfrak{v}_{A}\ x} \textsc{(C1)} \quad
\inferrule{ }{\Gamma\vdash \star\ \widehat{\mathcal{R}}^\mathfrak{v}_{\mathtt{unit}}\ \star}\textsc{(C2)}	\quad
\inferrule{\Gamma,x_1:A_1,\ldots,x_n:A_n \vdash s\ \mathcal{R}^\mathfrak{c}\ t}{\Gamma \vdash \lbd{\overrightarrow{x}}{\overrightarrow{A}}{s}\ \widehat{\mathcal{R}}^\mathfrak{v}_{\neg(A_1,\ldots,A_n)}\ \lbd{\overrightarrow{x}}{\overrightarrow{A}}{t}} \textsc{(C3)}	\\
\inferrule{ }{\Gamma\vdash \mathtt{zero}\ \widehat{\mathcal{R}}^\mathfrak{v}_{\mathtt{nat}}\ \mathtt{zero}} \textsc{(C4)} \quad
\inferrule{\Gamma\vdash v\ \mathcal{R}^\mathfrak{v}_{\mathtt{nat}}\ v'}{\Gamma\vdash \mathtt{succ}(v)\ \widehat{\mathcal{R}}^\mathfrak{v}_{\mathtt{nat}}\ \mathtt{succ}(v')} \textsc{(C5)}	\\
\inferrule{\Gamma\vdash v\ \mathcal{R}^\mathfrak{v}_{\neg(A_1,\ldots,A_n)}\ v' \\ \Gamma\vdash w_1\ \mathcal{R}^\mathfrak{v}_{A_1}\ w'_1,\ldots,\Gamma\vdash w_n\ \mathcal{R}^\mathfrak{v}_{A_n}\ w'_n}{\Gamma \vdash v(w_1,\ldots,w_n)\ \widehat{\mathcal{R}}^\mathfrak{c}\ v'(w'_1,\ldots,w'_n)} \textsc{(C6)} \\
\inferrule{\Gamma,x:\neg(\overrightarrow{A}) \vdash v\ \mathcal{R}^\mathfrak{v}_{\neg(\overrightarrow{A})}\ v' \\ \Gamma\vdash w_i\ \mathcal{R}^\mathfrak{v}_{A_i}\ w'_i \text{ for each i}}{\Gamma\vdash (\mufix{x}{v})(\overrightarrow{w})\ \widehat{\mathcal{R}}^\mathfrak{c}\ (\mufix{x}{v'})(\overrightarrow{w'})} \textsc{(C7)} \\
\inferrule{\Gamma\vdash v\ \mathcal{R}^\mathfrak{v}_{\mathtt{nat}}\ v' \\ \Gamma,x:\mathtt{nat}\vdash t\ \mathcal{R}^\mathfrak{c}\ t'}{\Gamma\vdash\sigma(v,x.t)\ \widehat{\mathcal{R}}^\mathfrak{c}\ \sigma(v',x.t')}\sigma\in\Sigma\ \textsc{(C8)} \quad
\inferrule{ }{\Gamma\vdash \downarrow\ \widehat{\mathcal{R}}^\mathfrak{c}\ \downarrow}\textsc{(C9)}	\\
\inferrule{\Gamma\vdash v\ \mathcal{R}^\mathfrak{v}_{\mathtt{nat}}\ v' \\ \Gamma\vdash s\ \mathcal{R}^\mathfrak{c}\ s' \\ \Gamma,x:\mathtt{nat}\vdash t\ \mathcal{R}^\mathfrak{c}\ t'}{\Gamma \vdash \caset{v}{s}{x}{t}\ \widehat{\mathcal{R}}^\mathfrak{c}\ \caset{v'}{s'}{x}{t'}}\\ \textsc{(C10)}
\end{gather*}
\caption{Compatible refinement rules.}\label{Fig_compat_refinement}
\end{figure}

\begin{definition}[Howe extension]
Given a well-typed closed relation $\mathcal{R}$, we define its Howe extension $\mathcal{R^H}$ to be the least relation $\mathcal{S}$ such that $\mathcal{S}=\mathcal{R}^\circ \circ \widehat{\mathcal{S}}$.
\end{definition}

It has been observed by Levy (\cite[Proposition~5.4]{DBLP:journals/entcs/Levy06b}) that the equation above determines a unique relation. The Howe extension can equivalently be defined inductively as \emph{the least relation closed under the rules}:
\begin{equation*}
\inferrule{\Gamma\vdash s\ \widehat{\mathcal{R}^{\mathcal{H}}}^{\mathfrak{c}}\ t	\\	\Gamma\vdash t\ \mathcal{R}^{\circ,\mathfrak{c}}\ r}{\Gamma\vdash s\ \mathcal{R}^{\mathcal{H},\mathfrak{c}}\ r}\textsc{(HC)}	\quad
\inferrule{\Gamma\vdash v\ \widehat{\mathcal{R}^{\mathcal{H}}}^{\mathfrak{v}}_A\ w	\\	\Gamma\vdash w\ \mathcal{R}^{\circ,\mathfrak{v}}_A\ u}{\Gamma\vdash v\ \mathcal{R}^{\mathcal{H},\mathfrak{v}}_A\ u} \textsc{(HV)}
\end{equation*} 
This is shown in \cite{DBLP:journals/corr/LagoGL17}.

Below are two lemmas about the open extension and the Howe extension of a relation. Their proofs appear in Appendix \ref{App_sec_howe}.

\begin{lemma}[{From \cite[Appendix]{SimV17}}]\label{Lemm_howe_extension_misc}
Given a well-typed relation $\mathcal{R}$ on closed terms that is reflexive:
\begin{enumerate}
\item The Howe extension of $\mathcal{R}$, $\mathcal{R}^{\mathcal{H}}$, is compatible and hence reflexive.
\item $\mathcal{R}^\circ\ \subseteq\ \mathcal{R}^{\mathcal{H}}$.
\end{enumerate}
\end{lemma}

\begin{lemma}[{From \cite[Appendix]{SimV17}}]\label{Lemm_howe_method_misc}
Given a well-typed relation $\mathcal{R}$ on closed terms that is transitive:
\begin{equation*}
\mathcal{R}^\circ\circ\mathcal{R}^{\mathcal{H}} \subseteq \mathcal{R}^{\mathcal{H}}.
\end{equation*}
\end{lemma}


\begin{lemma}[Substitutivity]\label{Lemm_howe_subsitutivity}
Given a well-typed relation $\mathcal{R}$ on closed terms that is transitive, its Howe extension satisfies the following two value-substitutivity properties:
\begin{enumerate}
\item $\overrightarrow{x_i:A_i},y:B \vdash s \mathrel{\mathcal{R}^{\mathcal{H},\mathfrak{c}}} t \text{ and } \overrightarrow{x_i:A_i}\vdash v \mathrel{\mathcal{R}^{\mathcal{H},\mathfrak{v}}_B} w \implies \overrightarrow{x_i:A_i} \vdash s[v/y] \mathrel{\mathcal{R}^{\mathcal{H},\mathfrak{c}}} t[w/y]$.
\item $\overrightarrow{x_i:A_i},y:B \vdash u \mathrel{\mathcal{R}^{\mathcal{H},\mathfrak{v}}_C} u' \text{ and } \overrightarrow{x_i:A_i}\vdash v \mathrel{\mathcal{R}^{\mathcal{H},\mathfrak{v}}_B} w \implies \overrightarrow{x_i:A_i} \vdash u[v/y] \mathrel{\mathcal{R}^{\mathcal{H},\mathfrak{v}}_C} u'[w/y]$.
\end{enumerate}
\end{lemma}
\begin{proof}
The proof is done by induction on the structure of $s$ and $u$. It can be found in Appendix \ref{App_sec_howe}.
\end{proof}

The following lemma will help prove that $\precsim^\mathcal{H}$ restricted to closed terms is a simulation. Its proof appears in Appendix~\ref{App_sec_howe}.

\begin{lemma}\label{Lemm_howe_ext_sim_nat}
Consider a well-typed closed relation $\leq$ that is a $\mathfrak{P}$-simulation. For any closed values $v$ and $w$:
\begin{equation*}
\vdash v \leq^{\mathcal{H},\mathfrak{v}}_{\mathtt{nat}} w \implies v=w.
\end{equation*}	
\end{lemma}

Using the domain theoretic definition of ECPS computation trees (Definition \ref{Def_ecps_tree_chain}) we can state the Key Lemma which will help us prove $\precsim^\mathcal{H}$ is a simulation. Recall that $\treet{s}_n$ is the tree resulting from $n$ steps of evaluation of $s$.

\begin{lemma}[Key Lemma]\label{Lemm_howe_key}
Consider a decomposable set of Scott-open observations $\mathfrak{P}$. Consider a well-typed closed relation $\leq$ that is a preorder and a $\mathfrak{P}$-simulation. For any closed computations $s$ and $t$, $\vdash s\leq^{\mathcal{H},\mathfrak{c}} t$ implies:
\begin{equation*}
\forall n\in\mathbb{N}.\ \forall P\in\mathfrak{P}.\ \treet{s}_n\in P \implies \treet{t}\in P.
\end{equation*}
\end{lemma}
\begin{proof}
The proof is done by induction on $n\in\mathbb{N}$. It uses the fact that observations are upwards closed and $\mathfrak{P}$ is decomposable. The full proof appears in Appendix~\ref{App_sec_howe}.
\end{proof}

\begin{proposition}\label{Howe_exten_sim}
Consider a decomposable set of Scott-open observations $\mathfrak{P}$. Consider a well-typed closed relation $\leq$ that is a preorder and a $\mathfrak{P}$-simulation. The Howe extension of $\leq$, $\leq^{\mathcal{H}}$, restricted to closed terms is an applicative simulation.
\end{proposition}
\begin{proof}
We need to verify that all four conditions in the definition of applicative simulation are satisfied by $\leq^{\mathcal{H}}$ for closed terms.
\begin{enumerate}
\item $\vdash v \leq^{\mathcal{H},\mathfrak{v}}_{\mathtt{unit}} w \implies v=w=\star$.
Assume $\vdash v \leq^{\mathcal{H},\mathfrak{v}}_{\mathtt{unit}} w$. The only closed value of type $\mathtt{unit}$ is $\star$ so $v=w=\star$.

\item\label{Howe_exten_sim2} $\vdash v \leq^{\mathcal{H},\mathfrak{v}}_{\mathtt{nat}} w \implies v=w$. This is Lemma \ref{Lemm_howe_ext_sim_nat}.

\item $\vdash s \leq^{\mathcal{H},\mathfrak{c}} t \implies \forall P\in\mathfrak{P}.\ (\treet{s}\in P \implies \treet{t}\in P)$.
Assume that $\vdash s \leq^{\mathcal{H},\mathfrak{c}} t$ and $\treet{s}\in P$ for some $P\in\mathfrak{P}$. From the Key Lemma (Lemma \ref{Lemm_howe_key}) we know that:
\begin{equation}\label{How_exten_sim_proof2}
\forall n\in\mathbb{N}.\ \forall P\in\mathfrak{P}.\ \treet{s}_n\in P \implies \treet{t}\in P.
\end{equation}
We know $\treet{s}=\bigsqcup_{m\in\mathbb{N}}\treet{s}_m\in P$ and that $\{\treet{s}_m\}_{m\in\mathbb{N}}$ is an ascending chain. By Scott-openness of $P$, there exists $j\in\mathbb{N}$ such that $\treet{s}_j\in P$.

Therefore, by equation \ref{How_exten_sim_proof2} we have the desired result: $\treet{t}\in P$.
\item $\vdash v \leq^{\mathcal{H},\mathfrak{v}}_{\neg(A_1,\ldots,A_n)} w \implies \forall \vdash u_1:A_1,\ldots,\vdash u_n:A_n.\ v(u_1,\ldots,u_n) \leq^{\mathcal{H},\mathfrak{c}} w(u_1,\ldots, u_n)$.
Assume $\vdash v \leq^{\mathcal{H},\mathfrak{v}}_{\neg(A_1,\ldots,A_n)} w$. By Lemma \ref{Lemm_howe_extension_misc} we know $\leq^{\mathcal{H}}$ is compatible and reflexive. Therefore, for all $i$, $u_i\leq^{\mathcal{H},\mathfrak{v}}_{A_i}u_i$ so by compatibility $v(u_1,\ldots,u_n) \leq^{\mathcal{H},\mathfrak{c}} w(u_1,\ldots, u_n)$. 
\end{enumerate}
\end{proof}

The following lemmas will be used to prove that similarity and bisimilarity are compatible. Their proofs can be found in Appendix \ref{App_sec_howe}.

\begin{lemma}\label{Lemm_howe_subst_impl_open_incl}
Given a well-typed open relation $\mathcal{R}$ that is reflexive and has the two substitutivity properties from Lemma \ref{Lemm_howe_subsitutivity}, and a well-typed closed relation $\mathcal{S}$ then:
\begin{center}
if $\mathcal{R}$ restricted to closed terms is included in $\mathcal{S}$ then $\mathcal{R}\subseteq\mathcal{S}^\circ$.
\end{center}
\end{lemma}

\begin{lemma}\label{Lemm_howe_reflstran_sim}
Given a $\mathfrak{P}$-simulation $\mathcal{R}$, its reflexive-transitive closure, $\mathcal{R}^*$ is also a $\mathfrak{P}$-simulation.
\end{lemma}

\begin{lemma}\label{Lemm_howe_refltran_compat}
Given a well-typed compatible relation $\mathcal{R}$, its reflexive-transitive closure $\mathcal{R}^*$ is also compatible.
\end{lemma}

\begin{lemma}[From \cite{LasPhd}]\label{Lemm_howe_lassen}
Given a well-typed closed relation $\mathcal{R}$ the following holds:
\begin{center}
if $\mathcal{R}^\circ$ is reflexive and symmetric, then $\mathcal{R}^{\mathcal{H}*}$ is symmetric.
\end{center}
Where $S^*$ denotes the reflexive-transitive closure of a relation $\mathcal{S}$.
\end{lemma}

Finally, we can prove similarity and bisimilarity are compatible. We first recall the formal statement of this:

\begin{reptheorem}{Thm_sim_compat}
Given a decomposable set of Scott-open observations $\mathfrak{P}$:
\begin{enumerate}
\item The open extension of applicative $\mathfrak{P}$-similarity, $\precsim^\circ$, is compatible, and hence it is a precongruence.
\item The open extension of applicative $\mathfrak{P}$-bisimilarity, $\sim^\circ$, is compatible, and hence it is a congruence.
\end{enumerate}
\end{reptheorem}
\begin{proof}
This proof has the same structure as the proof of Theorem~3 from \cite{SimV18}. Here, we present significant details that were missing. 
\begin{enumerate}
\item We need to prove that the open extension of applicative $\mathfrak{P}$-similarity, $\precsim^\circ$, is compatible. We know that $\precsim$ is a preorder (Lemma \ref{Lemm_(bi)sim_preord}) and a simulation. Therefore, we can apply  Proposition \ref{Howe_exten_sim} to deduce that the restriction of $\precsim^{\mathcal{H}}$ to closed terms is a simulation, so it is included in the greatest simulation, $\precsim$. 

Since $\precsim$ is transitive we know from Lemma \ref{Lemm_howe_subsitutivity} that $\precsim^{\mathcal{H}}$ has the two substitution properties. Because $\precsim$ is reflexive, we know from Lemma \ref{Lemm_howe_extension_misc} that $\precsim^{\mathcal{H}}$ is reflexive.

We can use all these to apply Lemma \ref{Lemm_howe_subst_impl_open_incl} for $\precsim^{\mathcal{H}}$ and $\precsim$ to deduce $\precsim^{\mathcal{H}}\subseteq\precsim^\circ$.

From Lemma \ref{Lemm_howe_extension_misc} we already know that $\precsim^\circ \subseteq \precsim^{\mathcal{H}}$ and that $\precsim^{\mathcal{H}}$ is compatible. Therefore, the open extension of applicative $\mathfrak{P}$-similarity equals the Howe extension, $\precsim^\circ = \precsim^{\mathcal{H}}$, and is compatible.

Since $\precsim$ is a preorder, $\precsim^\circ$ is also a preorder, so it is a precongruence.

\item We need to prove that $\sim^\circ$ is compatible. The relation $\sim$ is also a $\mathfrak{P}$-simulation. From Lemma \ref{Lemm_(bi)sim_preord} we know $\sim$ is an equivalence relation, hence a preorder. Therefore we can use Proposition \ref{Howe_exten_sim} to deduce that $\sim^{\mathcal{H}}$ restricted to closed terms is a simulation. 

From Lemma \ref{Lemm_howe_reflstran_sim}, we obtain that $\sim^{\mathcal{H}*}$ restricted to closed terms is a simulation, because restricting to closed terms and taking the reflexive-transitive closure are commutative operations.

Because $\sim$ is reflexive and symmetric it follows that $\sim^\circ$ is also reflexive and symmetric. We can then apply Lemma \ref{Lemm_howe_lassen} for $\sim$ to deduce that $\sim^{\mathcal{H}*}$ is symmetric. Therefore, $\sim^{\mathcal{H}*}$ restricted to closed terms is also symmetric.

As a result, we know that $\sim^{\mathcal{H}*}$ restricted to closed terms is a bisimulation, so it is included in the greatest bisimulation, $\sim$.

Now we would like to use Lemma \ref{Lemm_howe_subst_impl_open_incl} to deduce $\sim^{\mathcal{H}*}\subseteq\sim^\circ$. We know $\sim$ is transitive so we can apply Lemma \ref{Lemm_howe_subsitutivity} to deduce that $\sim^{\mathcal{H}}$ has the required substitutivity properties. Then it is easy to prove using transitivity and reflexivity that $\sim^{\mathcal{H}*}$ also has the required substitutivity properties. Moreover, $\sim^{\mathcal{H}*}$ is reflexive by definition. So we can apply Lemma \ref{Lemm_howe_subst_impl_open_incl}.

Because $\sim$ is reflexive we know from Lemma \ref{Lemm_howe_extension_misc} that $\sim^\circ\subseteq\sim^{\mathcal{H}}$ and that $\sim^{\mathcal{H}}$ is compatible. By definition of the reflexive transitive closure we know $\sim^{\mathcal{H}}\subseteq\sim^{\mathcal{H}*}$.

Therefore we know $\sim^{\mathcal{H}*}\subseteq\sim^\circ$ and $\sim^\circ\subseteq\sim^\mathcal{H}\subseteq\sim^{\mathcal{H}*}$. So $\sim^\circ=\sim^{\mathcal{H}*}$.

Because $\sim^{\mathcal{H}}$ is compatible, from Lemma \ref{Lemm_howe_refltran_compat},  $\sim^{\mathcal{H}*}$ is also compatible. Therefore, $\sim^\circ$ is compatible. Since $\sim$ is an equivalence relation, $\sim^\circ$ is also an equivalence relation, so it is a congruence.
\end{enumerate}
\end{proof}

\section{Chapter Summary}

This chapter started by defining a set $\mathfrak{P}$ of \emph{observations}, for each of the following effects: nondeterminisim, probabilistic choice, global store and I/O. In ECPS, we can observe termination and some effect operations, both of which are encoded by computation trees. Therefore, each observation is a set of trees. For example, for nondeterminism the observations are:
\begin{align*}
\Diamond &= \{\text{trees in which at least one execution path that can occur ends in }\downarrow\}	\\
\Box &= \{\text{trees in which all paths that can occur have finite height and end in }\downarrow\}.
\end{align*}

We used observations to define applicative bisimilarity for ECPS (Definition~\ref{Def_ecps_bisim}). Two values of ground type are bisimilar when they are equal. Two functions are bisimilar if, for all arguments, they yield bisimilar computations. Finally, two computations are bisimilar if they belong to exactly the same elements of $\mathfrak{P}$, in other words, they have the same observable behaviour. For example, two computations with the trees below are not bisimilar:
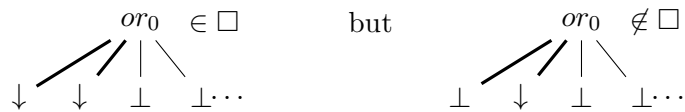
\begin{center}
\begin{tikzpicture}[level distance = 1cm, sibling distance=0.8cm]
\node (n) {$or_0$}
	child { node {$\downarrow$} edge from parent[very thick]}
	child { node {$\downarrow$} edge from parent[very thick]}
	child { node {$\bot$} }
	child { node {$\bot$} }
	child { node[left=0.1cm] {$\ldots$} edge from parent[draw=none]};
\node[right=0.1cm of n] {$\in\Box$};		
\node[right=2.2cm of n] {but};	
\node[right=5cm of n] (m) {$or_0$}
	child { node {$\bot$} edge from parent[very thick]}
	child { node {$\downarrow$} edge from parent[very thick]}
	child { node {$\bot$} }
	child { node {$\bot$} }
	child { node[left=0.1cm] {$\ldots$} edge from parent[draw=none]};
\node[right=0.1cm of m] {$\not\in\Box$};		
\end{tikzpicture}
\end{center}

Next, we defined compatibility formally (Definition~\ref{Def_compat}). The definition says that two related programs can be substituted in related contexts to yield another pair of related programs. We identified two sufficient conditions that the set of observations $\mathfrak{P}$ should satisfy in order for bisimilarity to be compatible. These are Scott-openness, which has to do with the topology of the elements of $\mathfrak{P}$, and \emph{decomposability} (Definition~\ref{Def_decomposability}).
This definition of decomposability is novel. It states that, for each tree $tr$ in an observation $P\in\mathfrak{P}$, there exist observations that characterise the restrictions that $P$ places on the children of $tr$.

The main result of this chapter is Theorem~\ref{Thm_sim_compat}, which states that, given a decomposable set $\mathfrak{P}$ of Scott-open observations, bisimilarity is compatible. The proof of this uses Howe's method (Section~\ref{Sec_howe}) and required significant effort. The idea is to define a new relation named the Howe extension of bisimilarity, which is compatible by construction, and prove it equal to bisimilarity.


\chapter{Logical Equivalence for ECPS}\label{Chap_mod_logic}

This chapter introduces the logic $\mathcal{F}$ whose formulas express properties of ECPS terms. The logic is defined using the set of observations $\mathfrak{P}$. We prove that program equivalence induced by $\mathcal{F}$ is an equivalence relation and coincides with applicative $\mathfrak{P}$-bisimilarity. Therefore, using the main theorem of Chapter~\ref{Chap_bisim}, we can deduce that $\mathcal{F}$-logical equivalence is compatible. The results in this chapter will be used in Chapter~\ref{Chap_ctx_equiv} to show that the logic $\mathcal{F}$ characterises contextual equivalence, which is our main goal.

\section{Two Logics for ECPS}\label{Sec_ecps_modal_log}

Recall the set of observations $\mathfrak{P}$, defined in Section \ref{Sec_ecps_observations}, which contains sets of ECPS effect trees. Each $P\in\mathfrak{P}$ specifies the shape of computation trees for a particular effect. Using $\mathfrak{P}$, we define two slightly different logics for ECPS named $\mathcal{V}$ and $\mathcal{F}$ respectively. 

In the logic $\mathcal{V}$ values appear inside logical formulas, whereas this is not the case in $\mathcal{F}$. Both logics make a distinction between value formulas and computation formulas. Value formulas are always associated an ECPS type.

\begin{definition}[Logic $\mathcal{F}$]\label{Def_logic_f}
The value formulas of the logic $\mathcal{F}$ are constructed from basic formulas $\phi=\{n\}$ and $(\phi_1,\ldots,\phi_n)\mapsto P$, where $P$ is an observation from $\mathfrak{P}$, according to the rules in Figure \ref{Fig_logic_formls}. In these rules, $A$ stands for an ECPS type. The computation formulas are the elements of $\mathfrak{P}$.

The satisfaction relation $\models$ relates a closed value $\vdash v : A$ to a value formula $\phi : A$ of the same type, or a closed computation $t$ to an observation $P$. The definition of $\models$ appears in Figure \ref{Fig_logical_sat}. Intuitively, $v\models\phi$ means that the program $v$ has property $\phi$.

Let $\mathcal{F}^+$ be the fragment of $\mathcal{F}$ without negation.
\end{definition}

\begin{figure}
\begin{gather*} 
\inferrule{n\in\mathbb{N}}{\{n\} : \mathtt{nat}}  \textsc{(nat)} \quad
\inferrule{\phi_1 : A_1 \ldots \phi_n : A_n}{(\phi_1,\ldots,\phi_n)\mapsto P : \neg(A_1,\ldots,A_n)}P\in\mathfrak{P} \ \textsc{(val)}	\\
\inferrule{(\phi_i : A)_{i\in I}}{\lor_{i\in I}\phi_i : A}\textsc{(disj)}	\quad
\inferrule{(\phi_i : A)_{i\in I}}{\land_{i\in I}\phi_i : A}\textsc{(conj)}	\quad
\inferrule{\phi : A}{\neg\phi : A}\textsc{(neg)}
\end{gather*}
\caption{Value formulas in the logic $\mathcal{F}$.}\label{Fig_logic_formls}
\end{figure}

\begin{figure}[t]
\begin{align*}
v \models \{n\} &\quad\Longleftrightarrow\quad v=\overline{n}	\\
v \models (\phi_1,\ldots,\phi_n)\mapsto P &\quad\Longleftrightarrow\quad \text{for all closed values } w_1,\ldots,w_n \text{ such that } w_i\models\phi_i	\\ &\qquad \qquad \text{ then } \treet{v(w_1\ldots w_n)} \in P	\\
v \models \lor_{i\in I}\phi_i &\quad\Longleftrightarrow\quad \text{there exists } j\in I \text{ such that } v\models \phi_j	\\
v \models \land_{i\in I}\phi_i &\quad\Longleftrightarrow\quad \text{for all } j\in I \text{, } v\models \phi_j	\\
v \models \neg\phi &\quad\Longleftrightarrow\quad\text{it is false that } v \models \phi	\\
t \models P &\quad\Longleftrightarrow\quad \treet{t} \in P	
\end{align*}
\caption{Satisfaction relation $\models$ for the logic $\mathcal{F}$.} \label{Fig_logical_sat}
\end{figure}

\begin{definition}[Logic $\mathcal{V}$]
The logic $\mathcal{V}$ is the same as $\mathcal{F}$ except that the $\textsc{(val)}$ rule is replaced by:
\begin{equation*}
\inferrule{\vdash w_1 : A_1 \ldots \vdash w_n : A_n}{(w_1,\ldots,w_n)\mapsto P : \neg(A_1,\ldots,A_n)}P\in\mathfrak{P}\ \textsc{(val')}
\end{equation*}

\begin{equation*}
v \models (w_1,\ldots,w_n)\mapsto P \quad\Longleftrightarrow\quad \treet{v(w_1\ldots w_n)} \in P.
\end{equation*}
That is, formulas of function type are now constructed using ECPS values.

Let $\mathcal{V}^+$ be the fragment of $\mathcal{V}$ without negation.
\end{definition}

Notice that a computation formula is just one of the observations $P\in\mathfrak{P}$. For example, it can be $\Box$, $\mathbf{P}_{>0.5}$, $(s\rightarrowtail r)$ depending on the effects present in the language.  Therefore, $t\models P$ tests the shape of the computation tree of $t$ without looking at its possible return values. This is consistent with the fact that ECPS computations are not expected to return. 

However, this is unlike computation formulas in EPCF logic, $o\phi$, which test whether return values satisfy $\phi$ (Section~\ref{Sec_epcf_results}). Therefore, $P$ is no longer a modality in the same sense as $o$ because it does not lift formulas. This is why we called $P$ an \emph{observation}. In the logic $\mathcal{F}$, it can still be argued that $P$ is a modality in the traditional sense because it takes the value formula $\phi$ to another value formula $\phi\mapsto P$.

In both $\mathcal{F}$ and $\mathcal{V}$ the value formulas of type $\mathtt{nat}$ are obtained from the natural numbers, arbitrary conjunctions and disjunctions, and negation. However, the basic formulas of function type are different. In $\mathcal{V}$:
\begin{equation*}
v\models(\lbd{x}{\mathtt{nat}}{\downarrow})\mapsto \Diamond
\end{equation*}
means that $v:\neg(\neg(\mathtt{nat}))$ \emph{may} terminate when given argument $\lbd{x}{\mathtt{nat}}{\downarrow}$. In $\mathcal{F}$, an analogous statement is:
\begin{equation*}
v \models ((\lor_{n\in\mathbb{N}}\{n\})\mapsto\Box)\mapsto\Diamond.
\end{equation*}
This says that $v$ may terminate  when given as argument a function that satisfies $\psi=(\lor_{n\in\mathbb{N}}\{n\})\mapsto\Box$. Notice that $(\lbd{x}{\mathtt{nat}}{\downarrow})$ indeed satisfies $\psi$.

There is no need to include logical connectives at the level on computation formulas because they can be encoded in value formulas. For example $\phi\mapsto \land_{i\in I}P_i$ can be expressed as:
\begin{equation*}
\land_{i\in I}(\phi\mapsto P_i).
\end{equation*}
The statement $t\models\land_{i\in I} P_i$ can instead be expressed as:
\begin{equation*}
\lbd{x}{\mathtt{unit}}{t} \models \land_{i\in I}(true\mapsto P_i).
\end{equation*}

The indexing set $I$ in $\land_{i\in I}$ may be uncountable. However, the sets of values and computations are countable. Since logical formulas are interpreted over values and computations, all conjunctions and disjunctions are semantically equivalent to countable ones.

The following example compares logical formulas for EPCF with $\mathcal{F}$-formulas:

\begin{example}[Nondeterminism]
Recall the EPCF logic formula $\phi_2=\{3\}\mapsto\Diamond\{2\}$ from Examples~\ref{Eg_boxdi} and \ref{Eg_nondet_epcf_logic} and the function:
\begin{equation*}
g =\lbd{n}{\mathbbm{N}}{or(\mathbf{pred}\ n,\ \mathbf{succ}\ n)}.
\end{equation*}
We have previously established that $g\models_{\textit{EPCF}}\phi_2$ because, when given argument $\overline{3}$ in the empty stack $id$, $(id,g\ \overline{3})$ \emph{may} return $\overline{2}$.

Now consider the CPS translation (Section~\ref{Sec_cps_trans}) of the function $g$ into ECPS, named $g^*$. Apart from the natural number argument that $g$ receives, $g^*$ also receives a continuation $k$ to which it passes its result.\\
\begin{tabular}{c}
\begin{lstlisting}[mathescape=true]
$g^*=\lambda(n,k){:}(\mathtt{nat},\neg\mathtt{nat}).$
   $(\lambda k'{:}\neg\mathtt{nat}.or(\overline{0},\ x.(\lambda(y,k''){:}(\mathtt{nat},\neg\mathtt{nat}).$
                   $(\lambda k''':\neg\mathtt{nat.}\mathtt{case}\ y\ \mathtt{in}\ \{\mathtt{zero}\Rightarrow(\mathbf{pred}\ n)^*\ k''',$
                            $\mathtt{succ}(y')\Rightarrow\mathtt{case}\ y'\ \mathtt{in}\ \{\mathtt{zero}\Rightarrow(\mathbf{succ}\ n)^*\ k''',$
                                             $\mathtt{succ}(y'')\Rightarrow loop^*\ k'''\}\}$
                   $)\ k''$                          
                 $)\ (x,k')$
            $)$
   $)\ k$.
\end{lstlisting}
\end{tabular}\\
The tree of $(g^*\ (\overline{3},k))$ is:
\begin{center}
\begin{tikzpicture}[level distance=1cm, sibling distance=1cm]
\node (n) {$or_0$} 
	child { node {$\treet{k\ \overline{2}}$}  edge from parent[very thick]}
	child { node {$\treet{k\ \overline{4}}$} edge from parent[very thick]}
	child {node {$\bot$} }
	child {node[left=0.05cm] {$\ldots$} edge from parent[draw=none] };
\node [left=0.1cm of n] {$\treet{g^*\ (\overline{3},k)}=$};
\node [right=0.5cm of n] {in particular};
\node [right=9cm of n] (n1) {$or_0$} 
	child { node {$\downarrow$} edge from parent[very thick]}
	child { node {$\downarrow$} edge from parent[very thick]}
	child {node {$\bot$} }
	child {node[left=0.05cm] {$\ldots$} edge from parent[draw=none] };
\node [left=0.1cm of n1] {$\treet{g^*\ (\overline{3},id^*)}=\treet{g^*\ (\overline{3},\lbd{x}{\mathtt{nat}}{\downarrow})}=$};	
\end{tikzpicture}
\end{center}

The formula $\phi_2$ could be translated to the $\mathcal{F}$ formula:
\begin{equation*}
\phi_2^*=((\{3\},\{2\}\mapsto\Diamond)\mapsto\Diamond) \land ((\{3\},\{2\}\mapsto\Box)\mapsto\Diamond).
\end{equation*}

Intuitively, $(\{3\},\{2\}\mapsto\Diamond)\mapsto\Diamond$ says that, when $g^*$ is given as arguments $\overline{3}$ and a continuation $k$ that satisfies $\{2\}\mapsto\Diamond$,
the computation \emph{may} eventually terminate. So if $g$ indeed returns $\overline{2}$ this formula will be satisfied. Similarly for $(\{3\},\{2\}\mapsto\Box)\mapsto\Diamond$.
Therefore, we can see that $g\models_{\textit{EPCF}}\phi_2$ implies $g^*\models_\mathcal{F}\phi_2^*$.

The conjunction in $\phi_2^*$ is over all nontrivial $P\in\mathfrak{P}$ that appear in $\{2\}\mapsto P$, namely $\Diamond$ and $\Box$.

 
\end{example}

\section{Logical Equivalence}

The logics $\mathcal{F}$ and $\mathcal{V}$ induce a notion of program equivalence defined below:

\begin{definition}[Logical preorder and equivalence]
Consider a fragment $\mathcal{L}$ of one of the logics $\mathcal{F}$ or $\mathcal{V}$. For any closed values $\vdash v_1:A$ and $\vdash v_2:A$:
\begin{equation*}
v_1 \sqsubseteq_{\mathcal{L}} v_2 \quad \Longleftrightarrow \quad \forall\phi:A \text{ in }\mathcal{L}.\ (v_1\models_{\mathcal{L}} \phi \implies v_2\models_{\mathcal{L}} \phi).
\end{equation*}
And for any closed computations $\vdash s_1$ and $\vdash s_2$:
\begin{equation*}
s_1 \sqsubseteq_{\mathcal{L}} s_2 \quad \Longleftrightarrow \quad \forall P\text{ in }\mathcal{L}.\ (s_1\models_{\mathcal{L}} P \implies s_2\models_{\mathcal{L}} P).
\end{equation*}
Two terms (values or computations) are logically equivalent, $t_1\equiv_{\mathcal{L}} t_2$, if $t_1\sqsubseteq_{\mathcal{L}} t_2$ and $t_2\sqsubseteq_{\mathcal{L}} t_1$.
\end{definition}

The definition of logical equivalence provides a convenient way of proving that two programs are \emph{not} equivalent: we just need to find a formula that one of them does not satisfy but the other does. For example:

\begin{example}[Probabilistic choice]
Consider the following ECPS functions, where $m_{\overline{1},\overline{2}}$, $n_{\overline{1},\overline{2},\overline{1},\overline{3}}$ and $n'_{\overline{1},\overline{2},\overline{1},\overline{3}}$ are defined as in Example \ref{Eg_prob_ecps_bisim}:
\begin{align*}
f_1 &=\lbd{x}{\mathtt{nat}}{m_{\overline{1},\overline{2}}}	\\
f_2 &=\lbd{x}{\mathtt{nat}}{\caset{x}{n_{\overline{1},\overline{2},\overline{1},\overline{3}}}{y}{n'_{\overline{1},\overline{2},\overline{1},\overline{3}}}}.
\end{align*}
Consider the $\mathcal{F}$-formula:
\begin{equation*}
\phi = \{4\}\mapsto \mathbf{P}_{>0.9}.
\end{equation*}
Formula $\phi$ distinguishes between these two functions because: $f_1\models_\mathcal{F}\phi$ but $f_2\not\models_\mathcal{F}\phi$. This can be seen by looking at their computation trees:
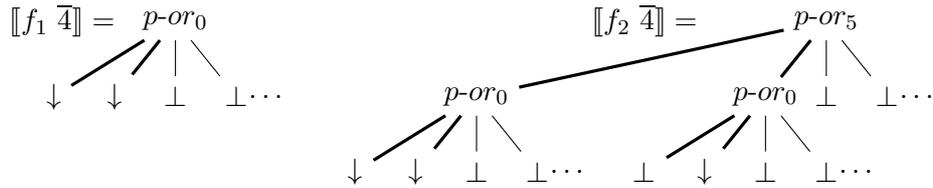
\begin{center}
\begin{tikzpicture}[level distance=1cm, sibling distance=0.8cm]
\node (m) {$p\text{-}or_0$}
	child {node {$\downarrow$} edge from parent[very thick]}
	child {node {$\downarrow$} edge from parent[very thick]}
	child {node {$\bot$} }
	child {node {$\bot$} }
	child {node[left=0.05cm] {$\ldots$} edge from parent[draw=none] };
\node [left=0.1cm of m] {$\treet{f_1\ \overline{4}}=$};	
\node (n) [right = 8cm] {$p\text{-}or_5$}
	child[sibling distance=2.3cm] {node {$p\text{-}or_0$} edge from parent[very thick]
		child[sibling distance=0.8cm] {node {$\downarrow$} }
		child[sibling distance=0.8cm] {node {$\downarrow$} } 
		child[sibling distance=0.8cm, thin] {node {$\bot$} }
		child[sibling distance=0.8cm, thin] {node {$\bot$} }
		child[sibling distance=0.8cm, thin] {node[left=0.05cm] {$\ldots$} edge from parent[draw=none] }}
	child {node (m) {$p\text{-}or_0$} edge from parent[very thick]
		child {node {$\bot$} }
		child {node {$\downarrow$} } 
		child[thin] {node {$\bot$} }
		child[thin] {node {$\bot$} }
		child[thin] {node[left=0.05cm] {$\ldots$} edge from parent[draw=none] }}
	child {node {$\bot$} }	
	child {node {$\bot$} }
	child {node[left=0.05cm] {$\ldots$} edge from parent[draw=none] };	
\node [left=1cm of n] {$\treet{f_2\ \overline{4}}=$};	
\end{tikzpicture}
\end{center}
Therefore, $f_1$ and $f_2$ are not $\mathcal{F}$-logically equivalent. However, we can see that $f_1\models_\mathcal{F} \{0\}\mapsto\mathbf{P}_{>0.9}$ and $f_2\models_\mathcal{F} \{0\}\mapsto\mathbf{P}_{>0.9}$.
\end{example}

It is the program equivalence induced by $\mathcal{F}$, rather than $\mathcal{V}$, that we are mostly interested in. This is because $\mathcal{F}$ enforces a natural separation between ECPS terms and program properties. Using $\mathcal{F}$, we can specify program properties without knowing the syntax of the programming language.

It can be easily seen that $\sqsubseteq_{\mathcal{F}^+}$ is a preorder and $\equiv_{\mathcal{F}}$ is an equivalence relation. Compatibility, the property that related programs can be substituted for variables in related contexts, is established in the main theorem of this chapter. The proof of this theorem appears at the end of the next section.

\begin{theorem}\label{Cor_bisim_is_log_equiv}
Given a decomposable set $\mathfrak{P}$ of Scott-open observations:
\begin{enumerate}
\item Applicative $\mathfrak{P}$-similarity, $\precsim$, coincides with the logical preorder induced by the logic $\mathcal{F}^+$, $\sqsubseteq_{\mathcal{F}^+}$. Hence, the open extension of the $\mathcal{F}^+$-logical preorder $\sqsubseteq_{\mathcal{F}^+}^\circ$ is compatible.

\item Applicative $\mathfrak{P}$-bisimilarity, $\sim$, coincides with the logical equivalence induced by the logic $\mathcal{F}$, $\equiv_{\mathcal{F}}$. Hence, the open extension of $\mathcal{F}$-logical equivalence $\equiv_{\mathcal{F}}^\circ$ is compatible. 
\end{enumerate}
\end{theorem}

This theorem is important because, when combined with the result of the next chapter, it shows that the logic $\mathcal{F}$ characterises contextual equivalence. This answers the main question asked in the introduction.

\section{Logical Equivalence Coincides with Bisimilarity}\label{Sec_log_equiv_is_bisim}

The aim of this section is to show that program equivalence induced by the logic $\mathcal{F}$ coincides with applicative $\mathfrak{P}$-bisimilarity, defined in the previous chapter. First, we show that this is the case for the logic $\mathcal{V}$. The proof appears in Appendix~\ref{App_modal_log}.

\begin{proposition}\label{Prop_sim_coinc_log}
Given a decomposable set $\mathfrak{P}$ of Scott-open observations:
\begin{enumerate}
\item Applicative $\mathfrak{P}$-similarity, $\precsim$, coincides with the logical preorder induced by the logic $\mathcal{V}^+$, $\sqsubseteq_{\mathcal{V}^+}$. Therefore, the open extension of $\sqsubseteq_{\mathcal{V}^+}$ is compatible.

\item Applicative $\mathfrak{P}$-bisimilarity, $\sim$, coincides with the logical equivalence induced by the logic $\mathcal{V}$, $\equiv_{\mathcal{V}}$. Therefore, the open extension of $\equiv_{\mathcal{V}}$ is compatible.
\end{enumerate}
\end{proposition}

Next, we show that the logics $\mathcal{F}$ and $\mathcal{V}$ are in fact equivalent. This is done by translating $\mathcal{F}$-formulas into $\mathcal{V}$, and vice-versa, and proving that the satisfaction relation is preserved.
Define a translation from $\mathcal{F}$ to $\mathcal{V}$, $(-)^\flat$, and a translation from $\mathcal{V}$ to $\mathcal{F}$, $(-)^\sharp$. The definition appears in Figure \ref{Fig_logic_trans}. It makes use of the following lemma, which is similar to a lemma for EPCF proved in \cite{SimV18}:

\begin{lemma}[Characteristic formula] \label{Lem_logic_char_form}
For any fragment $\mathcal{L}$ of $\mathcal{F}$ or $\mathcal{V}$  closed under countable conjunction it is true that
\begin{center}
for any closed value $v$ there exists a formula $\chi_v\in\mathcal{L}$ such that:
\end{center} 
\begin{equation*}
u \models_{\mathcal{L}} \chi_{v} \Longleftrightarrow v \sqsubseteq_{\mathcal{L}} u.
\end{equation*}
\end{lemma}
\begin{proof}
For each closed value $u$ such that $v\not\sqsubseteq_{\mathcal{L}}u$ we can choose a formula $\phi_u$ such that $v\models_{\mathcal{L}} \phi_u$ but $u\not\models_{\mathcal{L}} \phi_u$. Define $\chi_v$ to be:
\begin{equation*}
\chi_v = \land_{\{u\mid  v\not\sqsubseteq_{\mathcal{L}}u\}}\phi_u.
\end{equation*}
We can see that $u\not\models_{\mathcal{L}} \chi_v \Longleftrightarrow v\not\sqsubseteq_{\mathcal{L}}u$, which is what we need.
\end{proof}

\begin{figure}
\begin{gather*}
((\phi_1,\ldots,\phi_n)\mapsto P)^\flat = \bigwedge\{(w_1,\ldots,w_n)\mapsto P \mid w_1\models_{\mathcal{V}}\phi_1^\flat,\ldots,w_n\models_{\mathcal{V}}\phi_n^\flat \}	\\
((w_1,\ldots,w_n)\mapsto P)^\sharp = (\chi_{w_1},\ldots,\chi_{w_n}) \mapsto P
\end{gather*}
\begin{align*}
\{n\}^\flat &= \{n\}	&\{n\}^\sharp &= \{n\}	\\
P^\flat &= P	&P^\sharp &= P	\\
(\lor_{i\in I}\phi_i)^\flat &= \lor_{i\in I}\phi_i^\flat	&(\lor_{i\in I}\phi_i)^\sharp &= \lor_{i\in I}\phi_i^\sharp	\\
(\land_{i\in I}\phi_i)^\flat &= \land_{i\in I}\phi_i^\flat	&(\land_{i\in I}\phi_i)^\sharp &= \land_{i\in I}\phi_i^\sharp	\\
(\neg\phi)^\flat &= \neg\phi^\flat	&(\neg\phi)^\sharp&=\neg\phi^\sharp
\end{align*}
\begin{center}
The formula $\chi_{w_i}$ is the characteristic formula of $w_i$ in the logic $\mathcal{F}$, from Lemma \ref{Lem_logic_char_form}.
\end{center}
\caption{Translation from $\mathcal{F}$ to $\mathcal{V}$ and vice-versa.}	\label{Fig_logic_trans}
\end{figure} 

\begin{theorem}\label{Thm_logic_preord_equiv}
Given a decomposable set $\mathfrak{P}$ of Scott-open observations, the logics $\mathcal{F}^+$ and $\mathcal{V}^+$ are equi-expressive.
\begin{enumerate}
\item \label{rtp_equi_1} For any type $A$, for any formula $\phi$ in $\mathcal{F}^+$, $\phi:A$ implies that for any value $\vdash v:A$:
\begin{equation*}
v \models_{\mathcal{F}^+} \phi \Longleftrightarrow v \models_{\mathcal{V}^+} \phi^\flat.
\end{equation*}
For any $P\in\mathfrak{P}$ and any computation $\vdash t$:
\begin{equation*}
t \models_{\mathcal{F}^+} P \Longleftrightarrow t \models_{\mathcal{V}^+} P^\flat.
\end{equation*}
\item \label{rtp_equi_2} For any type $A$, for any formula $\phi$ in $\mathcal{V}^+$, $\phi:A$ implies that for any value $\vdash v:A$:
\begin{equation*}
v \models_{\mathcal{V}^+} \phi \Longleftrightarrow v \models_{\mathcal{F}^+} \phi^\sharp.
\end{equation*}
For any $P\in\mathfrak{P}$ and any computation $\vdash t$:
\begin{equation*}
t \models_{\mathcal{V}^+} P \Longleftrightarrow t \models_{\mathcal{F}^+} P^\sharp.
\end{equation*}
\end{enumerate}
\end{theorem}
\begin{proof}
\textbf{Statement \ref{rtp_equi_1}.} For computation formulas the result is immediate because they do not change when translated. 

For value formulas we prove the following property:
\begin{equation*}
\Phi(\phi, A) = (\phi:A \implies (\forall \vdash v:A.\ v \models_{\mathcal{F}^+} \phi \Longleftrightarrow v \models_{\mathcal{V}^+} \phi^\flat))
\end{equation*}
by induction on the rules in Figure \ref{Fig_logic_formls}, which specify when $\phi:A$ is well-formed.

In the case \textsc{(nat)}, $\phi=\{n\}$. The equivalence holds because $\{n\}^\flat = \{n\}$ and the satisfaction relation does not change with the translation. 

The cases for the logical connectives follow from the induction hypothesis. 

In the case \textsc{(val)}, $\phi = (\phi_1,\ldots,\phi_n)\mapsto P$. Let $v\models_{\mathcal{F}^+}\phi$ and consider some arbitrary $w_1\models_{\mathcal{V}^+}\phi_1^\flat,\ldots,w_n\models_{\mathcal{V}^+}\phi_n^\flat$. By the induction hypothesis we know ${w_1\models_{\mathcal{F}^+}\phi_1},\ldots,$ ${w_n\models_{\mathcal{F}^+}\phi_n}$. So by assumption $\treet{v(w_1,\ldots,w_n)} \in P$. Therefore it is true that, in $\mathcal{V}^+$, $v$ satisfies $(w_1,\ldots,w_n)\mapsto P$ so in general $v \models_{\mathcal{V}^+} \phi^\flat$.

For the reverse implication let $v\models_{\mathcal{V}^+}\phi^\flat$ and consider some arbitrary $w_1\models_{\mathcal{F}^+}\phi_1,\ldots,$ $w_n\models_{\mathcal{F}^+}\phi_n$. By the induction hypothesis $w_1\models_{\mathcal{V}^+}\phi_1^\flat,\ldots,w_n\models_{\mathcal{V}^+}\phi_n^\flat$, so by assumption $\treet{v(w_1,\ldots,w_n)} \in P$. Therefore $v\models_{\mathcal{F}^+}\phi$ as required.

\paragraph{Statement \ref{rtp_equi_2}.} For computation formulas $P^\sharp=P$ so the equivalence holds. 

For value formulas proceed by induction on the type $A$. If $A=\mathtt{nat}$, then the formulas $\phi$ and $\phi^\sharp$ represent the same set of natural numbers. Therefore, $v\models_{\mathcal{V}^+}\phi$ is equivalent to $v\models_{\mathcal{F}^+}\phi^\sharp$. For $A=\mathtt{unit}$ the only formulas are $true$ and $false$ so the equivalence holds trivially. For $A=\neg(B_1,\ldots,B_n)$ the induction hypothesis is, for each $B_i$:
\begin{equation*}
\text{For any formula }\phi'\text{ in }\mathcal{V}^+\text{, }\phi':B_i\text{ implies that for any value }\vdash v:B_i\text{:}
\end{equation*}
\begin{equation*}
v \models_{\mathcal{V}^+} \phi' \Longleftrightarrow v \models_{\mathcal{F}^+} \phi'^\sharp.
\end{equation*}

 We do an additional induction on $\phi$:

\paragraph{Case $\phi=(w_1,\ldots,w_n)\mapsto P : \neg(B_1,\ldots,B_n)$.} Assume $v\models_{\mathcal{V}^+}\phi$, that is $\treet{v(w_1,\ldots,w_n)} \in P$. We need to prove that for any $u_1:B_1,\ldots,u_n:B_n$ such that $u_i\models_{\mathcal{F}^+}\chi_{w_i}$  for all $i$, we have $\treet{v(u_1,\ldots,u_n)} \in P$.

By the definition of $\chi_{w_i}$ we know that $w_i \sqsubseteq_{\mathcal{F}^+} u_i$. We can show $w_i \sqsubseteq_{\mathcal{V}^+} u_i$ as follows: consider an arbitrary $\psi:B_i$ such that $w_i \models_{\mathcal{V}^+}\psi$. Then by the induction hypothesis for the type $B_i$ we know $w_i \models_{\mathcal{F}^+}\psi^\sharp$. Hence deduce $u_i \models_{\mathcal{F}^+}\psi^\sharp$ from $w_i \sqsubseteq_{\mathcal{F}^+} u_i$. Again from the induction hypothesis for $B_i$, we have $u_i \models_{\mathcal{V}^+}\psi$, as required.

Now that we have $w_i \sqsubseteq_{\mathcal{V}^+} u_i$ we can use compatibility of $\sqsubseteq_{\mathcal{V}^+}$, Proposition \ref{Prop_sim_coinc_log}, and reflexivity to deduce:
\begin{equation*}
v(w_1,\ldots,w_n) \sqsubseteq_{\mathcal{V}^+} v(u_1,\ldots,u_n).
\end{equation*}
So from $\treet{v(w_1,\ldots,w_n)} \in P$ we get the desired result $\treet{v(u_1,\ldots,u_n)} \in P$.

For the reverse implication assume $v\models_{\mathcal{F}^+}(\chi_{w_1},\ldots,\chi_{w_n})\mapsto P$. We need to prove that $\treet{v(w_1,\ldots,w_n)} \in P$. This follows from the fact that $w_i\models_{\mathcal{F}^+}\chi_{w_i}$ because $\sqsubseteq_{\mathcal{F}^+}$ is reflexive.

\paragraph{Case $\phi=\lor_{i\in I}\varphi_i : \neg(B_1,\ldots,B_n)$.} From the type of $\phi$ we know that for all $i$, $\varphi_i : \neg(B_1,\ldots,B_n)$. This means that the induction hypothesis for $\varphi_i$ gives us:
\begin{equation*}
\text{For any value }\vdash v:\neg(B_1,\ldots,B_n)\text{:}
\end{equation*}
\begin{equation*}
v \models_{\mathcal{V}^+} \varphi_i \Longleftrightarrow v \models_{\mathcal{F}^+} \varphi_i^\sharp.
\end{equation*}

Assume $v\models_{\mathcal{V}^+}\lor_{i\in I}\varphi_i$. There exists $j\in I$ such that $v\models_{\mathcal{V}^+}\varphi_j$. By the induction hypothesis for $\varphi_j$ we have that $v\models_{\mathcal{F}^+}\varphi_j^\sharp$. So $v\models_{\mathcal{F}^+} \phi^\sharp$.

The reverse implication is analogous.

\paragraph{Case $\phi=\land_{i\in I}\varphi_i : \neg(B_1,\ldots,B_n)$.} Analogous to the previous case.
\end{proof}

\begin{theorem}\label{Thm_logic_equiv_equiv}
Given a decomposable set $\mathfrak{P}$ of Scott-open observations, 
the logics $\mathcal{F}$ and $\mathcal{V}$ are equi-expressive.
\end{theorem}
\begin{proof}
We need to prove the same statements as in Theorem \ref{Thm_logic_preord_equiv}, where $\mathcal{V}^+$ is replaced by $\mathcal{V}$ and $\mathcal{F}^+$ is replaced by $\mathcal{F}$. The proof is very similar and the differences are pointed out in Appendix \ref{App_modal_log}.
\end{proof}

Using all the results in this section we can finally prove Theorem \ref{Cor_bisim_is_log_equiv}, which says that $\mathcal{F}$-logical equivalence coincides with bisimilarity. 

\begin{proof}[Proof of Theorem \ref{Cor_bisim_is_log_equiv}]
From Theorems \ref{Thm_logic_preord_equiv} and \ref{Thm_logic_equiv_equiv} we can deduce that:
\begin{gather*}
v \sqsubseteq_{\mathcal{V}^+} u \quad \Longleftrightarrow \quad v \sqsubseteq_{\mathcal{F}^+} u	\\
v \equiv_{\mathcal{V}} u \quad \Longleftrightarrow \quad v \equiv_{\mathcal{F}} u
\end{gather*}
and similarly for computations.

Then by Proposition \ref{Prop_sim_coinc_log} we have $(\precsim)=(\sqsubseteq_{\mathcal{F}^+})$ and $(\sim)=(\equiv_{\mathcal{F}})$. From Theorem~\ref{Thm_sim_compat} we know $\precsim^\circ$ and $\sim^\circ$ are compatible, so this is also the case for $\sqsubseteq_{\mathcal{F}^+}^\circ$ and $\equiv_{\mathcal{F}}^\circ$.
\end{proof}

Notice that the proofs of Theorems \ref{Thm_logic_preord_equiv} and \ref{Thm_logic_equiv_equiv} make use of compatibility of $\sqsubseteq_{\mathcal{V}^+}$ and $\equiv_\mathcal{V}$, which was established via Howe's method. So a direct proof of Theorem \ref{Cor_bisim_is_log_equiv} would require us to prove a compatibility property of $\sqsubseteq_{\mathcal{F}^+}$ first. As we have seen in the previous chapter, proofs of compatibility are laborious. Therefore, the method of going through the logic $\mathcal{V}$ to prove Theorem \ref{Cor_bisim_is_log_equiv} is justified.

\section{Chapter Summary}

This chapter defined a logic $\mathcal{F}$ in which each formula expresses a property of an ECPS program (Definition~\ref{Def_logic_f}). Formulas which concern computations are elements of the set of observations $\mathfrak{P}$. Formulas for function values have the form $(\phi_1,\phi_2,\ldots,\phi_n)\mapsto P$. They assert that, if the arguments $x_1,\ldots, x_n$ of a function satisfy $\phi_1,\ldots,\phi_n$ respectively, the resulting computation is in $P\in\mathfrak{P}$.

Recall the successor function from equation~\ref{Eq_ecps_succ_func}:
\begin{equation*}
f = \lbd{(n,k)}{(\mathtt{nat},\neg\mathtt{nat})}{(k\ \mathtt{succ}(n))} : \neg(\mathtt{nat},\neg\mathtt{nat}).
\end{equation*}
We can see that it satisfies the following formula:
\begin{equation*}
\phi=(\{2\},\{3\}\mapsto\Diamond)\mapsto\Diamond.
\end{equation*}
This says that, given argument $\overline{2}$ and a continuation $k$ which may terminate for input $\overline{3}$, the body of $f$ may terminate. Consider the formula $\phi'=\{2\}\mapsto\Diamond\{3\}$, which describes a program in direct style that may return $\overline{3}$. This is similar to the formulas discussed in Example~\ref{Eg_boxdi}, but it is not a valid $\mathcal{F}$-formula. Formula $\phi$ can be viewed as a translation of $\phi'$ describing a program in continuation-passing style instead.

The goal of this chapter was to prove that program equivalence induced by the logic $\mathcal{F}$ coincides with applicative bisimilarity (Theorem~\ref{Cor_bisim_is_log_equiv}). Hence, according to Theorem~\ref{Thm_sim_compat} from the previous chapter, for a decomposable set $\mathfrak{P}$ of Scott-open observations, $\mathcal{F}$-logical equivalence is compatible. 

To prove this, we defined a logic $\mathcal{V}$ which is similar to $\mathcal{F}$, but in which ECPS values can appear inside formulas. From this point of view $\mathcal{V}$ is not satisfactory as a specification language because we need knowledge of the syntax of ECPS to express properties of programs.

Nevertheless, it is relatively straightforward to prove that $\mathcal{V}$-logical equivalence coincides with applicative bisimilarity (Proposition~\ref{Prop_sim_coinc_log}), but we do not know any proof of a similar result for $\mathcal{F}$-logical equivalence. Instead, we proved that $\mathcal{F}$ and $\mathcal{V}$ are equi-expressive, using translations between the two logics (Theorem~\ref{Thm_logic_equiv_equiv}). Thus, we obtained a proof of Theorem~\ref{Cor_bisim_is_log_equiv}. This theorem will be used in the next chapter to prove the main result of the dissertation: that $\mathcal{F}$-logical equivalence coincides with contextual equivalence.

\chapter{Contextual Equivalence for ECPS}\label{Chap_ctx_equiv}

This chapter defines contextual equivalence for ECPS coinductively and proves that it is compatible and an equivalence relation. Then, contextual equivalence is proved to coincide with applicative bisimilarity for ECPS (Theorem \ref{Thm_bisim_is_ctx_equiv}). As a result, they both coincide with logical equivalence induced by the logic $\mathcal{F}$ (Corollary \ref{Cor_logeq_is_ctxeq}). Therefore, the logic $\mathcal{F}$ characterises contextual equivalence. Establishing this was the main goal of the dissertation. Finally, we present an alternative definition of contextual equivalence using program contexts, and prove it equivalent with the coinducitve definition (Theorem \ref{Thm_ctxctx_is_ctxcoind}).

\section{Contextual Equivalence Coinductively}\label{Sec_ctx}

This section presents a coinductive definition of contextual preorder and equivalence, initially proposed by Lassen \cite{LasPhd} and Gordon \cite{Gor98}. Contextual preorder is defined as the greatest compatible and adequate relation, for a suitable definition of adequacy.

The advantage of this definition is that we do not need to deal with contexts explicitly. This is important in the case of ECPS where contexts are duplicated because of the distinction between values and computations, as Section \ref{Sec_ctx_ctx} will show.

\begin{definition}[Adequacy]
Consider a set of observations $\mathfrak{P}$ and a well-typed relation on possibly open terms $\mathcal{R}=(\mathcal{R}^\mathfrak{v}_A,\mathcal{R}^\mathfrak{c})$, where $\mathcal{R}^\mathfrak{c}$ relates computations. The relation $\mathcal{R}$ is $\mathfrak{P}$-adequate if:
\begin{equation*}
\forall s,t.\ \vdash s\mathrel{\mathcal{R}^\mathfrak{c}} t \implies \forall P\in\mathfrak{P}.\ \treet{s}\in P \implies \treet{t}\in P.
\end{equation*}
The relation $\mathcal{R}$ is $\mathfrak{P}$-biadequate if:
\begin{equation*}
\forall s,t.\ \vdash s\mathrel{\mathcal{R}^\mathfrak{c}} t \implies \forall P\in\mathfrak{P}.\ \treet{s}\in P \Longleftrightarrow \treet{t}\in P.
\end{equation*}
\end{definition}

The definition of adequacy is motivated by the fact that, in ECPS, the observable behaviour of a program $s$, in the sense of Section \ref{Sec_prog_equiv_intro}, is whether $s\in P$, where  $P\in\mathfrak{P}$. So adequacy checks that $t$ simulates the observable behaviour of $s$.

\begin{definition}[Contextual preorder]
Let $\mathbb{CA}$ be the set of well-typed relations on possibly open terms that are both compatible and $\mathfrak{P}$-adequate. Define the contextual preorder $\ctxpre$ to be $\bigcup\mathbb{CA}$.
\end{definition}

The next proposition establishes that contextual preorder is a precongruence. This will help prove that contextual equivalence is a congruence, hence a well-behaved notion of  program equivalence.

\begin{proposition}\label{Prop_ctxpre_comp}
The contextual preorder $\ctxpre$ is a preorder, and is moreover compatible and $\mathfrak{P}$-adequate. Thus, it is the greatest compatible and $\mathfrak{P}$-adequate preorder.
\end{proposition}
\begin{proof}
The proof follows the structure of the proof of Proposition 4 from \cite{DBLP:journals/corr/LagoGL17}.
To prove reflexivity, we show that the open identity relation $\mathcal{I}$ is in $\mathbb{CA}$. To show transitivity it suffices to show that the composition of relations in $\mathbb{CA}$ is itself in $\mathbb{CA}$. The relation $\ctxpre$ is shown compatible using the definition of $\mathbb{CA}$. The complete proof of this proposition can be found in Appendix \ref{App_ctxt}.
\end{proof}

\begin{definition}[Contextual equivalence]\label{Def_ctxeq_coind}
Let $\mathbb{CAS}$ be the set of well-typed relations on possibly open terms that are both compatible and $\mathfrak{P}$-biadequate. Define contextual equivalence $\ctxeq$ to be $\bigcup\mathbb{CAS}$.
\end{definition}

\begin{proposition}\label{Prop_ctxeq_is_ctxpreop}
Contextual equivalence is the intersection of the contextual preorder with its converse:
\begin{equation*}
(\ctxeq) = (\ctxpre)\cap(\ctxpre)^{\textit{op}}.
\end{equation*}
\end{proposition}

The proof of the above proposition appears in Appendix \ref{App_ctxt}. This relationship between contextual equivalence and preorder in ECPS is expected. It also holds in the case of the untyped $\lambda$-calculus with generic effects, as shown by \cite{DBLP:conf/lics/LagoGL17}. Finally, we can prove contextual equivalence is a congruence.

\begin{proposition}
Contextual equivalence $\ctxeq$ is an equivalence relation, and is moreover compatible and $\mathfrak{P}$-biadequate. Thus, it is the greatest compatible and $\mathfrak{P}$-biadequate equivalence relation.
\end{proposition}
\begin{proof}
From Proposition \ref{Prop_ctxeq_is_ctxpreop} we know $(\ctxeq) = (\ctxpre)\cap(\ctxpre)^{\textit{op}}$. We have proved $\ctxpre$ is a preorder, so $(\ctxpre)^{\textit{op}}$ is also a preorder. The intersection of two preorders is another preorder so $\ctxeq$ is a preorder. Moreover, it is symmetric because it is the intersection of a relation and its converse. Thus $\ctxeq$ is an equivalence relation.
In the proof of Proposition~\ref{Prop_ctxeq_is_ctxpreop} we have shown $(\ctxpre)\cap(\ctxpre)^{\textit{op}}$ is compatible and biadequate, so $\ctxeq$ is as well.
\end{proof}

\section{Contextual Equivalence Coincides with Bisimilarity}\label{Sec_ctx_is_bisim}

This section shows that the coinductive notion of contextual equivalence (Definition~\ref{Def_ctxeq_coind}) coincides with applicative bisimilarity. Using the results from the previous chapter about logical equivalence, we can in fact deduce that all notions of program equivalence for ECPS considered so far are the same. 
Thus, we have defined a logic $\mathcal{F}$ whose induced program equivalence characterises contextual equivalence. This is the main contribution of this dissertation.

To obtain this result, the set of observations $\mathfrak{P}$ needs to satisfy one more condition apart from decomposability and Scott-openness, named \emph{consistency}.

\begin{definition}[Consistency]\label{Def_consistency}
A set of observations $\mathfrak{P}$ is consistent if there exists at least one observation $P_0\in\mathfrak{P}$ such that:
\begin{enumerate}
\item $P_0\not=\textit{Trees}_\Sigma$ and 
\item there exists at least one computation $t_0$ such that $\treet{t_0}\in P_0$.
\end{enumerate} 
\end{definition}

This definition says that $\mathfrak{P}$ contains a non-trivial observation $P_0$ which contains at least one computation tree $\treet{t_0}$. If this were not the case, then contextual equivalence would equate all terms, including natural numbers.

On the other hand, applicative bisimilarity would equate all computations, but not all natural numbers because $v \sim_\mathtt{nat}^\mathfrak{v} w \Longleftrightarrow v=w$. Hence, it would not be the case that applicative bisimilarity coincides with contextual equivalence.

However, program equivalence induced by a set of observations $\mathfrak{P}$ which does not satisfy consistency is not meaningful because not enough programs are distinguished. In particular, all computations would be equivalent to $loop$, the program that runs forever. Indeed, for all example of effects considered so far $\mathfrak{P}$ is consistent. Therefore consistency is a reasonable assumption.

\begin{theorem}\label{Thm_bisim_is_ctx_equiv}
Consider a decomposable set of Scott-open observations $\mathfrak{P}$ which is consistent. Then:
\begin{enumerate}
\item The open extension of applicative $\mathfrak{P}$-similarity, $\precsim^\circ$, coincides with the contextual preorder, $\ctxpre$. \item The open extension of applicative $\mathfrak{P}$-bisimilarity, $\sim^\circ$, coincides with contextual equivalence, $\ctxeq$.
\end{enumerate}
\end{theorem}
\begin{proof}
\textbf{We first show $(\precsim^\circ)=(\ctxpre)$}. We have shown in Theorem \ref{Thm_sim_compat} that $\precsim^\circ$ is compatible under the current assumptions. Consider $\vdash s\precsim^\circ t$. Then $\vdash s\precsim t$ so by the definition of simulation we know that:
\begin{equation*}
\forall P\in\mathfrak{P}.\ \treet{s}\in P \implies \treet{t}\in P.
\end{equation*}
Therefore, $\precsim^\circ$ is adequate. Being both compatible and adequate, $\precsim^\circ$ is included in $\ctxpre$.

Now we need to show $(\ctxpre)\subseteq(\precsim^\circ)$. We first show that $\ctxpre$ restricted to closed terms is included in $\precsim$, then extend this to open terms. To do this, we show $\ctxpre$ restricted to closed terms is a simulation by checking it satisfies the four conditions in the definition of simulation (Definition \ref{Def_ecps_simulation}). 

We will concentrate on the case for natural numbers because it is the most interesting. It makes use of the existence of computation $t_0$ and observation $P_0$ from the definition of consistency. The complete proof can be found in Appendix \ref{App_ctxt}.

\begin{enumerate}
\item[2.] Assume $\vdash v\mathrel{(\ctxpre)^\mathfrak{v}_{\mathtt{nat}}} u$.  Consider the computation:
\begin{equation*}
\textit{loop}=(\mufix{f}{\lbd{x}{\mathtt{nat}}{f\ x}})(\overline{0}).
\end{equation*}
This computation leads to an infinite chain of reductions:
\begin{multline*}
(\mufix{f}{\lbd{x}{\mathtt{nat}}{f\ x}})\ \overline{0} \longrightarrow (\lbd{x}{\mathtt{nat}}{(\lbd{y}{\mathtt{nat}}{(\mufix{f}{\lbd{x}{\mathtt{nat}}{f\ x}})\ y})\ x})\ \overline{0} 	\longrightarrow^2	\\
(\mufix{f}{\lbd{x}{\mathtt{nat}}{f\ x}})\ \overline{0} \longrightarrow^*
\end{multline*}
so $\treet{loop}=\bot$.

Since $v$ and $u$ are closed values there exist $m$ and $n$ in $\mathbb{N}$ such that $u=\overline{n}$ and $v=\overline{m}$. Consider the following computation:
\begin{multline*}
\ifte{x}{n}{t_0}{\textit{loop}} =	\\
\caset{x}{\textit{loop}}{x_1}{\mathtt{case}\ldots\\ x_n\ \mathtt{in}\ \{\mathtt{zero}\Rightarrow t_0,\ \mathtt{succ}(x_{n+1})\Rightarrow\textit{loop}\}\ldots}
\end{multline*}
which evaluates to $t_0$ if $x=\overline{n}$ or loops otherwise. We know by consistency that $\treet{t_0}\in P_0\not=\textit{Trees}_\Sigma$. It must be the case that $\bot\not\in P_0$ because otherwise by upwards-closure of $P_0$ we would obtain $P_0=\textit{Trees}_\Sigma$.

By compatibility and reflexivity of $\ctxpre$ we know that:
\begin{equation}\label{ctxpre_is_sim_prf1}
\vdash (\ifte{\overline{n}}{n}{t_0}{\textit{loop}})\mathrel{(\ctxpre)^\mathfrak{c}} (\ifte{\overline{m}}{n}{t_0}{\textit{loop}}).
\end{equation}

Suppose by contradiction that $n\not=m$. Then:
\begin{gather*}
\treet{\ifte{\overline{n}}{n}{t_0}{\textit{loop}}} = \treet{t_0} \in P_0	\\
\treet{\ifte{\overline{m}}{n}{t_0}{\textit{loop}}} = \treet{\textit{loop}} = \bot \not\in P_0.
\end{gather*}
But this contradicts equation \ref{ctxpre_is_sim_prf1} because $\ctxpre$ is adequate. Therefore, $n=m$ and $v=u$ as required.

\item[3.] By adequacy of $\ctxpre$.

\item[4.] By compatibility of $\ctxpre$.
\end{enumerate}

So we have established $(\ctxpre)\subseteq(\precsim)$ for closed terms. Now we need to prove $(\ctxpre)\subseteq(\precsim^\circ)$ in general. To do this, we consider computations and each type of value separately. Again, we concentrate on the case for natural numbers. The cases for computations and function values are proved using compatibility of $\ctxpre$. We give the proof of the computation case as an example. The proof for function values can be found in Appendix~\ref{App_ctxt}.

\paragraph{If $\protect\overrightarrow{x_j:A_j} \vdash v \mathrel{(\ctxpre)^\mathfrak{v}_{\mathtt{nat}}} w$} then for any $\vdash \overrightarrow{u_j:A_j}$ there exist natural numbers $m$ and $k$ such that $v[\overrightarrow{u_j/x_j}]=\overline{m}$ and $w[\overrightarrow{u_j/x_j}]=\overline{k}$. We want to show $m=k$ and therefore:
\begin{equation*}
\forall (\vdash\overrightarrow{u_j:A_j}).\ v[\overrightarrow{u_j/x_j}]=w[\overrightarrow{u_j/x_j}].
\end{equation*}
This would imply by the definition of simulation and of open extension that $\overrightarrow{x_i:A_i}\vdash v\mathrel{\precsim^\circ} w$, as required. To do this consider the computation:
\begin{multline*}
\textit{eq}(y,v) = (\mufix{f}{\lbd{(y_0,x_0)}{(\mathtt{nat},\mathtt{nat})}
	{\\ \caset{y_0}{\caset{x_0}{t_0}{x_1}{\textit{loop}}\\ }
		{y_1}{\caset{x_0}{\textit{loop}}{x_1}{f\ (y_1,x_1)}}}})\ (y,v).
\end{multline*}
If $y=v$ then $\treet{\textit{eq}(y,v)}=\treet{t_0}$, otherwise $\treet{\textit{eq}(y,v)}=\bot$. Using context weakening and compatibility and reflexivity of $\ctxpre$ deduce that:
\begin{equation*}
\overrightarrow{x_j:A_j}, y:\mathtt{nat} \vdash eq(y,v)\mathrel{(\ctxpre)^\mathfrak{c}} eq(y,w)
\end{equation*}
and then again by compatibility
\begin{equation*}
\vdash \lbd{(\overrightarrow{x_j},y)}{(\overrightarrow{A_j},\mathtt{nat})}{eq(y,v})\mathrel{(\ctxpre)^\mathfrak{v}_{\neg(\overrightarrow{A_j},\mathtt{nat})}} \lbd{(\overrightarrow{x_j},y)}{(\overrightarrow{A_j}, \mathtt{nat})}{eq(y,w)}.
\end{equation*}
From $(\ctxpre)\subseteq(\precsim)$ for closed terms we can now establish that:
\begin{equation*}
\vdash \lbd{(\overrightarrow{x_j},y)}{(\overrightarrow{A_j},\mathtt{nat})}{eq(y,v})\mathrel{\precsim^\mathfrak{v}_{\neg(\overrightarrow{A_j},\mathtt{nat})}} \lbd{(\overrightarrow{x_j},y)}{(\overrightarrow{A_j}, \mathtt{nat})}{eq(y,w)}
\end{equation*}
so from the definition of $\precsim$ and the fact that reduction preserves similarity (Lemma \ref{Lem_red_pres_sim}) we know that:
\begin{equation}\label{ctxpre_is_sim_prf2}
\forall (\vdash\overrightarrow{u_j:A_j}).\ \forall (\overline{n}:\mathtt{nat}).\ \vdash \textit{eq}(\overline{n},v[\overrightarrow{u_j/x_j}])\mathrel{\precsim^\mathfrak{c}} \textit{eq}(\overline{n},w[\overrightarrow{u_j/x_j}]).
\end{equation}
Assume by contradiction that there exists $\overrightarrow{u_j:A_j}$ such that $v[\overrightarrow{u_j/x_j}]\not=w[\overrightarrow{u_j/x_j}]$. Choose $\overline{n}=v[\overrightarrow{u_j/x_j}]$. Then:
\begin{gather*}
\treet{\textit{eq}(\overline{n},v[\overrightarrow{u_j/x_j}])}=\treet{t_0} \in P_0	\\
\treet{\textit{eq}(\overline{n},w[\overrightarrow{u_j/x_j}])}=\bot \not\in P_0
\end{gather*}
but this contradicts equation \ref{ctxpre_is_sim_prf2} because of the definition of $\precsim$. Therefore, it must be the case that $\forall \vdash\overrightarrow{u_j:A_j}.\ v[\overrightarrow{u_j/x_j}]=w[\overrightarrow{u_j/x_j}]$, which is what he had to prove.

\paragraph{If $\protect\overrightarrow{x_i:A_i} \vdash s \mathrel{(\ctxpre)^\mathfrak{c}} t$} then by compatibility of $\ctxpre$ we know that:
\begin{equation*}
\vdash \lbd{(\overrightarrow{x_i})}{\overrightarrow{A_i}}{s}\mathrel{ (\ctxpre)^\mathfrak{v}_{\neg(\overrightarrow{A_i})}} \lbd{(\overrightarrow{x_i})}{\overrightarrow{A_i}}{t}.
\end{equation*}
So using $(\ctxpre)\subseteq(\precsim)$ for closed terms we have:
\begin{equation*}
\vdash \lbd{(\overrightarrow{x_i})}{\overrightarrow{A_i}}{s}\mathrel{ \precsim^\mathfrak{v}_{\neg(\overrightarrow{A_i})}} \lbd{(\overrightarrow{x_i})}{\overrightarrow{A_i}}{t}.
\end{equation*}
By the definition of $\precsim$ for function values we can deduce:
\begin{equation*}
\forall(\vdash\overrightarrow{v_i:A_i}).\ \vdash (\lbd{(\overrightarrow{x_i})}{\overrightarrow{A_i}}{s})\ (\overrightarrow{v_i})\mathrel{\precsim^\mathfrak{c}} (\lbd{(\overrightarrow{x_i})}{\overrightarrow{A_i}}{t})\ (\overrightarrow{v_i}).
\end{equation*}
From Lemma \ref{Lem_red_pres_sim} we know that reduction preserves similarity so:
\begin{equation*}
\forall(\vdash\overrightarrow{v_i:A_i}).\ \vdash s[\overrightarrow{v_i/x_i}]\mathrel{\precsim^\mathfrak{c}} t[\overrightarrow{v_i/x_i}]
\end{equation*}
which by the definition of open extensions means that:
\begin{equation*}
\overrightarrow{x_i:A_i} \vdash s\mathrel{\precsim^{\circ,\mathfrak{c}}} t.
\end{equation*}

\paragraph{Now show that $(\sim^\circ)=(\ctxeq)$.} This is done using $(\precsim^\circ)=(\ctxpre)$ and the facts that $(\ctxeq)=(\ctxpre)\cap(\ctxpre)^{\textit{op}}$ (Proposition \ref{Prop_ctxeq_is_ctxpreop})
and $(\sim)=(\precsim)\cap(\precsim^\textit{op})$ (Proposition \ref{Prop_bisim_is_sim_and_simop}). The full proof can be found in Appendix \ref{App_ctxt}.
\end{proof}

\begin{corollary}\label{Cor_logeq_is_ctxeq}
For a decomposable set of Scott-open observations $\mathfrak{P}$ that is consistent, the logic $\mathcal{F}$ characterises contextual equivalence for ECPS. That is:
\begin{enumerate}
\item The open extension of $\mathcal{F}$-logical preorder coincides with the contextual preorder: $(\sqsubseteq_{\mathcal{F}^+}^\circ)=(\ctxpre)$.

\item  The open extension of $\mathcal{F}$-logical equivalence coincides with contextual equivalence: $(\equiv_\mathcal{F}^\circ)=(\ctxeq)$.
\end{enumerate}
Hence, applicative bisimilarity, logical equivalence and contextual equivalence all coincide.
\end{corollary}
\begin{proof}
Recall Theorem \ref{Cor_bisim_is_log_equiv} which says that, for a decomposable set of Scott-open observations $\mathfrak{P}$, $(\precsim)=(\sqsubseteq_{\mathcal{F}^+})$ and $(\sim)=(\equiv_\mathcal{F})$. From Theorem  \ref{Thm_bisim_is_ctx_equiv} we obtain the desired result:
\begin{gather*}
(\sqsubseteq_{\mathcal{F}^+}^\circ) = (\ctxpre) = (\precsim^\circ)	\\
(\equiv_\mathcal{F}^\circ) = (\ctxeq) = (\sim^\circ).
\end{gather*}
\end{proof}

\section{Contextual Equivalence via Contexts}\label{Sec_ctx_ctx}

For completeness, we give a more familiar definition of contextual equivalence using program contexts, following Crary and Harper \cite{DBLP:journals/entcs/CraryH07} and Pitts \cite{Pit11}. This way of defining contextual equivalence implements the intuition that equal programs behave the same in all contexts. This definition is proved to be the same as the coinductive one (Definition \ref{Def_ctxeq_coind}), a fact that is not immediately apparent and which justifies the use of the coinductive definition.

Informally, a context is an ECPS program with a hole. Because ECPS makes a distinction between values and computations, contexts are divided according to whether they accept and produce computations or values. This leads to a duplication of contexts which makes contextual equivalence tedious to work with. This is why we preferred working with the coinductive definition of contextual equivalence to establish the main results of this chapter.

\begin{definition}
Program contexts for ECPS are defined by the following grammar:
\begin{align*}
C^\mathfrak{v}_\mathfrak{v} &\coloneqq [-]^\mathfrak{v} \mid \mathtt{succ}(C^\mathfrak{v}_\mathfrak{v}) \mid \lbd{\overrightarrow{x_i}}{\overrightarrow{A_i}}{C^\mathfrak{v}_\mathfrak{c}}	\\
C^\mathfrak{v}_\mathfrak{c} &\coloneqq C^\mathfrak{v}_\mathfrak{v}(\overrightarrow{w_i}) \mid v(w_1,\ldots,C^\mathfrak{v}_\mathfrak{v},\ldots,w_n) \mid (\mufix{x}{C^\mathfrak{v}_\mathfrak{v}})(\overrightarrow{w_i}) \mid (\mufix{x}{v})(w_1,\ldots,C^\mathfrak{v}_\mathfrak{v},\ldots,w_n) \mid	\\
&\sigma(C^\mathfrak{v}_\mathfrak{v},x.t) \mid \sigma(v,x.C^\mathfrak{v}_\mathfrak{c}) \mid \caset{C^\mathfrak{v}_\mathfrak{v}}{s}{x}{t} \mid	\\
&\caset{v}{C^\mathfrak{v}_\mathfrak{c}}{x}{t} \mid \caset{v}{s}{x}{C^\mathfrak{v}_\mathfrak{c}}
\end{align*}
\begin{align*}
C^\mathfrak{c}_\mathfrak{c} &\coloneqq [-]^\mathfrak{c} \mid C^\mathfrak{c}_\mathfrak{v}(\overrightarrow{w_i}) \mid v(w_1,\ldots,C^\mathfrak{c}_\mathfrak{v},\ldots,w_n) \mid (\mufix{x}{C^\mathfrak{c}_\mathfrak{v}})(\overrightarrow{w_i}) \mid (\mufix{x}{v})(w_1,\ldots,C^\mathfrak{c}_\mathfrak{v},\ldots,w_n) \mid 	\\
&\sigma(v,x.C^\mathfrak{c}_\mathfrak{c}) \mid \caset{v}{C^\mathfrak{c}_\mathfrak{c}}{x}{t} \mid	\\
&\caset{v}{s}{x}{C^\mathfrak{c}_\mathfrak{c}}	\\
C^\mathfrak{c}_\mathfrak{v} &\coloneqq \lbd{\overrightarrow{x_i}}{\overrightarrow{A_i}}{C^\mathfrak{c}_\mathfrak{c}}.
\end{align*}
In the context $v(w_1,\ldots,C^\mathfrak{v}_\mathfrak{v},\ldots,w_n)$, $C^\mathfrak{v}_\mathfrak{v}$ can appear in any of the positions $1$ to $n$, for any $n\in\mathbb{N}$. Similarly for the other expressions which take multiple arguments.

The notation $C^\mathfrak{v}_\mathfrak{c}$ stands for a context whose hole can only be filled with a value, and the resulting term is a computation; $C^\mathfrak{v}_\mathfrak{v}$, $C^\mathfrak{c}_\mathfrak{v}$ and $C^\mathfrak{c}_\mathfrak{c}$ should be read analogously.
\end{definition}

Filling the whole of a context with a value or computation, as appropriate, is defined by recursion on the structure of contexts. The definition is standard so we present only a few cases:
\begin{gather*}
[-]^\mathfrak{v}[u] = u	\\
(\lbd{(\overrightarrow{x_i})}{(\overrightarrow{A_i})}{C^\mathfrak{v}_\mathfrak{c}})[u] = \lbd{(\overrightarrow{x_i})}{(\overrightarrow{A_i})}{C^\mathfrak{v}_\mathfrak{c}[u]}	\\
((\mufix{x}{C^\mathfrak{v}_\mathfrak{v}})(\overrightarrow{w_i}))[u] = (\mufix{x}{C^\mathfrak{v}_\mathfrak{v}[u]})(\overrightarrow{w_i})	\\
\ldots
\end{gather*}
Note that contexts can bind free variables in the term that fills their hole.

There are no contexts of the form $\sigma(C^\mathfrak{c}_\mathfrak{v},x.t)$, $\mathtt{case}\ C^\mathfrak{c}_\mathfrak{v}\ldots$ or $\mathtt{succ}(C^\mathfrak{c}_\mathfrak{v})$. In these cases $C^\mathfrak{c}_\mathfrak{v}[t]$ needs to be a value of type $\mathtt{nat}$, for some computation $t$. But the values of type $\mathtt{nat}$ are either variables $x$ or natural numbers $\overline{n}$, which do not contain any computation. Hence, they cannot be obtained from $C^\mathfrak{c}_\mathfrak{v}$.

We write $C[C']$ for composition of contexts. This means replacing the hole, $[-]$ in $C$ with the context $C'$. We can now define typing judgements for contexts:

\begin{definition}
The typing relation $C^\mathfrak{v}_\mathfrak{c}:(\Gamma'\vdash B)\Rightarrow(\Gamma\vdash)$ asserts that, given a value $\Gamma'\vdash u:B$, $C^\mathfrak{v}_\mathfrak{c}[u]$ is a well-formed computation in the environment $\Gamma$.

Similarly, define typing relations:
\begin{gather*}
C^\mathfrak{v}_\mathfrak{v}:(\Gamma'\vdash B)\Rightarrow(\Gamma\vdash A)	\\
C^\mathfrak{c}_\mathfrak{c}:(\Gamma'\vdash)\Rightarrow(\Gamma\vdash)	\\
C^\mathfrak{c}_\mathfrak{v}:(\Gamma'\vdash)\Rightarrow(\Gamma\vdash A).
\end{gather*}
These relations are the least relations closed under the rules in Figures \ref{Fig_ctx_type_val} and \ref{Fig_ctx_type_comp}.
\end{definition}

\begin{figure}
\begin{gather*}
\inferrule{ }{[-]^\mathfrak{v}:(\Gamma\vdash A)\Rightarrow(\Gamma\vdash A)} \textsc{(vv-id)}
\quad
\inferrule{C^\mathfrak{v}_\mathfrak{v}:(\Gamma'\vdash B)\Rightarrow(\Gamma\vdash\mathtt{nat})}{\mathtt{succ}(C^\mathfrak{v}_\mathfrak{v}) : (\Gamma'\vdash B)\Rightarrow(\Gamma\vdash\mathtt{nat})}\textsc{(vv-nat)}	\\
\inferrule{C^\mathfrak{v}_\mathfrak{c}:(\Gamma'\vdash B)\Rightarrow(\Gamma,\overrightarrow{x_i:A_i}\vdash)}{\lbd{\overrightarrow{x_i}}{\overrightarrow{A_i}}{C^\mathfrak{v}_\mathfrak{c}} : (\Gamma'\vdash B)\Rightarrow(\Gamma\vdash \neg(\overrightarrow{A_i}))}	\quad
\inferrule{C^\mathfrak{v}_\mathfrak{v}:(\Gamma'\vdash B)\Rightarrow(\Gamma\vdash\neg(\overrightarrow{A_i}))	\\	\Gamma\vdash \overrightarrow{w_i}:\overrightarrow{A_i}}{C^\mathfrak{v}_\mathfrak{v}(\overrightarrow{w_i}):(\Gamma'\vdash B)\Rightarrow(\Gamma\vdash)} 	\\[-0.7em]
 \textsc{(vv-lbd)} \qquad\qquad\qquad\qquad\qquad\qquad \textsc{(vc-appl)} \quad\quad	\\
\inferrule{\Gamma\vdash v:\neg(\overrightarrow{A_j})	\\	\Gamma\vdash(w_1,\ldots,w_{i-1}):(A_1,\ldots,A_{i-1})	\\	C^\mathfrak{v}_\mathfrak{v}:(\Gamma'\vdash B)\Rightarrow(\Gamma\vdash A_i)	\\	\Gamma\vdash(w_{i+1},\ldots,w_n):(A_{i+1},\ldots,A_n)}{v(w_1,\ldots,w_{i-1},C^\mathfrak{v}_\mathfrak{v},w_{i+1},\ldots,w_n) : (\Gamma'\vdash B)\Rightarrow(\Gamma\vdash)} \textsc{(vc-appr$_i$)}	\\
\inferrule{C^\mathfrak{v}_\mathfrak{v}:(\Gamma'\vdash B)\Rightarrow(\Gamma,x:\neg(\overrightarrow{A_i})\vdash)	\\	\Gamma\vdash \overrightarrow{w_i}:\overrightarrow{A_i}}{(\mufix{x}{C^\mathfrak{v}_\mathfrak{v}})(\overrightarrow{w_i}):(\Gamma'\vdash B)\Rightarrow(\Gamma\vdash)} \textsc{(vc-mul)}	\\
\inferrule{\Gamma,x:\neg(\overrightarrow{A_j}) \vdash v:\neg(\overrightarrow{A_j})	\\ \Gamma\vdash(w_1,\ldots,w_{i-1}):(A_1,\ldots,A_{i-1})	\\	C^\mathfrak{v}_\mathfrak{v}:(\Gamma'\vdash B)\Rightarrow(\Gamma\vdash A_i)	\\	\Gamma\vdash(w_{i+1},\ldots,w_n):(A_{i+1},\ldots,A_n)}{(\mufix{x}{v})(w_1,\ldots,w_{i-1},C^\mathfrak{v}_\mathfrak{v},w_{i+1},\ldots,w_n) : (\Gamma'\vdash B)\Rightarrow(\Gamma\vdash)} \textsc{(vc-mur$_i$)}	\\
\inferrule{C^\mathfrak{v}_\mathfrak{v}:(\Gamma'\vdash B)\Rightarrow(\Gamma\vdash\mathtt{nat})	\\	\Gamma,x:\mathtt{nat}\vdash t}{\sigma(C^\mathfrak{v}_\mathfrak{v},x.t):(\Gamma'\vdash B)\Rightarrow(\Gamma\vdash)}	\textsc{(vc-opl)} \\ 
\inferrule{\Gamma\vdash v:\mathtt{nat}	\\	C^\mathfrak{v}_\mathfrak{c}:(\Gamma'\vdash B)\Rightarrow(\Gamma,x:\mathtt{nat}\vdash)}{\sigma(v,x.C^\mathfrak{v}_\mathfrak{c}):(\Gamma'\vdash B)\Rightarrow(\Gamma\vdash)}	\textsc{(vc-opr)} \\
\inferrule{C^\mathfrak{v}_\mathfrak{v}:(\Gamma'\vdash B)\Rightarrow(\Gamma\vdash\mathtt{nat})	\\	\Gamma\vdash s	\\	\Gamma,x:\mathtt{nat}\vdash t}{\caset{C^\mathfrak{v}_\mathfrak{v}}{s}{x}{t} : (\Gamma'\vdash B)\Rightarrow(\Gamma\vdash)} \textsc{(vc-casev)}	\\
\inferrule{\Gamma\vdash v:\mathtt{nat}	\\	C^\mathfrak{v}_\mathfrak{c}:(\Gamma'\vdash B)\Rightarrow(\Gamma\vdash)	\\	\Gamma,x:\mathtt{nat}\vdash t}{\caset{v}{C^\mathfrak{v}_\mathfrak{c}}{x}{t} : (\Gamma'\vdash B)\Rightarrow(\Gamma\vdash)} \textsc{(vc-casel)}	\\
\inferrule{\Gamma\vdash v:\mathtt{nat}	\\	\Gamma\vdash s	\\	C^\mathfrak{v}_\mathfrak{c}:(\Gamma'\vdash B)\Rightarrow(\Gamma,x:\mathtt{nat}\vdash)}{\caset{v}{s}{x}{C^\mathfrak{v}_\mathfrak{c}} : (\Gamma'\vdash B)\Rightarrow(\Gamma\vdash)} \textsc{(vc-caser)}
\end{gather*}
\caption{Typing rules for contexts that accept a value.} \label{Fig_ctx_type_val}
\end{figure}

\begin{figure}
\begin{gather*}
\inferrule{ }{[-]^\mathfrak{c}:(\Gamma\vdash)\Rightarrow(\Gamma\vdash)} \textsc{(cc-id)}
\quad
\inferrule{C^\mathfrak{c}_\mathfrak{v}:(\Gamma'\vdash)\Rightarrow(\Gamma\vdash\neg(\overrightarrow{A_i}))	\\	\Gamma\vdash \overrightarrow{w_i}:\overrightarrow{A_i}}{C^\mathfrak{c}_\mathfrak{v}(\overrightarrow{w_i}):(\Gamma'\vdash)\Rightarrow(\Gamma\vdash)} \textsc{(cc-appl)}	\\
\inferrule{\Gamma\vdash v:\neg(\overrightarrow{A_j})	\\	\Gamma\vdash(w_1,\ldots,w_{i-1}):(A_1,\ldots,A_{i-1})	\\	C^\mathfrak{c}_\mathfrak{v}:(\Gamma'\vdash)\Rightarrow(\Gamma\vdash A_i)	\\	\Gamma\vdash(w_{i+1},\ldots,w_n):(A_{i+1},\ldots,A_n)}{v(w_1,\ldots,w_{i-1},C^\mathfrak{c}_\mathfrak{v},w_{i+1},\ldots,w_n) : (\Gamma'\vdash)\Rightarrow(\Gamma\vdash)} \textsc{(cc-appr$_i$)}	\\
\inferrule{C^\mathfrak{c}_\mathfrak{v}:(\Gamma'\vdash)\Rightarrow(\Gamma,x:\neg(\overrightarrow{A_i})\vdash)	\\	\Gamma\vdash \overrightarrow{w_i}:\overrightarrow{A_i}}{(\mufix{x}{C^\mathfrak{c}_\mathfrak{v}})(\overrightarrow{w_i}):(\Gamma'\vdash)\Rightarrow(\Gamma\vdash)} \textsc{(cc-mul)}	\\
\inferrule{\Gamma,x:\neg(\overrightarrow{A_j}) \vdash v:\neg(\overrightarrow{A_j})	\\ \Gamma\vdash(w_1,\ldots,w_{i-1}):(A_1,\ldots,A_{i-1})	\\	C^\mathfrak{c}_\mathfrak{v}:(\Gamma'\vdash)\Rightarrow(\Gamma\vdash A_i)	\\	\Gamma\vdash(w_{i+1},\ldots,w_n):(A_{i+1},\ldots,A_n)}{(\mufix{x}{v})(w_1,\ldots,w_{i-1},C^\mathfrak{c}_\mathfrak{v},w_{i+1},\ldots,w_n) : (\Gamma'\vdash)\Rightarrow(\Gamma\vdash)} \textsc{(cc-mur$_i$)}	\\
\inferrule{\Gamma\vdash v:\mathtt{nat}	\\	C^\mathfrak{c}_\mathfrak{c}:(\Gamma'\vdash)\Rightarrow(\Gamma,x:\mathtt{nat}\vdash)}{\sigma(v,x.C^\mathfrak{c}_\mathfrak{c}):(\Gamma'\vdash)\Rightarrow(\Gamma\vdash)}	\textsc{(cc-op)}	\\
\inferrule{\Gamma\vdash v:\mathtt{nat}	\\	C^\mathfrak{c}_\mathfrak{c}:(\Gamma'\vdash)\Rightarrow(\Gamma\vdash)	\\	\Gamma,x:\mathtt{nat}\vdash t}{\caset{v}{C^\mathfrak{c}_\mathfrak{c}}{x}{t} : (\Gamma'\vdash)\Rightarrow(\Gamma\vdash)} \textsc{(cc-casel)}	\\
\inferrule{\Gamma\vdash v:\mathtt{nat}	\\	\Gamma\vdash s	\\	C^\mathfrak{c}_\mathfrak{c}:(\Gamma'\vdash)\Rightarrow(\Gamma,x:\mathtt{nat}\vdash)}{\caset{v}{s}{x}{C^\mathfrak{c}_\mathfrak{c}} : (\Gamma'\vdash)\Rightarrow(\Gamma\vdash)} \textsc{(cc-caser)}	\\
\inferrule{C^\mathfrak{c}_\mathfrak{c}:(\Gamma'\vdash)\Rightarrow(\Gamma,\overrightarrow{x_i:A_i}\vdash)}{\lbd{\overrightarrow{x_i}}{\overrightarrow{A_i}}{C^\mathfrak{c}_\mathfrak{c}} : (\Gamma'\vdash)\Rightarrow(\Gamma\vdash \neg(\overrightarrow{A_i}))} \textsc{(cv-lbd)}
\end{gather*}
\caption{Typing rules for contexts that accept a computation.} \label{Fig_ctx_type_comp}
\end{figure}

\begin{lemma}\label{Lem_ctx_inst_comp}
Context instantiation and context composition yield well-typed contexts:
\begin{enumerate}[label=\arabic*a.]
\item If $C^\mathfrak{v}_\mathfrak{v}:(\Gamma'\vdash B)\Rightarrow(\Gamma\vdash A)$ and $\Gamma'\vdash u:B$ then $\Gamma\vdash C^\mathfrak{v}_\mathfrak{v}[u]:A$.

And the analogous statements for contexts $C^\mathfrak{v}_\mathfrak{c}$, $C^\mathfrak{c}_\mathfrak{c}$, $C^\mathfrak{c}_\mathfrak{v}$.

\item If $C^\mathfrak{c}_\mathfrak{v}:(\Gamma'\vdash)\Rightarrow(\Gamma\vdash A)$ and ${C'}^\mathfrak{v}_\mathfrak{c} : (\Gamma''\vdash B)\Rightarrow(\Gamma'\vdash)$ then 
\begin{equation*}
C^\mathfrak{c}_\mathfrak{v}[{C'}^\mathfrak{v}_\mathfrak{c}]:(\Gamma''\vdash B)\Rightarrow(\Gamma\vdash A).
\end{equation*}
And the analogous statements for all valid combinations of contexts $C$ and $C'$.
\end{enumerate}
\end{lemma}
\begin{proof}
In both cases proceed by induction on the typing derivation of the context C.
\begin{enumerate}
\item The two base cases \textsc{(vv-id)} and \textsc{(cc-id)} follow by assumption.
All the other cases are solved by applying the induction hypothesis then the appropriate typing rule for terms.
\item Again, the base cases follow by assumption. In the other cases apply the induction hypothesis then apply the context typing rule that matches $C$.
\end{enumerate}
\end{proof}

Contextual equivalence checks whether two terms have the same behaviour in all program contexts that yield closed computations. Only these contexts are used because we can only observe the behaviour of closed computations. The set of observations $\mathfrak{P}$ encodes the observable behaviour of programs, so contextual equivalence makes use of it.

\begin{definition}\label{Def_ctx_ctx}
Contextual preorder and equivalence are well-typed relations on possibly open terms defined as follows:
\begin{enumerate}
\item Given values $\Gamma\vdash v,u:A$, they are in the contextual preorder, $\Gamma\vdash v\ (\ctxprec)^\mathfrak{v}_A\ u$, if and only if:
\begin{equation*}
\forall C^\mathfrak{v}_\mathfrak{c}:(\Gamma\vdash A) \Rightarrow (\emptyset\vdash).\ \forall P\in\mathfrak{P}.\ \treet{C^\mathfrak{v}_\mathfrak{c}[v]}\in P \implies \treet{C^\mathfrak{v}_\mathfrak{c}[u]}\in P.
\end{equation*}
\item Given two computations $\Gamma\vdash s,t$, they are in the contextual preorder, $\Gamma\vdash s\ (\ctxprec)^\mathfrak{c}\ t$, if and only if:
\begin{equation*}
\forall C^\mathfrak{c}_\mathfrak{c}:(\Gamma\vdash) \Rightarrow (\emptyset\vdash).\ \forall P\in\mathfrak{P}.\ \treet{C^\mathfrak{c}_\mathfrak{c}[s]}\in P \implies \treet{C^\mathfrak{c}_\mathfrak{c}[t]}\in P.
\end{equation*}
\end{enumerate}
Two values $v$ and $u$ are contextually equivalent, $\Gamma\vdash v\ (\ctxeqc)^\mathfrak{v}_A\ u$, if $\Gamma\vdash v\ (\ctxprec)^\mathfrak{v}_A\ u$ and $\Gamma\vdash u\ (\ctxprec)^\mathfrak{v}_A\ v$. And similarly for computations. Therefore:
\begin{equation*}
(\ctxeqc)=(\ctxprec)\cap(\ctxprec)^\textit{op}.
\end{equation*}
\end{definition}

Next, we prove that the definition above coincides with the coinductive definition of contextual equivalence from the previous section (Definition \ref{Def_ctxeq_coind}). The proofs of the following two lemmas can be found in Appendix \ref{App_ctxt}.

\begin{lemma}\label{Lem_ctx_wctx_compat}
Contextual preorder defined with contexts, $\ctxprec$, is a compatible and adequate preorder. Hence, it is included in contextual preorder defined coinductively, $\ctxpre$.
\end{lemma}

\begin{lemma}\label{Lem_ctx_pres_ctx_coin}
Contextual preorder defined coinductively, $\ctxpre$, is closed under program contexts, that is:
\begin{enumerate}
\item If $\Gamma'\vdash v\ (\ctxpre)^\mathfrak{v}_A\ u$ and $C^\mathfrak{v}_\mathfrak{v}:(\Gamma'\vdash A)\Rightarrow(\Gamma\vdash B)$ then $\Gamma\vdash C^\mathfrak{v}_\mathfrak{v}[v]\ (\ctxpre)^\mathfrak{v}_B\ C^\mathfrak{v}_\mathfrak{v}[u]$.

And the analogous statement for $C^\mathfrak{v}_\mathfrak{c}$.

\item If $\Gamma'\vdash s\ (\ctxpre)^\mathfrak{c}\ t$ and $C^\mathfrak{c}_\mathfrak{c}:(\Gamma'\vdash)\Rightarrow(\Gamma\vdash)$ then $\Gamma\vdash C^\mathfrak{c}_\mathfrak{c}[s]\ (\ctxpre)^\mathfrak{c}\ C^\mathfrak{c}_\mathfrak{c}[t]$.

And the analogous statement for $C^\mathfrak{c}_\mathfrak{v}$.
\end{enumerate}
\end{lemma}

\begin{theorem}\label{Thm_ctxctx_is_ctxcoind}
Contextual preorder defined using program contexts, $\ctxprec$, coincides with contextual preorder defined coinductively, $\ctxpre$. Moreover, $(\ctxeqc)=(\ctxeq)$.
\end{theorem}
\begin{proof}
We have already shown in Lemma \ref{Lem_ctx_wctx_compat} that $(\ctxprec)\subseteq(\ctxpre)$. So it only remains to show the inverse inclusion.

Consider $\Gamma\vdash v\ (\ctxpre)^\mathfrak{v}_A\ u$. Then by Lemma \ref{Lem_ctx_pres_ctx_coin} we know that:
\begin{equation*}
\forall C^\mathfrak{v}_\mathfrak{c}:(\Gamma\vdash A)\Rightarrow(\emptyset\vdash).\ \emptyset\vdash C^\mathfrak{v}_\mathfrak{c}[v]\ (\ctxpre)^\mathfrak{c}\ C^\mathfrak{v}_\mathfrak{c}[u].
\end{equation*}
By adequacy of $\ctxpre$ we can deduce that:
\begin{equation*}
\forall C^\mathfrak{v}_\mathfrak{c}:(\Gamma\vdash A)\Rightarrow(\emptyset\vdash).\ \forall P\in\mathfrak{P}.\ \treet{C^\mathfrak{v}_\mathfrak{c}[v]}\in P \implies \treet{C^\mathfrak{v}_\mathfrak{c}[u]}\in P
\end{equation*}
which is equivalent to:
\begin{equation*}
\Gamma\vdash v\ (\ctxprec)^\mathfrak{v}_A\ u.
\end{equation*}
A similar reasoning can be applied for computations. So $(\ctxpre)\subseteq(\ctxprec)$ and hence $(\ctxpre)=(\ctxprec)$.

We know from Proposition \ref{Prop_ctxeq_is_ctxpreop} that $(\ctxeq)=(\ctxpre)\cap(\ctxpre)^\textit{op}$. By definition of $\ctxeqc$ we know $(\ctxeqc)=(\ctxprec)\cap(\ctxprec)^\textit{op}$. So since $(\ctxpre)^\textit{op}=(\ctxprec)^\textit{op}$ as well, we know that:
\begin{equation*}
(\ctxeqc)=(\ctxeq).
\end{equation*}
\end{proof}

\section{Chapter Summary}

Section~\ref{Sec_ctx} defined contextual equivalence for ECPS coinductively, as the greatest compatible relation that is adequate. Adequacy means that related computations are part of exactly the same observations from $\mathfrak{P}$. Thus, contextual equivalence is compatible by construction.

We then identified a condition on $\mathfrak{P}$ named \emph{consistency} (Definition~\ref{Def_consistency}), which requires contextual equivalence to distinguish between at least two computations. This is a reasonable condition; otherwise, all computations would be identified with the divergent one. 

Given a decomposable set $\mathfrak{P}$ of Scott-open observations which is consistent, we proved that contextual equivalence coincides with applicative bisimilarity (Theorem~\ref{Thm_bisim_is_ctx_equiv}). The proof proceeds by showing the equality of the two relations for closed terms, and then extends this to open terms. Consistency is used in the proof to construct ECPS computations which distinguish between two different natural numbers. The construction of these terms required some ingenuity.

Theorems~\ref{Thm_bisim_is_ctx_equiv} and \ref{Cor_bisim_is_log_equiv} allowed us to deduce the main result of the dissertation: that under the assumptions of Scott-openness, decomposability and consistency of $\mathfrak{P}$, $\mathcal{F}$-logical equivalence characterises contextual equivalence (Corollary~\ref{Cor_logeq_is_ctxeq}).

The coinductive definition of contextual equivalence (Definition~\ref{Def_ctxeq_coind}) might seem unintuitive, so we also presented the standard definition (Definition~\ref{Def_ctx_ctx}). This says that programs are related if they have the same \emph{observable} behaviour in all program contexts. In ECPS the observable behaviour is encoded by the set $\mathfrak{P}$. So two programs are contextually equivalent if, when substituted in an arbitrary context yielding closed computations, the resulting computations are part of the same observations $P\in\mathfrak{P}$.

This definition of contextual equivalence is difficult to work with because of the universal quantification over a large number of contexts. Therefore, we preferred working with the coinductive definition in order to establish the most important results in this chapter.

\chapter{Conclusion}\label{Chap_concl}

This chapter summarises the motivation and outcomes of the dissertation. It provides a comparison with previous work and outlines two directions for further research. The chapter concludes with an assessment of my personal development throughout the project.

\section{Summary}

Program equivalence, establishing when two programs are interchangeable, is one of the most important problems in the theory of programming languages. As explained in the introduction, the case of higher-order functions is especially difficult. Many definitions of program equivalence for increasingly complex higher-order languages have been proposed over the years. This variety of approaches has led researchers to investigate the relationships between different definitions of program equivalence.

This dissertation studied program equivalence for a higher-order language with algebraic effects. Algebraic effects are an approach to giving uniform formal semantics to impure operations such as input and output, probabilistic choice, nondeterminism etc.~not usually found in purely functional languages. This method was discovered relatively recently \cite{DBLP:conf/fossacs/PlotkinP01, DBLP:conf/fossacs/PlotkinP02, DBLP:journals/acs/PlotkinP03, DBLP:journals/tcs/HylandPP06} so work still needs to be done to understand program equivalence in the presence of generic algebraic effects. 

One open question we identified was: can we formulate a logic of program properties that characterises contextual equivalence for a higher-order language with generic algebraic effects? The language we chose to study is named ECPS (Section~\ref{Sec_ecps_def}). It is a call-by-value extension of the simply-typed $\lambda$-calculus with algebraic effects, natural numbers and general recursion. Moreover, ECPS is a continuation-passing language in which programs are not expected to return.

We answered the above question positively in the context of ECPS by developing the logic $\mathcal{F}$, defined in Chapter~\ref{Chap_mod_logic}. Thus, we obtained the first logic whose induced program equivalence coincides with contextual equivalence for algebraic effects (Corollary~\ref{Cor_logeq_is_ctxeq}).

The starting point of our work was a paper by Simpson and Voorneveld \cite{SimV18}. They proposed a modal logic that characterises applicative bisimilarity for the EPCF language, but not contextual equivalence. EPCF (Section~\ref{Sec_epcf_def}) is a variant of ECPS in which programs are written in direct style.

To justify the use of ECPS and to link our work to previous research, we first investigated the relationship between ECPS and EPCF, in Chapter~\ref{Chap_rel-epcf-ecps}. We proved that EPCF can be embedded into ECPS via a continuation-passing translation which preserves computation trees (Theorem~\ref{Thm_CPS_correct}). Then we argued informally that ECPS is more expressive than EPCF.

The remaining chapters were concerned with three forms of program equivalence for ECPS and the relationship between them: applicative bisimilarity, logical equivalence and contextual equivalence.
In Chapter~\ref{Chap_bisim}, we defined applicative bisimilarity for ECPS and proved it compatible (Theorem~\ref{Thm_sim_compat}). Compatibility is a fundamental property that a meaningful program equivalence should satisfy. It says that related programs can be substituted for a variable on opposite sides of a program equation.

Chapter~\ref{Chap_mod_logic} defined the logic $\mathcal{F}$, which expresses properties of ECPS programs. We proved that $\mathcal{F}$-logical equivalence coincides with applicative bisimilarity, and is therefore compatible (Theorem~\ref{Cor_bisim_is_log_equiv}). In Chapter~\ref{Chap_ctx_equiv}, we developed contextual equivalence and proved it compatible. Then, we showed that applicative bisimilarity coincides with contextual equivalence (Theorem~\ref{Thm_bisim_is_ctx_equiv}). This helped us prove the main result of the dissertation: that $\mathcal{F}$-logical equivalence coincides with contextual equivalence (Corollary~\ref{Cor_logeq_is_ctxeq}).

\section{Comparison with Previous Work}

Other logics of program properties for effects have been proposed. For example, the \emph{evaluation logic} of Pitts \cite{PittsAM:evall} concerns general computational effects. It uses two built-in modalities $\Box$ and $\Diamond$ to talk about evaluation of programs. Observations $P\in\mathfrak{P}$ from the logic $\mathcal{F}$ play a similar role to these modalities. The difference is that the contents of the set of observations $\mathfrak{P}$ depend on the effects present in the language, rather than being built into the logic. Moreover, program equivalence induced by evaluation logic is not compared with other operational notions of program equivalence.

Plotkin and Pretnar \cite{DBLP:conf/lics/PlotkinP08} propose a \emph{logic for algebraic effects} which is shown to be sound for establishing different kinds of program equivalence, but not complete in general. According to Pnueli's classification \cite{Pnu77}, this logic for algebraic effects is an exogenous logic because computations are allowed to appear inside formulas. In contrast, the logic $\mathcal{F}$ is endogenous because a formula concerns only one computation. 

Another difference is that in Plotkin's and Pretnar's logic there is a modality for each algebraic operation. In the logic $\mathcal{F}$, we instead adopt the view that each observation should express a behavioural property of programs. Thus, although observations depend on the effects in the language, they do not depend on the syntax of these effects. So we obtain a greater separation between the logic and the syntax of the programming language.

Simpson's and Voorneveld's work \cite{SimV18} is the most closely related to ours. They propose the EPCF logic (Section~\ref{Sec_epcf_results}), in which modalities specify shapes of computation trees, and show it characterises applicative bisimilarity. Therefore, modalities play a similar role to the observations from $\mathcal{F}$. 

The difference is that EPCF formulas, and applicative bisimilarity, check for the return values of programs. Because ECPS is a continuation-passing language, programs do not return, so neither $\mathcal{F}$-formulas nor applicative bisimilarity can check return values. As explained in Chapter~\ref{Chap_rel-epcf-ecps}, EPCF is in fact a fragment of ECPS. Therefore, contextual equivalence in ECPS is more restrictive since there are more program contexts. These differences provide an insight into why in ECPS we were able to obtain the relationship:
\begin{center}
applicative bisimilarity $=$ $\mathcal{F}$-logical equivalence $=$ contextual equivalence.
\end{center}
Whereas in EPCF the situation is:
\begin{center}
applicative bisimilarity $=$ $\textit{EPCF}$-logical equivalence $\not=$ contextual equivalence.
\end{center}

Compared to modal logics used to specify program properties in practice, such as LTL or CTL, the logic $\mathcal{F}$ is not directly suitable for verification because of its infinitary connectives. However, $\mathcal{F}$ achieves its goal of expressing behavioural properties of higher-order programs with algebraic effects. As a result, it can be used as the starting point for designing higher-level logics for algebraic effects that would be suitable for verification.

\section{Future Work}

The examples of effects considered in this work were nondeterminism, probabilistic choice, global store and I/O. A next step would be to consider local store, which is also an algebraic effect \cite{DBLP:conf/fossacs/PlotkinP02}. Relevant questions here are integrating local store into the framework for generic algebraic effects provided by the logic $\mathcal{F}$, and investigating whether the current relationship between different forms of program equivalence still holds. 

Operational notions of program equivalence for higher-order languages with local store have been studied by Pitts and Stark \cite{PitS98}, who obtained a characterisation of contextual equivalence using a logical relation. Yoshida, Honda and Berger \cite{DBLP:journals/lmcs/YoshidaHB08} propose an extension of Hoare logic to reason about local store. They prove that this logic characterises contextual equivalence.
However, local store remains a challenging effect. We anticipate that the notion of computation tree, on which observations from $\mathcal{F}$ are based, would need to be changed to account for local store.

Finally, one could study a fourth notion of program equivalence for ECPS called \emph{normal-form bisimilarity} \cite{DBLP:conf/csl/LassenL07}. This has not been extended to algebraic effects before. One difference between normal-form bisimilarity and applicative bisimilarity is that the latter only relates closed terms, whereas the former also considers open terms. 

The definition of normal-form bisimulation is related to game semantics \cite{AbrM99}, which models open programs as strategies in a two-player game. Similarly, the operational semantics of ECPS can be extended to open terms. A satisfaction relation between open terms and logical formulas can be defined based on this operational semantics and the observations in $\mathfrak{P}$. The question here is how does normal-form bisimilarity compare to the program equivalence induced by game satisfaction? and furthermore, how do they both compare to applicative bisimilarity? We conjecture that normal-form bisimilarity is closer to  game satisfaction than applicative bisimilarity, but leave these questions for future investigation.




\section{Personal Reflections}

From a technical perspective the project presented many challenges. I learnt about algebraic effects and became familiar with proof techniques such as logical relations, coinduction and Howe's method. I completed a significant number or fairly long inductive proofs which required care and organisation.

The proof of correctness of the CPS translation (Theorem~\ref{Thm_CPS_correct}) was more difficult than expected. Two approaches I tried initially failed: an inductive proof using domain theoretic techniques and a coinductive proof. I finally settled on the combination between logical relations and coinduction which proved successful.

The \emph{Categories, Proofs and Processes} course \cite{AbrT17} I took this year helped me understand some of the abstract material needed for the project, such as coinduction, while the \emph{Principles of Programming Languages} lectures I attended provided useful background on continuations. Other useful courses were \emph{Formal Verification} and \emph{Automata, Logic and Games}, where I learnt about modal logics used in verification to express program properties, such as LTL, CTL and the modal $\mu$-calculus.

Overall, my ability to understand new theoretical material and to carry out complex proofs has improved. Considering this and the novel theoretical contribution that the dissertation makes, I can say that the project has been a success.

\addcontentsline{toc}{chapter}{Bibliography}
\bibliography{modal}

\appendix
\numberwithin{equation}{chapter}
\addappheadtotoc 

\chapter{Proofs about the CPS translation}\label{App_cps}

\begin{replemma}{Lem_sim_eval}
If $(S,M)\rightarrowtail^k(S',M')$ and $t\longrightarrow^*t'$ where $k<n$ then:
\begin{equation*}
(S,M)\ \mathcal{S}^n_{\tau,\rho}\ t \iff (S',M')\ \mathcal{S}^{n-k}_{\tau,\rho}\ t'.
\end{equation*}
\end{replemma}
\begin{proof}
Assume $(S,M)\ \mathcal{S}^n_{\tau,\rho}\ t$. To show $(S',M')\ \mathcal{S}^{n-k}_{\tau,\rho}\ t'$ it suffices to show that the two conditions in Definition \ref{Def_step_indexed_sim} are satisfied for $((S',M'),t')$. 

Assume that for some $p<n-k$, $(S',M')\longrightarrow^p (S'',\sigma(V;W))$. Then we know $(S,M)\longrightarrow^{p+k} (S'',\sigma(V;W))$ and $p+k<n$. So we can deduce $t\longrightarrow^*t'\longrightarrow^*\sigma(v, x.t'')$ where $V=\overline{l}$ and $v=\overline{l}$, and $\forall l\in\mathbb{N}.\ (S'',W\ \overline{l})\ \mathcal{S}^{n-k-p}_{\tau'',\rho}\ t''[\overline{l}/x]$, as required. 

Assume that for some $p\leq n-k$, $(S',M')\longrightarrow^p(id,\return{V})$. Then $(S,M)\longrightarrow^{p+k}(id,\return{V})$. So from the initial assumption we know $t\longrightarrow^*t'\longrightarrow^*\downarrow$, as required.

The reverse implication is proved similarly.
\end{proof}

\begin{replemma}{Lem_step_simp_impl_sim}
For any configuration $(S,M)$ and closed computation $t$:
\begin{equation*}
(\forall n\in\mathbb{N}.\ (S,M)\ \mathcal{S}^n_{\tau,\rho}\ t) \implies (S,M)\ \mathcal{S}_{\tau,\rho}\ t.
\end{equation*}
\end{replemma}
\begin{proof}
Assume $\forall n\in\mathbb{N}.\ (S,M)\ \mathcal{S}^n_{\tau,\rho}\ t$.

To prove $(S,M)\ \mathcal{S}_{\tau,\rho}\ t$, it suffices to find a simulation relation respecting the conditions in Definition \ref{Def_sim} which contains the pair $((S,M),t)$. Because $\mathcal{S}$ is the greatest simulation, this relation will be contained in $\mathcal{S}$.

Let $\bigcap_{n\in\mathbb{N}}\mathcal{S}^n$ be the candidate simulation. Check the condititions in Definition \ref{Def_sim}:
\begin{enumerate}
\item Assume $(S,M)\rightarrowtail^*(S',\sigma(V;W))$. Then there exists $p\in\mathbb{N}$ such that $(S,M)\rightarrowtail^p(S',\sigma(V;W))$. By assumption we know $\forall m>0.\ (S,M)\ \mathcal{S}^{p+m}_{\tau,\rho}\ t$. So by the first condition in Definition \ref{Def_step_indexed_sim} we know that $t\longrightarrow^*\sigma(v,x.t')$ where $V=\overline{l}$ and $v=\overline{l}$, and that $\forall l\in\mathbb{N}.\ (S',W\ \overline{l})\ \bigcap_{m>0}\mathcal{S}^m_{\tau',\rho}\ t'[\overline{l}/x]$. For $m=0$, $(S',W\ \overline{l})\ \mathcal{S}^0_{\tau',\rho}\ t'[\overline{l}/x]$ can be easily seen to hold. So we have $\forall l\in\mathbb{N}.\ (S',W\ \overline{l})\ \bigcap_{m\in\mathbb{N}}\mathcal{S}^m_{\tau',\rho}\ t'[\overline{l}/x]$ as required.
\item Assume $(S,M)\rightarrowtail^p(id,\return{V})$ for some $p\in\mathbb{N}$. Then from $(S,M)\ \mathcal{S}^p_{\tau,\rho}\ t$ we have $t\longrightarrow^*\downarrow$ as required.
\end{enumerate}
\end{proof}

\begin{replemma}{Lem_fund_prop}[Fundamental property of the logical relation]
For any value $\overrightarrow{x_i:\tau_i} \vdash V :\rho$, any computation $\overrightarrow{x_i:\tau_i} \vdash M :\rho$ and any stack $\overrightarrow{x_i:\tau_i} \vdash S :\rho\Rightarrow\rho'$ in EPCF:
\begin{enumerate}
\item $\overrightarrow{x_i:\tau_i} \vdash V\ \mathcal{R}^\mathfrak{v}_\rho\ V^*$.
\item $\overrightarrow{x_i:\tau_i} \vdash M\ \mathcal{R}^\mathfrak{c}_\rho\ M^*$.
\item $\overrightarrow{x_i:\tau_i} \vdash S\ \mathcal{R}^\mathfrak{s}_{\rho,\rho'}\ S^*$.
\end{enumerate}
\end{replemma}
\begin{proof}
The proof is by induction on the typing derivations of $V$, $M$ and $S$. All cases are presented below, starting with $\mathbf{fix}$ which is the most interesting.

\paragraph{Case (fix), $M=\fix{W}$.} Assume $\overrightarrow{x_i:\tau_i} \vdash W\ \mathcal{R}^\mathfrak{v}_{(\rho_1{\rightarrow}\rho_2){\rightarrow}(\rho_1{\rightarrow}\rho_2)}\ W^*$. We need to prove:
\begin{equation*}
\forall n\in\mathbb{N}.\ \forall\overrightarrow{(V_i,v_i):\tau_i}.\ (\overrightarrow{(V_i,v_i)\in\mathcal{R}^{\mathfrak{v},n}_{\tau_i}}\implies (\fix{W}[\overrightarrow{V_i/x_i}],(\fix{W})^*[\overrightarrow{v_i/x_i}])\in\mathcal{R}^{\mathfrak{c},n}_{\rho_1\rightarrow\rho_2}).
\end{equation*}
This is proved by induction on $n$.

\paragraph{Base case $n=0$.} To prove $(\fix{W}[\overrightarrow{V_i/x_i}],(\fix{W})^*[\overrightarrow{v_i/x_i}])\in\mathcal{R}^{\mathfrak{c},0}_{\rho_1\rightarrow\rho_2}$ we need to check that:
\begin{equation*}
\forall (S,k)\in\mathcal{R}^{\mathfrak{s},0}_{\rho_1\rightarrow\rho_2,\rho}.\ (S,\fix{W}[\overrightarrow{V_i/x_i}])\ \mathcal{S}^0_{\rho_1\rightarrow\rho_2,\rho}\ ((\fix{W})^*[\overrightarrow{v_i/x_i}]\ k).
\end{equation*}
The first condition in the definition of step-indexed similarity holds because there is no $p<0$. The second condition holds because $(S,\fix{W}[\overrightarrow{V_i/x_i}])\longrightarrow^0(id,\return{V})$ is false.

\paragraph{Induction step $n>0$.}The induction hypothesis is:
\begin{equation*}
\forall m<n.\ \forall\overrightarrow{(V_i,v_i):\tau_i}.\ (\overrightarrow{(V_i,v_i)\in\mathcal{R}^{\mathfrak{v},m}_{\tau_i}}\implies (\fix{W}[\overrightarrow{V_i/x_i}],(\fix{W})^*[\overrightarrow{v_i/x_i}])\in\mathcal{R}^{\mathfrak{c},m}_{\rho_1\rightarrow\rho_2}).
\end{equation*}

Consider some arbitrary $\overrightarrow{(V_i,v_i)\in\mathcal{R}^{\mathfrak{v},n}_{\tau_i}}$. By the definition of a step-indexed relation we know $\forall m<n.\ \mathcal{R}^{\mathfrak{v},n}_{\tau_i}\subseteq\mathcal{R}^{\mathfrak{v},m}_{\tau_i}$.

We need to prove that:
\begin{equation*}
\forall p\leq n.\ \forall(S,k)\in\mathcal{R}^{\mathfrak{s},p}_{\rho_1{\rightarrow}\rho_2,\rho}.\ (S, \fix{W}[\overrightarrow{V_i/x_i}])\ \mathcal{S}^p_{\rho_1{\rightarrow}\rho_2,\rho}\ ((\fix{W})^*[\overrightarrow{v_i/x_i}]\ k).
\end{equation*}

For $p<n$ this follows immediately from the induction hypothesis if we choose $m=n-1$.

For the case $p=n$ let:
\begin{align*}
E &= \lbd{y}{\rho_1}{\letin{\fix{W[\overrightarrow{V_i/x_i}]}}{z}{z\ y}}	\\
F &= \lbd{x}{\rho_1}{\letin{W[\overrightarrow{V_i/x_i}]\ E}{w}{w\ x}}.	\\
\end{align*}
From the operational semantics and the CPS translation we can see that:
\begin{gather*}
(S, \fix{W}[\overrightarrow{V_i/x_i}]) \rightarrowtail (S,\return{F})	\\
(\fix{W})^*[\overrightarrow{v_i/x_i}]\ k \longrightarrow^* k\ F^*.
\end{gather*}
Therefore, using Lemma \ref{Lem_sim_eval} it suffices to prove:
\begin{equation*}
(S,\return{F})\ \mathcal{S}^{n-1}_{\rho_1{\rightarrow}\rho_2,\rho}\ k\ F^*.
\end{equation*}
We have assumed $(S,k)\in\mathcal{R}^{\mathfrak{s},n}_{\rho_1{\rightarrow}\rho_2,\rho}$ so by definition of $\mathcal{R}$ we know:
\begin{equation*}
\forall(V,v)\in\mathcal{R}^{\mathfrak{v},n-1}_{\rho_1{\rightarrow}\rho_2}.\ (S,\return{V})\ \mathcal{S}^{n-1}_{\rho_1{\rightarrow}\rho_2,\rho}\ k\ v.
\end{equation*}
So it is enough to show $(F, F^*)\in\mathcal{R}^{\mathfrak{v},n-1}_{\rho_1{\rightarrow}\rho_2}$, that is:
\begin{multline*}
\forall p_1<n-1.\ \forall(V_1,v_1)\in\mathcal{R}^{\mathfrak{v},p_1}_{\rho_1}.\ \forall p_2\leq p_1.\ \forall(S_2,k_2)\in\mathcal{R}^{\mathfrak{s},p_2}_{\rho_2,\rho'}.	\\
(S_2,F\ V_1)\ \mathcal{S}^{p_2}_{\rho_2,\rho'}\ (\lbd{k_1}{\neg\rho_2^*}{F^*\ (v_1,k_1)})\ k_2.
\end{multline*}
From the operational semantics we can deduce:
\begin{gather*}
(S_1,F\ V_1)\rightarrowtail^2(S_2\circ(\letin{(-)}{w}{wV_1}),W[\overrightarrow{V_i/x_i}]\ E)	\\
(\lbd{k_1}{\neg\rho_2^*}{F^*\ (v_1,k_1)})\ k_2 \longrightarrow^* W^*[\overrightarrow{v_i/x_i}]\ (E^*,\lbd{w}{(\rho_1{\rightarrow}\rho_2)^*}{((\lbd{l''}{\neg\rho_2^*}{w\ (v_1,l'')})\ k_2)}).
\end{gather*}

For $p_2<2$, we immediately have that $(S_2,F\ V_1)\ \mathcal{S}^{p_2}_{\rho_2,\rho'}\ (\lbd{k_1}{\neg\rho_2^*}{F^*\ (v_1,k_1)})\ k_2$ because in the definition of $\mathcal{S}^{p_2}$, the premises of both implications are false.

If $p_2 \geq 2$, then using Lemma \ref{Lem_sim_eval} it suffices to show:
\begin{multline*}
(S_2\circ(\letin{(-)}{w}{wV_1}),W[\overrightarrow{V_i/x_i}]\ E)\ \mathcal{S}^{p_2-2}_{\rho_1{\rightarrow}\rho_2,\rho''}	\\ W^*[\overrightarrow{v_i/x_i}]\ (E^*,\lbd{w}{(\rho_1{\rightarrow}\rho_2)^*}{((\lbd{l''}{\neg\rho_2^*}{w\ (v_1,l'')})\ k_2)}).
\end{multline*}
From the induction hypothesis for $W$ we know that:
\begin{equation}\label{Eq_rtp_1}
(W[\overrightarrow{V_i/x_i}],W^*[\overrightarrow{v_i/x_i}])\in\mathcal{R}^{\mathfrak{v},p_2-1}_{(\rho_1{\rightarrow}\rho_2){\rightarrow}(\rho_1{\rightarrow}\rho_2)}.
\end{equation} 
We would like to use this to prove the previous statement. Unpacking the definition of $\mathcal{R}^{\mathfrak{v},p_2-1}_{(\rho_1{\rightarrow}\rho_2){\rightarrow}(\rho_1{\rightarrow}\rho_2)}$ we see that if we prove:
\begin{enumerate}
\item \label{Enum_rtp_3} $(E,E^*)\in\mathcal{R}^{\mathfrak{v},p_2-2}_{\rho_1{\rightarrow}\rho_2}$.
\item \label{Enum_rtp_4} $(S_2\circ(\letin{(-)}{w}{wV_1}),\lbd{w}{(\rho_1{\rightarrow}\rho_2)^*}{((\lbd{l''}{\neg\rho_2^*}{w\ (v_1,l'')})\ k_2)}) \in\mathcal{R}^{\mathfrak{s},p_2-2}_{\rho_1{\rightarrow}\rho_2,\rho''}$.
\end{enumerate}
we can then use equation \ref{Eq_rtp_1} to obtain the desired result. 

We first prove the second statement, which is equivalent to:
\begin{multline*}
\forall p_3\leq p_2-2.\ \forall(V_3,v_3)\in\mathcal{R}^{\mathfrak{v},p_3}_{\rho_1{\rightarrow}\rho_2}.	\\
(S_2\circ(\letin{(-)}{w}{wV_1}),V_3)\ \mathcal{S}^{p_3}_{\rho_1{\rightarrow}\rho_2,\rho''}\ (\lbd{w}{(\rho_1{\rightarrow}\rho_2)^*}{((\lbd{l''}{\neg\rho_2^*}{w\ (v_1,l'')})\ k_2)})\ v_3.
\end{multline*}
From the operational semantics:
\begin{gather*}
(S_2\circ(\letin{(-)}{w}{wV_1}),V_3) \rightarrowtail (S_2,V_3\ V_1)	\\
(\lbd{w}{(\rho_1{\rightarrow}\rho_2)^*}{((\lbd{l''}{\neg\rho_2^*}{w\ (v_1,l'')})\ k_2)})\ v_3 \longrightarrow^* v_3\ (v_1,k_2).
\end{gather*}
Therefore, it is enough to show:
\begin{equation}\label{Eq_rtp_2}
(S_2,V_3\ V_1)\ \mathcal{S}^{p_3-1}_{\rho_2,\rho''}\ v_3\ (v_1,k_2).
\end{equation}
We have assumed $(V_3,v_3)\in\mathcal{R}^{\mathfrak{v},p_3}_{\rho_1{\rightarrow}\rho_2}$ and $(V_1,v_1)\in\mathcal{R}^{\mathfrak{v},p_1}_{\rho_1}\subseteq\mathcal{R}^{\mathfrak{v},p_3-1}_{\rho_1}$, so by definition of $\mathcal{R}^{\mathfrak{v}}$ we know:
\begin{equation*}
(V_3\ V_1, \lbd{k}{\neg\rho_2^*}{v_3\ (v_1,k)}) \in\mathcal{R}^{\mathfrak{c},p_3-1}_{\rho_2}.
\end{equation*}
We have also assumed $(S_2,k_2)\in\mathcal{R}^{\mathfrak{s},p_2}_{\rho_2,\rho'} \subseteq \mathcal{R}^{\mathfrak{s},p_3-1}_{\rho_2,\rho'}$ so by definition of $\mathcal{R}^\mathfrak{c}$ we know:
\begin{equation*}
(S_2,V_3\ V_1)\ \mathcal{S}^{p_3-1}_{\rho_2,\rho''}\ (\lbd{k}{\neg\rho_2^*}{v_3\ (v_1,k)})\ k_2.
\end{equation*}
So from Lemma \ref{Lem_sim_eval} we obtain the desired result, equation  \ref{Eq_rtp_2}.

Now prove the first statement, \ref{Enum_rtp_3}: this is equivalent to proving
\begin{multline*}
\forall p_4<p_2-2.\ \forall(V_4,v_4)\in\mathcal{R}^{\mathfrak{v},p_4}_{\rho_1}.\ \forall p_5 \leq p_4.\ \forall(S_5,k_5)\in\mathcal{R}^{\mathfrak{s},p_5}_{\rho_2,\sigma}.	\\
(S_5,E\ V_4)\ \mathcal{S}^{p_5}_{\rho_2,\sigma}\ (\lbd{k}{\neg\rho_2^*}{E^*\ (v_4,k)})\ k_5.
\end{multline*}
Using the operational semantics we have:
\begin{gather*}
(S_5,E\ V_4) \rightarrowtail^2 (S_5 \circ (\letin{(-)}{z}{zV_4}), \fix{W}[\overrightarrow{V_i/x_i}])	\\
(\lbd{k}{\neg\rho_2^*}{E^*\ (v_4,k)})\ k_5 \longrightarrow^* (\fix{W})^*[\overrightarrow{v_i/x_i}]\ \lbd{z}{(\rho_1{\rightarrow}\rho_2)^*}{((\lbd{p'}{\neg\rho_2^*}{z\ (v_4,p')})\ k_5)}.
\end{gather*}
So it is enough to show:
\begin{multline}\label{Eq_rtp_5}
(S_5 \circ (\letin{(-)}{z}{zV_4}), \fix{W}[\overrightarrow{V_i/x_i}])\ \mathcal{S}^{p_5-2}_{\rho_1{\rightarrow}\rho_2,\sigma}	\\ (\fix{W})^*[\overrightarrow{v_i/x_i}]\ \lbd{z}{(\rho_1{\rightarrow}\rho_2)^*}{((\lbd{p'}{\neg\rho_2^*}{z\ (v_4,p')})\ k_5)}.
\end{multline}
From the induction hypothesis on natural numbers we know that:
\begin{equation*}
(\fix{W}[\overrightarrow{V_i/x_i}],(\fix{W})^*[\overrightarrow{v_i/x_i}])\in\mathcal{R}^{\mathfrak{c},p_5-2}_{\rho_1\rightarrow\rho_2}.
\end{equation*}
Similarly to the proof for $S_2$ and $k_2$, statement \ref{Enum_rtp_4}, we can show that:
\begin{equation*}
(S_5 \circ (\letin{(-)}{z}{zV_4}), \lbd{z}{(\rho_1{\rightarrow}\rho_2)^*}{((\lbd{p'}{\neg\rho_2^*}{z\ (v_4,p')})\ k_5)}) \in\mathcal{R}^{\mathfrak{s},p_5-2}_{\rho_1{\rightarrow}\rho_2,\sigma}.
\end{equation*}
Combining the last two observations and using the definition of $\mathcal{R}^\mathfrak{c}$ we obtain the desired result, equation \ref{Eq_rtp_5}.

\paragraph{Case (op), $M=\sigma(V;W)$ for some $\sigma\in\Sigma$.} The induction hypothesis is:
\begin{equation*}
\overrightarrow{x_i:\sigma_i}\vdash V\ \mathcal{R}^\mathfrak{v}_{\mathbbm{N}}\ V^* \quad\text{ and }\quad \overrightarrow{x_i:\sigma_i}\vdash W\ \mathcal{R}^\mathfrak{v}_{\mathbbm{N}{\rightarrow}\tau}\ W^*.
\end{equation*}

We need to prove;
\begin{multline*}
\forall n\in\mathbb{N}.\ \forall\overrightarrow{(V_i,v_i):\sigma_i}.\ (\forall i.\ (V_i,v_i)\in\mathcal{R}^{\mathfrak{v},n}_{\sigma_i})\implies 	\\
(\sigma(V;W)[\overrightarrow{V_i/x_i}], \lbd{k}{\neg\tau^*}{\sigma(V^*;x.W^*(x,k))[\overrightarrow{v_i/x_i}]}) \in \mathcal{R}^{\mathfrak{c},n}_{\tau}.
\end{multline*}

Consider $p\leq n$ and $(S,k)\in\mathcal{R}^{\mathfrak{s},p}_{\tau,\rho}$. It suffices to prove:
\begin{equation*}
(S,\sigma(V;W)[\overrightarrow{V_i/x_i}])\ \mathcal{S}^{p}_{\tau,\rho}\ (\lbd{k}{\neg\tau^*}{\sigma(V^*;x.W^*(x,k))[\overrightarrow{v_i/x_i}]})\ k.
\end{equation*}

We prove this by checking the conditions in the definition of $\mathcal{S}^p$. If $p=0$ the first condition is satisfied immeditaly because of the strict inequality. If $p>0$ we can choose $p'=0$ such that $(S,\sigma(V;W)[\overrightarrow{V_i/x_i}])\rightarrowtail^{p'}(S,\sigma(V;W)[\overrightarrow{V_i/x_i}])$. We see that $(\lbd{k}{\neg\tau^*}{\sigma(V^*;x.W^*(x,k))}[\overrightarrow{v_i/x_i}])\ k \longrightarrow^* \sigma(V^*;x.W^*(x,k))[\overrightarrow{v_i/x_i}]$ as required. By the induction hypothesis for $V$ we know that $V[\overrightarrow{V_i/x_i}]=\overline{m}$ and $V^*[\overrightarrow{v_i/x_i}]=\overline{m}$ for some $m\in\mathbb{N}$. It remains to show that:
\begin{equation*}
\forall l\in\mathbb{N}.\ (S, W[\overrightarrow{V_i/x_i}]\ \overline{l})\ \mathcal{S}^{p}_{\tau,\rho}\ W^*[\overrightarrow{v_i/x_i}]\ (\overline{l},k).
\end{equation*}
This follows from the induction hypothesis for $W$, using Lemma \ref{Lem_sim_eval}.

The second condition in the definition of $\mathcal{S}^p$ is satisfied because the premise of the implication is false. $(S,\sigma(V;W)[\overrightarrow{V_i/x_i}])$ cannot reduce anymore according to the $\rightarrowtail$ relation, and is not of the form $(S,\return{V'})$.

\paragraph{Case (ret), $M=\return{W}$.} From the induction hypothesis we know:
\begin{equation*}
\forall n\in\mathbb{N}.\ \forall\overrightarrow{(V_i,v_i):\sigma_i}.\ (\forall i.\ (V_i,v_i)\in\mathcal{R}^{\mathfrak{v},n}_{\sigma_i})\implies (W[\overrightarrow{V_i/x_i}], W^*[\overrightarrow{v_i/x_i}]) \in \mathcal{R}^{\mathfrak{v},n}_{\tau}.
\end{equation*}
We need to prove that:
\begin{multline*}
\forall n\in\mathbb{N}.\ \forall\overrightarrow{(V_i,v_i):\sigma_i}.\ (\forall i.\ (V_i,v_i)\in\mathcal{R}^{\mathfrak{v},n}_{\sigma_i})\implies 	\\
(\return{W[\overrightarrow{V_i/x_i}]}, (\return{W})^*[\overrightarrow{v_i/x_i}]) \in \mathcal{R}^{\mathfrak{c},n}_{\tau}.
\end{multline*}
So it suffices to show:
\begin{multline*}
\forall n\in\mathbb{N}.\ (W[\overrightarrow{V_i/x_i}], W^*[\overrightarrow{v_i/x_i}]) \in \mathcal{R}^{\mathfrak{v},n}_{\tau} \implies	\\ (\return{W[\overrightarrow{V_i/x_i}]}, (\return{W})^*[\overrightarrow{v_i/x_i}]) \in \mathcal{R}^{\mathfrak{c},n}_{\tau}.
\end{multline*}

Consider $p\leq n$ and $(S,k)\in\mathcal{R}^{\mathfrak{s},p}_{\tau,\rho}$. We need to show that:
\begin{equation*}
(S,\return{W[\overrightarrow{V_i/x_i}]})\ \mathcal{S}^{p}_{\tau,\rho}\ (\lbd{k}{\neg\tau^*}{k\ W^*[\overrightarrow{v_i/x_i}]})\ k.
\end{equation*}
Using Lemma \ref{Lem_sim_eval} it is enough to show $(S,\return{W[\overrightarrow{V_i/x_i}]})\ \mathcal{S}^{p}_{\tau,\rho}\ k\ W^*[\overrightarrow{v_i/x_i}]$. This follows from $(S,k)\in\mathcal{R}^{\mathfrak{s},p}_{\tau,\rho}$ and $(W[\overrightarrow{V_i/x_i}], W^*[\overrightarrow{v_i/x_i}]) \in \mathcal{R}^{\mathfrak{v},n}_{\tau}\subseteq\mathcal{R}^{\mathfrak{v},p}_{\tau}$ by definition of $\mathcal{R}^{\mathfrak{s}}$.

\paragraph{Case (let), $M=\letin{N_1}{y}{N_2}$.} The induction hypothesis is:
\begin{equation*} 
\overrightarrow{x_i:\sigma_i}\vdash N_1\ \mathcal{R}^\mathfrak{c}_{\tau_1}\ N_1^* \quad\text{ and }\quad \overrightarrow{x_i:\sigma_i},y:\tau_1\vdash N_2\ \mathcal{R}^\mathfrak{c}_{\tau_2}\ N_2^*.
\end{equation*}

We need to prove:  
\begin{multline*}
\forall n\in\mathbb{N}.\ \forall\overrightarrow{(V_i,v_i):\sigma_i}.\ (\forall i.\ (V_i,v_i)\in\mathcal{R}^{\mathfrak{v},n}_{\sigma_i})\implies 	\\
(\letin{N_1[\overrightarrow{V_i/x_i}]}{y}{N_2[\overrightarrow{V_i/x_i}]}, (\letin{N_1}{y}{N_2})^*[\overrightarrow{v_i/x_i}]) \in \mathcal{R}^{\mathfrak{c},n}_{\tau_2}.
\end{multline*}

Consider $p\leq n$ and $(S,k)\in\mathcal{R}^{\mathfrak{s},p}_{\tau_2,\rho}$. For $p\geq 1$ it is enough to show:
\begin{equation*}
(S\circ\letin{(-)}{y}{N_2[\overrightarrow{V_i/x_i}]}, N_1)\ \mathcal{S}^{p-1}_{\tau_2,\rho}\ N_1^*[\overrightarrow{v_i/x_i}]\ (\lbd{y}{\tau_1^*}{N_2^*[\overrightarrow{v_i/x_i}]}\ k).
\end{equation*}

By induction hypothesis for $N_1$ it is enough to show:
\begin{equation*}
(S\circ\letin{(-)}{y}{N_2[\overrightarrow{V_i/x_i}]}, \lbd{y}{\tau_1^*}{N_2^*[\overrightarrow{v_i/x_i}]}\ k) \in \mathcal{R}^{\mathfrak{s},p-1}_{\tau_1,\rho}.
\end{equation*}

Consider $p_1\leq p$ and $(W_1,w_1)\in\mathcal{R}^{\mathfrak{v},p_1}_{\tau_1}$. For $p_1\geq 1$ it is enough to show:
\begin{equation*}
(S, N_2[\overrightarrow{V_i/x_i}, W_1/y])\ \mathcal{S}^{p_1-1}_{\tau_2,\rho}\ N_2^*[\overrightarrow{v_i/x_i}, w_1/y]\ k.
\end{equation*}
This follows from the induction hypothesis for $N_2$.

\paragraph{Case (app), $M=W\ U$.} If $W$ has type $\tau_1\rightarrow\tau_2$ and $U$ has type $\tau_1$, for each $n\in\mathbb{N}$, we need to prove:
\begin{equation*}
\forall p\leq n.\ \forall (S,k)\in\mathcal{R}^{\mathfrak{s},p}_{\tau_2,\rho}.\ (S,(W\ U)[\overrightarrow{V_i/x_i}])\ \mathcal{S}^p_{\tau_2,\rho}\ (\lbd{k}{\neg\tau_2^*}{W^*[\overrightarrow{v_i/x_i}]\ (U^*[\overrightarrow{v_i/x_i}],k)})\ k.
\end{equation*}

This follows from $(W[\overrightarrow{V_i/x_i}], W^*[\overrightarrow{v_i/x_i}])\in\mathcal{R}^{\mathfrak{v},p+1}_{\tau_1{\rightarrow}\tau_2,\rho}$ and $(U[\overrightarrow{V_i/x_i}], U^*[\overrightarrow{v_i/x_i}])\in\mathcal{R}^{\mathfrak{v},p}_{\tau_1}$, which we know from the induction hypothesis.

\paragraph{Case (case), $M=\casess{V}{N_1}{y}{N_2}$.} We need to prove for each $n\in\mathbb{N}$:
\begin{multline*}
\forall p\leq n.\ \forall (S,k)\in\mathcal{R}^{\mathfrak{s},p}_{\tau,\rho}.\ (S,\ (\casess{V}{N_1}{y}{N_2})[\overrightarrow{V_i/x_i}])\ \mathcal{S}^p_{\tau,\rho}	\\
(\lbd{k}{\neg\tau^*}{\caset{V^*}{N_1^*\ k}{y}{N_2^*\ k}})[\overrightarrow{v_i/x_i}]\ k.
\end{multline*}

We proceed by a case split on whether $V[\overrightarrow{V_i/x_i}]=\overline{0}$ or not. Then the result follows from the induction hypothesis for $N_1$ or $N_2$, respectively.

\paragraph{Case (sid), $S=id$.} In this case we need to prove:
\begin{equation*}
\forall n\in\mathbb{N}.\ \forall p\leq n.\ \forall(V,v)\in\mathcal{R}^{\mathfrak{v},p}_{\rho}.\ (id,\return{V})\ \mathcal{S}^p_{\rho,\rho}\ (\lbd{x}{\neg\rho^*}{\downarrow})\ v.
\end{equation*}

We show that the conditions in the definition of $\mathcal{S}^p$ hold. The first condition holds because $(id,\return{V})\not\rightarrowtail$ and $(id,\return{V})$ is not of the form $(S,\sigma(V;W))$. For the second condition we can choose $p'=0\leq p$ such that $(id,\return{V})\rightarrowtail^{p'}(id,\return{V})$ and we have $(\lbd{x}{\neg\rho^*}{\downarrow})\ v\longrightarrow\downarrow$, as required.

\paragraph{Case (slet), $S=S'\circ\letin{(-)}{y}{M}$.} The induction hypothesis is:
\begin{equation*}
\overrightarrow{x_i:\sigma_i}\vdash S'\ \mathcal{R}^\mathfrak{s}_{\tau_2,\rho}\ S'^* \quad\text{ and }\quad \overrightarrow{x_i:\sigma_i},y:\tau_1\vdash M\ \mathcal{R}^\mathfrak{c}_{\tau_2}\ M^*.
\end{equation*}

We need to prove:
\begin{multline*}
\forall n\in\mathbb{N}.\ \forall p\leq n.\ \forall(V,v)\in\mathcal{R}^{\mathfrak{v},p}_{\tau_1}.\\ (S'[\overrightarrow{V_i/x_i}]\circ\letin{(-)}{y}{M[\overrightarrow{V_i/x_i}]},\return{V})\ \mathcal{S}^p_{\tau_1,\rho}\ (\lbd{y}{\tau_1^*}{(M^*\ S'^*)[\overrightarrow{v_i/x_i}]})\ v.
\end{multline*}

If $p>0$, from Lemma \ref{Lem_sim_eval}, it is enough to show that:
\begin{equation*}
(S'[\overrightarrow{V_i/x_i}],M[\overrightarrow{V_i/x_i},V/y])\ \mathcal{S}^{p-1}_{\tau_2,\rho}\ M^*[\overrightarrow{v_i/x_i},v/y]\ S'^*[\overrightarrow{v_i/x_i}].
\end{equation*}
The result follows from the induction hypothesis for $M$ and $S'$.

\paragraph{Case (lbd), $V=\lbd{y}{\tau_1}{M}$.} The result follows from the induction hypothesis for $M$, unpacking the definition of $\mathcal{R}^\mathfrak{v}_{\tau_1{\rightarrow}\tau_2}$.

\paragraph{Case (var), $V=y$.} We know by assumption that $\overrightarrow{x_i:\sigma_i}\vdash y:\tau$. We need to prove:
\begin{equation*}
\forall n\in\mathbb{N}.\ (\forall(V_i,v_i)\in\mathcal{R}^{\mathfrak{v},n}_{\sigma_i})\implies(y[\overrightarrow{V_i/x_i}],y[\overrightarrow{v_i/x_i}])\in\mathcal{R}^{\mathfrak{v},n}_{\tau}.
\end{equation*}
This is true because $y[\overrightarrow{V_i/x_i}]=V_j$, $y[\overrightarrow{v_i/x_i}]=v_j$ and $\tau=\sigma_j$ for some $j$.

\paragraph{Case (zero), $V=Z$.} By definition of $\mathcal{R}^{\mathfrak{v},m}_{\mathbbm{N}}$ we know that $(Z,Z)\in\mathcal{R}^{\mathfrak{v},m}_{\mathbbm{N}}$ for any $m\in\mathbb{N}$.

\paragraph{Case (unit), $V=\star$.} Analogous to the previous case.

\paragraph{Case (succ), $V=S(W)$.} We need to prove that:
\begin{equation*}
\forall n\in\mathbb{N}.\ (\forall(V_i,v_i)\in\mathcal{R}^{\mathfrak{v},n}_{\sigma_i})\implies(S(W[\overrightarrow{V_i/x_i}]),\texttt{succ}(W^*[\overrightarrow{v_i/x_i}]))\in\mathcal{R}^{\mathfrak{v},n}_{\mathbbm{N}}.
\end{equation*} 

From the induction hypothesis for $W$ we know that $W[\overrightarrow{V_i/x_i}]=\overline{m}$ and $W^*[\overrightarrow{v_i/x_i}]=\overline{m}$ for some $m\in\mathbb{N}$. So $S(W[\overrightarrow{V_i/x_i}])=\overline{m+1}$ and $\texttt{succ}(W^*[\overrightarrow{v_i/x_i}])=\overline{m+1}$, which gives us the required result.
\end{proof}

\begin{replemma}{Lem_ecps_tree_coalgmor}
The function $\treet{-}:(\vdash)\longrightarrow \textit{Trees}_\Sigma$  is a coalgebra morphism in the category of coalgebras for the functor $T$.
\end{replemma}
\begin{proof}
We will use the definitions of the functor $T$ and coalgebra morphisms $c$ and $a$ from Section \ref{Sec_cps_correct}.

First, we describe the $\omega$-CPO structure on $\textit{Trees}_\Sigma$. The partial order $\leq$ satisfies the following:
\begin{align}
&\forall \textit{tr}\in\textit{Trees}_\Sigma.\ \bot\leq\textit{tr}	\text{ and } (\textit{tr} \leq \bot \Longleftrightarrow \textit{tr}=\bot)	\\
&\forall \textit{tr}\in\textit{Trees}_\Sigma.\  (\downarrow\leq\textit{tr}\Longleftrightarrow\textit{tr}=\downarrow) \text{ and } (\textit{tr}\leq\downarrow\ \Longleftrightarrow \textit{tr}=\bot \text{ or } \textit{tr}=\downarrow)	\\
&\forall \sigma,\sigma'\in\Sigma,\ n,n'\in\mathbb{N},\ \overrightarrow{\textit{tr}_k},\overrightarrow{\textit{tr'}_k}\in\textit{Trees}_\Sigma.	\nonumber\\ &
\sigma_n(\overrightarrow{\textit{tr}_k})\leq\sigma'_{n'}(\overrightarrow{\textit{tr'}_k}) \quad\text{iff}\quad \sigma=\sigma' \text{ and } n=n' \text{ and } \forall k\in\mathbb{N}.\ \textit{tr}_k\leq \textit{tr'}_k. \label{ecps_tree_prop1}
\end{align}

From these properties we see that a chain $\textit{tr}_0\leq\textit{tr}_1\leq\textit{tr}_2\leq\ldots$ can have one of the following three forms:
\begin{enumerate}
\item $\forall n\in\mathbb{N}.\ \textit{tr}_n=\bot$. In this case define the least upper bound of the chain to be $\bigsqcup_{n\in\mathbb{N}}\textit{tr}_n=\bot$.

\item $\exists k\in\mathbb{N}$ such that $\forall n\geq k.\ \textit{tr}_n=\downarrow$ and $\forall n<k.\ \textit{tr}_n=\bot$. So we define the least upper bound to be $\bigsqcup_{n\in\mathbb{N}}\textit{tr}_n=\downarrow$.

\item $\exists k\in\mathbb{N}$ such that $\forall n\geq k.\ \textit{tr}_n=\sigma_i(\overrightarrow{\textit{tr}^n_j})$ and $\forall n<k.\ \textit{tr}_n=\bot$. By property \ref{ecps_tree_prop1} of $\leq$, $\sigma$ and $i$ need to stay the same for all $n$. So we can define the least upper bound as: $\bigsqcup_{n\in\mathbb{N}}\textit{tr}_n=\sigma_i(\bigsqcup_{n\geq k}\textit{tr}^n_1,\ \bigsqcup_{n\geq k}\textit{tr}^n_2,\ldots)$.
\end{enumerate}

Now we can define a partial order, $\leq_T$, on $T(\textit{Trees}_\Sigma)=\{\downarrow,\bot\}+\Sigma\times\mathbb{N}\times(\textit{Trees}_\Sigma)^\mathbb{N}$ which makes it into an $\omega$-CPO:
\begin{align*}
&\forall x\in T(\textit{Trees}_\Sigma).\ \bot\leq_T x	\text{ and } (x \leq_T \bot \Longleftrightarrow x=\bot)\\
&\forall x\in T(\textit{Trees}_\Sigma).\  (\downarrow\leq_T x\Longleftrightarrow x=\downarrow) \text{ and } (x\leq_T\downarrow\ \Longleftrightarrow x=\bot \text{ or } x=\downarrow)	\\
&\forall \sigma,\sigma'\in\Sigma,\ n,n'\in\mathbb{N},\ \overrightarrow{\textit{tr}_k},\overrightarrow{\textit{tr'}_k}\in\textit{Trees}_\Sigma.	\nonumber\\ &
(\sigma,n,(\overrightarrow{\textit{tr}_k}))\leq_T(\sigma',n',(\overrightarrow{\textit{tr'}_k})) \quad\text{iff}\quad \sigma=\sigma' \text{ and } n=n' \text{ and } \forall k\in\mathbb{N}.\ \textit{tr}_k\leq \textit{tr'}_k.
\end{align*}

Given this definition, a chain $x_0\leq_T x_1 \leq_T x_2 \leq_T\ldots$ can only have one of the three forms below:
\begin{enumerate}
\item $\forall n\in\mathbb{N}.\ x_n=\bot$. Define the least upped bound as: $\bigsqcup_{n\in\mathbb{N}}x_n=\bot$.

\item $\exists k\in\mathbb{N}$ such that $\forall n\geq k.\ x_n=\downarrow$ and $\forall n<k.\ x_n=\bot$. Define the least upped bound as: $\bigsqcup_{n\in\mathbb{N}}x_n=\downarrow$.

\item $\exists k\in\mathbb{N}$ such that $\forall n\geq k.\ x_n=(\sigma,i,(\overrightarrow{\textit{tr}^n_j}))$ and $\forall n<k.\ x_n=\bot$. Again it must be the case that $\sigma$ and $i$ are the same for all $n\geq k$ so we define the least upper bound as: $\bigsqcup_{n\in\mathbb{N}}x_n=(\sigma,i,(\bigsqcup_{n\geq k}\textit{tr}^n_1,\ \bigsqcup_{n\geq k}\textit{tr}^n_2,\ldots))$.
\end{enumerate}

To prove $\treet{-}$ is a coalgebra morphism, we must show that the following diagram commutes:

\begin{equation*}
\begin{tikzcd}
(\vdash) \arrow[r, "\treet{-}"] \arrow[d, "a"] 	&	\textit{Trees}_\Sigma \arrow[d, "c"]	\\
\{\downarrow,\bot\}+\Sigma\times\mathbb{N}\times(\vdash)^\mathbb{N} \arrow[r, "T(\treet{-})"]	&	\{\downarrow,\bot\}+\Sigma\times\mathbb{N}\times(\textit{Trees}_\Sigma)^\mathbb{N}
\end{tikzcd}
\end{equation*}

First show that for any chain $\textit{tr}_0\leq\textit{tr}_1\leq\textit{tr}_2\leq\ldots$ in $\textit{Trees}_\Sigma$:
\begin{equation*}
c(\bigsqcup_{n\in\mathbb{N}}\textit{tr}_n)=\bigsqcup_{n\in\mathbb{N}}c(\textit{tr}_n).
\end{equation*}
We do a case split on the structure of the chain following the cases that we identified previously:
\begin{enumerate}
\item $\forall n\in\mathbb{N}.\ \textit{tr}_n=\bot$. Then $c(\bigsqcup_{n\in\mathbb{N}}\textit{tr}_n)=\bot$ and $c(\textit{tr}_n)=\bot$ for all $n$ so we are done.

\item $\exists k\in\mathbb{N}$ such that $\forall n\geq k.\ \textit{tr}_n=\downarrow$ and $\forall n<k.\ \textit{tr}_n=\bot$. Then $c(\bigsqcup_{n\in\mathbb{N}}\textit{tr}_n)=c(\downarrow)=\downarrow$ and $\bigsqcup_{n\in\mathbb{N}}c(\textit{tr}_n)=\bigsqcup_{n\geq k}\downarrow=\downarrow$.

\item $\exists k\in\mathbb{N}$ such that $\forall n\geq k.\ \textit{tr}_n=\sigma_i(\overrightarrow{\textit{tr}^n_j})$ and $\forall n<k.\ \textit{tr}_n=\bot$. In this case: 
\begin{gather*}
c(\bigsqcup_{n\in\mathbb{N}}\textit{tr}_n)=c(\bigsqcup_{n\geq k}\textit{tr}_n)=c(\sigma_i(\bigsqcup_{n\geq k}\textit{tr}^n_1,\ \bigsqcup_{n\geq k}\textit{tr}^n_2,\ldots))=(\sigma, i, (\bigsqcup_{n\geq k}\textit{tr}^n_1,\ \bigsqcup_{n\geq k}\textit{tr}^n_2,\ldots))	\\
\bigsqcup_{n\in\mathbb{N}}c(\textit{tr}_n) = \bigsqcup_{n\geq k}(\sigma,i,(\textit{tr}^n_1,\textit{tr}^n_2,\ldots)) = (\sigma, i, (\bigsqcup_{n\geq k}\textit{tr}^n_1,\ \bigsqcup_{n\geq k}\textit{tr}^n_2,\ldots)).
\end{gather*}
so the two sides are equal as required.
\end{enumerate}

Now we can deduce that:
\begin{equation}\label{ecps_tree_prf2}
c(\treet{t})=c(\bigsqcup_{n\in\mathbb{N}}\treet{t}_n) = \bigsqcup_{n\in\mathbb{N}}c(\treet{t}_n).
\end{equation}

On the other side of the diagram we have $T(\treet{-})(a(t))$ which by definition of $T$ is equal to:
\begin{equation*}
T(\treet{-})(a(t))=\begin{cases}
						\downarrow &\text{if } a(t)=\downarrow	\\
						\bot &\text{if } a(t)=\bot	\\
						(\sigma,k,\overrightarrow{\treet{t'[\overline{i}/x]}}) &\text{if } a(t)=(\sigma,k,\overrightarrow{t'[\overline{i}/x]}).
					\end{cases}
\end{equation*}
By definition of $a$ the side conditions can be expressed in terms of the reduction behaviour of $t$:
\begin{equation*}
T(\treet{-})(a(t))=\begin{cases}
						\downarrow &\text{if } t\longrightarrow^*\downarrow	\\
						\bot &\text{if } t\longrightarrow^\infty	\\
						(\sigma,k,\overrightarrow{\bigsqcup_{n\in\mathbb{N}}\treet{t'[\overline{i}/x]}_n}) &\text{if } t\longrightarrow^*\sigma(\overline{k},x.t').
					\end{cases}
\end{equation*}
By the definition of least upper bound in $T(\textit{Trees}_\Sigma)$ we can rewrite the last case as:
\begin{equation}\label{ecps_tree_prf3}
T(\treet{-})(a(t))=\begin{cases}
						\downarrow &\text{if } t\longrightarrow^*\downarrow	\\
						\bot &\text{if } t\longrightarrow^\infty	\\
						\bigsqcup_{n\in\mathbb{N}}(\sigma,k,\overrightarrow{\treet{t'[\overline{i}/x]}_n}) &\text{if } t\longrightarrow^*\sigma(\overline{k},x.t').
					\end{cases}
\end{equation}

To see that equations \ref{ecps_tree_prf2} and \ref{ecps_tree_prf3} are in fact equal we proceed by a case analysis on the reduction of $t$. By definition of $\longrightarrow$ the following cases are exhaustive and only one of them can occur:
\begin{enumerate}
\item $t\longrightarrow^*\downarrow$. There exists $m$ such that $t\longrightarrow^m\downarrow$. By definition of $\treet{-}$ we can deduce that $\forall n\geq m+1.\ \treet{t}_n=\downarrow$ so $c(\treet{t}_n)=\downarrow$. And $\forall n< m+1.\ \treet{t}_n=\bot$. Therefore we have:
\begin{equation*}
\bigsqcup_{n\in\mathbb{N}}c(\treet{t}_n)=\downarrow.
\end{equation*}
as required.

\item $t\longrightarrow^\infty$. For each $n\in\mathbb{N}$ there exists $t_n$ such that $t\longrightarrow^n t_n$. Therefore $\forall n\in\mathbb{N}.\ \treet{t}_n=\treet{t_n}_0=\bot$ so by the definition of least upper bound in $T(\textit{Trees}_\Sigma)$:
\begin{equation*}
\bigsqcup_{n\in\mathbb{N}}c(\treet{t}_n)=\bot.
\end{equation*}

\item $t\longrightarrow^*\sigma(\overline{k},x.t')$. There exists $m\in\mathbb{N}$ such that $t\longrightarrow^m\sigma(\overline{k},x.t')$. By definition of $\treet{-}_n$ we can deduce that:
\begin{gather*}
\forall n\geq m+1.\ \treet{t}_n=\sigma_k(\treet{t'[\overline{0}/x]}_{n-m-1},\treet{t'[\overline{1}/x]}_{n-m-1},\ldots)	\\
\forall n<m+1.\ \treet{t}_n = \bot.
\end{gather*}
From this we know that:
\begin{equation*}
\bigsqcup_{n\in\mathbb{N}}c(\treet{t}_n)=\bigsqcup_{n\geq m+1}c(\treet{t}_n)=\bigsqcup_{n\geq m+1}(\sigma,k,\overrightarrow{\treet{t'[\overline{i}/x]}_{n-m-1}})=T(\treet{-})(a(t)).
\end{equation*}
\end{enumerate}				
\end{proof}

\begin{repproposition}{Prop_bisim_impl_trees}
For any well-typed EPCF configuration $(S,M)$, where $S:\tau\Rightarrow\rho$, and any ECPS computation $t$:
\begin{equation*}
((S,M),t) \in \mathcal{S}_{\tau,\rho} \cap \mathcal{S'}_{\tau,\rho} \implies \treet{t} = \trees{S,M}[\downarrow/l_1,\downarrow/l_2,\ldots].
\end{equation*}
\end{repproposition}
\begin{proof}
By coinduction. Define the relation $\mathcal{B}\subseteq (Stack\times Comp)\times(\vdash)$ as:
\begin{equation*}
\mathcal{B} = \bigcup_{\tau,\rho}(\mathcal{S}_{\tau,\rho}\cap\mathcal{S'}_{\tau,\rho}).
\end{equation*}
We will show that $\mathcal{B}$ is a bisimulation in the abstract sense of Definition \ref{Def_bisim_abstract}. Consider the following morphism of type $\mathcal{B}\longrightarrow T(\mathcal{B})$:
\begin{equation*}
r((S,M),t) = \begin{cases}
				\downarrow	&\text{if } (S,M)\rightarrowtail^* (id,\return{V})	\\
				\bot	&\text{if } (S,M)\rightarrowtail^\infty	\\
				(\sigma,n,(\overrightarrow{(S',W\overline{l}),t'[\overline{l}/x]}))	&\text{if } (S,M)\rightarrowtail^* (S',\sigma(\overline{n};W)), \text{which implies	}	\\
					&\text{from definition of } \mathcal{S} \text{ that } t\longrightarrow^*\sigma(\overline{n},x.t').
			 \end{cases}
\end{equation*}
This morphism makes the following diagram commute, so it makes $\mathcal{B}$ into a bisimulation:
\begin{equation*}
\begin{tikzcd}
 \textit{Satck}\times\textit{Comp}  \arrow[r, leftarrow, "\pi_1"] \arrow[d, "b"] & \mathcal{B} \arrow[r, "\pi_2"] \arrow[d, "r"] & (\vdash) \arrow[d, "a"] \\
T(\textit{Stack}\times\textit{Comp}) \arrow[r, leftarrow, "T(\pi_1)"] & T(\mathcal{B}) \arrow[r, "T(\pi_2)"] & T(\vdash)
\end{tikzcd}
\end{equation*}
By definition of $b$ and $r$ we have:
\begin{equation*}
\forall((S,M),t)\in\mathcal{B}.\ b(\pi_1((S,M),t)) = T(\pi_1)(r((S,M),t)).
\end{equation*}
For the second equation: we know that for any $((S,M),t)\in\mathcal{B}$:
\begin{align}
a(\pi_2((S,M),t)) = a(t) = &\begin{cases}
		\downarrow	&	\text{if } t\longrightarrow^*\downarrow	\\
		\bot	&	\text{if } t\longrightarrow^\infty	\\		
		(\sigma,k,\overrightarrow{t'[\overline{n}/x]})	&	\text{if } t\longrightarrow^*\sigma(\overline{k},x.t').
	   \end{cases}	\label{eq_coind_prf1}\\
T(\pi_2)(r((S,M),t)) = &\begin{cases}
						\downarrow	&\text{if } (S,M)\rightarrowtail^* (id,\return{V})	\\
						\bot	&\text{if } (S,M)\rightarrowtail^\infty	\\
						(\sigma,n,\overrightarrow{t'[\overline{l}/x]})	&\text{if } (S,M)\rightarrowtail^* (S',\sigma(\overline{n};W)), \text{which implies	}	\\
					&\text{from definition of } \mathcal{S} \text{ that } t\longrightarrow^*\sigma(\overline{n},x.t').
					   \end{cases} \label{eq_coind_prf2}
\end{align}
Expressions \ref{eq_coind_prf1} and \ref{eq_coind_prf2} can be proved equal by using the definitions of $\mathcal{S}$ and $\mathcal{S}'$. If $(S,M)\rightarrowtail^* (S',\sigma(\overline{n};W))$ then we know by definition of $\mathcal{S}$ that $t\longrightarrow^*\sigma(\overline{n},x.t')$. If $(S,M)\rightarrowtail^* (id,\return{V})$ then again by definition of $\mathcal{S}$ we have $t\longrightarrow^*\downarrow$. If $(S,M)\rightarrowtail^\infty$, because the reduction relation $\rightarrowtail$ is deterministic we know that $(S,M)\not\rightarrowtail^* (id,\return{V})$ and $(S,M)\not\rightarrowtail^* (S',\sigma(\overline{n};W))$ for any $n\in\mathbb{N}$. Therefore, by definition of $\mathcal{S}'$, $t\not\longrightarrow^*\downarrow$ and $t\not\longrightarrow^*\sigma(v,x.t')$. So by the definition of $\longrightarrow$ it must be the case that $t\longrightarrow^\infty$, as required. We can prove the reverse implication by making an assumption about the reduction of $t$ and proceeding analogously.

By assumption, we know $((S,M),t)\in\mathcal{B}$. As discussed in Section~\ref{Sec_cps_correct}, $\textit{Trees}_\Sigma$ is a final coalgebra. From Lemmas~\ref{Lem_ecps_tree_coalgmor} and \ref{Lem_epcf_tree_colagmor} we know that $\treet{-}$ and $\beta$ are coalgebra morphisms into it.
Using the fact that $\mathcal{B}$ is a bisimulation, we can apply the coinduction proof principle, Proposition \ref{Prop_coind_princip}, to deduce:
\begin{equation*}
\treet{t} = \beta(S,M) = \trees{S,M}^* = \trees{S,M}[\downarrow/l_1,\downarrow/l_2,\ldots].
\end{equation*}
as required.
\end{proof}

\chapter{Proofs about Applicative Bisimilarity}\label{App_bisim}

\begin{repproposition}{Prop_bisim_is_sim_and_simop}
Applicative $\mathfrak{P}$-bisimilarity coincides with the intersection between applicative $\mathfrak{P}$-similarity and its converse:
\begin{equation*}
(\sim) = (\precsim)\cap(\precsim)^{\textit{op}}.
\end{equation*}
\end{repproposition}
\begin{proof}
We prove each inclusion in turn.

\paragraph{``$\subseteq$''.} Bisimilarity is a bisimulation, hence also a simulation so $(\sim)\subseteq(\precsim)$. Therefore $(\sim)^{\textit{op}}\subseteq(\precsim)^{\textit{op}}$. But bisimilarity is by definition symmetric so $(\sim)=(\sim)^{\textit{op}}$ so we can deduce that $(\sim)\subseteq(\precsim)^{\textit{op}}$. Hence, $(\sim) \subseteq (\precsim)\cap(\precsim)^{\textit{op}}$ as required.

\paragraph{``$\supseteq$''.} The strategy is to check that $(\precsim)\cap(\precsim)^{\textit{op}}$ is a symmetric simulation. Consider $(s,t)\in(\precsim)\cap(\precsim)^{\textit{op}}$. Then $s\precsim t$ and $t \precsim s$, which implies $t \precsim^{\textit{op}} s$ and $t \precsim s$. Therefore, $(t,s)\in (\precsim)\cap(\precsim)^{\textit{op}}$ and $(\precsim)\cap(\precsim)^{\textit{op}}$ is symmetric.

The relation $(\precsim)\cap(\precsim)^{\textit{op}}$ satisfies the first three conditions in the definition of simulation because we know $\precsim$ is a simulation. For the fourth condition assume: 
\begin{equation*}
(v,u)\in ((\precsim)\cap(\precsim)^{\textit{op}})^\mathfrak{v}_{\neg(\overrightarrow{A_i})}.
\end{equation*}
Then $v\precsim^\mathfrak{v}_{\neg(\overrightarrow{A_i})} w$ and $w \precsim^\mathfrak{v}_{\neg(\overrightarrow{A_i})} v$ and since $\precsim$ is a simulation we know that:
\begin{gather*}
\forall \vdash\overrightarrow{u_i:A_i}.\ v(\overrightarrow{u_i})\precsim^\mathfrak{c} w(\overrightarrow{u_i})	\\
\forall \vdash\overrightarrow{u_i:A_i}.\ w(\overrightarrow{u_i})\precsim^\mathfrak{c} v(\overrightarrow{u_i}).
\end{gather*}
Therefore:
\begin{equation*}
\forall \vdash\overrightarrow{u_i:A_i}.\ (v(\overrightarrow{u_i}), w(\overrightarrow{u_i})) \in ((\precsim)\cap(\precsim)^{\textit{op}})^\mathfrak{c}. 
\end{equation*}

So we have shown that $(\precsim)\cap(\precsim)^{\textit{op}}$ is a symmetric simulation and it is therefore included in the union of all symmetric simulations $\sim$.
\end{proof}

\begin{replemma}{Lemm_(bi)sim_preord}
Applicative $\mathfrak{P}$-similarity is a preorder. Applicative $\mathfrak{P}$-bisimilarity is an equivalence relation.
\end{replemma}
\begin{proof}
Prove that similarity is reflexive and transitive.

\paragraph{Reflexivity.} Let $\mathcal{I}^\mathfrak{v}_A$ and $\mathcal{I}^\mathfrak{c}$ be the identity relations. We will show $(\mathcal{I}^\mathfrak{v}_A,\mathcal{I}^\mathfrak{c})\subseteq\precsim$ so $\precsim$ is reflexive. For this it suffices to show $(\mathcal{I}^\mathfrak{v}_A,\mathcal{I}^\mathfrak{c})$ is a simulation. Conditions \ref{ecps_sim1} and \ref{ecps_sim2} are satisfied by definition. For condition \ref{ecps_sim3} assume $s\ \mathcal{I}^\mathfrak{c}\ t$. Then it must be the case that $s=t$ so $\treet{s} = \treet{t}$. Therefore $\forall P\in\mathfrak{P}.\ (\treet{s} \in P \implies \treet{t} \in P)$. For condition \ref{ecps_sim4} assume $v\ \mathcal{I}^\mathfrak{v}_{\neg(A_1,\ldots,A_n)}\ u$. Then $v=u$ so the result follows from the definition of $\mathcal{I}^\mathfrak{c}$.

\paragraph{Transitivity.} First assume that there exist some closed computations such that $r\precsim s\precsim t$. Then we know from the definition of similarity and transitivity of implication that:
\begin{equation*}
\forall P\in\mathfrak{P}.\ \treet{r}\in P \implies \treet{t}\in P.
\end{equation*}
Since similarity is the greatest relation with property \ref{ecps_sim3}, it must be the case that $r\precsim t$.

Assume $u\precsim v\precsim w$. Then there exist simulations $\mathcal{R}$ and $\mathcal{S}$ such that $u\ \mathcal{R}^\mathfrak{v}_A\ v$ and $v\ \mathcal{S}^\mathfrak{v}_A\ w$. Proceed by a case distinction on $A$.

If $A=\mathtt{nat}$ or $A=\mathtt{unit}$, then it must be the case that $u=v=w$. Let $\mathcal{T}^\mathfrak{v}_A=\{(u,w)\}$, $\mathcal{T}^\mathfrak{v}_{B\not=A}=\emptyset$, $\mathcal{T}^\mathfrak{c}=\emptyset$. The relation $(\mathcal{T}^\mathfrak{v}_A,\mathcal{T}^\mathfrak{v}_{B\not=A},\mathcal{T}^\mathfrak{c})$ is a simulation because $u=w$, and is therefore included in $\precsim$.

If $A=\neg(A_1,\ldots,A_n)$, consider the candidate simulation $\mathcal{T}^\mathfrak{v}_A=\{(u,w)\}$, $\mathcal{T}^\mathfrak{v}_{B\not=A}=\emptyset$, $\mathcal{T}^\mathfrak{c}=\{(u\ (w_1,\ldots,w_n), w\ (w_1,\ldots,w_n)) \mid \vdash w_1 : A_1,\ldots,\vdash w_n : A_n\}$. Condition \ref{ecps_sim4} is satisfied by definition of $\mathcal{T}^\mathfrak{c}$. Condition \ref{ecps_sim3} is satisfied because $u\ \mathcal{R}^\mathfrak{v}_A\ v$ and $v\ \mathcal{S}^\mathfrak{v}_A\ w$ imply:
\begin{equation*}
\forall \vdash w_1:A_1,\ldots,\vdash w_n:A_n.\ \forall P\in\mathfrak{P}.\ \treet{u(w_1,\ldots,w_n)}\in P \implies \treet{w(w_1,\ldots,w_n)} \in P.
\end{equation*}

From Proposition \ref{Prop_bisim_is_sim_and_simop} we know that:
\begin{equation*}
(\sim) = (\precsim)\cap(\precsim)^{\textit{op}}.
\end{equation*}
The relation $(\precsim)\cap(\precsim)^{\textit{op}}$ is reflexive and transitive because similarity is, and it is by definition symmetric. So $\sim$ is an equivalence relation.
\end{proof}

\begin{replemma}{Lem_comp_1premise_rules}
Consider a well-typed relation $\mathcal{R}$ that is a preorder. The compatibility rules \textsc{(comp6)}, \textsc{(comp7)}, \textsc{(comp8)} and \textsc{(comp10)} are equivalent to the conjunction of their single-premise versions. 
\end{replemma}
\begin{proof}
As an example, consider the case of (\textsc{(comp7)}$\Longleftrightarrow$\textsc{(comp7L)} and \textsc{(comp7R$_i$)}). The proof for the other cases is analogous.

Assume that \textsc{(comp7)} is true and that 
\begin{gather*}
\Gamma,x:\neg(\overrightarrow{A}) \vdash v\ \mathcal{R}^\mathfrak{v}_{\neg(\overrightarrow{A})}\ v'	\\
\Gamma\vdash w_i\ \mathcal{R}^\mathfrak{v}_{A_i}\ w'_i \text{ for each } i.
\end{gather*}
By reflexivity of $\mathcal{R}$ we know $\Gamma\vdash w_i\ \mathcal{R}^\mathfrak{v}_{A_i}\ w_i$ for each $i$. So we can apply rule \textsc{(comp7)} to deduce the conclusion of rule \textsc{(comp7l)}. 

Also by reflexivity, we know $\Gamma,x:\neg(\overrightarrow{A}) \vdash v\ \mathcal{R}^\mathfrak{v}_{\neg(\overrightarrow{A})}\ v$ and $\Gamma\vdash w_j\ \mathcal{R}^\mathfrak{v}_{A_j}\ w_j$ for each $j\not=i$. Now apply rule \textsc{(comp7)} to deduce the conclusion of rule \textsc{(comp7r$_i$)}.

For the reverse implication assume that the rules \textsc{(comp7L)} and \textsc{(comp7R$_i$)} are true and that:
\begin{gather*}
\Gamma,x:\neg(\overrightarrow{A}) \vdash v\ \mathcal{R}^\mathfrak{v}_{\neg(\overrightarrow{A})}\ v'	\\
\Gamma\vdash w_i\ \mathcal{R}^\mathfrak{v}_{A_i}\ w'_i \text{ for each } i.
\end{gather*}
Because $\mathcal{R}$ is well-typed we know that all the values considered above are well-typed. Therefore, we can apply \textsc{(comp7L)} and \textsc{(comp7R$_i$)} one by one to obtain:
\begin{gather*}
\Gamma\vdash (\mufix{x}{v})(\overrightarrow{w})\ \mathcal{R}^\mathfrak{c}\ (\mufix{x}{v'})(\overrightarrow{w})	\\
\Gamma\vdash (\mufix{x}{v'})(w_1,w_2,\ldots,w_n)\ \mathcal{R}^\mathfrak{c}\ (\mufix{x}{v'})(w'_1,w_2,\ldots,w_n)	\\
\ldots	\\
\Gamma\vdash (\mufix{x}{v'})(w'_1,\ldots,w'_i,\ldots,w'_{n-1},w_n)\ \mathcal{R}^\mathfrak{c}\ (\mufix{x}{v'})(w'_1,\ldots,w'_i,\ldots,w'_{n-1},w'_n)
\end{gather*}
so by transitivity of $\mathcal{R}$ we have $\Gamma\vdash (\mufix{x}{v})(\overrightarrow{w})\ \mathcal{R}^\mathfrak{c}\ (\mufix{x}{v'})(\overrightarrow{w'})$ as required.
\end{proof}

\section{Howe's Method}\label{App_sec_howe}

\begin{replemma}{Lemm_howe_extension_misc}[{From \cite[Appendix]{SimV17}}]
Given a well-typed relation $\mathcal{R}$ on closed terms that is reflexive:
\begin{enumerate}
\item The Howe extension of $\mathcal{R}$, $\mathcal{R}^{\mathcal{H}}$, is compatible and hence reflexive.
\item $\mathcal{R}^\circ\ \subseteq\ \mathcal{R}^{\mathcal{H}}$.
\end{enumerate}
\end{replemma}
\begin{proof}
\begin{enumerate}
\item From the definitions of compatibility and compatible refinement we see that $\widehat{\mathcal{R}^{\mathcal{H}}} \subseteq \mathcal{R}^{\mathcal{H}}$ implies that $\mathcal{R}^{\mathcal{H}}$ is compatible.

We know $\mathcal{R}$ is reflexive. Therefore $\mathcal{R}^\circ$ is reflexive. Let $\textit{Id}$ be the identity well-typed open relation. Then:
\begin{equation*}
\widehat{\mathcal{R}^{\mathcal{H}}} = \textit{Id} \circ \widehat{\mathcal{R}^{\mathcal{H}}}\ \subseteq\ \mathcal{R}^\circ\circ\widehat{\mathcal{R}^{\mathcal{H}}} =\ \mathcal{R}^{\mathcal{H}}.
\end{equation*}

To show $\mathcal{R}^{\mathcal{H}}$ is reflexive, that is, for all terms  $s$, $s\mathcal{R}^{\mathcal{H}} s$
we can proceed by induction on $s$ using the fact that $\mathcal{R}^{\mathcal{H}}$ is compatible.

\item Now we know that $\mathcal{R}^{\mathcal{H}}$ is reflexive. By the definition of compatible refinement we can easily see that this implies $\widehat{\mathcal{R}^{\mathcal{H}}}$ is reflexive. Therefore:
\begin{equation*}
\mathcal{R}^\circ\ =\ \mathcal{R}^\circ \circ\ \textit{Id}\ \subseteq\  \mathcal{R}^\circ \circ\  \widehat{\mathcal{R}^{\mathcal{H}}} =\ \mathcal{R}^{\mathcal{H}}.
\end{equation*}
\end{enumerate}
\end{proof}

\begin{replemma}{Lemm_howe_method_misc}[{From \cite[Appendix]{SimV17}}]
Given a well-typed relation $\mathcal{R}$ on closed terms that is transitive:
\begin{equation*}
\mathcal{R}^\circ\circ\mathcal{R}^{\mathcal{H}} \subseteq \mathcal{R}^{\mathcal{H}}.
\end{equation*}
\end{replemma}
\begin{proof}
Consider three terms $s$, $t$ and $r$, either values or computations, such that $s\mathrel{\mathcal{R}^{\mathcal{H}}} t \mathrel{\mathcal{R}^\circ} r$. By definition of $\mathcal{R}^{\mathcal{H}}$ this means there exists a term $t'$ such that:
\begin{equation*}
s\ \mathrel{\widehat{\mathcal{R}^{\mathcal{H}}}} t' \mathrel{\mathcal{R}^\circ} t \mathrel{\mathcal{R}^\circ} r.
\end{equation*}

Since $\mathcal{R}$ is transitive, $\mathcal{R}^\circ$ is also transitive. Therefore:
\begin{equation*}
s\ \mathrel{\widehat{\mathcal{R}^{\mathcal{H}}}}\ t' \mathrel{\mathcal{R}^\circ} r
\end{equation*}
so by definition of $\mathcal{R}^{\mathcal{H}}$ we have $s \mathrel{\mathcal{R}^{\mathcal{H}}} r$.
\end{proof}

\begin{replemma}{Lemm_howe_subsitutivity}[Substitutivity]
Given a well-typed relation $\mathcal{R}$ on closed terms that is transitive, its Howe extension satisfies the following two value-substitutivity properties:
\begin{enumerate}
\item $\overrightarrow{x_i:A_i},y:B \vdash s \mathrel{\mathcal{R}^{\mathcal{H},\mathfrak{c}}} t \text{ and } \overrightarrow{x_i:A_i}\vdash v\mathrel{\mathcal{R}^{\mathcal{H},\mathfrak{v}}_B} w \implies \overrightarrow{x_i:A_i} \vdash s[v/y] \mathrel{\mathcal{R}^{\mathcal{H},\mathfrak{c}}} t[w/y]$.
\item $\overrightarrow{x_i:A_i},y:B \vdash u \mathrel{\mathcal{R}^{\mathcal{H},\mathfrak{v}}_C} u' \text{ and } \overrightarrow{x_i:A_i}\vdash v\mathrel{\mathcal{R}^{\mathcal{H},\mathfrak{v}}_B} w \implies \overrightarrow{x_i:A_i} \vdash u[v/y] \mathrel{\mathcal{R}^{\mathcal{H},\mathfrak{v}}_C} u'[w/y]$.
\end{enumerate}
\end{replemma}
\begin{proof}
We prove the two statements by induction on $s$ and $u$:

If $s$ is a computation then $\overrightarrow{x_i:A_i},y:B \vdash s \mathrel{\mathcal{R}^{\mathcal{H},\mathfrak{c}}} t$ was derived using rule \textsc{(HC)} so it must be the case that:
\begin{gather}
\overrightarrow{x_i:A_i},y:B \vdash s \mathrel{\widehat{\mathcal{R}^{\mathcal{H}}}^{\mathfrak{c}}} t'	\label{Howe_subst_proof1}\\
\overrightarrow{x_i:A_i},y:B \vdash t' \mathrel{\mathcal{R}^{\circ,\mathfrak{c}}} t.	\label{Howe_subst_proof2}
\end{gather}

By the definition of open extension we know from equation \ref{Howe_subst_proof2} that:
\begin{equation*}
\overrightarrow{x_i:A_i} \vdash t'[w/y] \mathrel{\mathcal{R}^{\circ,\mathfrak{c}}} t[w/y].
\end{equation*}
If we can prove
\begin{equation*}
\overrightarrow{x_i:A_i} \vdash s[v/y] \mathrel{\widehat{\mathcal{R}^{\mathcal{H}}}^{\mathfrak{c}}} t'[w/y]
\end{equation*}
then we could use rule \textsc{(HC)} to deduce the desired result, $\overrightarrow{x_i:A_i} \vdash s[v/x] \mathrel{\mathcal{R}^{\mathcal{H},\mathfrak{c}}} t[w/x]$.

\paragraph{Case $s=(\mufix{z}{u})(\protect\overrightarrow{w_i})$.} It must be the case that equation \ref{Howe_subst_proof1} was obtained using rule \textsc{(C7)} so we know that:
\begin{gather*}
t'=(\mufix{z}{u'})(\overrightarrow{w'_i})	\\
\overrightarrow{x_i:A_i},y:B, z:\neg(\overrightarrow{C}) \vdash u \mathrel{\mathcal{R}^{\mathcal{H},\mathfrak{v}}_{\neg(\overrightarrow{C})}} u'	\\
\overrightarrow{x_i:A_i},y:B \vdash w_i \mathrel{\mathcal{R}^{\mathcal{H},\mathfrak{v}}_{C_i}} w'_i \text{ for each }i.
\end{gather*}

By induction hypothesis for $u$ and $\overrightarrow{w_i}$ we can deduce:
\begin{gather*}
\overrightarrow{x_i:A_i}, z:\neg(\overrightarrow{C}) \vdash u[v/y] \mathrel{\mathcal{R}^{\mathcal{H},\mathfrak{v}}_{\neg(\overrightarrow{C})}} u'[w/y]	\\
\overrightarrow{x_i:A_i} \vdash w_i[v/y] \mathrel{\mathcal{R}^{\mathcal{H},\mathfrak{v}}_{C_i}} w'_i[w/y] \text{ for each }i
\end{gather*}
and then apply rule \textsc{(C7)} to get $\overrightarrow{x_i:A_i} \vdash s[v/y] \mathrel{\widehat{\mathcal{R}^{\mathcal{H}}}^{\mathfrak{c}}} t'[w/y]$.

\paragraph{Case $s=\downarrow$.} In this case equation~\ref{Howe_subst_proof1} was obtained from rule \textsc{(C9)}, so $t'=\downarrow$ and $s[v/y]=t'[w/y]=\downarrow$. We can then deduce $\overrightarrow{x_i:A_i} \vdash s[v/y] \mathrel{\widehat{\mathcal{R}^{\mathcal{H}}}^{\mathfrak{c}}} t'[w/y]$ by rule \textsc{(C9)}.

\paragraph{Cases $s=u\ (\protect\overrightarrow{w_i})$, $s=\sigma(u,x.s')$, $s=\caset{u}{s'}{z}{s''}$.} Analogous to the case $s=(\mufix{z}{u})(\protect\overrightarrow{w_i})$.

\paragraph{}If $u$ is a value then $\overrightarrow{x_i:A_i},y:B \vdash u \mathrel{\mathcal{R}^{\mathcal{H},\mathfrak{v}}_C} u'$ was derived using rule \textsc{(HV)} so it must be the case that:
\begin{gather}
\overrightarrow{x_i:A_i},y:B \vdash u \mathrel{\widehat{\mathcal{R}^{\mathcal{H}}}^{\mathfrak{v}}_C} u''	\label{Howe_subst_proof3}\\
\overrightarrow{x_i:A_i},y:B \vdash u'' \mathrel{\mathcal{R}^{\circ,\mathfrak{v}}_C} u'. \label{Howe_subst_proof4}
\end{gather}

From equation~\ref{Howe_subst_proof4} by the definition of open extension we have:
\begin{equation}\label{Howe_subst_proof5}
\overrightarrow{x_i:A_i} \vdash u''[w/y] \mathrel{\mathcal{R}^{\circ,\mathfrak{v}}_C} u'[w/y]
\end{equation}
so using rule \textsc{(HV)} it suffices to prove:
\begin{equation}\label{Howe_subst_proof_rtp1}
\overrightarrow{x_i:A_i} \vdash u[v/y] \mathrel{\widehat{\mathcal{R}^{\mathcal{H}}}^{\mathfrak{v}}_C} u''[w/y].
\end{equation}

\paragraph{Case $u=\lbd{\protect\overrightarrow{z_j}}{\protect\overrightarrow{C_j}}{r}$.} It must be the case that equation \ref{Howe_subst_proof3} was obtained by rule \textsc{(C3)} so:
\begin{gather*}
u'' = \lbd{\overrightarrow{z_j}}{\overrightarrow{C_j}}{r'}	\\
\overrightarrow{x_i:A_i},y:B, \overrightarrow{z_j:C_j} \vdash r \mathrel{\mathcal{R}^{\mathcal{H},\mathfrak{c}}} r'.
\end{gather*}
By induction hypothesis for $r$ and applying rule \textsc{(C3)} we can deduce:
\begin{equation*}
\overrightarrow{x_i:A_i} \vdash \lbd{\overrightarrow{z_j}}{\overrightarrow{C_j}}{r}[v/y] \mathrel{\widehat{\mathcal{R}^{\mathcal{H}}}^{\mathfrak{v}}_{\neg(\overrightarrow{C_j})}} \lbd{\overrightarrow{z_j}}{\overrightarrow{C_j}}{r'}[w/y].
\end{equation*}

\paragraph{Case $u=\mathtt{succ}(u')$.} Analogous to the previous case.

\paragraph{Cases $u=\mathtt{zero}$ and $u=\star$.} Analogous to the case $s=\downarrow$.

\paragraph{Case $u=z$.} \textbf{If $z\not=y$} then $z=x_j$ for some $j$. Then equation \ref{Howe_subst_proof3} must have been the conclusion of rule \textsc{(C1)} so $u''=z$. By rule \textsc{(C1)} we have:
\begin{equation*}
\overrightarrow{x_i:A_i} \vdash x_j \mathrel{\widehat{\mathcal{R}^{\mathcal{H}}}^{\mathfrak{v}}_C} x_j
\end{equation*}
which is what we had to prove.

\textbf{If $z=y$} then instead of going through equation \ref{Howe_subst_proof_rtp1} we will prove directly $\overrightarrow{x_i:A_i} \vdash u[v/y] \mathrel{\mathcal{R}^{\mathcal{H},\mathfrak{v}}_C} u'[w/y]$, that is:
\begin{equation*}
\overrightarrow{x_i:A_i} \vdash v \mathrel{\mathcal{R}^{\mathcal{H},\mathfrak{v}}_C} u'[w/y].
\end{equation*}

Because $\mathcal{R}$ is transitive, we can apply Lemma \ref{Lemm_howe_method_misc} to obtain: $\mathcal{R}^\circ\circ\mathcal{R}^{\mathcal{H}}\subseteq\mathcal{R}^{\mathcal{H}}$. We already know equation \ref{Howe_subst_proof5}:
\begin{equation*}
\overrightarrow{x_i:A_i} \vdash u''[w/y] \mathrel{\mathcal{R}^{\circ,\mathfrak{v}}_C} u'[w/y]
\end{equation*}
so we just need to show:
\begin{equation*}
\overrightarrow{x_i:A_i} \vdash v \mathrel{\mathcal{R}^{\mathcal{H},\mathfrak{v}}_C} u''[w/y].
\end{equation*}

Equation \ref{Howe_subst_proof3} must be the conclusion of rule \textsc{(C1)} so $u''=y$. Then we are left to prove:
\begin{equation*}
\overrightarrow{x_i:A_i} \vdash v \mathrel{\mathcal{R}^{\mathcal{H},\mathfrak{v}}_C} w
\end{equation*}
which we know by the initial assumption.
\end{proof}

\begin{replemma}{Lemm_howe_ext_sim_nat}
Consider a well-typed closed relation $\leq$ that is a $\mathfrak{P}$-simulation. For any closed values $v$ and $w$:
\begin{equation*}
\vdash v \leq^{\mathcal{H},\mathfrak{v}}_{\mathtt{nat}} w \implies v=w.
\end{equation*}	
\end{replemma}
\begin{proof}
Since $v$ is a closed value of type $\mathtt{nat}$ there exists $n\in\mathbb{N}$ such that $v=\overline{n}$. The equation $\vdash \overline{n} \leq^{\mathcal{H},\mathfrak{v}}_{\mathtt{nat}} w$ must have been derived by rule \textsc{(HV)} so:
\begin{equation*}
\vdash \overline{n} \mathrel{\widehat{\leq^{\mathcal{H}}}^{\mathfrak{v}}_{\mathtt{nat}}} u \quad \text{and} \quad \vdash u \leq^{\circ,\mathfrak{v}}_{\mathtt{nat}} w.
\end{equation*}
But $u$ and $w$ are closed so $u\leq^{\mathfrak{v}}_{\mathtt{nat}} w$. Since $\leq$ is a simulation it follows that $u=w$. So we know:
\begin{equation}
\vdash \overline{n} \mathrel{\widehat{\leq^{\mathcal{H}}}^{\mathfrak{v}}_{\mathtt{nat}}} w. \label{Howe_exten_sim_proof1}
\end{equation}

We prove by induction on $n$ that:
\begin{equation*}
\forall \vdash w:\mathtt{nat}.\ \vdash \overline{n} \leq^{\mathcal{H},\mathfrak{v}}_{\mathtt{nat}} w \implies \overline{n}=w.
\end{equation*}
\textbf{Base case, $n=0$.} From $\vdash \overline{n} \leq^{\mathcal{H},\mathfrak{v}}_{\mathtt{nat}} w$ we deduce equation \ref{Howe_exten_sim_proof1}, which must be the conclusion of rule \textsc{(C4)}. Therefore $\overline{n}=w=\mathtt{zero}$.

\textbf{Induction step.} Assume $\overline{n+1} \leq^{\mathcal{H},\mathfrak{v}}_{\mathtt{nat}} w'$ for some arbitrary $\vdash w':\mathtt{nat}$. Then equation \ref{Howe_exten_sim_proof1} becomes:
\begin{equation*}
\vdash \mathtt{succ}(\overline{n}) \mathrel{\widehat{\leq^{\mathcal{H}}}^{\mathfrak{v}}_{\mathtt{nat}}} w'.
\end{equation*}
This must be the conclusion of rule \textsc{(C5)} so $w'=\mathtt{succ}(w'')$ and $\overline{n}\leq^{\mathcal{H},\mathfrak{v}}_{\mathtt{nat}} w''$.

We can instantiate the induction hypothesis with the last equation to obtain $w''=\overline{n}$. This means that $\overline{n+1}=\mathtt{succ}(\overline{n})=\mathtt{succ}(w'')=w'$ as required.
\end{proof}

\begin{replemma}{Lemm_howe_key}[Key Lemma]
Consider a decomposable set of Scott-open observations $\mathfrak{P}$. Consider a well-typed closed relation $\leq$ that is a preorder and a $\mathfrak{P}$-simulation. For any closed computations $s$ and $t$, $\vdash s\leq^{\mathcal{H},\mathfrak{c}} t$ implies:
\begin{equation*}
\forall n\in\mathbb{N}.\ \forall P\in\mathfrak{P}.\ \treet{s}_n\in P \implies \treet{t}\in P.
\end{equation*}
\end{replemma}
\begin{proof}
We prove by induction on $n\in\mathbb{N}$ that:
\begin{equation*}
\forall n\in\mathbb{N}.\ \forall \text{ computations }s',t'.\ \vdash s'\leq^{\mathcal{H},\mathfrak{c}} t' \implies (\forall P\in\mathfrak{P}.\ \treet{s'}_n\in P \implies \treet{t'}\in P).
\end{equation*}

\paragraph{Base case, $n=0$.} By definition $\treet{s'}_0=\bot$. Assume that $\treet{s'}_0\in P$. Then by upwards closure of $P$ it follows that $\textit{Tree}_\Sigma \subseteq P$, so we also have $\treet{t'}\in P$, as required.

\paragraph{Induction step.} The induction hypothesis is:
\begin{multline*}
\forall k<n+1\in\mathbb{N}.\ \forall \text{ computations }s',t'.\\
 \vdash s'\leq^{\mathcal{H},\mathfrak{c}} t' \implies (\forall P\in\mathfrak{P}.\ \treet{s'}_k\in P \implies \treet{t'}\in P).
\end{multline*}
Assume that $\vdash s\leq^{\mathcal{H},\mathfrak{c}} t$. This must be the conclusion of rule \textsc{(HC)} so there exists $r$ such that:
\begin{equation*}
\vdash s \mathrel{\widehat{\leq^{\mathcal{H}}}^\mathfrak{c}} r \quad \text{and} \quad \vdash r \leq^{\circ,\mathfrak{c}} t.
\end{equation*}
But $r$ and $t$ are closed terms so we in fact know $r \leq^{\mathfrak{c}} t$. Because $\leq$ is a simulation it follows that:
\begin{equation*}
\forall P'\in\mathfrak{P}.\ \treet{r}\in P' \implies \treet{t}\in P'.
\end{equation*}
Therefore, it suffices to show:
\begin{equation*}
\forall P\in\mathfrak{P}.\ \treet{s}_{n+1}\in P \implies \treet{r}\in P.
\end{equation*}
To do this we proceed by a case split on the structure of $s$.

\paragraph{Case $s=v(w_1,\ldots,w_n)$.} Because $s$ is a well-typed closed computation, it must be the case that $s=(\lbd{\overrightarrow{x_i}}{\overrightarrow{A_i}}{s'})\ (w_1,\ldots,w_n)$.

The equation $\vdash s \mathrel{\widehat{\leq^{\mathcal{H}}}^\mathfrak{c}} r$ must have been obtained by rule \textsc{(C3)} so it must be the case that:
\begin{gather}
r = (\lbd{\overrightarrow{x_i}}{\overrightarrow{A_i}}{r'})\ (w'_1,\ldots,w'_n)	\\
\vdash \lbd{\overrightarrow{x_i}}{\overrightarrow{A_i}}{s'} \leq^{\mathcal{H},\mathfrak{v}}_{\neg(\overrightarrow{A_i})} \lbd{\overrightarrow{x_i}}{\overrightarrow{A_i}}{r'}	\label{Howe_key_proof1} \\ 
\vdash w_i \leq^{\mathcal{H},\mathfrak{v}}_{A_i} w'_i \text{ for each } i. \label{Howe_key_proof2}
\end{gather}

By definition of $\treet{-}_{(-)}$ we know $\treet{s}_{n+1}=\treet{(\lbd{\overrightarrow{x_i}}{\overrightarrow{A_i}}{s'})\ (w_1,\ldots,w_n)}_{n+1}=\treet{s'[\overrightarrow{w_i/x_i}]}_n$, and similarly for $r$. So we only need to prove:
\begin{equation*}
\forall P\in\mathfrak{P}.\ \treet{s'[\overrightarrow{w_i/x_i}]}_{n}\in P \implies \treet{r'[\overrightarrow{w'_i/x_i}]}\in P.
\end{equation*}

Equation \ref{Howe_key_proof1} must have been obatined by rule \textsc{(HV)} do there exists $p$ such that:
\begin{gather}
\vdash \lbd{\overrightarrow{x_i}}{\overrightarrow{A_i}}{s'} \mathrel{\widehat{\leq^{\mathcal{H}}}^{\mathfrak{v}}_{\neg(\overrightarrow{A_i})}} p \label{Howe_key_proof3} \\ 
\vdash p \leq^{\circ,\mathfrak{v}}_{\neg(\overrightarrow{A_i})} \lbd{\overrightarrow{x_i}}{\overrightarrow{A_i}}{r'}. \label{Howe_key_proof4}
\end{gather}

Equation \ref{Howe_key_proof3} must be the conclusion of rule \textsc{(C3)} so:
\begin{equation*}
p = \lbd{\overrightarrow{x_i}}{\overrightarrow{A_i}}{p'} \quad \text{and} \quad \overrightarrow{x_i:A_i}\vdash s' \leq^{\mathcal{H},\mathfrak{c}} p'.
\end{equation*}
Because $\leq$ is transitive we can apply Lemma \ref{Lemm_howe_subsitutivity} to obtain the substitutivity property for $\leq^{\mathcal{H}}$. Using this and equation \ref{Howe_key_proof2} we can deduce:
\begin{equation*}
\vdash s'[\overrightarrow{w_i/x_i}] \leq^{\mathcal{H},\mathfrak{c}} p'[\overrightarrow{w'_i/x_i}].
\end{equation*}
Apply the induction hypothesis for this to obtain:
\begin{equation*}
\forall P\in\mathfrak{P}.\ \treet{s'[\overrightarrow{w_i/x_i}]}_{n}\in P \implies \treet{p'[\overrightarrow{w'_i/x_i}]}\in P.
\end{equation*}

We know that $\vdash \lbd{\overrightarrow{x_i}}{\overrightarrow{A_i}}{p'} \leq^{\mathfrak{v}}_{\neg(\overrightarrow{A_i})} \lbd{\overrightarrow{x_i}}{\overrightarrow{A_i}}{r'}$ from equation \ref{Howe_key_proof4}. By the definition of simulation we then have:
\begin{equation*}
\forall P\in\mathfrak{P}.\ \treet{p'[\overrightarrow{w'_i/x_i}]}\in P \implies \treet{r'[\overrightarrow{w'_i/x_i}]}\in P.
\end{equation*}
From here we obtain the required result:
\begin{equation*}
\forall P\in\mathfrak{P}.\ \treet{s'[\overrightarrow{w_i/x_i}]}_n\in P \implies \treet{r'[\overrightarrow{w'_i/x_i}]}\in P.
\end{equation*}

\paragraph{Case $s=\sigma(v,x.s')$.} In this case $\vdash s \mathrel{\widehat{\leq^{\mathcal{H}}}^\mathfrak{c}} r$ is the conclusion of rule \textsc{(C8)} so we know:
\begin{gather*}
r=\sigma(v',x.r')	\\
\vdash v \leq^{\mathcal{H},\mathfrak{v}}_{\mathtt{nat}} v'	\\
x:\mathtt{nat} \vdash s' \leq^{\mathcal{H},\mathfrak{c}} r'.
\end{gather*}
Using $\vdash v \leq^{\mathcal{H},\mathfrak{v}}_{\mathtt{nat}} v'$ and Lemma \ref{Lemm_howe_ext_sim_nat} deduce $v=v'=\overline{k}$.

By reflexivity of $\leq^{\mathcal{H}}$ (Lemma \ref{Lemm_howe_extension_misc}) we have $\vdash \overline{m} \leq^{\mathcal{H},\mathfrak{v}}_{\mathtt{nat}} \overline{m}$. Using $x:\mathtt{nat} \vdash s' \leq^{\mathcal{H},\mathfrak{c}} r'$ we obtain by substitutivity (Lemma \ref{Lemm_howe_subsitutivity}) that:
\begin{equation*}
\forall m\in\mathbb{N}.\ \vdash s'[\overline{m}/x] \leq^{\mathcal{H},\mathfrak{c}} r'[\overline{m}/x].
\end{equation*}
Applying the induction hypothesis to this we we obtain:
\begin{equation}\label{Howe_key_ih}
\forall m\in\mathbb{N}.\ \forall P'\in\mathfrak{P}.\ \treet{s'[\overline{m}/x]}_n\in P' \implies \treet{r'[\overline{m}/x]}\in P'.
\end{equation}

By definition of $\treet{-}_{(-)}$ it suffices to prove:
\begin{equation*}
\forall P\in\mathfrak{P}.\ \sigma_k(\treet{s'[\overline{0}/x]}_n,\ldots,\treet{s'[\overline{k}/x]}_n,\ldots)\in P \implies \sigma_k(\treet{r'[\overline{0}/x]},\ldots,\treet{r'[\overline{k}/x]},\ldots)\in P.
\end{equation*}

Assume $\sigma_k(\treet{s'[\overline{0}/x]}_n,\ldots,\treet{s'[\overline{k}/x]}_n,\ldots)\in P$. Then by decomposability of $\mathfrak{P}$, Definition \ref{Def_decomposability}, we know there exist observations $\overrightarrow{P'_m}$ such that:
\begin{gather*}
\forall \overrightarrow{p'_m}\in\overrightarrow{P'_m}.\ \sigma_k(\overrightarrow{p'_m})\in P	\\
\text{for each } m\in\mathbb{N} \quad \treet{s'[\overline{m}/x]}_n \in P'_m.
\end{gather*}
By equation \ref{Howe_key_ih} we can deduce that for all $m\in\mathbb{N}$, $\treet{r'[\overline{m}/x]}\in P'_m$. So we have 
\begin{equation*}
\sigma_k(\treet{r'[\overline{0}/x]},\ldots,\treet{r'[\overline{k}/x]},\ldots)\in P
\end{equation*}
as required.

\paragraph{Case $s=(\mufix{x}{v})(\protect\overrightarrow{w_i})$.} In this case $s \mathrel{\widehat{\leq^{\mathcal{H}}}^{\mathfrak{c}}} r$ is the conclusion of rule \textsc{(C7)} so:
\begin{gather}
r=(\mufix{x}{v'})(\overrightarrow{w'_i})	\\
x:\neg(\overrightarrow{A_i}) \vdash v \leq^{\mathcal{H},\mathfrak{v}}_{\neg(\overrightarrow{A_i})} v'	\label{Howe_key_proof5}	\\
w_i \leq^{\mathcal{H},\mathfrak{v}}_{A_i} w'_i \text{ for each } i.\label{Howe_key_proof6}
\end{gather}

By definition of $\treet{-}_{(-)}$ it suffices to prove:
\begin{equation*}
\forall P\in\mathfrak{P}.\ \treet{v[\lbd{\overrightarrow{y}}{\overrightarrow{A_i}}{(\mufix{x}{v})(\overrightarrow{y})}/x]\ (\overrightarrow{w_i})}_n \in P \implies \treet{v'[\lbd{\overrightarrow{y}}{\overrightarrow{A_i}}{(\mufix{x}{v'})(\overrightarrow{y})}/x]\ (\overrightarrow{w'_i})}\in P.
\end{equation*}

By equation \ref{Howe_key_proof5} we can deduce using context weakening that:
\begin{equation*}
\overrightarrow{y:A_i}, x:\neg(\overrightarrow{A_i}) \vdash v \leq^{\mathcal{H},\mathfrak{v}}_{\neg(\overrightarrow{A_i})} v'.
\end{equation*}
Because $\leq$ is reflexive, we can apply Lemma \ref{Lemm_howe_extension_misc} to deduce $\leq^{\mathcal{H}}$ is reflexive. Therefore:
\begin{equation*}
\overrightarrow{y:A_i} \vdash y_i \leq^{\mathcal{H},\mathfrak{v}}_{A_i} y_i.
\end{equation*}
From the last two equations and from rule \textsc{(C7)} we can deduce:
\begin{equation*}
\overrightarrow{y:A_i} \vdash (\mufix{x}{v})(\overrightarrow{y}) \mathrel{\widehat{\leq^{\mathcal{H}}}^{\mathfrak{c}}} (\mufix{x}{v'})(\overrightarrow{y}).
\end{equation*}
By reflexivity of $\leq$ we know:
\begin{equation*}
\overrightarrow{y:A_i} \vdash (\mufix{x}{v'})(\overrightarrow{y}) \leq^{\circ,\mathfrak{c}} (\mufix{x}{v'})(\overrightarrow{y}).
\end{equation*}
Using the last two equations and rule \textsc{(HC)} we obtain:
\begin{equation*}
\overrightarrow{y:A_i} \vdash (\mufix{x}{v})(\overrightarrow{y}) \leq^{\mathcal{H},\mathfrak{c}} (\mufix{x}{v'})(\overrightarrow{y}).
\end{equation*}
From this, using rule \textsc{(C3)}, we have:
\begin{equation*}
\vdash \lbd{\overrightarrow{y}}{\overrightarrow{A_i}}{(\mufix{x}{v})(\overrightarrow{y})} \mathrel{\widehat{\leq^{\mathcal{H}}}^\mathfrak{v}_{\neg(\overrightarrow{A_i})}} \lbd{\overrightarrow{y}}{\overrightarrow{A_i}}{(\mufix{x}{v'})(\overrightarrow{y})}.
\end{equation*}
By reflexivity of $\leq$ we have:
\begin{equation*}
\vdash \lbd{\overrightarrow{y}}{\overrightarrow{A_i}}{(\mufix{x}{v'})(\overrightarrow{y})} \leq^{\circ,\mathfrak{v}}_{\neg(\overrightarrow{A_i})} \lbd{\overrightarrow{y}}{\overrightarrow{A_i}}{(\mufix{x}{v'})(\overrightarrow{y})}.
\end{equation*}
From the last two equations, by rule \textsc{(HV)}, we obtain:
\begin{equation*}
\vdash \lbd{\overrightarrow{y}}{\overrightarrow{A_i}}{(\mufix{x}{v})(\overrightarrow{y})} \leq^{\mathcal{H},\mathfrak{v}}_{\neg(\overrightarrow{A_i})} \lbd{\overrightarrow{y}}{\overrightarrow{A_i}}{(\mufix{x}{v'})(\overrightarrow{y})}.
\end{equation*}

Using this last equation and equation \ref{Howe_key_proof5}: $x:\neg(\overrightarrow{A_i}) \vdash v \leq^{\mathcal{H},\mathfrak{v}}_{\neg(\overrightarrow{A_i})} v'$, we obtain by substitutivity for $\leq^{\mathcal{H}}$, Lemma \ref{Lemm_howe_subsitutivity}, that:
\begin{equation*}
\vdash v[(\lbd{\overrightarrow{y}}{\overrightarrow{A_i}}{(\mufix{x}{v})(\overrightarrow{y})})/x] \leq^{\mathcal{H},\mathfrak{v}}_{\neg(\overrightarrow{A_i})} v'[(\lbd{\overrightarrow{y}}{\overrightarrow{A_i}}{(\mufix{x}{v'})(\overrightarrow{y})})/x].
\end{equation*}
From Lemma \ref{Lemm_howe_extension_misc} we can deduce $\leq^{\mathcal{H}}$ is compatible. Using the last equation and equation \ref{Howe_key_proof6}, $\vdash w_i \leq^{\mathcal{H},\mathfrak{v}}_{A_i} w'_i$ for each $i$, we obtain by compatibility:
\begin{equation*}
\vdash v[(\lbd{\overrightarrow{y}}{\overrightarrow{A_i}}{(\mufix{x}{v})(\overrightarrow{y})})/x]\ (\overrightarrow{w_i}) \leq^{\mathcal{H},\mathfrak{c}} v'[(\lbd{\overrightarrow{y}}{\overrightarrow{A_i}}{(\mufix{x}{v'})(\overrightarrow{y})})/x]\ (\overrightarrow{w'_i}).
\end{equation*}

By applying the induction hypothesis to this we obtain the desired result:
\begin{equation*}
\forall P\in\mathfrak{P}.\ \treet{v[\lbd{\overrightarrow{y}}{\overrightarrow{A_i}}{(\mufix{x}{v})(\overrightarrow{y})}/x]\ (\overrightarrow{w_i})}_n \in P \implies \treet{v'[\lbd{\overrightarrow{y}}{\overrightarrow{A_i}}{(\mufix{x}{v'})(\overrightarrow{y})}/x]\ (\overrightarrow{w'_i})}\in P.
\end{equation*} 

\paragraph{Case $s=\caset{v}{s'}{y}{s''}$.} In this case $\vdash s \mathrel{\widehat{\leq^{\mathcal{H}}}^\mathfrak{c}} r$ is the conclusion of rule \textsc{(C10)} so we know:
\begin{gather*}
r=\caset{v'}{r'}{y}{r''}	\\
\vdash v \leq^{\mathcal{H},\mathfrak{v}}_{\mathtt{nat}} v'	\\
\vdash s' \leq^{\mathcal{H},\mathfrak{c}} r'	\\
y:\mathtt{nat} \vdash s'' \leq^{\mathcal{H},\mathfrak{c}} r''.
\end{gather*}
Using $\vdash v \leq^{\mathcal{H},\mathfrak{v}}_{\mathtt{nat}} v'$ we can apply Lemma \ref{Lemm_howe_ext_sim_nat} to deduce $v=v'$.

\textbf{If $v=v'=\mathtt{zero}$} then it suffices to prove:
\begin{equation*}
\forall P\in\mathfrak{P}.\ \treet{s'}_n\in P \implies \treet{r'}\in P.
\end{equation*}
We can deduce this using the induction hypothesis for $s'$ and $r'$ and $\vdash s' \leq^{\mathcal{H},\mathfrak{c}} r'$.

\textbf{If $v=v'=\mathtt{succ}(w)$ where $\vdash w:\mathtt{nat}$} then it suffices to prove: 
\begin{equation*}
\forall P\in\mathfrak{P}.\ \treet{s''[w/y]}_n\in P \implies \treet{r''[w/y]}\in P.
\end{equation*}
By reflexivity of $\leq^{\mathcal{H}}$ (Lemma \ref{Lemm_howe_extension_misc}) we have $\vdash w \leq^{\mathcal{H},\mathfrak{v}}_{\mathtt{nat}} w$. We can use this, $y:\mathtt{nat} \vdash s'' \leq^{\mathcal{H},\mathfrak{c}} r''$ and substitutivity for $\leq^{\mathcal{H}}$ (Lemma \ref{Lemm_howe_subsitutivity}), to deduce:
\begin{equation*}
\vdash s''[w/y] \leq^{\mathcal{H},\mathfrak{c}} r''[w/y].
\end{equation*}
We then get the desired result from the induction hypothesis for $s''[w/y]$ and $r''[w/y]$.

\paragraph{Case $s=\downarrow$.} In this case $\vdash s \mathrel{\widehat{\leq^{\mathcal{H}}}^\mathfrak{c}} r$ is the conclusion of rule \textsc{(C9)} so $r=\downarrow$. Then $\treet{s}_{n+1}=\treet{r}=\downarrow$. Therefore we have the required result:
\begin{equation*}
\forall P\in\mathfrak{P}.\ \treet{s}_{n+1}\in P \implies \treet{r}\in P.
\end{equation*}
\end{proof}

\begin{replemma}{Lemm_howe_subst_impl_open_incl}
Given a well-typed open relation $\mathcal{R}$ that is reflexive and has the two substitutivity properties from Lemma \ref{Lemm_howe_subsitutivity}, and a well-typed closed relation $\mathcal{S}$ then:
\begin{center}
if $\mathcal{R}$ restricted to closed terms is included in $\mathcal{S}$ then $\mathcal{R}\subseteq\mathcal{S}^\circ$.
\end{center}
\end{replemma}
\begin{proof}
Consider terms $s$ and $t$, values or computations, such that $\overrightarrow{x_i:A_i}\vdash s \mathrel{\mathcal{R}} t$. By reflexivity of $R$ we know that for any values $\vdash\overrightarrow{v_i:A_i}$ we have $\vdash v_i\mathrel{\mathcal{R}^\mathfrak{v}} v_i$. Therefore we can apply the substitutivity property of $\mathcal{R}$ to obtain:
\begin{equation*}
\forall \vdash\overrightarrow{v_i:A_i}.\ \vdash s[\overrightarrow{v_i/x_i}]\ \mathcal{R}\ t[\overrightarrow{v_i/x_i}].
\end{equation*}

From here we can deduce by assumption that:
\begin{equation*}
\forall \vdash\overrightarrow{v_i:A_i}.\ \vdash s[\overrightarrow{v_i/x_i}]\ \mathcal{S}\ t[\overrightarrow{v_i/x_i}]
\end{equation*}
so by the definition of open extension we have:
\begin{equation*}
\overrightarrow{x_i:A_i}\vdash s\ \mathcal{S}^\circ\ t.
\end{equation*}
\end{proof}

\begin{replemma}{Lemm_howe_reflstran_sim}
Given a $\mathfrak{P}$-simulation $\mathcal{R}$, its reflexive-transitive closure, $\mathcal{R}^*$ is also a $\mathfrak{P}$-simulation.
\end{replemma}
\begin{proof}
We check all the conditions in the definition of $\mathfrak{P}$-simulation in turn:
\begin{enumerate}
\item $\vdash v \mathrel{\mathcal{R}^{*,\mathfrak{v}}_{\mathtt{unit}}} w \implies v=w=\star$.
Assume $\vdash v \mathrel{\mathcal{R}^{*,\mathfrak{v}}_{\mathtt{unit}}} w$. The only closed value of type $\mathtt{unit}$ is $\star$ so $v=w=\star$.

\item $\vdash v \mathrel{\mathcal{R}^{*,\mathfrak{v}}_{\mathtt{nat}}} w \implies v=w$. Assume $v \mathrel{\mathcal{R}^{*,\mathfrak{v}}_{\mathtt{nat}}} w$. Then by the definition of reflexive-transitive closure there must exist a chain of values $v_1,\ldots,v_k$ such that $v_1=v$ and $v_k=w$ and  $v_{i-1}\mathrel{\mathcal{R}^\mathfrak{v}_{\mathtt{nat}}} v_i$ for each $i>1$. If $k=1$ then $v=w$. If $k>1$, we can use the fact that $\mathcal{R}$ is a simulation to deduce $v_{i-1}=v_i$ for each $i>1$. So by transitivity $v=w$.

\item $\vdash s \mathrel{\mathcal{R}^{*,\mathfrak{c}}} t \implies \forall P\in\mathfrak{P}.\ (\treet{s}\in P \implies \treet{t}\in P)$. Assume $s \mathrel{\mathcal{R}^{*,\mathfrak{c}}} t$. There exists a chain of computations $s_1,\ldots,s_k$ such that $s_1=s$ and $s_k=t$ and  $s_{i-1}\mathrel{\mathcal{R}^\mathfrak{c}} s_i$ for each $i>1$. If $k=1$ then $s=t$ so we have $\forall P\in\mathfrak{P}.\ (\treet{s}\in P \implies \treet{t}\in P)$ as required. If $k>1$ then for each $i>1$:
\begin{equation*}
\forall P\in\mathfrak{P}.\ (\treet{s_{i-1}}\in P \implies \treet{s_i}\in P).
\end{equation*}
From here we can deduce the desired result by transitivity of implication.

\item $\vdash v \mathrel{\mathcal{R}^{*,\mathfrak{v}}_{\neg(A_1,\ldots,A_n)}} w \implies \forall \vdash u_1:A_1,\ldots,\vdash u_n:A_n.\ v(u_1,\ldots,u_n) \mathrel{\mathcal{R}^{*,\mathfrak{c}}} w(u_1,\ldots, u_n)$. Similar to the previous two cases. It uses the fact that $v\mathrel{\mathcal{R}}u\mathrel{\mathcal{R}}w$ implies $v\mathrel{\mathcal{R}^*}w$.
\end{enumerate}
\end{proof}

\begin{replemma}{Lemm_howe_refltran_compat}
Given a well-typed compatible relation $\mathcal{R}$, its reflexive-transitive closure $\mathcal{R}^*$ is also compatible.
\end{replemma}
\begin{proof}
The proof is similar to the proof of Lemma \ref{Lemm_howe_reflstran_sim} in that it expands $s\mathrel{\mathcal{R}^*}t$ into a chain of terms related by $\mathcal{R}$. It also uses the fact that $\mathcal{R}$ compatible implies $\mathcal{R}$ is reflexive, Lemma \ref{Lemm_howe_extension_misc}.
\end{proof}

\begin{replemma}{Lemm_howe_lassen}[From \cite{LasPhd}]
Given a well-typed closed relation $\mathcal{R}$ the following holds:
\begin{center}
if $\mathcal{R}^\circ$ is reflexive and symmetric, then $\mathcal{R}^{\mathcal{H}*}$ is symmetric.
\end{center}
Where $S^*$ denotes the reflexive-transitive closure of a relation $\mathcal{S}$.
\end{replemma}
\begin{proof}
By examining the compatible refinement rules we can observe that for any relation $\mathcal{S}$:
\begin{equation}\label{lassen_prf1}
\widehat{\mathcal{S}^{\textit{op}}} = \widehat{\mathcal{S}}^{\textit{op}}.
\end{equation}

Since $\mathcal{R}^\circ$ is reflexive, $\mathcal{R}$ is also reflexive. Therefore we can apply Lemma \ref{Lemm_howe_extension_misc} to deduce:
\begin{equation}\label{lassen_prf2}
\mathcal{R}^\circ \subseteq \mathcal{R}^{\mathcal{H}}
\end{equation}
and $\mathcal{R}^{\mathcal{H}}$ compatible.

Because $\mathcal{R}^{\mathcal{H}}$ is compatible, $\mathcal{R}^{\mathcal{H}*}$ is also compatible using Lemma \ref{Lemm_howe_refltran_compat}. By definition of compatibility and compatible refinement we see that $\mathcal{R}^{\mathcal{H}*}$ compatible implies:
\begin{equation}\label{lassen_prf3}
\widehat{\mathcal{R}^{\mathcal{H}*}} \subseteq \mathcal{R}^{\mathcal{H}*}.
\end{equation}

Using the fact that $\mathcal{R}^\circ$ is symmetric, equations \ref{lassen_prf1}, \ref{lassen_prf3} and \ref{lassen_prf2}, and the fact that taking the reflexive-transitive closure and the converse of a relation are commutative operations we obtain:
\begin{multline*}
\mathcal{R}^\circ\circ\widehat{\mathcal{R}^{\mathcal{H}*\textit{op}}} = \mathcal{R}^{\circ\textit{op}}\circ\widehat{\mathcal{R}^{\mathcal{H}*\textit{op}}} = \mathcal{R}^{\circ\textit{op}}\circ\widehat{\mathcal{R}^{\mathcal{H}*}}^\textit{op} \subseteq	\\
\mathcal{R}^{\circ\textit{op}}\circ\mathcal{R}^{\mathcal{H}*\textit{op}} \subseteq \mathcal{R}^{\mathcal{H}\textit{op}}\circ\mathcal{R}^{\mathcal{H}*\textit{op}} = \mathcal{R}^{\mathcal{H}\textit{op}}\circ\mathcal{R}^{\mathcal{H}\textit{op}*} = \mathcal{R}^{\mathcal{H}\textit{op}*} = \mathcal{R}^{\mathcal{H}*\textit{op}}.
\end{multline*}

This means that $\mathcal{R}^{\mathcal{H}*\textit{op}}$ is a solution $\mathcal{S}$ to the inequation $\mathcal{R}^\circ\circ\widehat{\mathcal{S}}\subseteq\mathcal{S}$. So the relation $\mathcal{R}^{\mathcal{H}*\textit{op}}$ is closed under the rules \textsc{(HC)} and \textsc{(HV)}. But $\mathcal{R}^{\mathcal{H}}$ is the least relation closed under those rules. Therefore:
\begin{equation*}
\mathcal{R}^{\mathcal{H}} \subseteq \mathcal{R}^{\mathcal{H}*\textit{op}}.
\end{equation*}

Consider some terms $s$ and $t$, values or computations, such that $s \mathrel{\mathcal{R}^{\mathcal{H}*}} t$. Then there exists a sequence of terms $s_1,\ldots,s_n$ such that $s_1=s$ and $s_n=t$ and $s_1\mathrel{\mathcal{R}^{\mathcal{H}}}s_2\mathrel{\mathcal{R}^{\mathcal{H}}}\ldots\mathrel{\mathcal{R}^{\mathcal{H}}}s_n$. Therefore $s_1\mathrel{\mathcal{R}^{\mathcal{H}*\textit{op}}}s_2\mathrel{\mathcal{R}^{\mathcal{H}*\textit{op}}}\ldots\mathrel{\mathcal{R}^{\mathcal{H}*\textit{op}}}s_n$.

From here we can deduce $s\mathrel{\mathcal{R}^{\mathcal{H}*\textit{op}}}t$, which means $t\mathrel{\mathcal{R}^{\mathcal{H}*}}s$. So $\mathcal{R}^{\mathcal{H}*}$ is symmetric as required.
\end{proof}

\chapter{Proofs about Logical Equivalence}\label{App_modal_log}

\begin{repproposition}{Prop_sim_coinc_log}
Given a decomposable set $\mathfrak{P}$ of Scott-open observations:
\begin{enumerate}
\item Applicative $\mathfrak{P}$-similarity, $\precsim$, coincides with the logical preorder induced by the logic $\mathcal{V}^+$, $\sqsubseteq_{\mathcal{V}^+}$. Therefore, the open extension of $\sqsubseteq_{\mathcal{V}^+}$ is compatible.

\item Applicative $\mathfrak{P}$-bisimilarity, $\sim$, coincides with the logical equivalence induced by the logic $\mathcal{V}$, $\equiv_{\mathcal{V}}$. Therefore, the open extension of $\equiv_{\mathcal{V}}$ is compatible.
\end{enumerate}
\end{repproposition}
\begin{proof}
\begin{enumerate}
\item\label{log_coin_sim_prf} Consider two arbitrary closed computations $s$ and $t$. Since similarity is the greatest simulation:
\begin{align*}
s \precsim^\mathfrak{c} t \quad &\text{iff} \quad \forall P\in\mathfrak{P}.\ \treet{s}\in P \implies \treet{t}\in P	\\
s \sqsubseteq_{\mathcal{V}^+} t \quad &\text{iff} \quad \forall P\in\mathfrak{P}.\ \treet{s}\in P \implies \treet{t}\in P.
\end{align*}
So we can see that $s \precsim^\mathfrak{c} t \Longleftrightarrow s \sqsubseteq_{\mathcal{V}^+} t$ as required.

For values $v$ and $u$ assume first that $v\precsim^\mathfrak{v}_A u$. We need to prove $v\sqsubseteq_{\mathcal{V}^+}u$, that is:
\begin{equation*}
\forall \phi:A\in\mathcal{V}^+.\ (v\models_{\mathcal{V}^+}\phi \implies u\models_{\mathcal{V}^+}\phi).
\end{equation*}
We proceed by a case distinction on the type $A$.

\paragraph{Case $A=\mathtt{unit}$.} The only formulas of this type are $true$, the empty conjunction, and $false$, the empty disjunction. All values satisfy $true$ so in this case we are done. No values satisfy $false$ so the implication above holds trivially.

\paragraph{Case $A=\mathtt{nat}$.} From the definition of simulation we know that $v=u$. We continue by induction on the formula $\phi$.

\textbf{If $\phi=\{m\}$}, assume $v\models_{\mathcal{V}^+}\{m\}$. By definition of satisfaction this means $v=\overline{m}$ so $u=\overline{m}$. Therefore $u\models_{\mathcal{V}^+}\{m\}$ as required.

\textbf{If $\phi=\land_{i\in I}\phi_i$ or $\phi=\lor_{i\in I}\phi_i$}, the result follows from the induction hypothesis for $\phi_i$.

\paragraph{Case $A=\neg(B_1,\ldots,B_n)$.} From $v\precsim^\mathfrak{v}_A u$ we know that:
\begin{equation*}
\forall \vdash w_1:A_1,\ldots,\vdash w_n:A_n.\ \forall P\in\mathfrak{P}.\ \treet{v\ (\overrightarrow{w_i})}\in P \implies \treet{u\ (\overrightarrow{w_i})}\in P.
\end{equation*}
We proceed by induction on $\phi$.

\textbf{If $\phi=(w_1,\ldots,w_n)\mapsto P$}, assume $v\models_{\mathcal{V}^+} (w_1,\ldots,w_n)\mapsto P$, which means $\treet{v\ (\overrightarrow{w_i})}\in P$. Therefore, by assumption $\treet{u\ (\overrightarrow{w_i})}\in P$ so $u\models_{\mathcal{V}^+} (w_1,\ldots,w_n)\mapsto P$ as required.

\textbf{If $\phi=\land_{i\in I}\phi_i$ or $\phi=\lor_{i\in I}\phi_i$}, the result follows from the induction hypothesis.

Now assume that for values $v$ and $u$, $v \sqsubseteq_{\mathcal{V}^+} u$. To show $v\precsim^\mathfrak{v}_A u$ we proceed by a case distinction on $A$.

\paragraph{Case $A=\mathtt{unit}$.} The only closed value of type $\mathtt{unit}$ is $\star$ so $v=u=\star$. Since $\precsim$ is the greatest simulation, this is enough to establish $v\precsim^\mathfrak{v}_{\mathtt{unit}} u$.

\paragraph{Case $A=\mathtt{nat}$.} Since $v$ is closed $v=\overline{n}$ for some $n\in\mathbb{N}$. Therefore $v\models_{\mathcal{V}^+}\{n\}$. Since $v \sqsubseteq_{\mathcal{V}^+} u$ we have that $u\models_{\mathcal{V}^+}\{n\}$ so $u=\overline{n}=v$ as required.

\paragraph{Case $A=\neg(B_1,\ldots,B_n)$.} We need to prove that:
\begin{equation*}
\forall \vdash w_1:A_1,\ldots,\vdash w_n:A_n.\ \forall P\in\mathfrak{P}.\ \treet{v\ (\overrightarrow{w_i})}\in P \implies \treet{u\ (\overrightarrow{w_i})}\in P.
\end{equation*}
Assume $\treet{v\ (\overrightarrow{w_i})}\in P$. Then $v\models_{\mathcal{V}^+} (w_1,\ldots,w_n)\mapsto P$ so $w\models_{\mathcal{V}^+} (w_1,\ldots,w_n)\mapsto P$. Therefore $\treet{w\ (\overrightarrow{w_i})}\in P$ as required.

So we have proved $v\precsim^\mathfrak{v}_A u \Longleftrightarrow v \sqsubseteq_{\mathcal{V}^+} u$.

\item For computations we know that $s\sim^\mathfrak{c} t$ if and only if:
\begin{equation*}
\forall P\in\mathfrak{P}.\ \treet{s}\in P \Longleftrightarrow \treet{t}\in P.
\end{equation*}
So we can see this is equivalent to $s\equiv_{\mathcal{V}}t$.

For values, assume $v\equiv_{\mathcal{V}}u$, that is:
\begin{equation}\label{log_coinv_bisim_prf1}
\forall \phi:A\in\mathcal{V}.\ (v\models_{\mathcal{V}}\phi \Longleftrightarrow u\models_{\mathcal{V}}\phi).
\end{equation}
We need to prove $v\sim^\mathfrak{v}_A u$. We proceed by a case split on the type $A$. The proof is the same as in point \ref{log_coin_sim_prf}, except that in the case $A=\neg(B_1,\ldots,B_n)$ we need to prove an equivalence. This is done by using the fact that equation \ref{log_coinv_bisim_prf1} is now an equivalence.

Now assume $v\sim^\mathfrak{v}_A u$. We need to prove $v\equiv_{\mathcal{V}}u$, that is:
\begin{equation*}
\forall \phi:A\in\mathcal{V}.\ (v\models_{\mathcal{V}}\phi \Longleftrightarrow u\models_{\mathcal{V}}\phi).
\end{equation*}
As in point \ref{log_coin_sim_prf}, we proceed by a case distinction on the type $A$.

In case $A=\mathtt{unit}$ we have new formulas apart from $true$ and $false$ due to the addition of negation. However, they are all semantically equivalent to either $true$ or $false$, so we are done.

In cases $A=\mathtt{nat}$ and $A=\neg(B_1,\ldots,B_n)$ proceed by induction on $\phi:A$ as before. The $\Longrightarrow$ direction of the proof is the same. The $\Longleftarrow$ direction can be obtained from the previous one by choosing $\phi=\neg\varphi$. The inductions have an additional case:

\paragraph{If $\phi=\neg\phi'$.} By the induction hypothesis we know:
\begin{equation*}
v\models_{\mathcal{V}}\phi' \Longleftrightarrow u\models_{\mathcal{V}}\phi'
\end{equation*}
which is equivalent to
\begin{equation*}
v\not\models_{\mathcal{V}}\phi' \Longleftrightarrow u\not\models_{\mathcal{V}}\phi'
\end{equation*}
which is in turn equivalent to
\begin{equation*}
v\models_{\mathcal{V}}\neg\phi' \Longleftrightarrow u\models_{\mathcal{V}}\neg\phi'.
\end{equation*}
This is what we had to prove.
\end{enumerate}
\end{proof}

\begin{reptheorem}{Thm_logic_equiv_equiv}
Given a decomposable set $\mathfrak{P}$ of Scott-open observations, 
the logics $\mathcal{F}$ and $\mathcal{V}$ are equi-expressive.
\end{reptheorem}
\begin{proof}
We need to prove the same statements as in Theorem \ref{Thm_logic_preord_equiv}, where $\mathcal{V}^+$ is replaced by $\mathcal{V}$ and $\mathcal{F}^+$ is replaced by $\mathcal{F}$. We will point out where the proofs need to be modified.

\paragraph{Statement \ref{rtp_equi_1}.} The proof for computation formulas $P\in\mathfrak{P}$ remains the same because these formulas do not change when adding negation.

For values, we prove by induction on the derivation of $\phi:A$ the following property:
\begin{equation*}
\Phi(\phi, A) = (\phi:A \implies (\forall \vdash v:A.\ v \models_{\mathcal{F}} \phi \Longleftrightarrow v \models_{\mathcal{V}} \phi^\flat)).
\end{equation*}
All the cases from the proof of Theorem \ref{Thm_logic_preord_equiv} stay the same, but now we have an additional case \textsc{(neg)}.

In this case, $\phi=\neg\phi'$. Assume $\phi:A$. Then we know $\phi':A$ so we can apply the induction hypothesis to get:
\begin{equation*}
\forall \vdash v:A.\ v \models_{\mathcal{F}} \phi' \Longleftrightarrow v \models_{\mathcal{V}} \phi'^\flat.
\end{equation*}
This is equivalent to:
\begin{equation*}
\forall \vdash v:A.\ v \models_{\mathcal{F}} \neg\phi' \Longleftrightarrow v \models_{\mathcal{V}} \neg\phi'^\flat
\end{equation*}
which is what we had to prove.

\paragraph{Statement \ref{rtp_equi_2}.} For computation formulas the equivalence is proved the same as in Theorem \ref{Thm_logic_preord_equiv}.

For value formulas we proceed by induction on the type $A$, as in the proof of Theorem \ref{Thm_logic_preord_equiv}. The case $A=\mathtt{nat}$ stays the same. In the $A=\mathtt{unit}$ case, we now have more formulas than $true$ and $false$ because of the addition of negation. However, all these new formulas are semantically equivalent to $true$ and $false$, so their semantics does not change when translated. Therefore, the equivalence we need to prove is true. In the case $A=\neg(B_1,\ldots,B_n)$ we do an induction on the formula $\phi$.

We observe that $(\sqsubseteq_{\mathcal{F}})=(\equiv_{\mathcal{F}})$ because $\mathcal{F}$ contains negation. The proof of this goes as follows:
\begin{equation*}
\forall\phi\in\mathcal{F}.\ v\models_\mathcal{F}\phi \implies u\models_\mathcal{F}\phi
\end{equation*}
implies that
\begin{equation*}
\forall\phi'\in\mathcal{F}.\ v\models_\mathcal{F}\neg\phi' \implies u\models_\mathcal{F}\neg\phi'
\end{equation*}
so
\begin{equation*}
\forall\phi'\in\mathcal{F}.\ u\models_\mathcal{F}\phi' \implies v\models_\mathcal{F}\phi'.
\end{equation*}

Using $(\sqsubseteq_{\mathcal{F}})=(\equiv_{\mathcal{F}})$, we see that $\sqsubseteq_{\mathcal{F}}$ is compatible. Therefore, the proof of the case $\phi=(w_1,\ldots,w_n)\mapsto P$ is the same as in Theorem \ref{Thm_logic_preord_equiv}. The cases $\phi=\lor_{i\in I}\varphi_i$ and $\phi=\land_{i\in I}\varphi_i$ remain unchanged.

There is one new case, namely $\phi=\neg\phi'$. By the induction hypothesis for $\phi'$ we know that:
\begin{equation*}
\forall\vdash v:\neg(B_1,\ldots,B_n).\ v\models_\mathcal{V}\phi'\Longleftrightarrow v\models_\mathcal{F}\phi'^\sharp.
\end{equation*}
This is equivalent to:
\begin{equation*}
\forall\vdash v:\neg(B_1,\ldots,B_n).\ v\models_\mathcal{V}\neg\phi'\Longleftrightarrow v\models_\mathcal{F}\neg\phi'^\sharp
\end{equation*}
which is what we had to prove.
\end{proof}


\chapter{Proofs about Contextual Equivalence}\label{App_ctxt}

\begin{repproposition}{Prop_ctxpre_comp}
The contextual preorder $\ctxpre$ is a preorder, and is moreover compatible and $\mathfrak{P}$-adequate. Thus, it is the greatest compatible and $\mathfrak{P}$-adequate preorder.
\end{repproposition}
\begin{proof}
First prove that $\ctxpre=\bigcup\mathbb{CA}$ is a preorder. To prove reflexivity, we show that the open identity relation, $\mathcal{I}$, is in $\mathbb{CA}$. From the compatibility rules we can see $\mathcal{I}$ is compatible. Given $\vdash s\mathrel{\mathcal{I}}s$ it follows that ($\forall P\in\mathfrak{P}.\ \treet{s}\in P \implies \treet{s}\in P$) so $\mathcal{I}$ is adequate. Hence, $\mathcal{I}\in\mathbb{CA}$, as required.

To show transitivity it suffices to show that the composition of relations in $\mathbb{CA}$ is itself in $\mathbb{CA}$. Consider two relations $\mathcal{R}$ and $\mathcal{S}$ which are compatible and adequate.

We can show that $\mathcal{S}\circ\mathcal{R}$ is adequate. Consider $\vdash s\mathrel{\mathcal{R}^\mathfrak{c}} t\mathrel{\mathcal{S}^\mathfrak{c}} r$. Since $\mathcal{R}$ and $\mathcal{S}$ are adequate we know that:
\begin{gather*}
\forall P\in\mathfrak{P}.\ \treet{s}\in P \implies \treet{t}\in P	\\
\forall P\in\mathfrak{P}.\ \treet{t}\in P \implies \treet{r}\in P.
\end{gather*}
Therefore:
\begin{equation*}
\forall P\in\mathfrak{P}.\ \treet{s}\in P \implies \treet{r}\in P
\end{equation*}
which means $\mathcal{S}\circ\mathcal{R}$ is adequate.

To prove $\mathcal{S}\circ\mathcal{R}$ is compatible we check each of the rules in the definition of compatibility (Definition~\ref{Def_compat}) in turn. Since $\mathcal{R}$ and $\mathcal{S}$ are compatible we know $\Gamma\vdash x\mathrel{\mathcal{R}^\mathfrak{v}_A} x$ and $\Gamma\vdash x\mathrel{\mathcal{S}^\mathfrak{v}_A} x$ so $\Gamma\vdash x\mathrel{(\mathcal{S}\circ\mathcal{R})^\mathfrak{v}_A} x$. Therefore, $\mathcal{S}\circ\mathcal{R}$ satisfies \textsc{(comp1)}. Rules \textsc{(comp2)}, \textsc{(comp4)} and \textsc{(comp9)} are proved similarly.

Consider $\Gamma,\overrightarrow{x_i:A_i}\vdash s\ \mathcal{R}^\mathfrak{c}\ s'\ \mathcal{S}^\mathfrak{c}\ s''$. Then by compatibility of $\mathcal{R}$ and $\mathcal{S}$ we know that:
\begin{equation*}
\Gamma\vdash \lbd{(\overrightarrow{x_i})}{(\overrightarrow{A_i})}{s}\ \mathcal{R}^\mathfrak{v}_{\neg(\overrightarrow{A_i})}\ \lbd{(\overrightarrow{x_i})}{(\overrightarrow{A_i})}{s'}\  \mathcal{S}^\mathfrak{v}_{\neg(\overrightarrow{A_i})}\ \lbd{(\overrightarrow{x_i})}{(\overrightarrow{A_i})}{s''}.
\end{equation*}
So 
\begin{equation*}
\Gamma\vdash \lbd{(\overrightarrow{x_i})}{(\overrightarrow{A_i})}{s}\ (\mathcal{S}\circ\mathcal{R})^\mathfrak{v}_{\neg(\overrightarrow{A_i})}\ \lbd{(\overrightarrow{x_i})}{(\overrightarrow{A_i})}{s''}.
\end{equation*}
Therefore $\mathcal{S}\circ\mathcal{R}$ satisfies \textsc{(comp3)}. Rule \textsc{(comp5)} is proved similarly.

Consider $\Gamma,x:\neg(\overrightarrow{A_i}) \vdash v\  \mathcal{R}^\mathfrak{v}_{\neg(\overrightarrow{A_i})}\ v'\ \mathcal{S}^\mathfrak{v}_{\neg(\overrightarrow{A_i})}\ v''$ and $\Gamma\vdash w_i\ \mathcal{R}^\mathfrak{v}_{A_i}\ w'_i\ \mathcal{S}^\mathfrak{v}_{A_i}\ w''_i$ for each $i$. By compatibility of $\mathcal{R}$ and $\mathcal{S}$ we infer that:
\begin{gather*}
\Gamma \vdash (\mufix{x}{v})(\overrightarrow{w_i})\ \mathcal{R}^\mathfrak{c}\ (\mufix{x}{v})(\overrightarrow{w'_i})\ \mathcal{S}^\mathfrak{c}\ (\mufix{x}{v})(\overrightarrow{w''_i})
\end{gather*}
so
\begin{equation*}
\Gamma \vdash (\mufix{x}{v})(\overrightarrow{w_i})\   (\mathcal{S}\circ\mathcal{R})^\mathfrak{c}\  (\mufix{x}{v})(\overrightarrow{w''_i}).
\end{equation*} 
Therefore, $\mathcal{S}\circ\mathcal{R}$ satisfies \textsc{(comp7)}. Proving rules \textsc{(comp6)}, \textsc{(comp8)} and \textsc{(comp10)} are satisfied is similar.

Now we prove that $\ctxpre$ is compatible. We have already shown $\ctxpre$ is reflexive so rules \textsc{(comp1)}, \textsc{(comp2)}, \textsc{(comp4)} and \textsc{(comp9)} are satisfied.

Suppose $\Gamma,\overrightarrow{x_i:A_i}\vdash s\mathrel{(\ctxpre)^\mathfrak{c}} t$. Then there exists a relation $\mathcal{R}\in\mathbb{CA}$ such that $\Gamma,\overrightarrow{x_i:A_i}\vdash s\ \mathcal{R}^\mathfrak{c}\ t$. By compatibility of $\mathcal{R}$ we know that:
\begin{equation*}
\Gamma \vdash \lbd{\overrightarrow{x_i}}{\overrightarrow{A_i}}{s}\ \mathcal{R}^\mathfrak{v}_{\neg(\overrightarrow{A_i})}\ \lbd{\overrightarrow{x_i}}{\overrightarrow{A_i}}{t}
\end{equation*}
so since $\mathcal{R}\subseteq(\ctxpre)$ we have:
\begin{equation*}
\Gamma \vdash \lbd{\overrightarrow{x_i}}{\overrightarrow{A_i}}{s}\ (\ctxpre)^\mathfrak{v}_{\neg(\overrightarrow{A_i})}\ \lbd{\overrightarrow{x_i}}{\overrightarrow{A_i}}{t}.
\end{equation*}
Therefore, $\ctxpre$ satisfies rule \textsc{(comp3)}. Similarly, we can prove $\ctxpre$ satisfies \textsc{(comp5)}.

Since we have proved $\ctxpre$ is a preorder, we can apply Lemma \ref{Lem_comp_1premise_rules} to deduce that the compatibility clauses \textsc{(comp6)}, \textsc{(comp7)}, \textsc{(comp8)} and \textsc{(comp10)} are equivalent to their single-premise versions. We then use the same reasoning as before to show that these single-premise rules are satisfied.

Now show that $\ctxpre$ is adequate. Consider $\vdash s \mathrel{(\ctxpre)^\mathfrak{c}} t$. Then there exists a relation $\mathcal{R}\in\mathbb{CA}$ such that $\vdash s\ \mathcal{R}^\mathfrak{c}\ t$. So by adequacy of $\mathcal{R}$ we have:
\begin{equation*}
\forall P\in\mathfrak{P}.\ \treet{s}\in P \implies \treet{t}\in P
\end{equation*}
as required.

Therefore, we have shown that $\ctxpre$ is a compatible adequate preorder so we are done.
\end{proof}

\begin{repproposition}{Prop_ctxeq_is_ctxpreop}
Contextual equivalence is the intersection of the contextual preorder with its converse:
\begin{equation*}
(\ctxeq) = (\ctxpre)\cap(\ctxpre)^{\textit{op}}.
\end{equation*}
\end{repproposition}
\begin{proof}
We prove each inclusion in turn.

\paragraph{``$\subseteq$''.} Consider two terms, values or computations, $s\mathrel{\ctxeq} t$. By definition of $\ctxeq$ there exists a compatible and biadequate relation $\mathcal{R}$ such that $s\mathrel{\mathcal{R}}t$.
Since $\mathcal{R}$ is biadequate, it is also adequate, so $\mathcal{R}\in\mathbb{CA}$. Thus, $\mathcal{R}\subseteq(\ctxpre)$ and $s\mathrel{\ctxpre} t$.

Notice that $(\ctxpre)^{\textit{op}}=\bigcup\{\mathcal{S}^{\textit{op}} \mid \mathcal{S}\in\mathbb{CA}\}$. If a relation $\mathcal{S}$ is compatible then we see by the definition of compatibility that $\mathcal{S}^{\textit{op}}$ is also compatible. So $(\ctxpre)^{\textit{op}}$ is the union of all relations $\mathcal{S}$ that are compatible and have the property:
\begin{equation*}
\forall s',t'.\ \vdash s'\mathrel{\mathcal{R}^\mathfrak{c}} t' \implies \forall P\in\mathfrak{P}.\ \treet{t'}\in P \implies \treet{s'}\in P.
\end{equation*}
Because $\mathcal{R}$ is biadequate it has the above property. So $\mathcal{R}\subseteq(\ctxpre)^\textit{op}$. Therefore, $s\mathrel{(\ctxpre)^\textit{op}}t$. We have shown that $(s,t)\in(\ctxpre)\cap(\ctxpre)^{\textit{op}}$ as required.

\paragraph{``$\supseteq$''.} From Proposition \ref{Prop_ctxpre_comp} we know $\ctxpre$ is compatible so $(\ctxpre)^{\mathit{op}}$ is also compatible. From the definition of compatibility we can see that the intersection of two compatible relations is also compatible so $(\ctxpre)\cap(\ctxpre)^{\textit{op}}$ is compatible.

Let $s$ and $t$ be two closed computations such that $(s,t)\in(\ctxpre)\cap(\ctxpre)^{\textit{op}}$. Then by adequacy of $\ctxpre$ (Propostion \ref{Prop_ctxpre_comp}) we have:
\begin{gather*}
\forall P\in\mathfrak{P}.\ \treet{s}\in P \implies \treet{t}\in P	\\
\forall P\in\mathfrak{P}.\ \treet{t}\in P \implies \treet{s}\in P
\end{gather*}
which means that $(\ctxpre)\cap(\ctxpre)^{\textit{op}}$ is biadequate.

Therefore $(\ctxpre)\cap(\ctxpre)^{\textit{op}}\in\mathbb{CAS}$ so $(\ctxpre)\cap(\ctxpre)^{\textit{op}}\subseteq(\ctxeq)$.
\end{proof}

\begin{reptheorem}{Thm_bisim_is_ctx_equiv}
Consider a decomposable set of Scott-open observations $\mathfrak{P}$ that is consistent. Then:
\begin{enumerate}
\item The open extension of applicative $\mathfrak{P}$-similarity, $\precsim^\circ$, coincides with the contextual preorder, $\ctxpre$.
\item The open extension of applicative $\mathfrak{P}$-bisimilarity, $\sim^\circ$, coincides with contextual equivalence, $\ctxeq$.
\end{enumerate}
\end{reptheorem}
\begin{proof}
\textbf{We first show $(\precsim^\circ)=(\ctxpre)$}. We have shown in Section \ref{Sec_ctx_is_bisim} that $\precsim^\circ$ is included in $\ctxpre$.

Now we need to show $(\ctxpre)\subseteq(\precsim^\circ)$. We first show that $\ctxpre$ restricted to closed terms is included in $\precsim$, then extend this to open terms. To do this, we show $\ctxpre$ restricted to closed terms is a simulation by checking it satisfies the four conditions in the definition of simulation.

\begin{enumerate}
\item Assume $\vdash v\mathrel{(\ctxpre)^\mathfrak{v}_{\mathtt{unit}}} u$. The only closed value of type $\mathtt{unit}$ is $\star$ so $v=u=\star$ as required.

\item Assume $\vdash v\mathrel{(\ctxpre)^\mathfrak{v}_{\mathtt{nat}}} u$.  This is shown in Section \ref{Sec_ctx_is_bisim}.

\item Assume $\vdash s\mathrel{(\ctxpre)^\mathfrak{c}} t$. Because $\ctxpre$ is adequate (Proposition \ref{Prop_ctxpre_comp}) we have the desired result:
\begin{equation*}
\forall P\in\mathfrak{P}.\ \treet{s}\in P \implies \treet{t}\in P.
\end{equation*}

\item Assume $\vdash v\mathrel{(\ctxpre)^\mathfrak{v}_{\neg(A_1,\ldots,A_n)}} u$. Consider arbitrary values $\vdash w_1:A_1,\ldots,\vdash w_n:A_n$. Because $\ctxpre$ is a preorder it is reflexive so $w_i\mathrel{(\ctxpre)^\mathfrak{v}_{A_i}}w_i$ for each $i$. We know $\ctxpre$ is compatible so we obtain:
\begin{equation*}
v (w_1,\ldots,w_n)\ (\ctxpre)^\mathfrak{c}\ u (w_1,\ldots,w_n)
\end{equation*}
as required.
\end{enumerate}

So we have established $(\ctxpre)\subseteq(\precsim)$ for closed terms. Now we need to prove $(\ctxpre)\subseteq(\precsim^\circ)$ in general. To do this, we consider computations and each type of value separately. The cases for natural numbers and computations are shown in Section \ref{Sec_ctx_is_bisim} so we only show the cases for type $\mathtt{unit}$ and function values.

\paragraph{If $\protect\overrightarrow{x_j:A_j} \vdash v \mathrel{(\ctxpre)^\mathfrak{v}_{\mathtt{unit}}} w$} then each of $v$ and $w$ can either be $\star$ or a variable $x_i$. In any case, because $\star$ is the only closed value of type $\mathtt{unit}$, we know that:
\begin{equation*}
\forall (\vdash \overrightarrow{u_i:A_i}).\ v[\overrightarrow{u_i/x_i}]=w[\overrightarrow{u_i/x_i}]=\star.
\end{equation*}
This means that:
\begin{equation*}
\forall (\vdash \overrightarrow{u_i:A_i}).\ \vdash v[\overrightarrow{u_i/x_i}] \precsim^\mathfrak{v}_{\mathtt{unit}} w[\overrightarrow{u_i/x_i}]
\end{equation*}
so by the definition of open extension we have that $\overrightarrow{x_j:A_j} \vdash v \precsim^{\circ,\mathfrak{v}}_{\mathtt{unit}} w$ as required.

\paragraph{If $\protect\overrightarrow{x_j:A_j} \vdash v \mathrel{(\ctxpre)^\mathfrak{v}_{\neg(B_1,\ldots,B_n)}} w$} then by compatibility and reflexivity of $\ctxpre$ we know that
\begin{equation*}
\forall (\vdash \overrightarrow{u_i:B_i}).\ \overrightarrow{x_j:A_j} \vdash v(\overrightarrow{u_i})\ (\ctxpre)^\mathfrak{c}\ w(\overrightarrow{u_i})
\end{equation*}
and again by compatibility:
\begin{equation*}
\forall (\vdash \overrightarrow{u_i:B_i}).\ \vdash \lbd{(\overrightarrow{x_j})}{(\overrightarrow{A_j})}{v(\overrightarrow{u_i})}\ (\ctxpre)^\mathfrak{v}_{\neg(\overrightarrow{A_j})}\ \lbd{(\overrightarrow{x_j})}{(\overrightarrow{A_j})}{w(\overrightarrow{u_i})}.
\end{equation*}
Using $(\ctxpre)\subseteq(\precsim)$ we can deduce:
\begin{equation*}
\forall (\vdash \overrightarrow{u_i:B_i}).\ \vdash \lbd{(\overrightarrow{x_j})}{(\overrightarrow{A_j})}{v(\overrightarrow{u_i})}\ \precsim^\mathfrak{v}_{\neg(\overrightarrow{A_j})}\ \lbd{(\overrightarrow{x_j})}{(\overrightarrow{A_j})}{w(\overrightarrow{u_i})}
\end{equation*}
so by the definition of $\precsim$ we know that
\begin{equation*}
\forall (\vdash\overrightarrow{p_j:A_j}).\ \forall (\vdash \overrightarrow{u_i:B_i}).\ \vdash (\lbd{(\overrightarrow{x_j})}{(\overrightarrow{A_j})}{v(\overrightarrow{u_i})})(\overrightarrow{p_j})\ \precsim^\mathfrak{c}\ (\lbd{(\overrightarrow{x_j})}{(\overrightarrow{A_j})}{w(\overrightarrow{u_i})})(\overrightarrow{p_j}).
\end{equation*}
From Lemma \ref{Lem_red_pres_sim} we know reduction preserves similarity so:
\begin{equation*}
\forall (\vdash\overrightarrow{p_j:A_j}).\ \forall (\vdash \overrightarrow{u_i:B_i}).\ \vdash v[\overrightarrow{p_j/x_j}]\ (\overrightarrow{u_i})\ \precsim^\mathfrak{c}\ w[\overrightarrow{p_j/x_j}]\ (\overrightarrow{u_i}).
\end{equation*}
By the definition of similarity for function values this means that:
\begin{equation*}
\forall (\vdash\overrightarrow{p_j:A_j}).\ \vdash v[\overrightarrow{p_j/x_j}]\ \precsim^\mathfrak{v}_{\neg(\overrightarrow{B_i})}\ w[\overrightarrow{p_j/x_j}]
\end{equation*}
so by the definition of open extension
\begin{equation*}
\overrightarrow{x_j:A_j}\vdash v\ \precsim^{\circ,\mathfrak{v}}_{\neg(\overrightarrow{B_i})}\ w
\end{equation*}
as required.

\paragraph{Now show that $(\sim^\circ)=(\ctxeq)$.} We have shown that $(\precsim^\circ)=(\ctxpre)$ so:
\begin{equation*}
(\precsim^{\textit{op}})^\circ = (\precsim^\circ)^{\textit{op}} = (\ctxpre)^{\textit{op}}
\end{equation*}
because taking the converse of a relation and its open extension are commutative operations. From Proposition \ref{Prop_ctxeq_is_ctxpreop} we know that:
\begin{equation*}
(\ctxeq)=(\ctxpre)\cap(\ctxpre)^{\textit{op}} = (\precsim^\circ)\cap(\precsim^{\textit{op}})^\circ.
\end{equation*}

We can show that $(\precsim^\circ)\cap(\precsim^{\textit{op}})^\circ = ((\precsim)\cap(\precsim^\textit{op}))^\circ$. The equation:
\begin{equation*}
\overrightarrow{x_i:A_i}\vdash s\precsim^\circ t \quad\text{and}\quad \overrightarrow{x_i:A_i}\vdash s(\precsim^\textit{op})^\circ t
\end{equation*}
is equivalent to
\begin{equation*}
\forall (\vdash\overrightarrow{u_i:A_i}).\ \vdash s[\overrightarrow{u_i/x_i}] \precsim t[\overrightarrow{u_i/x_i}] \text{ and } \vdash s[\overrightarrow{u_i/x_i}] \precsim^\textit{op} t[\overrightarrow{u_i/x_i}]
\end{equation*}
which in turn is equivalent to $(s,t)\in ((\precsim)\cap(\precsim^\textit{op}))^\circ$, as required.

Therefore, we know that:
\begin{equation*}
(\ctxeq) = ((\precsim)\cap(\precsim^\textit{op}))^\circ
\end{equation*}
and by Proposition \ref{Prop_bisim_is_sim_and_simop} we know that:
\begin{equation*}
(\ctxeq) = ((\precsim)\cap(\precsim^\textit{op}))^\circ = (\sim)^\circ
\end{equation*}
which is what we had to prove.
\end{proof}

\begin{replemma}{Lem_ctx_wctx_compat}
Contextual preorder defined with contexts, $\ctxprec$, is a compatible and adequate preorder. Hence, it is included in contextual preorder defined coinductively, $\ctxpre$.
\end{replemma}
\begin{proof}
First prove $\ctxprec$ is a preorder. Given terms $\Gamma\vdash s=t$ we know that for any context $C$, $\treet{C[s]}=\treet{C[t]}$. Therefore $\Gamma\vdash s\ \ctxprec\ t$, so $\ctxprec$ is reflexive. Because implication is transitive we can see that $\ctxprec$ is also transitive.

To prove $\ctxprec$ is adequate consider closed computations $\emptyset\vdash s\mathrel{(\ctxprec)^\mathfrak{c}}t$. Then by definition of $\ctxprec$, choosing $C^\mathfrak{c}_\mathfrak{c}=[-]^\mathfrak{c}$, we know that:
\begin{equation*}
\forall P\in\mathfrak{P}.\ \treet{s}\in P \implies \treet{t}\in P
\end{equation*}
as required.

To prove $\ctxprec$ is compatible show that is satisfies each compatibiliy rule. Rules \textsc{(comp1)}, \textsc{(comp2)}, \textsc{(comp4)} and \textsc{(comp9)} are satisfied by reflexivity.

For rule \textsc{(comp3)} assume $\Gamma,\overrightarrow{x_i:A_i}\vdash s\ (\ctxprec)^\mathfrak{c}\ t$. Consider an arbitrary context ${C'}^\mathfrak{v}_\mathfrak{c}:(\Gamma\vdash \neg(\overrightarrow{A_i}))\Rightarrow(\emptyset\vdash)$. Now consider the context:
\begin{equation*}
C^\mathfrak{c}_\mathfrak{c}= {C'}^\mathfrak{v}_\mathfrak{c}[\lbd{\overrightarrow{x_i}}{\overrightarrow{A_i}}{[-]^\mathfrak{c}}] : (\Gamma,\overrightarrow{x_i:A_i}\vdash)\Rightarrow(\emptyset\vdash).
\end{equation*}
Instantiate the assumption that $s$ and $t$ are in the contextual preorder with $C^\mathfrak{c}_\mathfrak{c}$ to deduce that:
\begin{equation*}
\forall P\in\mathfrak{P}.\ \treet{{C'}^\mathfrak{v}_\mathfrak{c}[\lbd{\overrightarrow{x_i}}{\overrightarrow{A_i}}{[s]^\mathfrak{c}}]}\in P \implies \treet{{C'}^\mathfrak{v}_\mathfrak{c}[\lbd{\overrightarrow{x_i}}{\overrightarrow{A_i}}{[t]^\mathfrak{c}}]}\in P.
\end{equation*}
Then we know $\Gamma\vdash \lbd{\overrightarrow{x_i}}{\overrightarrow{A_i}}{s}\ (\ctxprec)^\mathfrak{v}_{\neg(\overrightarrow{A_i})}\ \lbd{\overrightarrow{x_i}}{\overrightarrow{A_i}}{t}$ which is what we had to prove. Rule \textsc{(comp5)} can be proved similarly choosing:
\begin{equation*}
C^\mathfrak{v}_\mathfrak{c}={C'}^\mathfrak{v}_\mathfrak{c}[\mathtt{succ([-]^\mathfrak{v})}].
\end{equation*} 

Using the fact that $\ctxprec$ is a preorder we can apply Lemma \ref{Lem_comp_1premise_rules} to replace the four compatibility rules that we still need to prove with their single-premise versions. Proving these single-premise rules hold is analogous to proving \textsc{(comp3)}. As an example, we prove rule \textsc{(comp7l)}.

Assume $\Gamma,x:\neg(\overrightarrow{A_i})\vdash v\ (\ctxprec)^\mathfrak{v}_{\neg(\overrightarrow{A_i})}\ v'$ and $\Gamma\vdash\overrightarrow{w_i:A_i}$. Consider an arbitrary context ${C'}^\mathfrak{c}_\mathfrak{c}:(\Gamma\vdash)\Rightarrow(\emptyset\vdash)$ and the context:
\begin{equation*}
C^\mathfrak{v}_\mathfrak{c}={C'}^\mathfrak{c}_\mathfrak{c}[(\mufix{x}{[-]^\mathfrak{v}})(\overrightarrow{w_i})].
\end{equation*}
Using the context typing rules \textsc{(vv-id)}, \textsc{(vc-mul)} and Lemma \ref{Lem_ctx_inst_comp} we can deduce that:
\begin{equation*}
C^\mathfrak{v}_\mathfrak{c} : (\Gamma,x:\neg(\overrightarrow{A_i})\vdash)\Rightarrow(\emptyset\vdash).
\end{equation*}
Therefore we can apply the assumption about $v$ and $v'$ to get:
\begin{equation*}
\forall P\in\mathfrak{P}.\ \treet{{C'}^\mathfrak{c}_\mathfrak{c}[(\mufix{x}{[v]^\mathfrak{v}})(\overrightarrow{w_i})]}\in P \implies \treet{{C'}^\mathfrak{c}_\mathfrak{c}[(\mufix{x}{[v']^\mathfrak{v}})(\overrightarrow{w_i})]}\in P
\end{equation*}
which means $\Gamma\vdash (\mufix{x}{v})(\overrightarrow{w_i})\ (\ctxprec)^\mathfrak{c}\ (\mufix{x}{v'})(\overrightarrow{w_i})$ as required.

We know that $\ctxpre$ is the greatest compatible and adequate relation (Lemma \ref{Prop_ctxpre_comp}) and we have shown $\ctxprec$ is compatible and adequate. Therefore $(\ctxprec)\subseteq(\ctxpre)$.
\end{proof}

\begin{replemma}{Lem_ctx_pres_ctx_coin}
Contextual preorder defined coinductively, $\ctxpre$, is closed under program contexts, that is:
\begin{enumerate}
\item If $\Gamma'\vdash v\ (\ctxpre)^\mathfrak{v}_A\ u$ and $C^\mathfrak{v}_\mathfrak{v}:(\Gamma'\vdash A)\Rightarrow(\Gamma\vdash B)$ then $\Gamma\vdash C^\mathfrak{v}_\mathfrak{v}[v]\ (\ctxpre)^\mathfrak{v}_B\ C^\mathfrak{v}_\mathfrak{v}[u]$.

And the analogous statement for $C^\mathfrak{v}_\mathfrak{c}$.

\item If $\Gamma'\vdash s\ (\ctxpre)^\mathfrak{c}\ t$ and $C^\mathfrak{c}_\mathfrak{c}:(\Gamma'\vdash)\Rightarrow(\Gamma\vdash)$ then $\Gamma\vdash C^\mathfrak{c}_\mathfrak{c}[s]\ (\ctxpre)^\mathfrak{c}\ C^\mathfrak{c}_\mathfrak{c}[t]$.

And the analogous statement for $C^\mathfrak{c}_\mathfrak{v}$.
\end{enumerate}
\end{replemma}
\begin{proof}
By induction on the typing derivation of $C$. The two base cases \textsc{(vv-id)} and \textsc{(cc-id)} follow from the assumptions $\Gamma'\vdash v\ (\ctxpre)^\mathfrak{v}_A\ u$ and $\Gamma'\vdash s\ (\ctxpre)^\mathfrak{c}\ t$ respectively.

In the case \textsc{(vv-lbd)}:
\begin{gather*}
C^\mathfrak{v}_\mathfrak{v}=\lbd{\overrightarrow{x_i}}{\overrightarrow{A_i}}{{C'}^\mathfrak{v}_\mathfrak{c}}:(\Gamma'\vdash A)\Rightarrow(\Gamma\vdash \neg(\overrightarrow{A_i}))	\\
\text{where } {C'}^\mathfrak{v}_\mathfrak{c}:(\Gamma'\vdash A)\Rightarrow(\Gamma,x:\neg(\overrightarrow{A_i})\vdash).
\end{gather*} 
Assume $\Gamma'\vdash v\ (\ctxpre)^\mathfrak{v}_A\ u$. Then by induction hypothesis for ${C'}^\mathfrak{v}_\mathfrak{c}$ we know that:
\begin{equation*}
\Gamma,x:\neg(\overrightarrow{A_i})\vdash {C'}^\mathfrak{v}_\mathfrak{c}[v]\ (\ctxpre)^\mathfrak{c}\ {C'}^\mathfrak{v}_\mathfrak{c}[u]
\end{equation*}
so by compatibility of $\ctxpre$, rule \textsc{(comp3)} we can deduce that:
\begin{equation*}
\Gamma\vdash \lbd{\overrightarrow{x_i}}{\overrightarrow{A_i}}{{C'}^\mathfrak{v}_\mathfrak{c}[v]}\ (\ctxpre)^\mathfrak{c}_{\neg(\overrightarrow{A_i})}\ \lbd{\overrightarrow{x_i}}{\overrightarrow{A_i}}{{C'}^\mathfrak{v}_\mathfrak{c}[u]}
\end{equation*}
which is what we had to prove. Cases \textsc{(cv-lbd)} and \textsc{(vv-nat)} are analogous.

For the remaining cases, we use the fact that $\ctxpre$ is a preorder, so the single-premise compatibility rules from Lemma \ref{Lem_comp_1premise_rules} hold. The proof then proceeds similarly to the proof of \textsc{(vv-lbd)}: apply the induction hypothesis then use one of the single-premise compatibility rules.
\end{proof}

\end{document}
